\documentclass[11pt]{article}
\usepackage[numbers,sort&compress]{natbib}
\usepackage{float}
\usepackage{zzt}
\usepackage{microtype}
\usepackage{mathrsfs}
\usepackage{xurl}
\usepackage{booktabs,array}
\usepackage{braket}
\usetikzlibrary{arrows.meta,positioning,calc}
\usepackage{xcolor,colortbl,hhline}

\numberwithin{theorem}{section}
\numberwithin{definition}{section}
\makeatletter
\let\c@lemma\c@theorem
\let\c@corollary\c@theorem
\let\c@proposition\c@theorem
\makeatother

\numberwithin{equation}{section}

\allowdisplaybreaks[1]
\title{Entrywise Logarithmic Matrix Algebra and Dichotomy of \\Planar Graph Homomorphisms}
\author{
Jin-Yi Cai\thanks{Department of Computer Sciences,
University of Wisconsin-Madison.}\\
\texttt{jyc@cs.wisc.edu}
\and Zhuxiao Tang\footnotemark[1]\\
\texttt{zztang@wisc.edu}
}
\date{}

\begin{document}
\pagenumbering{roman}
\pagestyle{empty}
\maketitle
\thispagestyle{empty}
\begin{abstract}
We prove a complexity classification of counting planar graph homomorphisms with non-negative weights.
For a real symmetric matrix $M$ with nonnegative entries, the problem $\GH(M)$ is either (1) P-time computable over all graphs, or (2) \#P-hard in general but  P-time computable over planar graphs, or (3) \#P-hard over planar graphs. 
%%%then maybe we also used \GH for the next sentence?
Furthermore, $\GH(M)$ in (2) consists of precisely  
those that involve the P-time FKT algorithm to count planar perfect matchings with a holographic transformation.

The dichotomy for entrywise positive and full rank matrices lies at the core of the proof. In this setting, 
the main tool is
%achieved by forming 
a matrix algebra  $\mathcal{A}$ (a vector space with bilinear multiplication) formed by taking the real span of (centered) entrywise logarithms of all realizable matrices from $M$ using planar edge gadgets and polynomial interpolation.
We then define a multiplicative homomorphism $F(\cdot)$ on the algebra $\mathcal{A}$, which turns out to be a polynomial map.
The exponents of
%from the unique 
prime factorization of component polynomials
of $F$ induce graph distances for distance-regular graphs.
These graphs are either the Hamming graphs or Doob graphs
using a classification result for distance-regular graphs (if
$\PlGH(M)$ is not \#P-hard),
after obtaining their intersection numbers.
%By computing their intersection numbers and a classification result for distance-regular graphs, we show these graphs are either Hamming graphs or Doob graphs.
%In either case, we can establish the complexity dichotomy.
The dichotomy is proved in either case.
We then extend the  complexity classification to all nonnegative matrices.
%%% sounds ok?
\end{abstract}

\clearpage
\setcounter{tocdepth}{2}
\tableofcontents
\clearpage
\pagenumbering{arabic}
\pagestyle{plain}

\section{Introduction}\label{sec:introduction}

Given graphs $G$ and $H$,
a mapping from $V(G)$ to $V(H)$ is called
a \emph{homomorphism}
if the edges of $G$ are mapped to the edges of $H$.
More generally,
let $M = (M_{ij})$ be an $n \times n$ symmetric matrix
with entries $M_{ij} \in \mathbb{R}_{\geq 0}$.
We interpret $M$ as defining a weighted graph $H = H_M$ on $[n]$,
where $M_{ij}$ is the weight of edge $(i,j)$.
Given $M$, the associated \emph{partition function}
$Z_{M}(G)$ for an input undirected multigraph $G = (V,E)$ is
$$Z_{M}(G) = \sum_{\sigma: V \rightarrow [n]} \prod_{(u, v) \in E} M_{\sigma(u) \sigma(v)}.$$

Isomorphic graphs $G \cong G'$ have the same value
$Z_{M}(G) = Z_{M}(G')$, thus every $M$
defines a graph property $Z_{M}(\cdot)$.
For a $0$--$1$ matrix $M$,
$Z_{M}(G)$ counts the number of homomorphisms from $G$ to $H$.
Graph homomorphism ($\GH$) encompasses a great deal of
graph properties and counting problems
arising in combinatorics and statistical physics
~\cite{lovasz1967operations, lovasz2012large, hell2004graphs}.

Each $M$ defines a
computational problem denoted by $\GH(M)$: given an input graph $G$,
output $Z_{M}(G)$.
The  complexity of 
$\GH(M)$ has been a major focus of research.
A number of increasingly general complexity dichotomy theorems have been achieved~\cite{dyer2000complexity,bulatov2005complexity,goldberg2010complexity,cai2013graph}.
A central feature of these results is the use of
\emph{edge gadgets}:
%small graphs
%%% JYC: doesn't have to be small
graph fragments that replace edges of the input instance
and enable polynomial-time reductions.
%Graph homomorphism 
$\GH$ can also be viewed as
a special case of counting $\CSP$ (\#$\CSP$), where
%a series of results established 
a complexity dichotomy is proved for
any set of constraint functions ${\cal F}$
~\cite{bulatov2013complexity,dyer2010complexity,dyer2011decidability,dyer2013complexity,bulatov2012complexity,cai2016nonnegative,cai2017complexity}.

Parallel to this development,
Valiant~\cite{valiant2008holographic}
introduced \emph{holographic algorithms}.
It is well known that counting the number of perfect
matchings (\#PM) is \#P-complete~\cite{valiant1979complexity}. 
On the other hand,
%it is known 
%since the 60's 
%that
the
famous FKT algorithm~\cite{kasteleyn1961statistics,temperley1961dimer,kasteleyn1963dimer,kasteleyn1967graph} from the 1960's
 can
compute \#PM on {\it planar graphs} in P-time.
%Valiant's h
Holographic algorithms
greatly extended its reach.
%in fact so much so that a 
Let $\PlGH(M)$ denote the problem $\GH(M)$ when the
input graphs $G$ are restricted to planar graphs.
Holographic algorithms using FKT  reveal non-trivial tractable cases for $\PlGH$ problems. 
For example, for  $M=\left(\begin{smallmatrix}
    a & b\\
    b & a
\end{smallmatrix}\right)$, the problem $\PlGH(M)$ is  tractable in P-time  for any $a,b\in \mathbb{R}$.
Indeed, this problem is equivalent to computing the partition function of the Ising model on planar graphs \cite{cai2017complexitybook}.
This tractability also extends to direct sums and tensor products
with other tractable cases.
A most intriguing question arises:
Is holographic transformation + FKT algorithm a \emph{universal} algorithm that 
\emph{every} \#P-hard counting problem expressible as a partition function (a sum-of-products)
that \emph{can be solved} in P-time on planar graphs
%(but  in general)
\emph{is solved} by this method alone?

After a series of work~\cite{cai2009holant,cai2016complete, backens2017new, backens2018complete, yang2022local, fu2019blockwise, fu2014holographic, cai2019holographic}
it was established that for every set of complex valued constraint functions ${\cal F}$
on the Boolean domain (i.e., domain size $n = 2$) there is a 3-way exact
classification for  
%counting 
%$\CSP$(${\cal F}$):
\#$\CSP$(${\cal F}$): 
 (1) P-time solvable, 
 (2) P-time solvable over planar graphs but \#P-hard over general graphs, 
 (3) \#P-hard over planar graphs.
Moreover,
category (2) consists of precisely  those problems that can be solved by
Valiant's holographic algorithm using FKT.

Extending this understanding to larger domains has
%proven
turned out to be remarkably difficult,
even in the restricted setting of $\GH$.
%graph homomorphisms.
The difficulty lies in two aspects.
First is the complication of more tractable classes. 
It is easy to check from the definition of partition function that 
\[
Z_{A\otimes B}(G)=Z_A(G)Z_B(G)
\] 
for any symmetric matrices $A$ and $B$.
Consequently, if $\PlGH(A)$ and $\PlGH(B)$ are both tractable, then so does $\PlGH(A\otimes B).$
This results in highly nontrivial tractable $\PlGH(M)$ problems.
For example, $\PlGH(M)$ is tractable for $M=\left(\begin{smallmatrix}
%\left(\begin{matrix}
    8 & 6 & 4 & 3\\
    6 & 8 & 3 & 4\\
    4 & 3 & 8 & 6\\
    3 & 4 & 6 & 8
    \end{smallmatrix}\right).$
%\end{matrix}\right)
%=\left(\begin{matrix}
%    2 & 1\\
%    1 & 2\\
%\end{matrix}\right)\otimes \left(\begin{matrix}
%    4 & 3\\
%    3 & 4\\
%\end{matrix}\right).$$

%Indeed, $\PlGH(M)$ is tractable when $M=\bigotimes_{i=1}^d\left(\begin{smallmatrix}
%    a_i & b_i\\
%    b_i & a_i
%\end{smallmatrix}\right).$
The second difficulty concerns proof techniques for $\#$P-hardness.
%is the lack of techniques of proving $\#$P-hardness.
Hardness proofs in the planar setting typically rely on
\emph{planar edge gadgets},
which must preserve planarity while enabling reductions.
%Even for
With significant effort and the construction of
individually tailored gadgets and 
highly specific analyses, this complexity classification was 
established for $n = 3$ and $n = 4$~\cite{CaiM23,CaiM24}.
%,
%hardness results were only possible 
%through the construction of
%individually tailored gadgets and 
%highly specific analyses~\cite{CaiM23,CaiM24}.

To better understand 
%the totality of 
planar edge gadgets,
%ultimate reach of planar edge gadgets,
%address this, 
Cai, Maran and Young~\cite{cai-maran-young}
studied  their power and ultimate reach.
%initiated a systematic study of these.
%of these planar edge gadgets
%to understand the fundamental challenges encountered
%when moving beyond the $n = 2$ case.
%Cai, Maran and Young 
They 
%established 
showed that
$\PlGH(M)$ is \#P-hard whenever
planar gadgets can be used to \emph{separate} the
diagonal entries of $M$.
This is proved by inductively reducing from matrices of smaller order.
%effectively distinguishing between different labels.
We call this the \emph{vertex separation method}.
However, they also discovered that,
%a formidable barrier:
for many matrices, such a separation is impossible.
And in this separability lies a formidable barrier: 
this separation is governed by
%impossible precisely
%of 
%This occurs 
%when the matrix possesses nontrivial symmetries described by
an abstract construct called the \emph{quantum automorphism group} $\qut(M)$~\cite{cai-maran-young}.
%~\cite{wang1995free, wang1998quantum,bichon2003quantum, banica2005quantum,atserias2019quantum, chassaniol2019study,lupini2020nonlocal, manvcinska2020quantum, kar2026npa}.
%%% JYC i edited a bit here. i don't want people to think this has to do literally with quantum computing
% or any change of model of computation, thus P and #P-hardness is not the classical established
% notion. this is especially the case since we had a section called model of computation.
When this  $\qut(M)$ is nontrivial, the diagonal
entries remain inseparable by
\emph{all} planar gadgets, creating a ``quantum gap''
%%% JYC I added a word here "all"
where the standard strategy of finding reductions
from smaller domain \#P-hard problems
becomes provably impossible.
In particular, the vertex separation method 
%is unable to 
cannot separate $\PlGH(M)$ for $M=\bigotimes_{i=1}^d\left(\begin{smallmatrix}
    a_i & b_i\\
    b_i & a_i
\end{smallmatrix}\right)$
from the \#P-hard problems.

To address the non-trivial Qut case, Cai, Maran and Tang \cite{spectral-manuscript} studied polynomial interpolation based on spectral relations of $M\in \Sym{n}{}{\R_{\ge 0}}$, to reduce $\PlGH(M)$ from the $n$-vertex coloring problem.
They proved a \emph{spectral criterion}, which says if the Perron–Frobenius unique largest eigenvalue of $M\in \Sym{n}{pd}{\Rp}$ is multiplicatively independent of the remaining eigenvalues, then $\PlGH(M)$ is \#P-hard.
This criterion applies even when $\qut(M)$ is non-trivial.
However, if $M$ is a tensor product of smaller matrices, then there will be non-trivial multiplicative relations among the eigenvalues of $M.$
%For example, when $M=A\otimes B$, where $A,B\in \Sym{2}{}{\R}$ with eigenvalues $\lambda_1,\lambda_2$ and $\mu_1,\mu_2$ respectively, the four eigenvalues 
Nevertheless, they proved when $M\in \Sym{n}{pd}{\Rp}$ is a tensor product of 
arbitrary 2 by 2 matrices, then $\PlGH(M)$ is \#P-hard unless $M$ is
%\cong 
$\bigotimes_{i=1}^d\left(\begin{smallmatrix}
    a_i & b_i\\
    b_i & a_i
\end{smallmatrix}\right)$, up to a 
%simultaneous permutation of rows and columns.
permutation.
%(up to a simultaneous permutation of rows and columns).
In our unpublished result, we also proved that if every irreducible tensor factor of $M$ with order $>2$
%greater than 2 
satisfies the spectral criterion, then $\PlGH(M)$ is \#P-hard.
We call this \emph{spectral method}.

The work \cite{spectral-manuscript} motivates us to think: given a matrix $M$ for which we want to prove \#P-hardness, can we \emph{realize} a matrix $M'$ by planar edge gadgets and polynomial interpolation such that every irreducible tensor factor of $M'$ with order $>2$
%greater than 2 
satisfies the spectral criterion?
The word ``realize'' contextually means $\PlGH(M')\le_T \PlGH(M)$, and thus,
if $\PlGH(M')$ is $\#$P-hard then $\PlGH(M)$ is also $\#$P-hard.
However, it is unclear how to characterize the entire set of all matrices we can realize by planar edge gadgets and interpolation.
This motivates us to explicitly formulate an algebraic structure to describe all matrices that are \emph{realizable} from $M.$

Technically speaking, if both $M_1$ and $M_2$ are realizable, then we can realize $M_1\circ M_2$ and $M_1M_2$ by thickening gadgets and  stretching gadgets, respectively. 
Here $M_1\circ M_2$ is the entrywise product and $M_1M_2$ is the ordinary matrix product.
Further, polynomial interpolation allows us to realize the entrywise power $M^{\circ t}$ for any $t\in \R$ if $M_{ij}>0$.
%$M\in\Sym{n}{}{\Rp}$.
We would like to have a vector space structure on them. But $M_1 + M_2$ is generally not realizable from $M_1$ and $M_2$.
It turns out that we \emph{can} obtain a matrix algebra  (a vector space \emph{and} a compatible bilinear product structure) by going to 
%a linear closure in 
an entrywise logarithm setting, denoted by $\log_\circ$.
%Denoting entrywise logarithm by , 
We have $\log_\circ(M_1\circ M_2)=\log_\circ(M_1)+\log_\circ(M_2)$ and $\log_\circ(M^{\circ t})=t\log_{\circ}(M),\,\forall t\in \R$, getting
a vector space structure.
%Therefore, the entrywise logarithm (if well-defined) of all realizable matrices form a $\R$-vector space.
%Stretching gadget realizes matrix product, but can
Miraculously, we can also perform the ordinary matrix product after taking $\log_\circ$, making the entrywise logarithm of realizable matrices an $\mathbb{R}$-algebra.
To make this work we have to take a slight detour.
We will first do a ``small thickening" then do stretching.
For realizable matrices $M_1$ and $M_2$, $L_1(s)=s\log_\circ M_1$ and $L_2(t)=t\log_\circ M_2$ are in the entrywise logarithmic vector space.
By stretching, the matrix product of entrywise exponents $M(s,t)=\exp_\circ (L_1(s))\exp_\circ(L_2(t))$ is realizable for every $s,t\in\mathbb{R}$.
The derivative $\partial_s\partial_t\log_{\circ}M(s,t)$ at $(s,t)=(0,0)$ is in the entrywise logarithmic vector space, since the $\mathbb{R}$-vector space is closed as a metric space.
For more details about the entrywise logarithmic matrix algebra, see \cref{sec:algebra}.

After we define the entrywise logarithmic matrix algebra $\mathcal{A}$, we further define a polynomial homomorphism $F(\cdot)$ on the algebra
using entrywise exponential map.
The map $F(X)$ for every $X \in \mathcal{A}$ separates out the action into pairwise orthogonal subspaces. The projection from $\mathbb{R}^{n \times n}$ into $\mathcal{A}$
creates a way to define a convex polytope $\mathcal C$, then the unique
polynomial factorization of each entry of $F(X)$
defines a metric (as in a metric space), from which we can ultimately show that it is a graph metric. They must satisfy stringent conditions as very special families of
distance-regular graphs. After excluding one particular family called the
Doob graphs using a special planar gadget (which had been used before~\cite{CaiM24}),
%what is 
we are  left with  only the Hamming graphs. This finally allowed us to
conclude that the only tractable families left are the tensor
products of the Ising model.
See \cref{sec:polyhom,sec:graph_distance,sec:hamming,sec:dichotomy} for more details.
%%% may be more than just sec 6 for details?
%Therefore, we will present the 
The following is the main theorem in the body text:

\begin{theorem}[Dichotomy for entrywise positive and full rank matrices]\label{thm:main}
Let $M\in\Sym{n}{F}{\Rp}$.
Then $\PlGH(M)$ is \#P-hard unless after one simultaneous permutation of rows and columns,
\begin{equation}
M=\lambda\bigotimes_{i=1}^{d}
\begin{pmatrix}a_i&b_i\\b_i&a_i\end{pmatrix},
\qquad \lambda>0,\quad a_i,b_i>0,\quad a_i\neq b_i.
\label{eq:ising-tensor}
\end{equation}
in which case $\PlGH(M)$ is tractable.
\end{theorem}

The dichotomy for entrywise positive and full rank matrices lies at the core of the dichotomy for non-negative matrices (this was also seen in \cite{CaiM23,CaiM24,cai-maran-young}).
The following theorem extends this classification to nonnegative matrices
and its proof will be presented in \autoref{app: nnextension}.
\begin{theorem}[Dichotomy for nonnegative matrices]\label{thm: nonnegative}
Let $M\in\Sym{n}{}{\mathbb{R}_{\geq 0}}$. After a simultaneous permutation of rows and columns, we can express $M$ as a direct sum of blocks representing
connected components $C$ of size at least 2
in addition to  a possibly empty set of isolated vertices with or without  self-loops.
The problem  $\PlGH(M)$ is 
 \#P-hard unless each component $C$ is a certain rank-one modification
 of the matchgate tensor in \eqref{eq:ising-tensor}, or its bipartite form.

\end{theorem}

\section{Preliminaries}\label{sec:preliminaries}

\subsection{Definitions and Notations}
For positive integer $n$, let $[n]$ denote the set $\{1,2,\ldots, n\}.$
Let $\Sym{n}{}{X}$ denote the set of
$n \times n$ symmetric matrices with entries from
$X \subseteq \mathbb{R}$.
For example, we can have $X = \mathbb{R}$, $\mathbb{R}_{\geq 0}$ or
$\mathbb{R}_{\neq 0}$.
We then let $\Sym{n}{F}{X} \subset \Sym{n}{}{X}$ 
and $\Sym{n}{pd}{X} \subset \Sym{n}{F}{X}$ denote, respectively, the
subsets of full rank and positive definite symmetric matrices.
We use $A\circ B$ to denote the entrywise product of $A$ and $B$.
We consider arbitrary real valued
matrices (rather than just algebraic real valued matrices)\footnote{However, our results stay strictly within the classical Turing machine model in terms of bit-complexity.
 See Cai, Maran, and Young~\cite[Appendix A]{cai-maran-young}. }.
We use $\leq_{T}$ and $\equiv_T$ to denote polynomial-time Turing reducibility and P-time Turing equivalence, respectively.
We use $\mathbb{I}[E]$ to denote the indicator of the event $E$.

Given $M\in\Sym{n}{}{\R}$, the associated \emph{partition function}
$Z_{M}(G)$ for an input undirected multi-graph $G = (V,E)$ is
$$Z_{M}(G) = \sum_{\sigma: V \rightarrow [n]} \prod_{(u, v) \in E} M_{\sigma(u) \sigma(v)}.$$
Each $M$ defines a
computational problem denoted by $\GH(M)$: given an input graph $G$,
output $Z_{M}(G)$.
Let $\PlGH(M)$ denote the problem $\GH(M)$ when the
input graphs $G$ are restricted to planar graphs.
\iffalse
The symmetric matrix $M$ is also called the edge weight.
The vertex weight is a vector $w=(w_1,\ldots, w_n)\in \Rp^n.$ 
The associated partition function $Z_{M,w}(G)$ for an input undirected multi-graph $G=(V,E)$ is
$$Z_{M,w}(G) = \sum_{\sigma: V \rightarrow [n]} \prod_{(u, v) \in E} M_{\sigma(u) \sigma(v)}\prod_{v\in V}w_{\sigma(v)}.$$
\fi
It is sometimes convenient to view the $\PlGH$ problem as a \emph{planar Holant} problem, especially when we allow the edge weight $M$ to be asymmetric.
We next introduce signatures and Holant problems.

A (real) {\it constraint function} $f$, or a {\it signature}, of arity $k$ over domain $[n]$
is a map $[n]^k  \rightarrow \mathbb{R}$.
Fix a set $\mathcal{F}$ of constraint functions. 
A {\it signature grid}
$\Omega=(G, \pi)$
 is a tuple, where $G = (V,E)$
is a graph called the {\it underlying graph} of $\Omega$, $\pi$ labels each $v\in V$ with a function
$f_v\in\mathcal{F}$ of arity ${\operatorname{deg}(v)}$,
and labels the incident edges
$E(v)$ at $v$ with input variables of $f_v$.
We consider all edge assignments $\sigma:E\to [n]$,
each gives an evaluation
$\prod_{v\in V}f_v(\sigma|_{E(v)})$, where $\sigma|_{E(v)}$
denotes the restriction of $\sigma$ to $E(v)$. 
The counting problem on an instance $\Omega$ is to compute the {partition function}

\begin{equation}\label{eq: def of Holant value}
\Holant_{\Omega}=\Holant({\Omega};\mathcal{F})=\sum_{\sigma:E\rightarrow [n]}\prod_{v\in V}f_v(\sigma|_{E_{(v)}}).
\end{equation}

The Holant problem $\Omega \mapsto \Holant_{\Omega}$ is a computational problem parameterized by a set $\mathcal{F}$ and is denoted by $\Holant(\mathcal{F})$.
If $\mathcal{F}=\{f\}$ is a singleton set, for simplicity, we write $\{f\}$ as $f$ directly and also write $\mathcal{F} \cup \{g\}$ as $\mathcal{F}\cup g$ or $\mathcal{F},g$.
We denote by
$\PlHolant(\mathcal{F})$  the restriction of $\Holant(\mathcal{F})$ to planar graphs $G$.
We use $\Holant{(\mathcal{F}\mid \mathcal{G})}$ to denote the Holant problem over signature grids with a bipartite graph $H = (U,V,E)$,
where each vertex in $U$ (or $V$) is assigned a signature in $\mathcal{F}$ (or $\mathcal{G}$
respectively).
$\PlHolant{(\mathcal{F}\mid \mathcal{G})}$ denotes its planar restriction.
Let
$\mathcal{EQ} = \{{=}_1, {=}_2, {=}_3, \dotsc\}$ be the set of \textsc{Equality} signatures of all arities.

For $M\in \Sym{n}{}{\R}$, the problem $\PlGH(M)$ is equivalent to $\PlHolant(\mathcal{EQ}\mid m)$, where $m$ is the binary signature over domain $[n]$ satisfying $m(i,j)=M_{ij},\,\forall i,j\in[n].$ 
Indeed, given a planar graph $G=(V,E)$, every vertex in $V$ may be viewed as an \textsc{Equality} signature, and we label the signature $m$ in the middle of every edge in $E.$
When $M$ is symmetric (as a matrix), the signature $m$ is symmetric in that $m(i,j)=m(j,i).$
If we allow $m$ be asymmetric as a binary signature, then the problem $\PlHolant(\mathcal{EQ}\mid m)$ naturally generalizes $\PlGH(M).$
%Also, given a matrix $M$ and a vertex weight $w$, the problem asking the partition function $Z_{M,w}(G)$ for planar input graphs $G$ can be reduced to the problem $\PlHolant(\mathcal{EQ}\mid m,w)$, where $m$ is the binary signature as above, and $w$ is viewed as a unary signature.
For convenience, we do not distinguish between the matrix $M\in \mathbb{R}^{n\times n}$ and the binary signature $m(\cdot,\cdot)$, in the sense that  $\PlHolant(\mathcal{EQ}\mid M)$ is the same as the problem $\PlHolant(\mathcal{EQ}\mid m)$.

\iffalse
We start with a trivial
lemma which shows that scaling $M$ has no effect on the complexity of $\PlGH(M)$:

\begin{restatable}{lemma}{MequivalentCM}\label{lemma: MequivalentCM}
    Let $M \in \Sym{n}{}{\mathbb{R}}$, and
    let $c \in \mathbb{R}_{\neq 0}$.
    Then $\PlGH(M) \equiv_T \PlGH(cM)$.
\end{restatable}
\begin{proof}
    Note that given any $G = (V, E),$
    $$Z_{cM}(G) = \sum_{\sigma: V \rightarrow [q]}
    \prod_{(u, v) \in E}(cM)_{\sigma(u)\sigma(v)}
    = c^{|E|} \sum_{\sigma: V \rightarrow [q]}
    \prod_{(u, v) \in E}cM_{\sigma(u)\sigma(v)}
    = c^{|E|}Z_{M}(G).$$
    Therefore, oracle access to one of $\PlGH(M)$ or $\PlGH(cM)$
    gives the other.
\end{proof}
\fi

In the following proofs, $n$ is a fixed positive integer and $M\in \Sym{n}{F}{\Rp}$ is fixed unless otherwise stated.
We will prove \cref{thm:main} for this particular $M$ in the body text.

\subsection{Linear Algebra}

\paragraph{Matrix conventions and spectral theorems}
\mbox{} \\
For every integer $m\geq1$, let $\mathbf1=(1,\ldots,1)^{\tt T}\in\mathbb R^m$, let $J_m=\mathbf1\mathbf1^{\tt T}$ be the all-ones matrix, $O_m$ be the all zero matrix, $U_m=J_m/m$ and $P_m=I_m-U_m$.
When the context is clear, we omit the subscript $m.$
The matrices $U$ and $P$ are orthogonal projections to 
$\operatorname{Span}_{\mathbb R}\{\mathbf1\}$ and 
$\mathbf{1}^\perp$, respectively.
%$\operatorname{Span}_{\mathbb R}\{\mathbf1\}^{\perp}$, respectively.

\begin{definition}\label{def: entrywise function}
Let $f: \mathbb{R} \to \mathbb{R}$. For a real matrix or vector $K$, let $f_\circ(K)$ denote the entrywise application of $f$. In particular, $\exp_\circ$ and $\log_\circ$ denote entrywise exponentiation and logarithm, respectively, where $\log_\circ$ is defined when all entries of $K$ are strictly positive.
\end{definition}

\begin{definition}\label{def:matrix-conventions}
The matrix inner product and its norm are
\[
\langle C,D\rangle=\operatorname{Tr}(C^{\tt T}D)
=\sum_{x=1}^n\sum_{y=1}^n C_{xy}D_{xy},
\qquad \|C\|=\sqrt{\langle C,C\rangle}
\quad(C,D\in\mathbb R^{n\times n}).
\]
\end{definition}

\begin{theorem}[Real spectral theorem; {\cite[Corollary 2.5.11(a)]{horn-johnson}}]\label{thm:spectral-theorem}
For every integer $n\geq1$ and every $A\in\Sym{n}{}{\R}$, let $\lambda_1,\ldots,\lambda_s$ be its distinct real eigenvalues, where $1\leq s\leq n$. 
Then $\R^n=E_{\lambda_1}\oplus \cdots \oplus E_{\lambda_s}$
as a direct sum of orthogonal eigen-subspaces, where  $E_{\lambda_i}=\ker(\lambda_iI-A)\subseteq \R^n$.
\end{theorem}

\begin{theorem}[Simultaneous spectral theorem; {\cite[Theorem 2.5.15]{horn-johnson}}]\label{thm:simultaneous-spectral}
For every integer $n\geq1$, any pairwise commuting family of real symmetric $n\times n$ matrices has a common orthonormal eigenbasis of $\mathbb R^n$.
\end{theorem}

\paragraph{Convex analysis}
\begin{definition}[Affine span and convex hull]\label{def: affine span}
    Let $V$ be a real vector space, and let $v_1,\ldots, v_n\in V.$
    The (real) \emph{affine span} of $v_1,\ldots, v_n$ is 
    \[\operatorname{Aff}_{\R}\{v_1,\ldots, v_n\}=\left\{\sum_{i=1}^n a_iv_i\mid a_i\in \R \text{ for every } i\in[n]\text{ and } \sum_{i=1}^n a_i=1\right\}.\]
    The (real) \emph{convex hull} of $v_1,\ldots, v_n$ is 
    \[\operatorname{Conv}_{\R}\{v_1,\ldots, v_n\}=\left\{\sum_{i=1}^n a_iv_i\mid a_i\ge 0 \text{ for every } i\in[n]\text{ and } \sum_{i=1}^n a_i=1\right\}.\]
\end{definition}

\begin{theorem}[{\cite[Theorem 6.9]{convex}}]\label{thm: convex hull nonempty}
    Let $V$ be a real vector space, and let $v_1, \ldots, v_n \in V$. If $u = \sum_{i=1}^n a_i v_i$ with $a_i \in (0,1)$ for all $i \in [n]$ and $\sum_{i=1}^n a_i = 1$, then $u$ belongs to the relative interior of $\operatorname{Conv}_{\mathbb{R}}\{v_1, \ldots, v_n\}$ in $\operatorname{Aff}_{\mathbb{R}}\{v_1, \ldots, v_n\}$.
\end{theorem}

\begin{definition}[Affine form {\cite[Sections 2.3.2 and 3.1.1]{boyd2004convex}}]\label{def: affine form}
    Let $V$ be a finite dimensional real vector space with an inner product.
    The map $\ell:V\to \R$ is called a (real) affine form, if it satisfies the following equivalent conditions:
    \begin{enumerate}
        \item There exists a linear map $L:V\to \R$ and $c\in \R$ such that $\ell(v)=L(v)+c$ for every $v\in V$;
        \item For all $u,v\in V$ and every $t\in \R$, $\ell(tu+(1-t)v)=t\ell(u)+(1-t)\ell(v)$.
        \item 
        There exists a unique vector $g\in V$ and $c\in \R$ such that $\ell(v)=\braket{g,v}+c.$
    \end{enumerate}
\end{definition}

\subsection{Analytic Functions}

All vector spaces in the analytic arguments are finite dimensional. Matrix spaces and their subspaces carry the norm in \cref{def:matrix-conventions}.

\begin{definition}[Fr\'echet derivatives; {\cite[Sections 3.1--3.2]{buffoni-toland}}]\label{def:frechet-derivative}
Let $V,W$ be finite-dimensional real normed vector spaces, let $\mathcal O\subseteq V$ be open, and let $g:\mathcal O\to W$. At $x\in\mathcal O$, the Fr\'echet derivative $dg(x):V\to W$ is the linear map, when it exists, satisfying
\[
\lim_{\substack{h\to0\\h\ne0}}
\frac{\|g(x+h)-g(x)-dg(x)(h)\|}{\|h\|}=0.
\]
For every integer $k\geq2$, $d^kg(x)$ denotes the $k$-th derivative when it exists, regarded as a $k$-linear map from $V^k$ to $W$. We write $d^0g(x)=g(x)$. For smooth $g$ and every $h\in V$,
$dg(x)(h)=\left.\frac{d}{dt}g(x+th)\right|_{t=0}$.
\end{definition}

\begin{definition}[Real analytic maps;
    {\cite[Definition 4.3.1]{buffoni-toland}}]
Let $V,W$ be finite-dimensional real normed vector spaces,
let $\mathcal O\subseteq V$ be open, and let $g:\mathcal O\to W$.
The map $g$ is \emph{real analytic} if, for every $x\in\mathcal O$,
there exist $\varepsilon>0$ and symmetric $k$-linear maps
$A_k:V^k\to W$ ($k\geq0$) such that
$\sup_{k\geq0}\varepsilon^k\|A_k\|<\infty$ and
\[
g(x+h)=\sum_{k=0}^{\infty}A_k(h,\ldots,h)
\qquad( h\in V,\,\|h\|<\varepsilon).
\]
Here $A_0\in W$ is regarded as a constant map.
\end{definition}

\begin{theorem}[Analytic Taylor expansion; {\cite[Proposition 4.3.4]{buffoni-toland}}]\label{thm:taylor-expansion}
Let $V,W$ be finite-dimensional real normed vector spaces, let $\mathcal O\subseteq V$ be open, and let $g:\mathcal O\to W$ be real analytic. For every $x\in\mathcal O$, there exists $\varepsilon>0$ such that the open ball $B(x,\epsilon)\subseteq\mathcal O$ and
\[
g(x+h)=\sum_{k=0}^{\infty}\frac1{k!}d^kg(x)(h,\ldots,h)
\qquad(h\in V,\ \|h\|<\varepsilon).
\]
The term for $k=0$ is $g(x)$. 
This expansion is unique, and its term of degree $k$ is homogeneous of degree $k$ for every integer $k\geq0$. 
%For every integer $m\geq0$, subtracting the terms of degrees $0,\ldots,m$ leaves $O(\|h\|^{m+1})$ as $h\to0$.
\end{theorem}

\begin{theorem}[Inverse function theorem; {\cite[Theorem 4.5.3]{buffoni-toland}}]\label{thm:analytic-inverse}
Let $V,W$ be finite-dimensional real vector spaces of equal dimension, let $\mathcal O\subseteq V$ be open, and let $g:\mathcal O\to W$ be real analytic. If $x\in\mathcal O$ and $dg(x)$ is invertible, then $g$ has a real analytic inverse $g^{-1}$ between neighborhoods of $y=g(x)$ and $x.$
Further, $dg^{-1}(y)=(dg(x))^{-1}.$
\end{theorem}

\begin{theorem}[Analytic identity theorem; {\cite[(9.4.3), p. 208]{dieudonne1969}}]\label{thm:analytic-identity}
Suppose $U\subseteq\mathbb R^m$ is a domain, and $g:U\to \R$ is real analytic.
If $g(x)=0$ on a nonempty open subset of $U$, then $g(x)=0$ for every $x\in U$.
\end{theorem}

\subsection{Distance-regular Graphs and the Hamming Classification}

\begin{definition}[Distance-regular graphs and intersection numbers; \cite{DRGbook}]\label{def:distance-regular}
Let $\Gamma$ be a connected simple graph of diameter $D\geq1$, and let $d_\Gamma(x,y)$ denote the length of a shortest path between vertices $x,y$. The graph is \emph{distance-regular} if, for every integer $0\leq i,j,k\leq D$ and every pair of vertices $x,y$ with $d_\Gamma(x,y)=k$, the number
\[
p_{ij}^k:=|\{z\in V(\Gamma)\mid d_\Gamma(x,z)=i \text{ and } d_\Gamma(z,y)=j\}|
\]
is independent of $x$ and $y$. 
These counts $p_{ij}^k$ are called \emph{intersection numbers}.
\end{definition}

\begin{definition}[Hamming graphs; \cite{DRGbook}]\label{def:hamming-graphs}
For integers $D\geq1$ and $Q\geq2$, the Hamming graph $H(D,Q)$ has vertex set $[Q]^D$; two tuples are adjacent exactly when they differ in one position.
\end{definition}

Clearly, the Hamming graph $H(D,Q)$ is distance-regular. 
Egawa \cite{egawa} proved that any distance-regular graph of diameter $D$ with the same intersection numbers as $H(D,Q)$ is isomorphic to the Hamming graph $H(D,Q)$, %except in the case of 
or the Doob graph. 
Doob graphs are defined using the Shrikhande graph:

\begin{definition}[The Shrikhande graph; \cite{DRGbook,egawa}]\label{def:shrikhande-graph}
The Shrikhande graph has vertex set $\mathbb Z_4^2$ and edges joining $x,y\in\mathbb Z_4^2$ exactly when $y-x\in\mathcal D$, where
\(
\mathcal D=\{(1,0),(3,0),(0,1),(0,3),(1,1),(3,3)\}.
\)
%All coordinates and differences in this definition are taken modulo four.
\end{definition}

See \autoref{app:doob} for a figure of the Shrikhande graph.

\begin{definition}[Doob graphs; \cite{Doob1972,egawa}]\label{def:hamming-doob}
A \emph{Doob graph} $I(p,q)$ is a Cartesian product of $p\geq1$ Shrikhande graphs and $q\geq0$ copies of the complete graph $K_4$.
%where $p,s$ are integers. Its diameter is $2p+s$.
\end{definition}

\begin{theorem}[Hamming--Doob classification; \cite{egawa}]\label{thm:hamming-classification}
Let $\Gamma$ be a distance-regular graph which has the same intersection numbers as $H(D,Q)$.
Then one of the following holds:
\begin{enumerate}
    \item $\Gamma$ is isomorphic to $H(D,Q)$; or
    \item 
    $Q=4$ and $\Gamma$ is isomorphic to the Doob graph $I(p,q)$ with $p\geq1$, $q\geq0$ and $2p+q=D$.
\end{enumerate}

\end{theorem}

In \cref{sec:hamming} we will define coordinate graphs $\{\Gamma_\omega\}_{\omega\in \Omega}$ and show that each of them has the same intersection numbers as $H(D_\omega,Q_\omega)$.
The case $\Gamma_\omega\cong H(D_\omega,Q_\omega)$ is used to establish the dichotomy for $\PlGH(M)$ in \cref{sec:dichotomy}, and the Doob case is excluded by directly showing $\PlGH(M)$ is \#P-hard in \autoref{app:doob}.

\subsection{Known Dichotomies and Hardness Results}

\begin{theorem}[Guo--Williams~{\cite[Theorem 2.11]{guo-williams}}]\label{thm:boolean-source}
Let $W=\left(\begin{smallmatrix}x&y\\y&z\end{smallmatrix}\right)\in \Sym{2}{F}{\Rp}$. 
Then $\PlGH(W)$ is \#P-hard unless $x=z$, in which case $\PlGH(W)$ is tractable.
\end{theorem}

\begin{theorem}[Vertigan~\cite{vertigan}]\label{thm:potts-source}
  The $n$-state Potts model
     $\PlGH({\tt{Potts}}_{n}(x))$ is $\#$P-hard for any integer $n \geq 3$,
    and non-negative real $x \neq 1$,
    where ${\tt{Potts}}_{n}(x) \in \text{Sym}_{n}(\mathbb{R})$
    is the matrix with entries
    $({\tt{Potts}}_{n}(x)_{ij})_{i,j\in[n]}$ such that
    ${\tt{Potts}}_{n}(x)_{ij} = 1$ if $i \neq j$,
    and ${\tt{Potts}}_{n}(x)_{ii} = x$.
    %otherwise.
\end{theorem}

We say a matrx $M\in \R^{n\times n}$ has nonconstant diagonal, if $|\{M_{ii}\mid 1\le i \le n\}|\ge 2.$

\begin{restatable}{theorem}{nonconstantdiag}\label{thm:nonconstant}
If $M\in\Sym{n}{F}{\Rp}$ and $M$ has nonconstant diagonal, then $\PlGH(M)$ is \#P-hard.
\end{restatable}

\autoref{app:nonconstant} includes the proof of \cref{thm:nonconstant}.

\section{Planar Edge Gadgets and the Homogeneous Setting}\label{sec:reductions}
This section provides basic preliminaries before we define the entrywise logarithmic algebra in \cref{sec:algebra}.
In \cref{subsec:thickening,subsec:stretching}, we discuss
%provide basic knowledge for 
two basic reduction techniques: planar edge gadgets and polynomial interpolation.
In particular, we focus on thicking gadgets and stretching gadgets.
In contrast to
%comparison with 
\cite{CaiM23,CaiM24,cai-maran-young}, we allow realizable matrices to be asymmetric.
This is for the convenience of defining the entrywise logarithmic matrix algebra in the next section.
The matrix product of two symmetric matrices may be asymmetric. 
So even if we start from symmetric matrices, the algebra generated by them may contain asymmetric ones.
We will also prove in \cref{subsec:homogeneous} that every realizable $K\in \mathbb{R}^{n\times n}$ has constant diagonal, constant row sum and constant column sum, which we call the homogeneous setting.

\begin{definition}[realizable]\label{def:realizable}
Fix an integer $n\geq1$ and $M\in\mathbb{R}^{n\times n}$. 
%Let $k(i,j)$ be the binary signature over domain $n$ satisfying $k(i,j)=K_{ij}.$
We say the signature $f$ is realizable (relative to $M$), if for an arbitrary signature set $\mathcal{G}$ containing $M,$
\(
\PlHolant(\mathcal{EQ}\mid f,\mathcal{G})
\leq_T
\PlHolant(\mathcal{EQ}\mid \mathcal{G}).
\)
In particular, if $N\in \mathbb{R}^{n\times n}$, we say $N$ is realizable (relative to $M$) if the binary signature $N(i,j)=N_{ij}$ is realizable.
\end{definition}

Unless otherwise stated, in the following all realizable signatures are relative to the fixed matrix $M\in \Sym{n}{F}{\Rp}$.

\subsection{Thickening Gadgets}\label{subsec:thickening}
\begin{lemma}\label{lem:thickening}
    Let $k\in\{1,2\}$.
    Let $f,g$ and $h = f \circ g$ be real valued signatures of arity $k$ on domain $[n]$, satisfying $h({\bf x})=f({\bf x})\cdot g({\bf x}),\,\forall {\bf x}\in[n]^k.$
    Then for arbitrary signature set $\mathcal{G}$,
    $\PlHolant(\mathcal{EQ}\mid h,\mathcal{G})\le \PlHolant(\mathcal{EQ}\mid f,g,\mathcal{G})$.
\end{lemma}
\begin{proof}
    Suppose $\Omega=(G,\pi)$ is a signature grid of $\PlHolant(\mathcal{EQ}\mid h,\mathcal{G}).$
    For every occurrence of the unary or binary signature $h$, replace it by one copy of $f$ and $g$.
    When $h$ is a binary signature, every edge constraint labeled by $h$ thickens into a pair of edges labeled by $f$ and $g$.
    This preserves planarity and the value of the partition function. 
\end{proof}

\begin{corollary}\label{cor:thickening}
    If $M_1,M_2\in \R^{n\times n}$ are both realizable, so is $M_1\circ M_2.$
\end{corollary}
\begin{proof}
    Fix an arbitrary signature set $\mathcal{G}$ containing $M.$
    By \cref{lem:thickening}, we have $\PlHolant(\mathcal{EQ}\mid M_1\circ M_2,\mathcal{G})\le_T \PlHolant(\mathcal{EQ}\mid M_1,M_2,\mathcal{G}).$
    Since $M_1$ is realizable, $\PlHolant(\mathcal{EQ}\mid M_1,M_2,\mathcal{G})\le_T \PlHolant(\mathcal{EQ}\mid M_2,\mathcal{G}).$
    Since $M_2$ is realizable, $\PlHolant(\mathcal{EQ}\mid M_2,\mathcal{G})\le_T \PlHolant(\mathcal{EQ}\mid \mathcal{G}).$
    Thus, the reuduction chain gives $\PlHolant(\mathcal{EQ}\mid M_1\circ M_2,\mathcal{G})\le_T \PlHolant(\mathcal{EQ}\mid \mathcal{G}).$
    This shows $M_1\circ M_2$ is realizable.
\end{proof}

The following lemma can be proved using standard polynomial interpolation method (for example, see Lemma 2, \cite{spectral-manuscript}):

\begin{lemma}\label{lem:thickening_interpolation}
    Let $k\in\{1,2\}$.  
    Let $f$ be a real  signature of arity $k$ with positive values on domain $[n].$
    Then for arbitrary signature set $\mathcal{G},$
    $\PlHolant(\mathcal{EQ}\mid f^t,\mathcal{G})\le \PlHolant(\mathcal{EQ}\mid f,\mathcal{G})$ for every $t\in \R$, where $f^t$ is the arity $k$ signature satisfying $f^t({\bf x})=(f({\bf x}))^t,\,\forall {\bf x}\in[n]^k.$
\end{lemma}

\begin{corollary}\label{cor:thickening_interpolation}
If $K\in \Rp^{n\times n}$ is realizable, so is $K^{\circ t}$ for every $t\in \R.$
\end{corollary}
%The proof is similar to \cref{cor:thickening}.

\subsection{Stretching Gadgets}\label{subsec:stretching}

\begin{lemma}\label{lem:stretching}
    If $M_1,M_2\in \mathbb{R}^{n\times n}$ are both realizable, so is $M_1M_2.$
\end{lemma}
\begin{proof}
    For every occurrence of the binary signature $M_1M_2$ on a vertex $v$ of degree 2
    with incident edges $e$ and $e'$ labeled by the first and second variable (row and column index of $M_1M_2$ respectively), replace
    %``stretch" it 
    $e,v,e'$ by a path $e, u, e^*, u', e'$ with vertices $u$ and $u'$ labeled by $M_1$ and
    $M_2$ respectively. The incident edges $e$ and $e^*$ of $u$ (respectively,  $e^*$ and $e'$ of $u'$) are labeled by
    the  first and second variables of $M_1$ (respectively, $M_2$).
    %of by a concatenation of two binary signatures $M_1$ and $M_2$.
    This preserves planarity and the value of the partition function.
    Thus, $\PlHolant(\mathcal{EQ}\mid M_1M_2,\mathcal{G})\le_T \PlHolant(\mathcal{EQ}\mid M_1,M_2,\mathcal{G})$, for arbitrary signature set $\mathcal
    G.$
    Similar to the proof of \cref{cor:thickening}, this implies that $M_1 M_2$ is realizable whenever $M_1$ and $M_2$ are both realizable.
\end{proof}

The following lemma is a consequence of standard polynomial interpolation (e.g., see Lemma 4, \cite{spectral-manuscript}; even though $t$ can be non-algebraic, all complexity claims are strictly in the Turing bit model):
\begin{lemma}\label{lem:spectral-powers}
    If $K\in \Sym{n}{pd}{\R}$ is realizable, so is the matrix power $K^t$, for every $t\in\mathbb R$.
\end{lemma}

%%% JYC ^pd is needed.

\subsection{The Homogeneous Setting}\label{subsec:homogeneous}
In this subsection, we prove every realizable matrix $K\in \mathbb{R}^{n\times n}$ (relative to the fixed $M\in \Sym{n}{F}{\Rp}$) must have constant diagonal, constant row sum and constant column sum, otherwise $\PlGH(M)$ is \#P-hard.
This is what we called the homogeneous setting.
This homogeneity brings a lot of technical convenience when we define the entrywise logarithmic matrix algebra in
%the following 
\cref{sec:algebra}.

\begin{lemma}\label{lem:realize_diag_K_diag}
    Let $u,v\in \mathbb{R}^n$ be realizable unary signatures, and let $K\in \R^{n\times n}$ be a realizable binary signature.
    Then $N=\mathrm{diag}(u_1,\ldots, u_n)\cdot K\cdot \mathrm{diag}(v_1,\ldots, v_n)$ is a realizable binary signature. 
\end{lemma}
\begin{proof}
    %We only need to prove the reduction $\PlHolant(\mathcal{EQ}\mid N)\le_T\PlHolant(\mathcal{EQ}\mid  K,u,v).$
    Fix an arbitrary signature set $\mathcal{G}$ containing $M.$
    Let 
$\Omega=(G,\pi)$ be a signature grid of $\PlHolant(\mathcal{EQ}\mid N,\mathcal{G}).$
    Every occurrence of
    %the edge constraint 
    $N$  labels  a vertex of degree 2. Suppose  its first and second variable 
    are associated to the edges connected to $w_1$ and $w_2$ respectively, both 
    having a label from $\mathcal{EQ}$.
    We replace $N$ by $K$ and attach unaries $u$ and $v$ to $w_1$ and $w_2$ by two new edges
    respectively.
  %  the two endpoints of $K$.
    This preserves planarity, and we obtain a signature grid of $\PlHolant(\mathcal{EQ}\mid K,u,v,\mathcal{G})$ with the partition function unchanged.
\end{proof}

\begin{lemma}\label{lem:realize_diag_rowsum_colsum}
    Suppose $K\in \mathbb{R}^{n\times n}$ is a realizable binary signature, then $d=(d_1,d_2,\ldots,d_n),\,r=(r_1,r_2,\ldots, r_n),$ and $c=(c_1,c_2,\ldots, c_n)$ are all realizable, where $d_i=K_{ii},\, r_i=\sum_{j\in[n]}K_{ij},\,c_i=\sum_{j\in[n]}K_{ji}$, for every $i\in[n].$
\end{lemma}
\begin{proof}
    Fix an arbitrary signature set $\mathcal{G}$ containing $M.$
    We first show that $d$ is realizable.
    We only need to prove the reduction $\PlHolant(\mathcal{EQ}\mid d,\mathcal{G})\le_T\PlHolant(\mathcal{EQ}\mid K,\mathcal{G}).$
    Suppose $\Omega$ is a signature grid of $\PlHolant(\mathcal{EQ}\mid d,\mathcal{G})$.
    Any occurrence of $d$ is attached to a vertex with label from $\mathcal{EQ}$.
    We replace every occurrence of $d$  by a  loop of length 2 with $K$ attached to the 
    mid point of the loop. The {\sc Equality} constraint reproduces $d$ with $d_i=K_{ii}$.
    %endpoint of $d$.
    %We obtain a signature grid of $\PlHolant(\mathcal{EQ}\mid K)$, and this 
    This is a signature grid of $\PlHolant(\mathcal{EQ}\mid K,\mathcal
    G)$.
    This replacement preserves planarity and the value of the partition function.
    
    Next, we prove that $r$ is realizable.
    We only need to show the reduction $\PlHolant(\mathcal{EQ}\mid r,\mathcal{G})\le_T\PlHolant(\mathcal{EQ}\mid K,\mathcal{G}).$
    Suppose $\Omega'$ is a signature grid of $\PlHolant(\mathcal{EQ}\mid r,\mathcal{G})$.
    For every occurrence of $r$, replace it by one copy of $K$ with an {\sc Equality} signature $=_1$ attached to the second input of $K$. 
    We obtain a signature grid of $\PlHolant(\mathcal{EQ}\mid K,\mathcal{G})$, and this replacement preserves planarity and the value of partition function.
    Thus the row sum signature $r$ of $K$ is realizable.
    Similarly, 
    the column sum  $c$ is also realizable.
\end{proof}

\begin{lemma}\label{lem:unary_constanat}
Fix $M\in \Sym{n}{F}{\Rp}$. 
If a unary signature $u\in\mathbb R_{>0}^n$ is realizable relative to $M$, then 
either there exists $c>0$ such that $u_i=c$ for every $i\in[n]$, or
%Otherwise,
$\PlGH(M)$ is \#P-hard.
\end{lemma}

\begin{proof}
\cref{thm:nonconstant} gives $M_{ii}=a$ for some $a>0$ and every $i\in[n]$, otherwise $\PlGH(M)$ is \#P-hard.
Write $u=(u_1,\ldots, u_n).$
By \cref{lem:realize_diag_K_diag}, the binary signature $N=\mathrm{diag}(u_1,\ldots, u_n)\cdot M\cdot \mathrm{diag}(u_1,\ldots,u_n)$ is realizable.
Clearly,  $N\in \Sym{n}{F}{\Rp}$ since $u\in \Rp^n.$
Its diagonal is $N_{ii}=au_i^2$ for every $i\in[n]$. 
If $u$ is non-constant then  $N$ has non-constant diagonal,
and thus $\PlGH(N)$ \#P-hard by \cref{thm:nonconstant}.
It follows that $\PlGH(M)$ is \#P-hard.
\end{proof}

\begin{lemma}[Homogeneous setting]\label{lem:homogeneous}
Fix $M\in \Sym{n}{F}{\Rp}$. Assume $\PlGH(M)$ is not \#P-hard. Then,
for every binary signature $K\in\mathbb R_{>0}^{n\times n}$ realizable relative to $M$, $K$ has constant diagonal, constant row sum and constant column sum.
\end{lemma}
\begin{proof}
    By \cref{lem:realize_diag_rowsum_colsum}, the diagonal, row sum and column sum of $K$ are realizable as unary signatures.
    \cref{lem:unary_constanat} then shows that these unaries must be constant, otherwise $\PlGH(M)$ is \#P-hard.
\end{proof}

The purpose of introducing the auxiliary signature set $\mathcal{G}$ in \cref{def:realizable} is to prove realizability of the signature in discuss in rigor (e.g., see the proof of \cref{cor:thickening}).
However, all realizability proofs in this section proceed via planar edge gadgets and polynomial interpolation from $M$, and the auxiliary signature set $\mathcal{G}$ is immatierial.
It is the same case in the reduction proofs in \autoref{app:nonconstant} and \autoref{app: nnextension}.
For simplicity, we omit the auxiliary signature set $\mathcal{G}$ in the lemma statements and proofs there.

Throughout the following proof, by being in 
the \emph{homogeneous setting} we mean
the matrix under discussion has constant diagonal, constant row  and  column sums. 
If $\PlGH(M)$ is not \#P-hard, every entrywise positive realizable matrix is in the homogeneous setting.
For simplicity, the phrase ``Assume $\PlGH(M)$ is not \#P-hard'' is omitted from lemma statements and proofs below.

\section{The Entrywise Logarithmic Matrix Algebra}\label{sec:algebra}

In this section, we define
%formulate the definition of 
the entrywise logarithmic matrix algebra, and explore its basic properties, including the existence of a multiplicative unit and its dimension.
All realizability in this section is relative to the fixed $M\in\Sym{n}{F}{\Rp}$ unless otherwise stated.

\begin{definition}\label{def:mathcalR}
    Let $$\mathcal R:=\{K\in \Rp^{n\times n}\mid K \text{ is realizable from }M\}$$ be the set of positive valued binary signatures realizable relative to $M$. 
\end{definition}
\subsection{Definition of the Entrywise Logarithmic Matrix Algebra}

Recall in \cref{def:matrix-conventions}, $P=I-U$ is the orthogonal projection onto $\mathbf{1}^\perp\subseteq \R^n.$
Instead of directly taking $\log_\circ$ of all realizable matrices $K$, we take the \emph{centered} entrywise logarithm by 
%restricting 
projecting $\log_\circ(K)$ on the subspace $\mathbf{1}^\perp\subseteq\mathbb{R}^n.$
This normalizes the row sums and column sums 
%of $K$ 
to $0$ (\cref{lem:normalize_rumsum_colsum_to_0}), bringing convenience when we prove the centered logarithmic  algebra $\mathcal{A}$ forms an $\mathbb{R}$-algebra, i.e.,
a vector space
with bilinear multiplication, in \cref{lem:universal-algebra}. 
The centralization doesn't affect realizability after entrywise exponentiating back, as shown in \cref{lem:algebra-exponentials}.

\begin{definition}[Entrywise logarithmic matrix algebra]\label{def: entrywise log matrix algebra}
Define
\begin{equation*}
\mathcal A:=\operatorname{Span}_{\mathbb R}
\{P\log_\circ(K)P\mid K\in\mathcal R\}.
\label{eq:algebra}
\end{equation*}
Every member in $\mathcal A$ 
%in this definition  consists of 
is a finite $\mathbb{R}$-linear combination of elements of the form 
$P\log_\circ(K)P$.
\end{definition}

\begin{lemma}\label{lem:algebra-exponentials}
For every $C\in\mathcal A$, $\exp_\circ C$ is realizable.
\end{lemma}

\begin{proof}
Let 
$$\mathcal L=\mathrm{Span}_{\R}\{\log_\circ u\mid u\in \Rp^{n} \text{ is a realizable unary signature}\}\subseteq\mathbb R^n$$
be the real span of  entrywise logarithms of positive realizable unaries. 
If $w=\sum_{j=1}^s a_j\log_\circ u^{(j)}\in\mathcal L$, where $s\geq0$, $a_j\in\mathbb R$, and each $u^{(j)}$ is a positive realizable unary, then
\[
\exp_\circ w=\mathop{\circ}_{j=1}^s (u^{(j)})^{\circ a_j}.
\]
\cref{lem:thickening,lem:thickening_interpolation}
guarantee that $\exp_\circ w$ is realizable for every $w\in\mathcal L$.

Fix $K\in\mathcal R$. 
For every $t\in\mathbb R$, $K^{\circ t}$ is realizable by \cref{cor:thickening_interpolation}.
By \cref{lem:realize_diag_rowsum_colsum}, the row sums of $K^{\circ t}$ is realizable as a unary signature.
The unary has entries
$u_x(t)=\sum_{y=1}^n K_{xy}^t$ for $x\in[n]$.
The curve $\log_\circ u(t),\,t\in \mathbb{R}$ lies in $\mathcal L$. 
Let $$v=\lim_{t\to 0}{\frac{1}{t}}\left(\log_\circ u(t)-\log_\circ u(0)\right)=\left(\left.\frac{d}{dt}\log u_1(t)\right|_{t=0},\ldots,\left.\frac{d}{dt}\log u_n(t)\right|_{t=0}\right)$$ be the derivative of $\log_\circ u(t)$ at zero.
For every $t\in \R$, ${\frac{1}{t}}\left(\log_\circ u(t)-\log_\circ u(0)\right)\in \mathcal{L}$ since $\mathcal{L}$ is a $\R$-vector space.
Therefore $v\in \mathcal L$, as $\mathcal{L}$ is closed as a metric space.
Its $x$-th coordinate is
\[
v_x=\left.\frac{d}{dt}\log u_x(t)\right|_{t=0}
=\frac{\sum_{y=1}^n K_{xy}^0\log K_{xy}}{\sum_{y=1}^n K_{xy}^0}
=\frac1n\sum_{y=1}^n\log K_{xy}.
\]
Consequently, $r=\exp_\circ(v)$ with the $x$-coordinate being
$r_x=\exp(n^{-1}\sum_{y=1}^n\log K_{xy})$ is a realizable unary signature. 
Similarly, the unary signature $c$ with
$c_y=\exp(n^{-1}\sum_{x=1}^n\log K_{xy})$ is also realizable.
Their reciprocals are realizable by taking $t=-1$ in \cref{lem:thickening_interpolation}.

Writing $L=\log_\circ K$, the centering operation is explicitly
\[
(PLP)_{xy}=L_{xy}-\frac1n\sum_{b=1}^nL_{xb}
-\frac1n\sum_{a=1}^nL_{ay}
+\frac1{n^2}\sum_{a=1}^n\sum_{b=1}^nL_{ab}
\quad(x,y\in[n]).
\]
It follows that
\[
\bigl[\exp_\circ(PLP)\bigr]_{xy}
=\frac{e^\mu K_{xy}}{r_xc_y}, \text{ where } \mu=\frac1{n^2}\sum_{a=1}^n\sum_{b=1}^n\log K_{ab}.
\]
Thus, $\exp_\circ(PLP)=e^\mu\cdot\mathrm{diag}(r_1^{-1},\ldots,r_n^{-1})\cdot K\cdot \mathrm{diag}(c_1^{-1},\ldots,c_n^{-1})$ is realizable by \cref{lem:realize_diag_K_diag}.

For a finite linear combination
%expression
$C=\sum_{j=1}^s a_jP\log_\circ(K^{(j)})P$ with $K^{(j)}\in\mathcal R$, exponentiation gives the entrywise product of the $a_j$-th powers of 
 realizable signatures $\exp_\circ(P\log_\circ(K^{(j)})P)$.
%these centered signatures.
Another application of \cref{lem:thickening,lem:thickening_interpolation} proves the claim.
\end{proof}

\begin{lemma}\label{lem:normalize_rumsum_colsum_to_0}
%$C=P\log_\circ(K)P\in 
    For every $C \in \mathcal{A}$, $C$ has constant diagonal, zero row sums and zero column sums.
    For $C=P\log_\circ(K)P$, we have
    $C = \log_\circ K-\mu J$, where $\mu$ is the average of all entries of $\log_\circ K.$
    %Otherwise, $\PlGH(M)$ is \#P-hard.
\end{lemma}
\begin{proof}
    Let $L=\log_\circ K$.
    By \cref{lem:homogeneous}, $K$ has constant diagonal, otherwise $\PlGH(M)$ is \#P-hard.
    So $L$ also has constant diagonal.
    By the proof of \cref{lem:algebra-exponentials}, the unary signature $r$ where $r_x=\exp(n^{-1}\sum_{y=1}^n\log K_{xy})$ is realizable.
    By \cref{lem:unary_constanat}, $r$ is a constant unary signature. %otherwise $\PlGH(M)$ is \#P-hard.
    
    Then $L$ has a constant row sum $r_L$.
    Similarly, $L$ also has a constant column sum $c_L$.
    Then $n r_L = n c_L$ is the sum of all entries in $L.$
    Thus, $r_L=c_L=n\mu.$
    By the proof of \cref{lem:algebra-exponentials},
    $$
    C_{xy}=(PLP)_{xy}=L_{xy}-\frac{r_L}{n}-\frac{c_L}{n}+\mu=L_{xy}-\mu.
    $$
    Therefore, $C$ has constant diagonal and $C=\log_\circ K-\mu J.$
    A general $C \in \mathcal{A}$ is a linear combination of the form
    $P\log_\circ(K)P$ and hence also has constant diagonal, zero row  and  column sums.
\end{proof}

\begin{lemma}\label{lem:universal-algebra}
The $\R$-vector space $\mathcal A$ is a transpose-closed matrix algebra.
\end{lemma}
\begin{proof}
    Suppose $K\in \mathcal{R}$ is realizable, then $K^{\tt T}$ is also realizable, simply by swapping two variables of the binary signature $K.$
    Thus $\mathcal{R}$ is transpose closed, 
    %as well as 
    and since $P^{\tt T} =P$, so is $\mathcal{A}$.

    Fix $C,D\in\mathcal A$, we show $CD\in \mathcal{A}.$
    For $s,t\in\mathbb R$, we have $sC, tD\in \mathcal{A}$ since $\mathcal{A}$ is a $\R$-vector space.
    By \cref{lem:algebra-exponentials}, $\exp_\circ(sC)$ and $\exp_\circ(tD)$ are both realizable.
    By stretching (\cref{lem:stretching}), $M(s,t)=\exp_\circ(sC)\exp_\circ(tD)$ is realizable, for every $s,t\in \mathbb{R}$.
    We have
    \[
    M(s,t)_{xy}=\sum_{v=1}^n e^{sC_{xv}+tD_{vy}}
    \quad(x,y\in[n]).
    \]
    At $(s,t)=(0,0)$, $M(0,0)_{xy}=n,$ and direct differentiation gives, for all $x,y\in[n]$,
    \[
    \qquad
    \partial_s M(0,0)_{xy}=\sum_{v=1}^n C_{xv},\qquad
    \partial_t M(0,0)_{xy}=\sum_{v=1}^n D_{vy}.
    \]
    By \cref{lem:normalize_rumsum_colsum_to_0}, $C$ and $D$ have zero row  and column sums, so $\partial_s M(0,0)_{xy}=
    \partial_t M(0,0)_{xy}=0.$
    Differentiating twice at $(s,t)=(0,0)$,
    \[
    \partial_s\partial_tM(0,0)_{xy}=\sum_{v=1}^n C_{xv}D_{vy}=(CD)_{xy}.\]
    Then
    \[
    \partial_s\partial_t\log M(0,0)_{xy}
    =\frac{\partial_s\partial_tM(0,0)_{xy}}{M(0,0)_{xy}}
    -\frac{\partial_sM(0,0)_{xy}\cdot\partial_tM(0,0)_{xy}}{M(0,0)_{xy}^2}
    =\frac{(CD)_{xy}}n.
    \]
    Since $P$ is a constant matrix independent of $s$ and $t$, we have 
    \[
    \partial_s\partial_t(P\log_\circ M(0,0)P)=P\cdot\partial_s\partial_t\log M(0,0)\cdot P=\frac{{PCDP}}{n}.
    \]
    Since $C$ and $D$ have zero row sums and zero column sums, $PC=C$ and $DP=D$.
    Therefore, $\partial_s\partial_t(P\log M(0,0)P)=CD/n.$
    Notice that for every $s,t\in \mathbb{R}$, $P\log_\circ M(s,t) P\in \mathcal{A}$.
    So $\partial_s\partial_t(P\log_\circ M(0,0)P)\in \mathcal{A}$, since $\mathcal{A}$ is closed as a metric space.
    Hence $CD\in\mathcal A$.
    Therefore, $\mathcal{A}$ is a transpose-closed matrix algebra.
\end{proof}

%%%%%%JYC here

\subsection{Basic Properties of the Matrix Algebra}

In this subsection, we study some basic properties of the entrywise logarithmic matrix algebra $\mathcal{A}$ defined in \cref{def: entrywise log matrix algebra}, including its dimension, support, and the existence of a unit.
To study $\mathcal{A}$,
it is convenient to define the real span $\mathcal{B}$ of realizable entrywise positive matrices:
\begin{definition}\label{def:binary-matrix-span}
Let 
\(
\mathcal B=\operatorname{span}_{\mathbb R}\mathcal R.
\)
\end{definition}

Note that $\mathcal{B}$ may contain matrices that are not realizable, since realizability is not necessarily preserved under addition.
However, taking real span preserves the homogeneous setting, namely  matrices in $\mathcal{B}$ have constant diagonal, constant row sum and column sum (\cref{lem:homogeneous-space}).

\begin{lemma}\label{lem:matrix-span-closure}
$\mathcal B$ is a matrix algebra closed under transpose and entrywise product, and it contains $J$.
\end{lemma}

\begin{proof}
    Let $C=\sum_{i}a_i K_i$ and $D=\sum_{j}b_jN_j$ both in $\mathcal{B}$, where $K_i,N_j\in \mathcal{R}.$
    Then $CD=\sum_{i,j}a_ib_jK_iN_j\in \mathcal{B}$, as $K_iN_j\in \mathcal{R}$ by \cref{lem:stretching}.
    Hence $\mathcal{B}$ is a matrix algebra.
    Because $\mathcal{R}$ is transpose closed, $\mathcal{B}$ inherits this property by taking linear span.
    Also, $C\circ D=\sum_{i,j}a_ib_jK_i\circ N_j\in \mathcal{B}$, as $K_i\circ N_j\in \mathcal{R}$ by \cref{cor:thickening}.
    Finally, by \cref{cor:thickening_interpolation}, $J=M^{\circ 0}\in \mathcal{R}\subseteq \mathcal{B}.$
\end{proof}

\begin{lemma}\label{lem:homogeneous-space}
Every $C\in\mathcal B$ has constant diagonal, constant row sum and column sum.
\end{lemma}
\begin{proof}
\cref{lem:homogeneous} gives these properties for each $K\in\mathcal R$.
Taking linear span preserves these properties.
\end{proof}

\begin{lemma}\label{lem:row-isometry}
For every $x\in[n]$ and all $C,D\in\mathcal B$,
\(
\langle C,D\rangle=n\sum_{y=1}^n C_{xy}D_{xy}.
\)
\end{lemma}

\begin{proof}
The matrix $C\circ D$ lies in $\mathcal B$ by \cref{lem:matrix-span-closure} and has constant row sum by \cref{lem:homogeneous-space}. Consequently,
\[
\langle C,D\rangle
=\sum_{a=1}^n\sum_{y=1}^n C_{ay}D_{ay}
%=\sum_{a=1}^n\sum_{y=1}^n (C\circ D)_{ay}
=n\sum_{y=1}^n C_{xy}D_{xy}\quad \text{for every }x\in[n].
\]
\end{proof}

\begin{corollary}\label{cor:row-determines-matrix}
For every $x\in[n]$, the map $\iota_x:$ $\mathcal{B}\to \mathbb{R}^n,C\mapsto(C_{x1},C_{x2},\ldots, C_{xn})$ is injective.
In particular, $\dim\mathcal B\leq n$.
\end{corollary}
\begin{proof}
Note that $\iota_x$ is a linear map, so we only need to prove that $\ker(\iota_x)=0.$
Suppose $C\in\ker (\iota_x)$.
Then the $x$-th row of $C\in\mathcal B$ vanishes.
\cref{lem:row-isometry} with $D=C$ gives $\|C\|^2=0$.
So $C=0.$
The map is therefore injective.
In particular, $\dim(\mathcal{B})\le\dim(\mathbb{R}^n)=n.$
\end{proof}

\begin{lemma}\label{lem:log-space-in-span}
$\mathcal A\subseteq\mathcal B$.
Consequently, $\dim\mathcal{A}\le n.$
\end{lemma}
\begin{proof}
Fix $K\in\mathcal R$ and let its distinct entry values be $a_1,\ldots,a_s$, where $s\geq1$. The polynomial
\[
p(t)=\sum_{j=1}^s (\log a_j)
\prod_{\substack{1\leq k\leq s\\k\ne j}}\frac{t-a_k}{a_j-a_k}
\]
interpolates $\log t$ at $a_1,\ldots,a_s$, namely it
satisfies $p(a_j)=\log a_j$ for every $j\in[s]$. 
Thus $\log_\circ K=p_\circ(K)\in\mathcal B$ by entrywise product closure and $J\in\mathcal B$. \cref{lem:homogeneous-space} gives a common row and column sum $b$ for $\log_\circ K$. 
Similar to the proof of \cref{lem:normalize_rumsum_colsum_to_0},
we have
\[
P\log_\circ(K)P=\log_\circ K-\frac bnJ\in\mathcal B.
\]
Taking real spans proves the inclusion.
\end{proof}

\begin{definition}[Kernel, support and unit]\label{def:algebra-support}
We define $$\mathcal{N}:=\bigcap_{C\in \mathcal{A}}\ker C\subseteq \R^n, \quad \text{and}\quad \mathcal{S}:=\mathcal{N}^{\perp}\subseteq \R^n$$
$\mathcal{S}$ is the orthogonal complement of $\mathcal{N}$ in $\R^n.$
We call $\mathcal{N}$ and $\mathcal{S}$ the \emph{kernel} and \emph{support} of $\mathcal{A}$, respectively.
Let $\mathcal{W}$ be the orthogonal complement of $\mathrm{Span}_\mathbb{R}\{\mathbf{1}\}$ in $\mathcal{N}.$
Let $E\in \mathbb{R}^{n\times n}$ and $V\in \R^{n\times n}$ be the orthogonal projections from $\R^n$ onto $\mathcal{S}$ and $\mathcal{W}$, respetively.
%and let $\mathcal W=\ker U\cap\ker E$. Thus $\mathcal W$ is the orthogonal complement of the constant direction and the support of $\mathcal A$.
\end{definition}

Notice that the all-one vector $\mathbf{1}\in \mathcal{N}$, since every matrix in $\mathcal{A}$ has zero row sum by \cref{lem:normalize_rumsum_colsum_to_0}.
So $\mathcal{W}$ is well defined. 
Hence, $\mathcal{N}=\mathrm{Span}_{\mathbb R}\{\mathbf{1}\}\oplus \mathcal{W}$, $\R^n=\mathcal{S}\oplus\mathrm{Span}_{\mathbb R}\{\mathbf{1}\}\oplus \mathcal{W}$ as direct sums of pairwise orthogonal subspaces, and $I=E+U+V.$
We next show that $E$ is a unit in the algebra $\mathcal{A}:$

\begin{lemma}\label{lem:support-identity} 
$E\in\mathcal A$ and $EC=CE=C$ for every $C\in\mathcal A$.
\end{lemma}
\begin{proof}
Choose a basis $\{C_1,\ldots,C_d\}$ of $\mathcal A$ and put
$A=\sum_{j=1}^s C_j^{\tt T}C_j$.
Since $\mathcal{A}$ is transpose closed matrix algebra by \cref{lem:universal-algebra}, $A\in \mathcal{A}.$
For each $v\in\mathbb R^n$,
\begin{equation}\label{eq: v transpose S v}
    v^{\tt T}Av=\sum_{j=1}^d\|C_jv\|_2^2.
\end{equation}
So $v\in \ker A\implies v^{\tt T}Av=0\implies C_jv=0 \text{ for every }j\in[d] \implies v\in \mathcal{N}$.
%%% JYC I edited S to A. when you changd notation from S to A in \eqref{eq: v transpose S v}
The converse is also true since $A\in \mathcal{A}$.
Thus, $\ker A=\mathcal{N}$. 
The equation \eqref{eq: v transpose S v} shows that $A$ is positive semi-definite.  Since $A\mathbf{1}=0$ by \cref{lem:normalize_rumsum_colsum_to_0}, $0$ is an eigenvalue of $A.$
Let $\lambda_0 = 0,\lambda_1,\ldots,\lambda_m$ be the distinct eigenvalues of $A$.
We have $\lambda_1,\ldots,\lambda_m>0$. 
By Lagrange interpolation, there exists a polynomial
$p(t)\in \R[t]$ such that $p(0)=0$ and $p(\lambda_i)=1\,(\forall i\in [m])$.
The polynomial $p(t)$ has 0 constant term.
Since $\mathcal{A}$ is a matrix algebra by \cref{lem:universal-algebra}, $p(A)\in\mathcal A$. 

We claim that $p(A)=E.$
Let $E_{\lambda_i}=\ker(\lambda_i I-A)\subseteq \R^n$, for every $0\le i\le m.$
By \cref{thm:spectral-theorem}, $\R^n=E_{0}\oplus E_{\lambda_1}\oplus \cdots \oplus E_{\lambda_m}=\mathcal{N}\oplus E_{\lambda_1}\oplus \cdots \oplus E_{\lambda_m}.$
For any $v\in \mathcal{N}$, $p(A)v=p(0)v=0=Ev$.
For any $i\in[m]$ and $v\in E_{\lambda_i}$, $p(A)v=p(\lambda_i)v=v=Ev.$
Thus, $p(A)$ and $E$ act identically on $\mathbb{R}^n$. 
Therefore, $E=p(A)\in \mathcal{A}$.

Fix $C\in \mathcal{A}$, we next show that $CE=EC=C$.
We only need to show that $\mathrm{im}\,C\subseteq \mathcal{S}.$
If so, for every $v\in {\R^n}$, write $v=v_1+v_2$ where $v_1\in \mathcal{N}$ and $v_2\in \mathcal{S}$.
Then $Cv=Cv_2$, $CEv=Cv_2$ and $ECv=ECv_2=C{v_2}$, and thus $C=CE=EC.$
To see $\mathrm{im}\,C\subseteq S=\mathcal{N}^\perp$, pick an arbitrary $x\in \R^n$ and $y\in \mathcal{N}.$
We have $\langle{Cx,y}\rangle=\langle x,C^{\tt T }y\rangle=\langle x,0\rangle=0$, since $C^{\tt T}\in \mathcal{A}$. 
This shows $\mathrm{im}\,C\subseteq \mathcal{N}^\perp =\mathcal{S}.$
\end{proof}

\subsection{Exponentiating Back}

For every $C\in \mathcal{A}$, $\exp_\circ C$ is realizable by \cref{lem:algebra-exponentials}.
The homogeneous setting (\cref{lem:homogeneous}) allows us to denote the constant row sum of $\exp_\circ C$ by $r(C)$. This equals the
constant column sum as shown in~\cref{lem:normalize_rumsum_colsum_to_0}.

\begin{definition}\label{def:stochastic-family}
For every $C\in\mathcal A$, let $r(C)$ be the constant row sum of $\exp_\circ C$.
Define
\[
K(C)=\frac{\exp_\circ C}{r(C)}.
\]
\end{definition}

\begin{lemma}\label{lem:stochastic-transpose}
For every $C\in \mathcal{A}$,
$K(C)^{\tt T}=K(C^{\tt T})$.
\end{lemma}
\begin{proof}
The matrices $\exp_\circ C$ and $\exp_\circ(C^{\tt T})$ have the same common row and column sum, so $r(C^{\tt T})=r(C)$. 
Transposition therefore commutes with the normalization in \cref{def:stochastic-family}.
\end{proof}

\begin{lemma}\label{lem:normalized-signatures}
For every $C\in\mathcal{A}$, $K(C)$ is an entrywise positive realizable doubly stochastic matrix.
\end{lemma}
\begin{proof}
It is clear that $K(C)$ is entrywise positive. 
By \cref{lem:algebra-exponentials}, $\exp_\circ C$ is realizable.
Normalizing by $r(C)>0$ does not affect realizability.
By \cref{lem:homogeneous}, $\exp_\circ C$ has constant row sum and constant column sum.
The two constants are both equal to $r(C).$
Therefore, $K(C)$ has constant row sum and constant column sum both equal to $1,$ i.e., $K(C)$ is a doubly stochastic matrix.
\end{proof}

\begin{lemma}\label{lem:stochastic-family}
For all $C,D\in\mathcal A$, we have $K(C)K(D)=K(H)$ for $H=P\log_\circ(K(C)K(D))P\in\mathcal A$.
%Indeed, $H=P\log_\circ(K(C)K(D))P.$
\end{lemma}
\begin{proof}
Let $Q=K(C)K(D)$, $H=P\log_\circ (Q) P.$
By \cref{lem:normalized-signatures}, both $K(C)$ and $K(D)$ are realizable and entrywise positive.
Hence $Q=K(C)K(D)\in \R_{>0}^{n\times n}$ is realizable by \cref{lem:stretching}.
Thus, $Q\in \mathcal{R}$ and $H=P\log_\circ (Q)P\in \mathcal{A}.$
Next we show $K(C)K(D)=K(H).$
By \cref{lem:normalize_rumsum_colsum_to_0}, $H=\log_\circ Q-\mu J$, where $\mu$ is the average of $\log_\circ Q.$
Then $\exp_\circ H=e^{-\mu} Q=e^{-\mu}K(C)K(D)$.
Since $K(C)$ and $K(D)$ are doubly stochastic by \cref{lem:normalized-signatures}, so is $Q=K(C)K(D)$.
In particular, $Q$ has constant row sum 1.
Hence, $\exp_\circ H = e^{-\mu} Q$ has constant row sum $e^{-\mu}$, which is by definition $r(H)$.
Hence, $K(H)=\exp_\circ H/r(H) = Q = K(C)K(D)$.
%Since $K(H)=\exp_\circ H/r(H)$ also has constant row sum 1, we know $r(H)=e^{-\mu}$ and $K(H)=K(C)K(D).$
\end{proof}

\begin{lemma}\label{lem:stochastic-derivative}
$K(0)=U$ and, for every $C\in\mathcal A$,
\begin{equation}\label{eq: KtC derivative}
    \left.\frac{d}{dt}K(tC)\right|_{t=0}=\frac{C}{n}.
\end{equation}
\end{lemma}
\begin{proof}
For the all-0 matrix 0, we have $\exp_\circ(0) = J$, which has row sum $n$.
Hence $r(0)=n$ and $K(0)=\exp_\circ(0)/r(0)=J/n=U.$ 
Fix $C\in \mathcal{A}$.
By definition, $K(tC)_{ij}=e^{tC_{ij}}/r(tC)$ for every $t\in \R$ and $i,j\in [n].$
We have
$r(tC)=\sum_{j=1}^ne^{tC_{ij}}$ for arbitrary $i\in [n]$.
Then $r(tC)=\sum_{j=1}^n (1+tC_{ij}+O(t^2))=n+O(t^2)$ as $t\to 0$, since $C$ has zero row sum by \cref{lem:normalize_rumsum_colsum_to_0}.
Thus for every $i,j\in [n]$, $$K(tC)_{ij}=\frac{e^{tC_{ij}}}{r(tC)}=\frac{1+tC_{ij}+O(t^2)}{n+O(t^2)}=\frac{1}{n}+\frac{tC_{ij}}{n}+O(t^2).$$
This proves \eqref{eq: KtC derivative}.
\end{proof}

\begin{lemma}\label{lem:algebra-invariance}
For all $C,D\in\mathcal A$, $K(C)D\in \mathcal{A}$ and $DK(C)\in\mathcal A.$ 
\end{lemma}
\begin{proof}
Set $Q(t)=K(C)K(tD)$ and $H(t)=P\log_\circ( Q(t)) P$ for every $t\in \R.$ 
%Since $C,tD\in \mathcal{A}$ for every $t\in \R$, \cref{lem:normalized-signatures} shows that both $K(C)$ and $K(tD)$ are realizable.
%Hence $Q(t)=K(C)K(tD)$ is realizable by \cref{lem:stretching}.
%Also $Q(t)$ is entry
Similar to the proof of \cref{lem:stochastic-family}, we can show that $Q(t)\in \mathcal{R}$ and $H(t)\in \mathcal{A}$ for every $t\in \R.$
Denote by $Q'(t)$ the entrywise derivative matrix whose $ij$ entry is $\frac{d}{dt}
(Q(t)_{ij})$, for all $i,j\in [n]$. 
Similarly for $H'(t)$. 
Then 
\[H'(t)_{ij}= \sum_{k, \ell \in [n]} P_{ik} \frac{d}{dt} (\log_\circ( Q(t)))_{k\ell} P_{\ell j}
= \sum_{k, \ell \in [n]} P_{ik} (\log ( Q(t)_{k\ell}))' P_{\ell j}
= \sum_{k, \ell \in [n]} P_{ik} \frac{Q'(t)_{k\ell}}{Q(t)_{k\ell}} P_{\ell j}\]
We have 
\[
H'(t) =P {\tilde Q(t)} P,\qquad H'(0)=P{\tilde Q(0)} P.
\]
where ${\tilde Q(t)}_{ij}= \frac{Q'(t)_{ij}}{Q(t)_{ij}}$.
\cref{lem:stochastic-derivative} gives $Q(0)=K(C)U$ and $Q'(0)=K(C)D/n$. 
Since $K(C)$ has row sum 1 by \cref{lem:normalized-signatures}, $K(C)U=U.$
Therefore ${\tilde Q(0)} = n Q'(0) = K(C)D$.
Thus,
\[
H'(0)=PK(C)DP.
\]
Since $K(C)$ is doubly stochastic, $PK(C)=K(C)P$.
Since $D$ has zero row sum and column sum by \cref{lem:normalize_rumsum_colsum_to_0}, $PDP=D.$
Thus, $H'(0)=K(C)PDP=K(C)D.$
For every nonzero $t\in \R$, ${(H(t)-H(0))}/{t}\in \mathcal{A}$ since $\mathcal{A}$ is a vector space.
Then $H'(0)=\lim_{t\to 0}{(H(t)-H(0))}/{t}\in \mathcal{A}$, since $\mathcal{A}$ is closed as a metric space.
Therefore, $K(C)D\in \mathcal{A}.$
Similarly, we can prove $DK(C)\in \mathcal{A}$ by defining $Q(t)=K(tD)K(C)$.
\end{proof}

\begin{corollary}\label{cor: EK=KE=EKE}
   For every $C\in \mathcal{A}$, \(K(C)E=EK(C)E=EK(C)\in \mathcal{A}.\)
\end{corollary}
\begin{proof}
    Taking $D=E$ in \cref{lem:algebra-invariance} gives $K(C)E\in \mathcal{A}$ and $EK(C)\in \mathcal{A}.$
    By \cref{lem:support-identity}, $E$ is a multiplicative unit in
    $\mathcal{A}.$ So,
    $EK(C)E=K(C)E=EK(C).$
\end{proof}

Recall that $U=J/n$ is the orthogonal projection from $\R^n$ onto
$\mathrm{Span}_{\R}\{\mathbf{1}\}$. 
For every $C\in \mathcal{A}$,
we have $K(C)U=UK(C)=U$ since $K(C)$ is doubly stochastic. 
Together with $K(C)E=EK(C)$ from \cref{cor: EK=KE=EKE},
this shows that $K(C)$ commutes with the three pairwise orthogonal projections
$U$, $E$, and $V = I  -U - E$. 
Therefore, $K(C)$ preserves each summand in the
orthogonal decomposition
\(
\R^n=\mathrm{Span}_{\R}\{\mathbf{1}\}\oplus\mathcal S\oplus\mathcal W.
\)
In an orthonormal basis  adapted to this decomposition, the projections have the forms
\[
U=\begin{pmatrix}
1&0&0\\
0&0&0\\
0&0&0
\end{pmatrix},
\qquad
E=\begin{pmatrix}
0&0&0\\
0&I_{\mathcal S}&0\\
0&0&0
\end{pmatrix},
\qquad
V=\begin{pmatrix}
0&0&0\\
0&0&0\\
0&0&I_{\mathcal W}
\end{pmatrix},
\]
and $K(C)$ is block diagonal:
\[
K(C)=\begin{pmatrix}
1&0&0\\
0&K(C)|_{\mathcal S}&0\\
0&0&K(C)|_{\mathcal W}
\end{pmatrix}.
\]
\section{Polynomial Homomorphism on the Matrix Algebra}\label{sec:polyhom}

The normalized entrywise exponential map $K(\cdot)$ is not a
homomorphism: in general, $K(CD)\neq K(C)K(D)$ for $C,D\in\mathcal A$.
However, by pulling back through the $\mathcal S$-block of $K(\cdot)$
and then applying $K(\cdot)$, we obtain a multiplicative homomorphism $F(\cdot)$ (\cref{lem: local homomorphism,lem: homomorphism property}).
This map is first defined in a neighborhood of $0$ and then
analytically extended to the entire vector space $\mathcal A$ (\cref{def: local F,def: polynomial homomorphism}).
It is still block diagonal under any orthonormal basis respecting  the orthogonal decomposition
\(
\R^n=\mathrm{Span}_{\R}\{\mathbf{1}\}\oplus\mathcal S\oplus\mathcal W
\)
(\cref{lem:block decomposition}).
Moreover, this process does not substantially complicate the form
of $F(\cdot)$. 
In particular, after choosing a basis for the real
vector space $\mathcal A$, we show that $F(X)$ is a multivariate
polynomial in the coordinates of $X\in\mathcal A$
(\cref{lem: hkE mutually orth proj}).
These are results in \cref{subsec:def_of_polyhom}.

In \cref{subsec: when is F realizable}, we specify a domain $\operatorname{int}\mathcal{C}$ in $\mathcal{A}$ in which $F(X)$ is realizable.
In particular, we show that $F(tE)$ is realizable for every $t\in (0,1)$ (\cref{lem:polytope-positivity}).
\cref{sec:dichotomy} uses this to prove \#P-hardness of $\PlGH(M)$, when $M$ is not an Ising tensor.
The domain $\operatorname{int}\mathcal C$ is defined as the interior in
$\mathcal A$ of the convex hull of the ``statistic matrices" $T_{xy}$ (\cref{def:statistic-polytope}),
each equal to $n$ times the orthogonal projection of the matrix unit
$\mathbf e_x\mathbf e_y^{\tt T}$ onto $\mathcal A$ (\cref{def:entry-statistics}).
In \cref{lem:row-normalization,lem:mean-identities}, we show that
each row of $F(X)$ sums to one and that the corresponding weighted
mean of the statistic matrices $T_{xy}$ is $X$. These two identities
are crucial to proving the realizability of $F(X)$ for
$X\in \operatorname{int}\mathcal C$ and establishing
the affine factorization of $F(X)_{xy}$ in
\cref{subsec: affine factorization}.

In \cref{subsec: affine factorization},
we show that every entry of $F(X)$ can be factored into affine forms, possibly with multiplicities (\cref{lem:affine-factors}).
Further, we show that all entries $F(X)_{xy}$ have common degree $d$ (\cref{lem:common-degree}).
Finally, we show that taking transpose on the input $X$ 
%affine factorization 
induces an involution map $\tau$ on the set of irreducible affine linear factors (\cref{lem:factor-transpose}).

The homomorphism property (\cref{lem: homomorphism property}) and the affine factorization (\cref{lem:affine-factors}) are the main results of this section. The following sections rely heavily on these two properties to show that the integer exponents in the factorization induce graph distances.
%Also, the multiplicity of the affine form $\ell_a$ in $F(X)_{xy}$ is equal to the $\ell_a(T_{xy})$
%The homomorphism property (\cref{lem: homomorphism property}) and the affine factorization are main conclusions in this section.

\subsection{Definition of the Polynomial Homomorphism}\label{subsec:def_of_polyhom}

\begin{definition}\label{def:projected-coordinate}
Define $\Phi:\mathcal A\to\mathcal A$ by
\[
\Phi(C)=EK(C)E\qquad(C\in\mathcal A).
\]
\cref{cor: EK=KE=EKE} ensures that its values lie in $\mathcal A$.
\end{definition}

\begin{lemma}\label{lem:local-coordinate}
$\Phi$ has a real analytic inverse $\Phi^{-1}$ between neighborhoods of $0$ in $\mathcal A$.
\end{lemma}
\begin{proof}
First note that $\Phi(0)=EK(0)E=EUE=0$ by \cref{lem:stochastic-derivative} and the fact that $E$ and $U$ are mutually orthogonal projections.
We next compute the derivative of $\Phi$ at $0.$
Fix an arbitrary $C\in \mathcal{A}.$
By \cref{lem:stochastic-derivative}, $$\left.\frac{d}{dt}K(tC)\right|_{t=0}=\frac{C}{n}.$$
Consequently, 
$$\left.\frac{d}{dt}\Phi(tC)\right|_{t=0}=E\left.\frac{d}{dt}K(tC)\right|_{t=0}E=\frac{ECE}{n}=\frac{C}{n}.$$
The last equality holds because $E$ is unit in $\mathcal{A}$ by \cref{lem:support-identity}.
%Also $\Phi(0)=EUE=0$. By Lemma~\ref{lem:stochastic-derivative}, its derivative sends $D\in\mathcal A$ to $EDE/n=D/n$.
Thus, $d\Phi(0)(C)=\frac{C}{n}$ for every $C\in \mathcal{A}$.
This is invertible, so the inverse function theorem (\cref{thm:analytic-inverse}) gives the local analytic inverse.
\end{proof}

\begin{definition}\label{def: local F}
Let $\Phi^{-1}:\mathcal{O}\to \mathcal{O}$ be the analytic inverse of $\Phi$ supplied by \cref{lem:local-coordinate}, where $\mathcal{O}\subseteq \mathcal{A}$ is a neighborhood of $0$ in $\mathcal{A}.$ 
Define
$$F:\mathcal{O}\to \mathcal{B},\qquad X\mapsto K(\Phi^{-1}(X)).$$
\end{definition}
Note that for $X \in \mathcal{O}$, by \cref{lem:normalized-signatures}, $K(\Phi^{-1}(X))$ is  realizable,
and thus $F(X) \in \mathcal B=\operatorname{span}_{\mathbb R}\mathcal R$.

%By definition, it is not yet clear that $F(\cdot)$ is a polynomial
\begin{lemma}\label{lem: F0 and dF0}
    $F$ is real analytic.
    $F(0)=U$ and $dF(0)=\mathrm{id}.$
\end{lemma}
\begin{proof}
    Both $\Phi^{-1}$ and $K$ are  real analytic, so is $F.$
    Since $\Phi(0)=0$,
    $F(0)=K(\Phi^{-1}(0))=K(0)=U.$
    By the chain rule and  inverse function theorem (\cref{thm:analytic-inverse}), $dF(0)=dK(\Phi^{-1}(0))\circ d\Phi^{-1}(0)=dK(0)\circ (d\Phi(0))^{-1}.$
    Here ``$\circ$'' means composition of linear maps instead of entrywise product.
    In \cref{lem:local-coordinate}, we proved $d\Phi(0)(C)=\frac{C}{n}$.
    So $(d\Phi(0))^{-1}(C)=nC.$
    By \cref{lem:stochastic-derivative}, $dK(0)(C)=\frac{C}{n}.$
    Thus, $dK(0)\circ (d\Phi(0))^{-1}(C)=\frac{nC}{n}=C$, for every $C\in \mathcal{O}.$
    Therefore, $dF(0)=\mathrm{id}.$  
\end{proof}

We view $F(\cdot)$ as a function of vector coordinates on $\mathcal A$.
More precisely, choose an arbitrary basis $\{C_1,\ldots,C_m\}$ of
$\mathcal A$ and write each $X\in\mathcal A$ as
$X=\sum_{i=1}^m x_iC_i$. We then regard
$F(X)=F(x_1,\ldots,x_m)$ as a multivariate analytic function of the
coordinates $(x_1,\ldots,x_m)$ near 0. 
Since $F(0)=U$ and $dF(0)=\mathrm{id}$, we can write the Taylor expansion (\cref{thm:taylor-expansion}) of $F$ at $0$ as
\[
F(X)=U+X+h(X)=U+X+\sum_{k\geq2}h_k(X),
\]
where
\[
h_k(X)=\frac{1}{k!}d^kF(0) (X,\ldots,X)
\]
is a homogeneous multivariate polynomial (in the coordinates of $X$) of degree $k$.

The following lemma shows that the three summands $U,X,$ and $h(X)$ in $F(X)$ respect the three subspaces $\mathrm{Span}_{\mathbb{R}}\{\mathbf{1}\}$, $\mathcal{S}$, and $\mathcal{W}$, respectively. 
\begin{lemma}[Block decomposition]\label{lem:block decomposition}
    For every $X\in\mathcal{O}$, write $F(X)=U+X+h(X)$.
    Then
    \begin{enumerate}
        \item\label{item: block decomposition 1} $UF(X)U=UF(X)=F(X)U=U$;
        \item \label{item: block decomposition 2} $EF(X)E=EF(X)=F(X)E=X$;
        \item \label{item: block decomposition 3} $VF(X)V=VF(X)=F(X)V=h(X).$
    \end{enumerate}
    Thus, $F(X)=U+X+h(X)$ is the decomposition of the operator $F(X)$ into its actions
    in the pairwise orthogonal subspaces $\mathrm{Span}_{\mathbb{R}}\{\mathbf{1}\}$, $\mathcal{S}$, and $\mathcal{W}$.
\end{lemma}
\begin{proof}
    By \cref{def: local F} and \cref{lem:stochastic-family}, $F(X)=K(\Phi^{-1}(X))$ is doubly stochastic. 
    Hence $F(X)U=U$, $UF(X)=U$ and $UF(X)U= U$.
    By \cref{cor: EK=KE=EKE}, $EF(X)=F(X)E=EF(X)E$.
    Indeed, $EF(X)E=EK(\Phi^{-1}(X))E=\Phi(\Phi^{-1}(X))=X.$
    Finally, $U+X+h(X)=F(X)=(U+E+V)F(X)=U+X+VF(X)$, so $h(X)=VF(X)$.
    Similarly, $VF(X)V=F(X)V=h(X).$
\end{proof}
In particular, $EU = 0$ and $E h(X) =0$. Also $EX = XE = X$.

Intuitively, in a basis adapted to the decomposition $\R^n=\mathrm{Span}_{\mathbb{R}}\{\mathbf{1}\}\oplus\mathcal{S}\oplus\mathcal{W}$, we have
\[
U=\begin{pmatrix}
1&0&0\\
0&0&0\\
0&0&0
\end{pmatrix},
\qquad
X=\begin{pmatrix}
0&0&0\\
0&X_{\mathcal S}&0\\
0&0&0
\end{pmatrix},
\qquad
h(X)=\begin{pmatrix}
0&0&0\\
0&0&0\\
0&0&h(X)_{\mathcal W}
\end{pmatrix},
\]
where $X_{\mathcal S}$ and $h(X)_{\mathcal W}$ are the restrictions of $F(X)$
on the subspaces  $\mathcal{S}$ and $\mathcal{W}$, respectively, and $F(X)$ is block diagonal:
\[
F(X)=\begin{pmatrix}
1&0&0\\
0&X_{\mathcal S}&0\\
0&0&h(X)_{\mathcal W}
\end{pmatrix}.
\]
The following lemma shows that $F(\cdot)$ is a homomorphism near $0.$ 
\begin{lemma}[Local homomorphism]\label{lem: local homomorphism}
$F(X)F(Y)=F(XY)$ for all $X,Y$ near 0 in $\mathcal{A}$.
\end{lemma}
\begin{proof}
For $X,Y\in \mathcal{O}$, let $C=\Phi^{-1}(X)$ and $D=\Phi^{-1}(Y)$.
\cref{lem:stochastic-family} gives
$K(C)K(D)=K(H)$ for $H=P\log_\circ (K(C)K(D))P\in \mathcal{A}$.
When $X,Y$ tend to zero, $H$ tends to $P\log_\circ (K(0)K(0))P=P\log_\circ (U)P=0.$
Note that $$K(C)K(D)=K(\Phi^{-1}(X))K(\Phi^{-1}(Y))=F(X)F(Y).$$
So
\[
\Phi(H)=EK(H)E=EK(C)K(D)E=EF(X)F(Y)E.
\]
By \cref{lem:block decomposition}, in view of the block decomposition,
\[
EF(X)F(Y)E=E(U+X+h(X))(U+Y+h(Y))E=E(U+XY+h(X)h(Y))E=XY.
\]
Thus, $\Phi(H)=XY.$
Without loss of generality, we may assume that $XY,H\in\mathcal{O}$, possibly after the neighborhood $\mathcal{O}$ is reduced. 
Then $H=\Phi^{-1}(XY)$, and thus
\[
F(XY)=K(\Phi^{-1}(XY))=K(H)=K(C)K(D)=F(X)F(Y).
\]
The lemma is proved.
\end{proof}

We have proved that $F(\cdot)$ is a local multiplicative homomorphism. 
Next, we will show that $F(\cdot)$ is a polynomial near $0$, which allows us to extend $F(\cdot)$ from a neighborhood of $0$ to the entire space $\mathcal{A}$ while preserving the homomorphism property. 

Recall $h(X)=\sum_{k=2}^{+\infty}h_k(X)$, where $h_k(X)$ is a multivariate homogeneous polynomial of degree $k$ in the coordinates of $X$.
Therefore, $h_k(X)$ is well-defined on the entire vector space $\mathcal{A}.$
To establish that $F(\cdot)$ is a polynomial, we will demonstrate that $\{h_k(E)\}_{k \ge 2}$ are mutually orthogonal projections (\cref{lem: hkE mutually orth proj}), thereby ensuring that the Taylor expansion of $F$ at $0$ terminates after finitely many terms.

\begin{lemma}\label{lem:homogeneous-products}
For all $X,Y\in\mathcal A$ and all integers $k,l\geq2$, we have
\(
h_k(XY)=h_k(X)h_k(Y)\)
and
\(h_k(X)h_l(Y)=0\) if $k\ne l$.
\end{lemma}
\begin{proof}
For $X,Y$ sufficiently close to zero, \cref{lem: local homomorphism} gives $F(X)F(Y)=F(XY).$
By block decomposition (\cref{lem:block decomposition}), $F(X)F(Y)=(U+X+h(X))(U+Y+h(Y))=U+XY+h(X)h(Y).$
Thus, $F(XY)=U+XY+h(XY)=U+XY+h(X)h(Y)$.
It follows that $h(XY)=h(X)h(Y)$.

Fix arbitrary $X,Y\in\mathcal A$ and substitute $sX,tY$ for sufficiently small real $s,t>0$. 
We obtain \(h(stXY)=h(sX)h(tY)\).
Because $h(\cdot)=\sum_{k=2}^{+\infty}h_k(\cdot)$, where $h_k(\cdot)$ is a homogeneous polynomial of degree $k$, we have
\[
\sum_{k=2}^{+\infty}s^kt^k h_k(XY)=h(stXY)=h(sX)h(tY)=\sum_{i=2}^{+\infty}\sum_{j=2}^{+\infty}s^it^j h_i(X)h_j(Y).
\]

This series is absolutely convergent when $s,t$ are sufficiently small.
The coefficient of $s^kt^k$ gives the first identity; that of $s^kt^l$ with $k\ne l$ gives the second. 
These calculations hold for arbitrary fixed $X,Y\in \mathcal{A}$, not necessarily close to 0.
\end{proof}

\begin{lemma}\label{lem:homogeneous-transpose}
For every $X\in\mathcal A$ and every integer $k\geq2$, $h_k(X^{\tt T})=h_k(X)^{\tt T}.$ 
\end{lemma}

\begin{proof}
\cref{lem:stochastic-transpose} and $E^{\tt T}=E$ imply $\Phi(C^{\tt T})=\Phi(C)^{\tt T}$ for every $C\in\mathcal A$. 
Fix an arbitrary $X\in \mathcal{A}$ near 0, let $C=\Phi^{-1}(X).$
Then $\Phi(C)=X$ and $\Phi(C^{\tt T})=X^{\tt T}.$
Hence $\Phi^{-1}(X^{\tt T})=C^{\tt T}=(\Phi^{-1}(X))^{\tt T}.$
By \cref{lem:stochastic-transpose}, $K(\Phi^{-1}(X^{\tt T}))=K((\Phi^{-1}(X))^{\tt T})=K(\Phi^{-1}(X))^{\tt T}.$
Thus, $F(X^{\tt T})=F(X)^{\tt T}.$
This implies $h(X^{\tt T})=h(X)^{\tt T}$, for every $X\in \mathcal{A}$ near 0.

Now fix an arbitrary $X\in\mathcal A$ (not necessarily near 0), and substitute $sX$ for sufficiently small $s>0.$
We have $h(sX^{\tt T})=h(sX)^{\tt T}$, for every $X\in \mathcal{A}$ and sufficiently small $s.$
Then
\[
\sum_{k=2}^{\infty}s^k h_k(X^{\tt T})
=\sum_{k=2}^{\infty}s^k h_k(X)^{\tt T}.
\]
Comparing coefficients of $s^k$ gives $h_k(X^{\tt T})=h_k(X)^{\tt T}.$
\end{proof}

\begin{lemma}\label{lem: hkE mutually orth proj}
The nonzero matrices among $\{h_k(E)\}_{k\geq2}$ are mutually orthogonal projections.
%, and\[h_k(X)=e_kh_k(X)e_k\quad\text{for every }X\in\mathcal A\text{ and integer }k\geq2.\]
In particular, only finitely many $h_k$ are nonzero.
\end{lemma}
\begin{proof}
Since $E^2=E=E^{\tt T}$, \cref{lem:homogeneous-products,lem:homogeneous-transpose}  give
$h_k(E)^2=h_k(E)=h_k(E)^{\tt T}$ and $h_k(E)h_l(E)=0$ for distinct $k,l\geq2$. 
So the nonzero matrices among $(h_k(E))_{k\geq2}$ are mutually orthogonal projections.
Consequently, there are finitely many nonzero matrices among $\{h_k(E)\}_{k\geq2}$.
%All $h_k$ are supported on $\mathcal W$, hence so are these projections. 
Using $EX=X$ in the same lemma gives
\(h_k(X)=h_k(E)h_k(X).\)
Thus $h_k\ne0$ implies $h_k(E)\ne0$. 
So there are finitely many nonzero $h_k.$
\end{proof}

\begin{definition}[Polynomial homomorphism on $\mathcal{A}$]\label{def: polynomial homomorphism}
    Let $d=\max\{\{1\}\cup\{k\mid h_k(E)\neq 0\}\}$.
    Then $d<+\infty$ by \cref{lem: hkE mutually orth proj}.
    Define the polynomial (in coordinates of $X$)
    $$F:\mathcal{A}\to \R^{n\times n},\qquad X\mapsto U+X+\sum_{k=2}^dh_k(X).$$
\end{definition}

By definition, $F(\cdot)$ is a polynomial, and it agrees with \cref{def: local F} near 0.
Next, we extend the local properties of $F(\cdot)$ to the entire space $\mathcal{A}.$

\begin{lemma}\label{lem: block decomposition global}
    The block decomposition in 
    \cref{lem:block decomposition} holds for every $X\in \mathcal{A}.$
\end{lemma}
\begin{proof}
    Fix $X\in \mathcal{A}.$
    Each entry of $UF(sX)U-U$ is an analytic function in $s$ vanishing in a small neighborhood of $0$.
    The analytic identity theorem (\cref{thm:analytic-identity}) implies $UF(sX)U-U=0$ for every $s\in \R$.
    In particular, letting $s=1$ proves $UF(X)U=U$.
    The other identities follow similarly.
\end{proof}

\begin{lemma}\label{lem: homomorphism property}
    For every $X,Y\in \mathcal{A}$, $F(X)F(Y)=F(XY).$
\end{lemma}
\begin{proof}
    Fix $X,Y\in \mathcal{A}.$
    For sufficiently small $s>0$, \cref{lem: local homomorphism} gives $F(sX)F(sY)=F(s^2 XY).$
    Each entry of $g(s)=F(sX)F(sY)-F(s^2XY)$ is an analytic function vanishing near $0.$
    By the analytic identity theorem (\cref{thm:analytic-identity}), $g(s)\equiv 0$ for $s\in \R$. 
    In particular, $g(1)=0$.
    That is, $F(X)F(Y)=F(XY).$
\end{proof}

\begin{lemma}\label{lem:polynomial-agreement}
$K(C)=F(\Phi(C))$ for every $C\in\mathcal A$.
\end{lemma}

\begin{proof}
Fix $C\in\mathcal A$. 
Each entry of
$K(sC)-F(\Phi(sC))$ is real analytic for every $s\in\mathbb R$.
The difference vanishes for $s$ near zero by \cref{def: local F}.
\cref{thm:analytic-identity} makes it zero for every real $s$, in particular for $s=1$.
\end{proof}

\begin{lemma}\label{lem:polynomial-in-span}
$F(X)\in\mathcal B$ for every $X\in\mathcal A$.
\end{lemma}
\begin{proof}
Choose a basis $\{C_1,\ldots,C_m\}$ of $\mathcal A$.
For $X\mathcal{\in \mathcal{A}}$, write $X=\sum_{i=1}^mx_iC_i.$
Since $F$ is a polynomial in $(x_1,\ldots, x_m)$, we have
\(
F(X)=F(x_1,\ldots, x_m)=\sum_{\alpha\in\mathbb Z_{\geq0}^m}x^\alpha M_\alpha\),
where
$x^\alpha=\prod_{i=1}^d x_i^{\alpha_i}$ and $M_\alpha$ are coefficient matrices.
Only finitely many $M_{\alpha}$
are nonzero. 
We will show that every $M_\alpha$ belongs to $\mathcal B$.
On a neighborhood of zero, $F(X)=K(\Phi^{-1}(X))\in\mathcal B$ by \cref{lem:normalized-signatures}. 
Each coefficient matrix satisfies $M_\alpha=\frac{1}{\alpha!}\partial^{\alpha}F(0)$, where $\alpha!=\prod_{i=1}^m\alpha_i!$ and $\partial^{\alpha}=\partial_{1}^{\alpha_1}\cdots \partial_{d}^{\alpha_d}$. 
Therefore, each $M_\alpha$ is a limit of finite linear combinations of matrices in $\mathcal{B}$.
Since $\mathcal B$ is a finite-dimensional
linear subspace of $\mathbb R^{n\times n}$, it is closed. 
Hence $M_\alpha \in \mathcal{B}$ for every $\alpha.$
Finally,
\(
F(X)=\sum_\alpha x^\alpha M_\alpha\in\mathcal B,
\)
as a finite linear combination of matrices in $\mathcal B$.
\end{proof}

\begin{lemma}\label{lem: F transpose}
    For every $X\in \mathcal{A}$, $F(X^{\tt T})=F(X)^{\tt T}$.
\end{lemma}
\begin{proof}
    By \cref{lem:homogeneous-transpose}, $F(X^{\tt T})=U+X^{\tt T}+\sum_{k=2}^d h_k(X^{\tt T})=U+X^{\tt T}+\sum_{k=2}^d h_k(X)^{\tt T}=F(X)^{\tt T}.$
\end{proof}

\begin{lemma}\label{lem:polynomial-identity}
$F(E)=I$.
\end{lemma}
\begin{proof}
By \cref{lem: homomorphism property}, $F(E)F(X)=F(X)$ for every $X\in \mathcal{A}.$
We claim there exists $X\in \mathcal{A}$ such that $F(X)$ is invertible. 
Then multiplying $F(X)^{-1}$ on both sides of $F(E)F(X)=F(X)$ gives $F(E)=I.$
Let $C_0=P\log_\circ(M)P\in\mathcal A$, where $M\in \Sym{n}{F}{\Rp}$ is the originally fixed matrix.
By \cref{lem:normalize_rumsum_colsum_to_0}, there is $\mu\in\mathbb R$ with $C_0=\log_\circ M-\mu J$. 
If $r_M>0$ is the common row sum of $M$, then
\(
K(C_0)=\frac{M}{r_M}.
\)
This matrix is invertible. \cref{lem:polynomial-agreement} implies that $F(\Phi(C_0))=K(C_0)$ is invertible, where $\Phi(C_0)\in \mathcal{A}.$
Therefore $F(E)=I$.
\end{proof}

%%%%%%%JYC here.

\iffalse
For all $X,Y\in\mathcal A$, the established identities can be summarized as
\begin{equation}
\begin{aligned}
F(0)=U,& F(E)=I, &F(X)\in \mathcal{B} \\
F(X)F(Y)&=F(XY),&F(X^{\tt T})&=F(X)^{\tt T}.
\end{aligned}
\label{eq:polynomial-monoid}
\end{equation}
\fi

\subsection{When is \texorpdfstring{$F(\cdot)$}{F(·)} Realizable}\label{subsec: when is F realizable}

We proved that $F$ is a map from $\mathcal{A}$ to $\mathcal{B}$ by \cref{lem:polynomial-in-span}.
However, not every element in $\mathcal{B}$ is realizable.
Nevertheless, \cref{lem:polynomial-agreement} shows that $K(C) = F(\Phi(C))$ for all $C \in \mathcal{A}$, and $K(C)$ is realizable by \cref{lem:normalized-signatures}.
So $F(X)$ is realizable whenever $\Phi^{-1}(X)$ exists.
In this subsection, we will specify a subregion $\mathcal{C}$ of $\mathcal{A}$ in which $\Phi$ has an inverse, and thus $F(X)=K(\Phi^{-1}(X))$ for $X\in \mathcal{C}$ (\cref{lem: Phi has inverse in C}).
In particular, we will prove that $F(tE)$ is realizable for $t\in (0,1)$ (\cref{lem:polytope-positivity}).
%The ultimate \#P-hardness comes from $\PlGH(F(tE))$.

\begin{definition}\label{def:entry-statistics}
Let $\pi_{\mathcal A}:\mathbb R^{n\times n}\to\mathcal A$ be the orthogonal projection onto $\mathcal{A}$. 
For every $x\in[n]$, let $\mathbf e_x\in\mathbb R^n$ be the $x$-th coordinate vector. 
For all $x,y\in[n]$, define
\(
T_{xy}=n\pi_{\mathcal A}(\mathbf e_x\mathbf e_y^{\tt T}).\)
\end{definition}

\begin{lemma}\label{lem:statistic-pairing}
For every $C\in\mathcal A$ and every $x,y\in[n]$, $\langle C,T_{xy}\rangle=nC_{xy}$. 
\end{lemma}
\begin{proof}
Fix $C\in \mathcal{A}.$
Orthogonal projection onto $\mathcal{A}$ preserves inner product with $C$, so
\(
\langle C,T_{xy}\rangle
=n\langle C,\pi_{\mathcal A}(\mathbf e_x\mathbf e_y^{\tt T})\rangle
=n\langle C,\mathbf e_x\mathbf e_y^{\tt T}\rangle=nC_{xy}.
\)
\end{proof}

\begin{lemma}\label{lem:diagonal-statistic}
For all $x,y\in[n]$, $T_{xx}=E$ and $T_{yx}=T_{xy}^{\tt T}$.
\end{lemma}
\begin{proof}
For every $C\in\mathcal A$, by constant diagonal (\cref{lem:normalize_rumsum_colsum_to_0}), \cref{lem:statistic-pairing}, and the fact that $E$ is a multiplicative unit (\cref{lem:support-identity}), we have
\(
\langle C,T_{xx}\rangle=nC_{xx}=\operatorname{Tr}C
=\operatorname{Tr}(C^{\tt T}E)=\langle C,E\rangle.
\)
Since $T_{xx},E\in\mathcal A$, equality against every $C\in\mathcal A$ proves $T_{xx}=E$.
Also, taking transpose is an isometry for the inner product $\langle \cdot, \cdot\rangle$,
and it preserves $\mathcal A$ and its orthogonal complement, thus
it commutes with the projection operation $\pi_{\mathcal A}$. Hence,
$T_{yx}=n\pi_{\mathcal{A}}(\mathbf{e}_y\mathbf{e}_x^{\tt T})=n\pi_{\mathcal{A}}((\mathbf{e}_x\mathbf{e}_y^{\tt T})^{\tt T})=
n(\pi_{\mathcal{A}}(\mathbf{e}_x\mathbf{e}_y^{\tt T}))^{\tt T}= T_{xy}^{\tt T}$.
\end{proof}

\begin{lemma}\label{lem:row-normalization}
For every $X\in\mathcal A$ and every $x\in[n]$,
$\sum_{y=1}^nF(X)_{xy}=1$.
\end{lemma}
\begin{proof}
By \cref{def: local F}, $F(X)=K(\Phi^{-1}(X))$ for $X$ near 0.
So $F(X)$ has row sum one for $X$ near $0.$ by \cref{lem:normalized-signatures}.
Each row sum is a polynomial in the coordinates of $X$ by \cref{def: polynomial homomorphism}, so the lemma is proved by the analytic identity theorem (\cref{thm:analytic-identity}).
\end{proof}

\begin{lemma}\label{lem:mean-identities}
For every $X\in\mathcal A$ and every $x\in[n]$, $\sum_{y=1}^nF(X)_{xy}T_{xy}=X$. 
\end{lemma}

\begin{proof}
Fix arbitrary $C\in\mathcal A, X\in \mathcal{A}$ and $x\in [n]$. 
By \cref{lem:log-space-in-span,lem:polynomial-in-span}, both $C$ and $F(X)$ are in $\mathcal B$.
By \cref{lem:row-isometry,lem:statistic-pairing},
\[
\langle C,F(X)\rangle
=n\sum_{y=1}^nF(X)_{xy}C_{xy}
=\sum_{y=1}^nF(X)_{xy}\langle C,T_{xy}\rangle
=\left\langle C,\sum_{y=1}^nF(X)_{xy}T_{xy}\right\rangle.
\]
On the other hand, we have the block decomposition $F(X)=U+X+h(X)$.
We have $\braket{C,U}=\mathrm{Tr}(C^{\tt T}U)=0$ and $\braket{C,h(X)}=\mathrm{Tr}(C^{\tt T}h(X))=\mathrm{Tr}(C^{\tt T}EVF(X))=0$ by \cref{lem: block decomposition global}.
This gives $\langle C,F(X)\rangle=\langle C,X\rangle$. 
Thus, $\langle C,X \rangle=\left\langle C,\sum_{y=1}^nF(X)_{xy}T_{xy}\right\rangle$ for every $C\in \mathcal{A}$.
This implies $\sum_{y=1}^nF(X)_{xy}T_{xy}=X$. 
\end{proof}

\begin{lemma}\label{lem:statistics-span}
For every $x\in[n]$, $\operatorname{Aff}_{\mathbb{R}}\{T_{x1},\ldots,T_{xn}\}=\mathcal{A}$, and $\sum_{y=1}^n\frac{1}{n}T_{xy}=0.$
\end{lemma}
\begin{proof}
Fix $x\in[n]$.
Take $X=0$ in \cref{lem:mean-identities}, we have $\sum_{y=1}^nF(0)_{xy}T_{xy}=0.$ 
Since $F(0)=U$ by \cref{lem: F0 and dF0},
$\sum_{y=1}^n \frac{1}{n}T_{xy}=0$. 
This implies $0\in \operatorname{Aff}_\R\{T_{x1},\ldots,T_{xn}\}.$
Therefore, $\operatorname{Aff}_\R\{T_{x1},\ldots,T_{xn}\}=\operatorname{Span}_\R\{T_{x1},\ldots,T_{xn}\}.$
Suppose that $\mathrm{Span}_{\mathbb{R}}\{T_{x1},\ldots,T_{xn}\}$ is a proper subspace of $\mathcal{A}$.
Then take a nonzero $C$ in its orthogonal complement in $\mathcal{A}$, we have $\langle C,T_{xy}\rangle=0$ for every $y\in[n]$. 
\cref{lem:statistic-pairing} would make the $x$-th row of $C$ zero.
By \cref{cor:row-determines-matrix}, $C=0$, a contradiction.
Therefore, $\mathrm{Aff}_{\mathbb{R}}\{T_{x1},\ldots,T_{xn}\}=\mathrm{Span}_{\mathbb{R}}\{T_{x1},\ldots,T_{xn}\}=\mathcal{A}$.
\end{proof}

\begin{lemma}\label{lem:row-statistic-multisets}
The multisets $\{T_{xy}\}_{y=1}^n$ agree for all $x\in[n]$, including multiplicities.
\end{lemma}
\begin{proof}
Fix $x,x'\in [n]$ and $x\neq x'.$
Let $\{S_1,\ldots, S_m\}$ be all distinct elements among the multisets $\{T_{xy}\}_{y\in[n]}\cup \{T_{x'y}\}_{y\in[n]}.$
For every $j\in [m]$, define $n_j:=|\{y\in [n]\mid T_{xy}=S_j\}|$ and $n'_j:=|\{y\in [n]\mid T_{x'y}=S_j\}|.$
We prove $n_j=n'_j$ for every $j\in[m].$

Note that for every distinct $i,j\in [m]$, $\mathcal{H}_{ij}:=\{C\in\mathcal{A}\mid \braket{C,S_i-S_j}=0\}$ is a proper subspace of $\mathcal{A}.$
It is well-known that no vector space over an infinite field is a finite union of proper subspaces.
Therefore, there exists $C\in \mathcal{A}\setminus \bigcup_{i,j\in [m],i\neq j}\mathcal{H}_{ij}.$
Define $\lambda_{i}:=\frac{1}{n}\braket{C,S_i},$
then $\lambda_1,\lambda_2,\ldots, \lambda_m$ are multually distinct.

For every $t\in \R$, $tC\in \mathcal{A}.$
The matrix $\exp_\circ(tC)$ is realizable by \cref{lem:algebra-exponentials}.
Hence it has constant row sum by \cref{lem:homogeneous}.
Combined with \cref{lem:statistic-pairing}, we have
\[
\sum_{y=1}^n \exp\left(\frac{t}{n}\braket{C,T_{xy}}\right)=\sum_{y=1}^n \exp(tC_{xy})=\sum_{y=1}^n \exp(tC_{x'y})=\sum_{y=1}^n \exp\left(\frac{t}{n}\braket{C,T_{x'y}}\right)\quad(t\in \R).
\]
Grouping the sum by distinct elements $S_1,\ldots, S_m$, we have
\[
\sum_{j=1}^m (n_j-n_j')\exp\left(\lambda_j t\right)=0\qquad(t\in \R).
\]
Since $\lambda_1,\lambda_2,\ldots, \lambda_m$ are distinct, the only possibility is $n_j=n_j',\,\forall j\in [m].$
Therefore, $\{T_{xy}\}_{y\in [n]}$ and $\{T_{x'y}\}_{y\in [n]}$ agree as multisets.
\end{proof}

\begin{definition}\label{def:statistic-polytope}
Define the convex hull of $\{T_{x1},\ldots, T_{xn}\}$
\[
\mathcal C:=\operatorname{Conv}_\R\{T_{x1},\ldots, T_{xn}\}=\left\{\sum_{y=1}^n a_yT_{xy}\mid a_y\ge 0\text{ and }\sum_{y=1}^n a_y=1\right\}.
\]
using any fixed $x\in[n]$. 
\cref{lem:row-statistic-multisets} makes the choice of $x$ immaterial. 
Define $\operatorname{int}\mathcal C$ be the interior of $\mathcal{C}$ in the real vector space $\mathcal A$. %Interior is taken in the real vector space $\mathcal A$; Lemma~\ref{lem:statistics-span} makes $\mathcal C$ full dimensional there.
\end{definition}

\begin{lemma}\label{lem: int C nonempty}         $\operatorname{int}\mathcal{C}$ is a nonempty open subset of $\mathcal{A}.$
In particular, $0\in \operatorname{int}\mathcal{C}.$
\end{lemma}
\begin{proof}
It is observed in the proof of \cref{lem:statistics-span} that,
for any $x \in [n]$,
$0 = \sum_{y=1}^n \frac{1}{n}T_{xy}$, and 
 $\mathrm{Aff}_{\mathbb{R}}\{T_{x1},\ldots,T_{xn}\}=\mathrm{Span}_{\mathbb{R}}\{T_{x1},\ldots,T_{xn}\}=\mathcal{A}$.
%This implies $0\in \operatorname{Aff}_\R\{T_{x1},\ldots,T_{xn}\}.$
We know that $T_{xx} = E \ne 0$ as $F(0) = U$ and  $F(E) = I$, by \cref{lem:diagonal-statistic} and 
\cref{lem:polynomial-identity}. 
%Also, By \cref{lem:statistics-span}, $0=\sum_{y=1}^n \frac{1}{n}T_{xy}$.
%So the affine span $\mathrm{Aff}_{\mathbb{R}}\{T_{x1},\ldots,T_{xn}\}$
%is the same as the linear span 
 Then $\mathcal{C}$ has dimension at least 1, and therefore
  $\operatorname{int}\mathcal{C}$ is open in its affine span.
%Also by \cref{lem:mean-identities}, $\{T_{x1},\ldots, T_{xn}\}$ linearly spans $\mathcal{A}.$
%   We only need to show $0\in \operatorname{int}\mathcal{C}$.
    By the expression $0=\sum_{y=1}^n \frac{1}{n}T_{xy}$, and
     $\{T_{x1},\ldots, T_{xn}\}$ is a  spanning set, 
     $0$ belongs to the relative interior of $\mathcal{C}$ in $\mathrm{Aff}_{\mathbb{R}}\{T_{x1},\ldots,T_{xn}\}$ by \cref{thm: convex hull nonempty}, i.e., $0\in \operatorname{int}\mathcal{C}$.
 %   But \cref{lem:statistics-span} also gives $\mathrm{Aff}_{\mathbb{R}}\{T_{x1},\ldots,T_{xn}\}=\mathcal{A}$.
 %   Thus, the relative interior of $\mathcal{C}$ in $\mathrm{Aff}_{\mathbb{R}}\{T_{x1},\ldots,T_{xn}\}$ is $\operatorname{int}\mathcal{C}$.
 %   So $0\in \operatorname{int}\mathcal{C}$.
\end{proof}

%%% JYC edited till here

\begin{lemma}\label{lem: Phi has inverse in C}
For every $X\in \operatorname{int}\mathcal C$, there exists $C\in \mathcal{A}$ such that $\Phi(C)=X.$
\end{lemma}
\begin{proof}
Fix $X\in \operatorname{int}\mathcal C$.
The proof constructs a function $g(C)$ with gradient $\nabla g(C) = \Phi(C) - X$ and shows that $g$ attains a minimum at some $C^* \in \operatorname{int}\mathcal{C}$, thereby proving $\Phi(C^*) = X$.

Fix $x\in[n]$. 
For every $C\in \mathcal{A}$, we have
\[
K(C)_{xy}=\frac{\exp(C_{xy})}{\sum_{z=1}^n\exp(C_{xz})}=\frac{\exp(\frac{1}{n}\braket{C,T_{xy}})}{\sum_{z=1}^n\exp(\frac{1}{n}\braket{C,T_{xz}})} \qquad(y\in[n])
\]
by \cref{def:stochastic-family,lem:statistic-pairing}.
Set $L(C):=\log\left(\sum_{z=1}^n\exp(\frac{1}{n}\braket{C,T_{xz}})\right)$ within this proof.
For an arbitrary direction $D\in \mathcal{A}$, the directional derivative of $L$ at $C$ is
\[
\begin{aligned}
    \left.\frac{d}{dt}L(C+tD)\right|_{t=0}
    =&\left.\frac{d}{dt}\log\left(\sum_{z=1}^n\exp\left(\frac{1}{n}\braket{C+tD,T_{xz}}\right)\right)\right|_{t=0}\\
    =&\frac{\sum_{z=1}^n\exp\left(\frac{1}{n}\braket{C,T_{xz}}\right)\frac{1}{n}\braket{D,T_{xz}}}{\sum_{z=1}^n\exp\left(\frac{1}{n}\braket{C,T_{xz}}\right)}\\
    =&\left\langle\frac{1}{n}\frac{\sum_{z=1}^n\exp\left(\frac{1}{n}\braket{C,T_{xz}}\right){T_{xz}}}{\sum_{z=1}^n\exp\left(\frac{1}{n}\braket{C,T_{xz}}\right)},D\right\rangle\\
    =& \left\langle \frac{1}{n}\sum_{z=1}^n K(C)_{xz}T_{xz},D \right\rangle\\
    =& \left\langle \frac{1}{n}\sum_{z=1}^n F(\Phi(C))_{xz}T_{xz},D \right\rangle\\
    =&\left\langle \frac{1}{n}\Phi(C),D \right\rangle
\end{aligned}
\]
The second-to-last equality holds by \cref{lem:polynomial-agreement}, and the final equality follows from \cref{lem:mean-identities}.
Therefore, $\nabla L(C)=\frac{1}{n}\Phi(C)$.

We define $g(C):=nL(C)-\braket{C,X}$ for $C\in \mathcal{A}.$
Note that 
\[
L(C)=\log\left(\sum_{z=1}^n\exp\left(\frac{1}{n}\braket{C,T_{xz}}\right)\right)
\ge \log\left(\max_{y\in [n]}\exp\left(\frac{1}{n}\braket{C,T_{xy}}\right)\right)
= \max_{y\in [n]}\frac{1}{n}\braket{C,T_{xy}}.
\]
Thus, 
\[
g(C)\ge \max_{y\in[n]}\braket{C,T_{xy}}-\braket{C,X}=\max_{y\in[n]}\braket{C,T_{xy}-X}.
\]
Since $X\in \operatorname{int}\mathcal C$, there exists $\epsilon>0$ such that
$B(X,\epsilon)\subseteq \mathcal{C}$, where
$B(X,\epsilon)$ is the ball of radius $\epsilon$ centered at $X$ inside the linear space $\mathcal{A}$. 
In particular, $Y=X+\epsilon\frac{C}{\|C\|}\in \mathcal{C}.$
%Suppose 
We can write $Y=\sum_{y=1}^na_yT_{xy}$, for some $a_y\ge 0 \,(y\in [n])$ and $\sum_{y=1}^n a_y=1.$
Then 
\[
\braket{C,Y-X}=\sum_{y=1}^n a_y\braket{C,T_{xy}-X}\le \max_{y\in [n]}\braket{C,T_{xy}-X}.
\]
Consequently,
\[g(C)\ge \braket{C,Y-X}=\left\langle{C,\epsilon \frac{C}{\|C\|}}\right\rangle=\epsilon \|C\| \qquad (C\in \mathcal{A}).\]
Thus, $g$ is a continuous coercive function in $\mathcal{A}$. 
So it attains a minimum at some $C^*\in \mathcal{A}$.
At the minimum point $C^*$ its gradient is 0, so
we have $\nabla g(C^*)=n\nabla L(C^*)-X=\Phi(C^*)-X=0$.
Therefore, for every $X\in \operatorname{int}\mathcal{C}$, there exists $C^*\in \mathcal{A}$ such that $\Phi(C^*)=X.$
\end{proof}

\begin{lemma}\label{lem:polytope-positivity}
$F(X)$ is entrywise positive and realizable for every $X\in\operatorname{int}\mathcal C$.
In particular, $F(tE)$ is entrywise positive and realizable for every $t\in (0,1)$.
\end{lemma}
\begin{proof}
For $X\in \operatorname{int}\mathcal{C}$, \cref{lem: Phi has inverse in C} supplies $C\in\mathcal A$ with $\Phi(C)=X$. 
Then $F(X)=K(C)$ by \cref{lem:polynomial-agreement}.
So $F(X)$ is entrywise positive and realizable by \cref{lem:normalized-signatures}.
By \cref{lem: int C nonempty}, $0\in \operatorname{int}\mathcal{C}.$
This implies
there exists $\epsilon>0$ such that $B(0,\epsilon)\subseteq \mathcal{C}.$
Fix an arbitrary $t\in (0,1).$
As $E = T_{xx} \in \mathcal{C}$ for any $x \in [n]$,
and $0 \in \mathcal{C}$, then $tE = (1-t)0 + t E \in \mathcal{C}$.
In fact, $B(tE, (1-t)\epsilon)\subseteq \mathcal{C}$.
To see this, let $X\in \mathcal{A}$ with $\|X\|\le (1-t)\epsilon$ and let $Y=\frac{X}{1-t}\in B(0,\epsilon)\subseteq \mathcal{C}.$
Then $tE+X=tE+(1-t)Y\in \mathcal{C}$, as a convex combination of two elements in $\mathcal{C}.$
Therefore, $tE\in \operatorname{int}\mathcal{C}$ for every $t\in (0,1)$.
Hence $F(tE)$ is entrywise positive and realizable.
\end{proof}

\subsection{Affine Factorization of the Polynomial Homomorphism}\label{subsec: affine factorization}

In this subsection, we  derive an explicit form of the polynomial homomorphism $F(\cdot).$ We  show that every irreducible polynomial factor of every entry of $F(\cdot)$ is an affine linear form (\cref{lem:affine-factors}).

Fix a basis of $\mathcal A$.
Write $F(X)=F(x_1,\ldots,x_m)$ for $X\in \mathcal{A}$, where $(x_1,\ldots,x_m)$ is the coordinate of $X$ under this basis.
Then $F(X)_{xy}\in \R[x_1,\ldots,x_m]$, for every $x,y\in [n]$.
We define the irreducible factors of $F(X)$ and their multiplicities:

\begin{definition}[Irreducible factors and multiplicities]\label{def:entry-factors}
For all $x,y\in [n]$, in 
the unique factorization domain $\mathbb{R}[x_1,\ldots,x_m]$ 
factor $F(X)_{xy}=F(x_1,\ldots,x_m)_{xy}\in \mathbb{R}[x_1,\ldots,x_m]$
 as  multivariate polynomials.
%$ \subset \mathbb{C}[x_1,\ldots,x_m]$ as complex multivariate polynomials.
Let $\mathcal{F}$ denote the set of indices corresponding to all 
nonconstant irreducible factors of $F(X)_{xy}$ across all $x, y \in [n]$.
Let $\{f_a\mid a\in \mathcal{F}\}$ denote the set of all irreducible factors.
For every $a\in \mathcal{F}$ and all $x,y\in[n]$, let $h_a(x,y)\in \Z_{\ge 0}$ be the multiplicity of $f_a$ in $F(X)_{xy}.$
\end{definition}

%Note that we factor $F(X)_{xy}$ over $\mathbb{C}[x_1,\ldots,x_m]$ rather than $\mathbb{R}[x_1,\ldots, x_m]$, because $\mathbb{C}[x_1,\ldots, x_m]$ is a unique factorization domain, which ensures the irreducible factors are well-defined.
%However, later we will show that every irreducible factor is a real affine form.
Note that $\mathbb{R}[x_1,\ldots, x_m]$ 
%%%% or $\mathbb{R}[x_1,\ldots, x_m]$ 
is a unique factorization domain, so every irreducible factor $f_a$ ($a \in \mathcal{F}$) is prime. Each irreducible factor is defined up to a nonzero 
%%% real?
real scalar, so we choose an arbitrary representative for each. Furthermore, $F(X)$ is non-zero since $F(0) = U$, and $dF(0)=\mathrm{id}$
(\cref{lem: F0 and dF0}) and so $F$ is not a constant function.
    This implies that $\mathcal{F} \neq \emptyset$.

\begin{lemma}\label{lem: affine multiplicities pre}
    For all real $(b_{xy})_{x,y\in[n]}$ satisfying $\sum_{x,y\in [n]}b_{xy}=0$ and $\sum_{x,y\in [n]}b_{xy}T_{xy}=0$, we have $\sum_{x,y\in [n]}b_{xy}h_a(x,y)=0$ for every $a\in \mathcal{F}.$
\end{lemma}
\begin{proof}
    For every $X\in \operatorname{int}\mathcal{C}$, there exists $C\in \mathcal{A}$ such that $\Phi(C)=X$ by \cref{lem: Phi has inverse in C}.
    Then
    \[    F(X)_{xy}=K(C)_{xy}=\frac{\exp(C_{xy})}{\sum_{z=1}^n\exp (C_{xz})}=\frac{\exp(\frac{1}{n}\braket{C,T_{xy}})}{\sum_{z=1}^n\exp(\frac{1}{n}\braket{C,T_{xz}})}\qquad (x,y\in[n])
    \]
    by \cref{lem:polynomial-agreement,lem:statistic-pairing}.
    Hence
    \[
    \log F(X)_{xy}=\frac{1}{n}\braket{C,T_{xy}}-\log\left(\sum_{z=1}^n\exp\left(\frac{1}{n}\braket{C,T_{xz}}\right)\right),
    \]
    and then 
    \[
    \sum_{x,y=1}^nb_{xy}\log F(X)_{xy}=\frac{1}{n}\braket{C,\sum_{x,y=1}^nb_{xy}T_{xy}}-\left(\sum_{x,y=1}^nb_{xy}\right)\cdot\log\left(\sum_{z=1}^n\exp\left(\frac{1}{n}\braket{C,T_{xz}}\right)\right)=0
    \]
    for every $x,y\in [n]$, by the assumption that $\sum_{x,y\in[n]}b_{xy}=0$ and $\sum_{x,y\in [n]}b_{xy}T_{xy}=0.$
    Taking the differential with respect to $(x_1,\ldots, x_m)$ yields
    \begin{equation}\label{eq: differentiate F}
        \sum_{x,y=1}^nb_{xy} d\log F(X)_{xy}=0\qquad \text{for every }X\in \operatorname{int}\mathcal{C}\text{ and } x,y\in [n].
    \end{equation}
    
    On the other hand, for every $X\in \operatorname{int}\mathcal{C}$ and $x,y\in [n]$, we have the unique prime factorization
    \[
    F(X)_{xy}=c_{xy}\prod_{a\in \mathcal{F}}f_a(X)^{h_a(x,y)},\qquad c_{xy}\neq 0.
    \]
    Then 
    \begin{equation}\label{eq: dF over F}
        d\log F(X)_{xy}=\frac{dF(X)_{xy}}{F(X)_{xy}}=\sum_{a\in \mathcal{F}}h_a(x,y)\frac{df_a(X)}{f_a(X)} \qquad (X\in \operatorname{int}\mathcal{C}).
    \end{equation}
    Note that $F(X)$ is entrywise positive when $X\in \operatorname{int}\mathcal{C}$ by \cref{lem:polytope-positivity}.
    Combined with \eqref{eq: differentiate F},
    \[
    \sum_{x,y\in [n]}\sum_{a\in \mathcal{F}}b_{xy}h_a(x,y)\frac{df_a(X)}{f_a(X)}=0.
    %\sum_{x,y}b_{xy}d\log F(X)_{xy}=0.
    \]
    Therefore,
    \begin{equation}\label{eq: sum of dfaX}
        \sum_{a\in \mathcal{F}}
    \left(\sum_{x,y=1}^nb_{xy}h_a(x,y)\right)\frac{df_a(X)}{f_a(X)}=0\quad \text{for every }X\in \operatorname{int}\mathcal{C}\text{ and } x,y\in [n].
    \end{equation}
    Fix $a_0\in \mathcal{F}$.
    Since $f_{a_0}$ is nonconstant, there exists $j\in [m]$ such that $\partial_jf_{a_0}\neq 0.$
    %For these fixed $a_0$ and $j$, equation \eqref{eq: sum of dfaX} gives
    %%% JYC the next eqn doesn't isolate a_0
    Then equation \eqref{eq: sum of dfaX} implies
    \[
       \sum_{a\in \mathcal{F}}
    \left(\sum_{x,y=1}^nb_{xy}h_a(x,y)\right)\frac{\partial_jf_a(X)}{f_a(X)}=0.\]
    After multiplying the polynomial $\prod_{b\in \mathcal{F}}f_b(X)$, we have 
    \[
    \sum_{a\in \mathcal{F}}
    \left(\sum_{x,y=1}^nb_{xy}h_a(x,y)\right){\partial_jf_a(X)}\prod_{b\in \mathcal{F},b\neq a}f_b(X)=0.
    \]
    The above identity holds for every $X\in \operatorname{int}\mathcal{C}$.
    Since $\operatorname{int}\mathcal{C}$ is a nonempty open subset of $\mathcal{A}$ by \cref{lem: int C nonempty}, this extends to $X\in \mathcal{A}$ by the analytic identity theorem (\cref{thm:analytic-identity}).
    Therefore, we obtain the polynomial divisibility relation
    \begin{equation}\label{eq: polynomial divisibility}
        f_{a_0}(X) \Biggm|
        %\mid 
    \left(\sum_{x,y=1}^nb_{xy}h_{a_0}(x,y)\right){\partial_j f_{a_0}(X)}\prod_{b\in \mathcal{F},b\neq a_0} f_b(X).\\
    \end{equation}
    Since $f_a$ is prime, it divides one of the factors on the RHS of \eqref{eq: polynomial divisibility}.
    However, because $\partial_j f_{a_0}$ is a nonzero polynomial of total degree one less than $f_{a_0}$, and $f_b$ is coprime to $f_{a_0}$ for all $b \neq a_0$, the only possibility is that $f_{a_0}\mid \sum_{x,y=1}^nb_{xy}h_{a_0}(x,y).$
    Since $f_{a_0}$ is a nonconstant polynomial, this implies $\sum_{x,y=1}^nb_{xy}h_{a_0}(x,y)=0.$
    Since $a_0\in \mathcal{F}$ is arbitrary,
    the lemma is proved.
\end{proof}

\begin{definition}\label{def:affine-factors}
For every $X\in \mathcal{A}$, we may write $X=\sum_{x,y\in [n]}b_{xy}T_{xy}$ with $b_{xy}\in \R$ for all $x,y\in[n]$ and $\sum_{x,y\in [n]}b_{xy}=1$, since $\mathcal{A}=\operatorname{Aff}_\R\{T_{xy}\mid x,y\in [n]\}$ by \cref{lem:statistics-span,lem:row-statistic-multisets}.  
    For every $a\in \mathcal{F}$,
    we define $$\ell_a:\mathcal{A}\to \R,\qquad X\mapsto \sum_{x,y\in [n]}b_{xy}h_a(x,y).$$
\end{definition}
    By \cref{lem: affine multiplicities pre}, $\ell_a(X)$ is well-defined.
    We can check that if $X$ has two expressions 
    $X=\sum_{x,y\in [n]}b_{xy}T_{xy}
    = \sum_{x,y\in [n]}b'_{xy}T_{xy}$, where $\sum_{x,y\in [n]}b_{xy}=
    \sum_{x,y\in [n]}b'_{xy}=1$, then we get the same value 
    $\sum_{x,y\in [n]}b_{xy}h_a(x,y) = \sum_{x,y\in [n]}b'_{xy}h_a(x,y)$.

\begin{lemma}\label{lem:affine-multiplicities}
For every $a\in \mathcal{F}$, $\ell_a$ is an affine form such that $\ell_a(T_{xy})=h_a(x,y)$ for all $x,y\in[n]$.
\end{lemma}
\begin{proof}
    Fix $x_0,y_0\in [n]$.
    Let $b_{x_0y_0}=1$ and $b_{xy}=0$ for all $x,y\in [n]$ with $(x,y)\neq (x_0,y_0)$.
    Then for every $a\in \mathcal{F},$  $\ell_a(T_{x_0y_0})=\ell_a(\sum_{x,y\in [n]}b_{xy}T_{xy})=\sum_{x,y\in[n]}b_{xy}h_a(x,y)=h_a(x_0,y_0).$
    Next fix arbitrary $a\in \mathcal{F}$, we show that $\ell_a$ is an affine form.
    By the definition of $\ell_a$, it is easy to check that $\ell_a(tX+(1-t)Y)=t\ell_a(X)+(1-t)\ell_a(Y)$ for all $X,Y\in \mathcal{A}$ and every $t\in \R$.
    Therefore, $\ell_a$ is an affine form by \cref{def: affine form}.
\end{proof}

\begin{lemma}\label{lem:irreducible-affine}
For every $a\in\mathcal F$, $f_a$ is a nonzero real scalar multiple of $\ell_a$.
\end{lemma}
\begin{proof}
    By \cref{lem:row-normalization}, $\sum_{y\in [n]}F(X)_{xy}=1$ for every $X\in \mathcal{A}$ and every $x\in[n]$.
    Averaging over $x\in [n]$ gives $\frac{1}{n}\sum_{x,y\in [n]}F(X)_{xy}=1$.
    Take differential on both sides, we obtain $\frac{1}{n}\sum_{x,y\in[n]}dF(X)_{xy}=0$.
    On the other hand, for every $X\in\mathcal{A}$ and $x,y\in [n]$, we have the unique prime factorization
    \(
    F(X)_{xy}=c_{xy}\prod_{a\in \mathcal{F}}f_a(X)^{h_a(x,y)}, c_{xy}\neq 0.
    \)
    Then by \eqref{eq: dF over F}, 
    \[
    dF(X)_{xy}=F(X)_{xy}\sum_{a\in \mathcal{F}}h_a(x,y)\frac{df_a(X)}{f_a(X)}\qquad (X\in \operatorname{int}\mathcal{C}).
    \]
    Thus, 
    \begin{equation}\label{eq: irreducible-affine}
        0=\frac{1}{n}\sum_{x,y\in [n]}dF(X)_{xy}=\frac{1}{n}\sum_{x,y\in[n]}\sum_{a\in \mathcal{F}}h_a(x,y)F(X)_{xy}\frac{df_a(X)}{f_a(X)} \qquad (X\in \operatorname{int}\mathcal{C}).
    \end{equation}
    \begin{claim}\label{claim: affine forms and multiplicities}
        For every $X\in \mathcal{A}$ and $a\in \mathcal{F}$, $\ell_a(X)=\frac{1}{n}\sum_{x,y\in[n]}h_a(x,y)F(X)_{xy}.$
    \end{claim}
    \begin{claimproof}{\ref{claim: affine forms and multiplicities}}
        Fix $X\in\mathcal A$ and $a\in \mathcal{F}.$
        By \cref{lem:mean-identities}, $\sum_{y=1}^nF(X)_{xy}T_{xy}=X$ for every $x\in[n]$.
        Averaging over $x\in[n],$ we have $X=\frac{1}{n}\sum_{x,y\in[n]}F(X)_{xy}T_{xy}.$
        By definition of $\ell_a$, we have $\ell_a(X)=\frac{1}{n}\sum_{x,y\in [n]}F(X)_{xy}h_a(x,y).$
    \end{claimproof}
    Combining \eqref{eq: irreducible-affine} and Claim~\ref{claim: affine forms and multiplicities},
    \begin{equation}\label{eq: fa and la}
        \sum_{a\in \mathcal{F}}\ell_a(X)\frac{df_a(X)}{f_a(X)}=0 \qquad (X\in\operatorname{int}\mathcal{C}).
    \end{equation}
    Hence
    \[
    \sum_{a\in \mathcal{F}}\ell_a(X){df_a(X)}\prod_{b\in \mathcal{F},b\neq a}{f_b(X)}=0 \qquad (X\in\mathcal{A}).
    \]
    The above equation first holds on $\operatorname{int}\mathcal{C}$, then extends to $\mathcal{A}$ by the analytic identity theorem (\cref{thm:analytic-identity}).
    Now fix arbitrary $a_0\in \mathcal{F}.$
    The irreducible factor $f_{a_0}$ is not constant, so there exists $j\in[m]$ such that $\partial_j f_{a_0}\neq 0.$
    Fix one such $j,$ we have the following polynomial divisibility relation:
    \[
    f_{a_0}(X)\mid \ell_{a_0}(X){\partial_jf_{a_0}(X)}\prod_{b\in \mathcal{F},b\neq a_0}{f_b(X)}.
    \]
    Similar to the proof of \cref{lem: affine multiplicities pre}, the only possibility is that $f_{a_0}(X)\mid \ell_{a_0}(X).$
    Since $\ell_{a_0}(X)$ has total degree one and $f_{a_0}(X)$ is a nonconstant polynomial, we conclude that $f_{a_0}(X)$ is a nonzero real scalar multiple of $\ell_{a_0}(X)$.
    Since $a_0\in \mathcal{F}$ is arbitary, the lemma is proved.
\end{proof}

\begin{lemma}\label{lem:factor-positivity}
For every $a\in\mathcal F$, $\ell_a(0)\in \Q_{> 0}$.
\end{lemma}
\begin{proof}
    Fix $a\in \mathcal{F}.$
    By \cref{lem:statistics-span}, $0=\frac{1}{n^2}\sum_{x,y=1}^n T_{xy}$.
    Then $\ell_a(0)=\sum_{x,y=1}^n \frac{1}{n^2}\ell_a(T_{xy})=\frac{1}{n^2}\sum_{x,y\in[n]}h_a(x,y)$ by \cref{lem:affine-multiplicities}.
    So $\ell_a(0)\in \Q_{\ge 0}$, and $\ell_a(0)=0\iff h_a(x,y)=0$ for all $x,y\in [n]\iff f_a $ has multiplicity $0$ in $F(X)_{xy}$ for all $x,y\in [n].$
    This is impossible.
    Therefore, $\ell_a(0)>0.$
\end{proof}

\begin{lemma}[Affine factorization]\label{lem:affine-factors}
For every $X\in\mathcal A$ and all $x,y\in[n]$,
\begin{equation*}
F(X)_{xy}=\frac1n\prod_{a\in\mathcal F}
\left(\frac{\ell_a(X)}{\ell_a(0)}\right)^{h_a(x,y)}.
\end{equation*}
A factor with exponent zero is interpreted as the constant 1 (a  polynomial
of degree 0).
\end{lemma}
\begin{proof}
For all $x,y\in \mathcal{A}$,
the polynomial factorization from \cref{def:entry-factors,lem:irreducible-affine} give
\[
F(X)_{xy}=c_{xy}\prod_{a\in\mathcal F}\ell_a(X)^{h_a(x,y)}
\]
for some nonzero constant $c_{xy}$. At $X=0$, $F(0)_{xy}=1/n$ by \cref{lem: F0 and dF0}, so
\[
\frac{1}{n}=c_{xy}\prod_{a\in\mathcal F}\ell_a(0)^{h_a(x,y)}.
\]
Since $c_{xy} \neq 0$ and $\ell_a(0) > 0$ for all $a \in \mathcal{F}$, dividing these two equations yields the desired result.
\end{proof}

The following lemma shows that $F(X)_{xy}$ have common degree $d$ for all $x,y\in[n].$
\begin{lemma}\label{lem:common-degree}
There exists $d\in \Zp$ such that
$\sum_{a\in\mathcal F}\ell_a(X)=d$ for every $X\in \mathcal{A}$, and
$\sum_{a\in\mathcal F}h_a(x,y)=d$ for all $x,y\in [n].$
\end{lemma}
\begin{proof}
By \cref{lem:irreducible-affine}, for every $a\in \mathcal{A}$, there exists $c_a\neq 0$ such that $f_a=c_a\ell_a.$
In \eqref{eq: fa and la}, replace $f_a(X)$ by $c_a\ell_a(X)$, we obtain
\[
\sum_{a\in \mathcal{F}}d\ell_a(X)=\sum_{a\in \mathcal{F}}\ell_a(X)\frac{d\ell_a(X)}{\ell_a(X)}=0 \qquad(X\in \operatorname{int}\mathcal{C}).
\]
This identity extends to $X\in \mathcal{A}$.
Thus $\sum_{a\in\mathcal F}\ell_a(X)=d$ is constant on $\mathcal A$. 
Evaluating at $T_{xy}$ gives $\sum_{a\in \mathcal{F}}h_a(x,y)=\sum_{a\in \mathcal{F}}\ell_a(T_{xy})=d$ for all $x,y\in[n]$ by \cref{lem:affine-multiplicities}.
%Since each $h_a(x,y)\in \Z_{\ge 0}$, $d\in \Z_{\ge 0}.$
Therefore, $d\in \Zp$ is the common degree of $F(X)_{xy}$ for all $x,y\in[n]$.
\end{proof}

\begin{lemma}\label{lem:factor-transpose}
There is an involution map $\tau:\mathcal F\to\mathcal F$ such that, for every $a\in\mathcal F$, every $X\in\mathcal A$, and all $x,y\in[n]$,
$\ell_{\tau(a)}(X)=\ell_a(X^{\tt T})$ and 
$h_{\tau(a)}(x,y)=h_a(y,x).$
\end{lemma}

\begin{proof}
For every $a\in \mathcal{F}$, the irreducible factor $\ell_a(X)$ appears in some $F(X)_{xy}.$
By \cref{lem: F transpose}, $F(X^{\tt T})=F(X)^{\tt T}.$
Thus, $\ell_a(X^{\tt T})\mid F(X^{\tt T})_{xy}\implies \ell_a(X^{\tt T})\mid F(X)_{yx}$. By its form, $\ell_a(X^{\tt T})$ is also affine linear.
%I mean by Definition 2.4 , we check X --> l_a(X^t)  satifies
% item 2 of def 2.4 :ℓ(tu + (1 − t)v) = tℓ(u) + (1 − t)ℓ(v)
Let $\tau(a)\in \mathcal{F}$ be the index of the irreducible factor in $F(X)_{yx}$ that is proportional to $\ell_a(X^{\tt T})$.
Note that the definition of $\tau$ is independent of the choice of $x,y.$
This is because the index of $\ell_a(X^{\tt T})$ in $\mathcal{F}$ only depends on $a.$
Now by \cref{lem:affine-factors}, 
\[
\frac1n\prod_{a\in\mathcal F}
\left(\frac{\ell_a(X^{\tt T})}{\ell_a(0)}\right)^{h_a(x,y)}=\frac1n\prod_{a\in\mathcal F}
\left(\frac{\ell_a(X)}{\ell_a(0)}\right)^{h_a(y,x)}
\]
for every $X\in \mathcal{A}$ and all $x,y\in[n]$.
Both sides of the equation are unique prime factorizations.
Thus, for every $a\in \mathcal{F}$ and all $x,y\in[n]$, $\ell_a(X^{\tt T})=c_a \ell_{\tau(a)}(X)$, and $h_a(x,y)=h_{\tau(a)}(y,x)$, for some $c_a\neq 0.$ 
Evaluating $\ell_a(X^{\tt T})=c_a \ell_{\tau(a)}(X)$ at $X=T_{xy}$ for arbitrary $x,y\in[n]$, we have
\[
\ell_a(T_{xy}^{\tt T})=\ell_a(T_{yx})=h_a(y,x)=h_{\tau(a)}(x,y)=\ell_{\tau(a)}(T_{xy})
\]
by \cref{lem:diagonal-statistic,lem:affine-multiplicities}.
Since $\{T_{xy}\}_{x,y\in [n]}$ affinely span $\mathcal{A}$ by \cref{lem:statistics-span,lem:row-statistic-multisets}, $\ell_a(X^{\tt T})=\ell_{\tau(a)}(X)$, for every $X\in \mathcal{A}.$

Finally, applying transpose twice, we have $\ell_a(X)=\ell_{\tau(a)}(X^{\tt T})=\ell_{\tau^2(a)}(X)$ for every $a\in \mathcal{F}$ and $X\in \mathcal{A}.$
Thus, $\tau^2(a)=a$ for every $a\in \mathcal{F}$, i.e., $\tau$ is an involution.
\end{proof}

\section{Distances and Color Bijection Induced by the Polynomial Homomorphism}\label{sec:graph_distance}

In this section, we will show that from the exponents in the affine factorization of $F(\cdot)$, we can define some binary functions $\delta(\cdot,\cdot)$ and $\Delta_\omega(\cdot,\cdot)$.
%These distances turn out to serve as graph distances for highly regular graphs. 
In \cref{subsec: endpoint indices},
we first partition the set of affine factors into ``endpoint factors" and ``complementary factors", according to their evaluation at $E$ is zero or not (\cref{def: endpoint and complementary indices}).
The function $\delta(x,y)$ is defined as the sum of the multiplicities of endpoint factors in $F(X)_{xy}$ (\cref{def: delta and alpha}).
By $F(E)=I$, the function $\delta(x,y)$ captures the asymptotic order at which the perturbation $F(E - \epsilon X)_{xy}$ approaches $I_{xy}$ (\cref{lem:endpoint-leading-term}).
Expand the asymptotics in the homomorphism identity (\cref{lem: homomorphism property}) and compare the lowest degree terms, we prove that $\delta$ is a distance on $[n]$ (\cref{lem:endpoint-distance}).

Each $\omega$ is an orbit induced by the involution $\tau$ in \cref{lem:factor-transpose} restricted 
%on complementary indices.
on the set of indices for complementary factors.
The function $\Delta_\omega(x,y)$ is defined as the decrease in the total
multiplicity of the complementary factors in the orbit $\omega$
when passing from the diagonal entry $F(X)_{xx}$ to the entry
$F(X)_{xy}$.
In \cref{subsec: complenmentray indices}, we will prove that each zero relation  $\Delta_\omega(\cdot,\cdot)=0$, denoted by $L_\omega$, is an equivalence relation (\cref{lem: expression for Delta omega}).

In \cref{subsec: color bijection}, we show that there is a bijection from the set of colors $[n]$ to the Cartesian product of the equivalent classes of $L_\omega$ (\cref{lem: phi is bijection}).
We call this the ``color bijection."
This captures the bijection between the indices of $M$ and its tensor factors, when $M$ is a tensor product.
Further, we show that the value of $\Delta_\omega(x,y)$ depends only on the $L_\omega$-equivalence classes of $x$ and $y$ (\cref{lem:coordinate-locality}).
This defines a binary function $\Delta_\omega(\cdot,\cdot)$ on the set of $L_\omega$-equivalence classes.
In \cref{sec:hamming},
we will show that they must be highly regular graph distances 
 if $\PlGH(M)$ is not \#P-hard.

\subsection{Endpoint Indices and the Distance \texorpdfstring{$\delta(\cdot,\cdot)$}{δ(.,.)}}\label{subsec: endpoint indices}

\begin{definition}[Endpoint and complementary indices]\label{def: endpoint and complementary indices}
    Let $\mathcal{F}$ be the index  set of irreducible factors in \cref{def:entry-factors}.
    Partition $\mathcal{F}$ into \(\mathcal I=\{a\in\mathcal F\mid \ell_a(E)=0\}\) and \(\mathcal F_+=\mathcal F\setminus\mathcal I.\)
    We call $\mathcal{I}$ the set of \emph{endpoint indices} and $\mathcal{F}_+$ the set of \emph{complementary indices}.
\end{definition}

Note that for every $a\in \mathcal{F},$ $\ell_a(E)=\ell_a(T_{xx})=h_a(x,x)\in \mathbb{Z}_{\ge 0}$ by \cref{lem:diagonal-statistic,lem:affine-multiplicities}.
Thus, $\mathcal{F}_+=\{a\in \mathcal{F}\mid \ell_a(E)\in \Zp\}.$

\begin{definition}\label{def: delta and alpha}
    For all $x,y\in [n]$,
    define $\delta(x,y)=\sum_{i\in\mathcal I}h_i(x,y)$ and $\mathbf{\alpha}(x,y)=(h_i(x,y))_{i\in\mathcal I}$.
\end{definition}

We will show that $\delta$ is a distance on $[n].$
First we can prove positivity:

\begin{lemma}\label{lem: positivity of delta}
For all $x, y \in [n]$, $\delta(x, y) \ge 0$, with equality holding if and only if $x = y$.
\end{lemma}
\begin{proof}
By definition, $\delta(x,y)=\sum_{i\in\mathcal I}h_i(x,y)\in \Z_{\ge 0}.$
If $x=y$, then $\delta(x,x)=\sum_{i\in \mathcal{I}}h_i(x,x)=\sum_{i\in \mathcal{I}}\ell_i(E)=0.$
Now suppose $\delta(x,y)=0$ for some $x,y\in [n]$.
We next show $x=y$.
\cref{lem:polynomial-identity} gives $F(E)=I.$
Substituting $X=E$ into the affine factorization \cref{lem:affine-factors}, we obtain
\[
I_{xy}=F(E)_{xy}=\frac{1}{n}\prod_{i\in \mathcal{I}}\left(\frac{\ell_i(E)}{\ell_i(0)}\right)^{h_i(x,y)}\prod_{r\in \mathcal{F}_+}\left(\frac{\ell_r(E)}{\ell_r(0)}\right)^{h_r(x,y)}.
\]
Suppose for a contradiction that $x\neq y.$
Then the left-hand side of the equation evaluates to $0$.
In the right-hand side the second product over $\mathcal{F}_+$ is positive,
the first product over $\mathcal{I}$ must be nonempty and not all exponents
$h_i(x,y)$ are zero. Then 
$\delta(x,y)=\sum_{i\in\mathcal I}h_i(x,y) >0$, a 
contradiction.
%
%Hence $\prod_{i\in\mathcal{I}}\left(\frac{\ell_i(X)}{\ell_i(0)}\right)^{h_i(x,y)}$ is not the constant one polynomial.
%Thus, there exists $i\in \mathcal{I}$ such that $h_i(x,y)>0.$
%This implies $\delta(x,y)>0$, a 
%contradiction.
Therefore, $\delta(x,y)=0$ if and only if $x=y$.
\end{proof}

To prove symmetry and the triangle inequality, we rely on the asymptotics in \cref{lem:endpoint-leading-term}. 
We then expand the homomorphism identity $F(X)F(Y) = F(XY)$ for perturbations $X$ and $Y$ at $E$, and compare the lowest degree terms (\cref{lem:endpoint-distance}).

\begin{definition}\label{def: dual basis}
For every $i\in\mathcal I$ and every $X\in\mathcal A$, define
$g_i(X)=\ell_i(E-X)$ and $g(X)=(g_i(X))_{i\in\mathcal I}$.
Let $\mathbf u_i\in\mathbb Z^{\mathcal I}$ denote the unit vector at $i \in {\mathcal I}$. 
For $g=(g_i)_{i\in\mathcal I}$ and $\beta\in\mathbb Z_{\geq0}^{\mathcal I}$, write $g^\beta=\prod_{i\in\mathcal I}g_i^{\beta_i}$ and $|\beta|=\sum_{i\in\mathcal I}\beta_i$. 
\end{definition}

Note that each $g_i(\cdot)$ is a nonzero linear form because $\ell_i(E)=0$ and $\ell_i$ is nonconstant. 
%%% I added a \cdot b/c you had in the def another $g=(g_i)_{i\in\mathcal I}$ 
%%% without specifying its type.

\begin{lemma}\label{lem:endpoint-leading-term}
For every $X\in\mathcal A$ and all $x,y\in[n]$, there exists $\kappa_{xy}>0$ such that
\begin{equation*}
F(E-\varepsilon X)_{xy}
=\varepsilon^{\delta(x,y)}\kappa_{xy}g(X)^{\alpha(x,y)}
+O\!\left(\varepsilon^{\delta(x,y)+1}\right), \quad\text{as }\varepsilon\to 0.
%\label{eq:endpoint-leading-term}
\end{equation*}
%The implicit constant may depend on $A,x,y$ and the fixed target.
\end{lemma}

\begin{proof}
Fix $X\in \mathcal{A}$ and $x,y\in [n]$.
By affine factorization (\cref{lem:affine-factors}),
\begin{equation}\label{eq: endpoint-leading-term}
    F(E-\varepsilon X)_{xy}=\frac{1}{n}\prod_{i\in\mathcal{I}}\left(\frac{\ell_i(E-\varepsilon X)}{\ell_i(0)}\right)^{h_i(x,y)}\prod_{r\in \mathcal{F}_+}\left(\frac{\ell_r(E-\varepsilon X)}{\ell_r(0)}\right)^{h_r(x,y)}.
\end{equation}
Since $g_i$ is a linear form for $i\in \mathcal{I}$, $\ell_i(E-\varepsilon X)=g_i(\varepsilon X)=\varepsilon g_i(X)$.
While for $r\in\mathcal F_+$, 
\(
\ell_r(E-\varepsilon X)
=\ell_r(E)+O(\varepsilon).
\)
Substitution in \eqref{eq: endpoint-leading-term} yields
\[
\begin{aligned}
    F(E-\varepsilon X)_{xy}=&\frac{1}{n}\prod_{i\in\mathcal{I}}\left(\frac{\varepsilon g_i(X)}{\ell_i(0)}\right)^{h_i(x,y)}\prod_{r\in \mathcal{F}_+}\left(\frac{\ell_r(E)+O(\varepsilon)}{\ell_r(0)}\right)^{h_r(x,y)}\\
    =&\frac{1}{n}\varepsilon^{\delta(x,y)} g(X)^{\alpha(x,y)}\prod_{i\in\mathcal{I}}{\ell_i(0)}^{-h_i(x,y)}\prod_{r\in \mathcal{F}_+}\left(\frac{\ell_r(E)}{\ell_r(0)}\right)^{h_r(x,y)}+O\left(\varepsilon^{\delta(x,y)+1}\right)
\end{aligned}
\]
as $\varepsilon\to 0$.
Define $\kappa_{xy}$ as
\[
\kappa_{xy}:=\frac1n
\prod_{i\in\mathcal I}\ell_i(0)^{-h_i(x,y)}
\prod_{r\in\mathcal F_+}\left(\frac{\ell_r(E)}{\ell_r(0)}\right)^{h_r(x,y)}.
\]
Then $\kappa_{xy}>0$ and $F(E-\varepsilon X)_{xy}
=\varepsilon^{\delta(x,y)}\kappa_{xy}g(X)^{\alpha(x,y)}
+O\!\left(\varepsilon^{\delta(x,y)+1}\right)$, as $\varepsilon\to 0.$
\end{proof}

\begin{lemma}\label{lem:endpoint-distance}
$\delta$ is a distance function on $[n]$: for all $x,y,z\in[n]$,
\begin{enumerate}
    \item $\delta(x,y)\geq0,$ and {\rm [}\,$ \delta(x,y)=0\iff x=y$\,{\rm ]};
    \item $\delta(x,y)=\delta(y,x);$
    \item $\delta(x,y)\leq\delta(x,z)+\delta(z,y).$
\end{enumerate}
\end{lemma}

\begin{proof}
The first item is proved in \cref{lem: positivity of delta}. 
Next we prove symmetry.
Fix arbitrary $x,y\in [n]$.
By \cref{lem:endpoint-leading-term}, 
\begin{equation}\label{eq: delta triangle 1}
    F((1-\varepsilon)E)_{xy}=\kappa_{xy}g(E)^{\alpha(x,y)}\varepsilon^{\delta(x,y)}+O(\varepsilon^{\delta(x,y)+1})
\quad \text{as }\varepsilon\to 0,
\end{equation}
where $\kappa_{xy}>0.$
Note that $g(E)^{\alpha(x,y)}=\prod_{i\in \mathcal{I}}g_i(E)^{h_i(x,y)}=\prod_{i\in \mathcal{I}}\ell_i(0)^{h_i(x,y)}>0.$
The matrix $F((1-\varepsilon) E)$ is symmetric by \cref{lem: F transpose} and $E=E^{\tt T}$.
So $F((1-\varepsilon E))_{xy}$ and $F((1-\varepsilon E))_{yx}$ have the same vanishing order as $\epsilon\to 0$. 
Thus, $\delta(x,y)=\delta(y,x)$.

Next we show the triangle inequality.
By \cref{lem: homomorphism property}, $F((1-\varepsilon)E)^2=F((1-\varepsilon)^2 E).$
For fixed $x,y\in [n]$,
\begin{equation}\label{eq: delta triangle 2}
    \sum_{z=1}^nF((1-\varepsilon)E)_{xz}F((1-\varepsilon)E)_{zy}
=F((1-2\varepsilon+\epsilon^2)E)_{xy}.
\end{equation}
Let $p=\min_{z\in[n]}(\delta(x,z)+\delta(z,y))$. 
Plug \eqref{eq: delta triangle 1} into \eqref{eq: delta triangle 2} as $\epsilon\to 0$.
The coefficient of the lowest degree term $\varepsilon^p$ on the left is
$$\sum_{\substack{z\in[n],\delta(x,z)+\delta(z,y)=p}}\kappa_{xz}\kappa_{zy}g(E)^{\alpha(x,z)}g(E)^{\alpha(z,y)}>0.$$
On the right, the lowest degree term is
$\kappa_{xy}g(E)^{\alpha(x,y)}(2\epsilon)^{\delta(x,y)}$,
whose coefficient is also nonzero. 
When $\epsilon$ approaches to 0, the lowest degree term on both sides of \eqref{eq: delta triangle 2} must match.
Therefore, $\delta(x,y)=p=\min_{z\in[n]}(\delta(x,z)+\delta(z,y)).$
\end{proof}

\begin{remark}
    In \cref{lem:endpoint-distance}, we establish not only the triangle inequality, but also that for every $x, y \in [n]$, there exists some $z \in [n]$ satisfying $\delta(x,y) = \delta(x,z) + \delta(z,y)$. Moreover, for this choice of $z$, we have $\alpha(x,y) = \alpha(x,z) + \alpha(z,y)$. In \cref{lem:endpoint-box}, we generalize this result: for every $0\le\beta\le \alpha(x,y)$ (entrywise), there exists $z \in [n]$ such that $\delta(x,y) = \delta(x,z) + \delta(z,y)$, with $\alpha(x,z) = \beta$ and $\alpha(z,y) = \alpha(x,y) - \beta$.
    To establish this, we first prove the following preliminary \cref{lem:endpoint-basis}.
\end{remark}

\begin{lemma}\label{lem:endpoint-basis}
Let $\mathcal{A}^*$ be the dual space of $\mathcal{A}.$
Then
$\{g_i\}_{i\in\mathcal I}$ is a basis of $\mathcal A^*$. 
%In particular, $|\mathcal I|=\dim\mathcal A$.
\end{lemma}

\begin{proof}
Fix a row $x\in[n]$. 
\cref{lem:row-normalization,lem:mean-identities} give $\sum_{y=1}^nF(X)_{xy}=1$ and $\sum_{y=1}^nF(X)_{xy}T_{xy}=X.$
Substituting $X=E-\varepsilon A$ for arbitrary $A\in \mathcal{A}$, we obtain
\begin{equation}\label{eq: endpoint-basis}
    \begin{aligned}
        E-\varepsilon A&=\sum_{y=1}^nF(E-\varepsilon A)_{xy}T_{xy}\\
        \implies
        \varepsilon A&=\sum_{y=1}^nF(E-\varepsilon A)_{xy}(E-T_{xy}).
    \end{aligned}
\end{equation}
Let $\varepsilon$ approach to 0. 
By \cref{lem:endpoint-leading-term}, the RHS of \eqref{eq: endpoint-basis} is 
\[
\text{RHS}=\sum_{y=1}^n\left(\varepsilon^{\delta(x,y)}\kappa_{xy}g(A)^{\alpha(x,y)}(E-T_{xy})
+O\!\left(\varepsilon^{\delta(x,y)+1}\right)\right)
\]
where $\kappa_{xy}>0$ for every $y\in[n].$
However, the LHS of \eqref{eq: endpoint-basis} has vanishing order one.
Thus, only the terms with $\delta(x,y)=1$ contribute to the RHS of \eqref{eq: endpoint-basis}.
When $\delta(x,y)=\sum_{i\in \mathcal{I}}h_i(x,y)=1$, there exists a unique $i$ such that $h_i(x,y)=1$ and other coordinates are $0$; that is, $\alpha(x,y)=\mathbf{u}_i$ and $g(A)^{\alpha(x,y)}=g_i(A).$
Grouping \eqref{eq: endpoint-basis} by $\mathbf{u}_i,$ we have
\begin{equation}\label{eq: endpoint-basis 2}
    A=\sum_{i\in \mathcal{I}}g_i(A)\sum_{y\in[n],\alpha(x,y)=\mathbf{u}_i}\kappa_{xy}(E-T_{xy})\qquad (A\in \mathcal{A}).
\end{equation}
For every $i\in \mathcal{I}$, let $V_i=\sum_{y\in[n],\alpha(x,y)=\mathbf{u}_i}\kappa_{xy}(E-T_{xy})\in \mathcal{A}.$
Then $A=\sum_{i\in \mathcal{I}}g_i(A)V_i,\,\forall A\in \mathcal{A}.$
Consequently, any linear functional $f \in \mathcal{A}^*$ satisfies $f = \sum_{i \in \mathcal{I}} f(V_i) g_i$, which proves that $\{g_i\}_{i \in \mathcal{I}}$ generates $\mathcal{A}^*$.

Next we show that $\{g_i\}_{i\in \mathcal{A}}$ are linearly independent.
For every $j\in \mathcal{I}$, apply $g_j$ on both sides of \eqref{eq: endpoint-basis 2},
\[
g_j(A)=\sum_{i\in \mathcal{I}}g_i(A)\sum_{y\in [n],\alpha(x,y)=\mathbf{u}_i}\kappa_{xy}g_j(E-T_{xy}) \qquad(A\in\mathcal{A}).
\]
When $\alpha(x,y)=\mathbf{u}_i$, $$g_j(E-T_{xy})=\ell_j(T_{xy})=h_j(x,y)=\begin{cases}
    1, &j=i;\\
    0, &j\neq i.
\end{cases}$$
Thus,
\[
g_j(A)=g_j(A)\sum_{y\in [n],\alpha(x,y)=\mathbf{u}_j}\kappa_{xy}\qquad (A\in \mathcal{A}).
\]
Since each $g_j$ is a nonzero linear form, $\sum_{y\in [n],\alpha(x,y)=\mathbf{u}_j}\kappa_{xy}=1.$
In particular, for every $j\in \mathcal{I}$, there exists $y_j\in[n]$ such that $\alpha(x,y_j)=\mathbf{u}_j.$
Fix any such $y_j$ and
define $B_j:=E-T_{xy_j}$. 
Then for all $i,j\in\mathcal{I}$, $$g_i(B_j)=\ell_i(T_{xy_j})=h_i(x,y_j)=\begin{cases}
    1, &j=i;\\
    0, &j\neq i.
\end{cases}$$
Therefore, $\{g_i\}_{i\in \mathcal{I}}$ is linearly independent, 
and therefore a basis of  $\mathcal{A}^*$,
and $\{B_i\}_{i\in \mathcal{I}}$ is a  basis 
of  $\mathcal{A}$ dual to $\{g_{i}\}_{i\in\mathcal{I}}.$
\end{proof}

\begin{definition}\label{def:endpoint-basis}
Let $\{B_i\}_{i\in\mathcal{I}}$ be the dual basis to $\{g_i\}_{i\in \mathcal{I}}$ defined in \cref{lem:endpoint-basis}.
\end{definition}

\begin{corollary}\label{lem:endpoint-coordinates}
For every $i\in\mathcal{I}$, there exist $x,y\in [n]$ such that $\alpha(x,y)=\mathbf{u}_i$ and $B_i=E-T_{xy}.$
\end{corollary}

\begin{lemma}\label{lem: E-Txy}
    For all $x,y\in [n]$, $E-T_{xy}=\sum_{i\in \mathcal{I}}h_i(x,y)B_i.$
\end{lemma}
\begin{proof}
    Expand $E-T_{xy}$ in the dual basis $\{B_i\}_{i\in \mathcal{I}},$ we have
    $E-T_{xy}=\sum_{i\in\mathcal{I}}g_i(E-T_{xy})B_i=\sum_{i\in \mathcal{I}}\ell_i(T_{xy})B_i=\sum_{i\in \mathcal{I}}h_i(x,y)B_i.$
    The last equality holds by \cref{lem:affine-multiplicities}.
\end{proof}
\begin{lemma}\label{lem:endpoint-box}
For all $x,y\in[n]$ and every $\beta\in\mathbb Z_{\geq0}^{\mathcal I}$ with $0\leq\beta\leq\alpha(x,y)$ coordinate wise, there exists $z\in[n]$ such that
$\alpha(x,z)=\beta$ and
$\alpha(z,y)=\alpha(x,y)-\beta$.
\end{lemma}

\begin{proof}
For arbitrary $X,Y\in \mathcal{A}$, write
$X=\sum_{i\in\mathcal I}u_iB_i$ and $Y=\sum_{i\in\mathcal I}v_iB_i$ by \cref{lem:endpoint-basis}, where $u_i=g_i(X)$ and $v_i=g_i(Y)$ for every $i\in \mathcal{I}.$
Since $X,Y\in \mathcal{A}$ are arbitrary $\{u_i\}_{i\in\mathcal{I}}$ and $\{v_i\}_{i\in\mathcal{I}}$ are independent real variables.
By \cref{lem: homomorphism property},
\begin{equation}\label{eq: endpoint-box 1}
    \sum_{z=1}^n F(E-X)_{xz}F(E-Y)_{zy}=F(E-X-Y+XY)_{xy}.
\end{equation}
First fix $x,y,z\in[n].$
After picking the basis $\{B_i\}_{i\in \mathcal{I}}$, $F(E-X)_{xz}$ becomes a multivariate polynomial in $\{u_i\}_{i\in \mathcal{I}}$.
Let $P(u)$ denote this polynomial, and write $P(u)=\sum_{k=0}^d P_k(u)$, where each $P_k(u)$ is a homogeneous polynomial of degree $k.$
By \cref{lem:endpoint-leading-term},
\[
P(\varepsilon u)=F(E-\varepsilon X)_{xz}=\varepsilon^{\delta(x,z)}\kappa_{xz}u^{\alpha(x,z)}+O\left(\varepsilon^{\delta(x,z)+1}\right)\quad \text{as }\varepsilon\to 0,
\]
where $\kappa_{xz}>0.$
Also, $P(\varepsilon u)=\sum_{k=0}^d \varepsilon^k P_k(u).$
Thus, $P_k(u)\equiv 0$ whenever $k<\delta(x,z)$, and $P_{\delta(x,z)}(u)=\kappa_{xz}u^{\alpha(x,z)}.$
Therefore,
\begin{equation}\label{eq: endpoint-box 2}
    F(E- X)_{xz}=\sum_{k=0}^d P_k(u)=\kappa_{xz}u^{\alpha(x,z)}+\text{terms of total degree greater than $\delta(x,z)$}.
\end{equation}
Similarly,
\begin{equation}\label{eq: endpoint-box 3}
    F(E-Y)_{zy}=\kappa_{zy}v^{\alpha(z,y)}+\text{terms of total degree greater than $\delta(z,y)$},
\end{equation}
and
\begin{equation}\label{eq: endpoint-box 4}
    F(E-(X+Y-XY))_{xy}=\kappa_{xy}(u+v)^{\alpha(x,y)}+\text{terms of total degree greater than $\delta(x,y)$},
\end{equation}
where $\kappa_{zy},\kappa_{xy}>0.$

Substitute equations \eqref{eq: endpoint-box 2} - \eqref{eq: endpoint-box 4} into \eqref{eq: endpoint-box 1} and compare the lowest total degree term on both sides.
On the LHS, the term indexed by $z$ has total degree $\delta(x,z)+\delta(z,y)\ge \delta(x,y)$ by \cref{lem:endpoint-distance}.
Also, by the proof of \cref{lem:endpoint-distance}, there exists $z\in [n]$ such that $\delta(x,z)+\delta(z,y)= \delta(x,y).$
Thus, the lowest total degree term on the LHS equals
\[
\sum_{\substack{z\in[n]\\\delta(x,z)+\delta(z,y)=\delta(x,y)}}
\kappa_{xz}\kappa_{zy}u^{\alpha(x,z)}v^{\alpha(z,y)}.
\]
On the other hand, the lowest total degree term on the RHS equals $\kappa_{xy}(u+v)^{\alpha(x,y)}.$
Therefore,
\begin{equation}
\kappa_{xy}(u+v)^{\alpha(x,y)}
=\sum_{\substack{z\in[n]\\\delta(x,z)+\delta(z,y)=\delta(x,y)}}
\kappa_{xz}\kappa_{zy}u^{\alpha(x,z)}v^{\alpha(z,y)}.
\label{eq:box-polynomial}
\end{equation}
For a prescribed $0\leq\beta\leq\alpha(x,y)$, the coefficient of
$u^\beta v^{\alpha(x,y)-\beta}$ on the LHS of \eqref{eq:box-polynomial} equals
\[
\kappa_{xy}\prod_{i\in\mathcal I}
\binom{h_i(x,y)}{\beta_i}>0.
\]
Hence the monomial $u^\beta v^{\alpha(x,y)-\beta}$ also appears on the RHS of \eqref{eq:box-polynomial}.
Thus, there exists $z\in[n]$ such that $\alpha(x,z)=\beta$ and $\alpha(z,y)=\alpha(x,y)-\beta.$
\end{proof}

\subsection{Complementary Indices and the Distance \texorpdfstring{$\Delta_\omega(\cdot,\cdot)$}{Delta w(.,.)}}\label{subsec: complenmentray indices}

In this subsection, we will define the binary functions $\Delta_\omega(\cdot,\cdot)$ (\cref{def:central-gradients}) and show that the zero relation 
$\Delta_\omega(\cdot,\cdot) =0$ is an equivalence relation (\cref{lem: expression for Delta omega}).
This is achieved by deriving a closed form for $\Delta_\omega(x,y)$ (see \eqref{eq:squared-distance}).
\cref{lem: complementary degree} to \cref{lem:positive-gradients} are all preparations for this closed form.
Note that $\Delta_\omega(\cdot,\cdot)$ may not be a distance on $[n].$
In the following section, we will prove that $\Delta_\omega$ is a distance on the set of equivalence classes of its zero relation.

\begin{lemma}\label{lem: complementary degree}
    $\sum_{r\in \mathcal{F}_+}\ell_r(E)=d$, where $d$ is the common degree in \cref{lem:common-degree}.
\end{lemma}
\begin{proof}
    By \cref{lem:common-degree}, $\sum_{a\in \mathcal{F}}\ell_a(X)=d$ for every $X\in\mathcal{A}.$
    Since $\ell_i(E)=0$ for every $i\in \mathcal{I}$,
    evaluating the identity at $E$ gives $\sum_{r\in\mathcal F_+}\ell_r(E)=d$.
\end{proof}

\begin{lemma}\label{lem: trace FX}
For every $X\in \mathcal{A}$,
$$\operatorname{Tr}F(X)
=\prod_{r\in\mathcal F_+}\left(\frac{\ell_r(X)}{\ell_r(0)}\right)^{\ell_r(E)}.$$
\end{lemma}
\begin{proof}
    By the affine factorization (\cref{lem:affine-factors}), 
    \[
    F(X)_{xx}=\frac{1}{n}\prod_{a\in \mathcal{F}}\left(\frac{\ell_a(X)}{\ell_a(0)}\right)^{h_a(x,x)}
    =\frac{1}{n}\prod_{a\in \mathcal{F}}\left(\frac{\ell_a(X)}{\ell_a(0)}\right)^{\ell_a(E)}
    =\frac{1}{n}\prod_{r\in\mathcal{F}_+}\left(\frac{\ell_r(X)}{\ell_r(0)}\right)^{\ell_r(E)}.
    \]
    The second equality holds because $h_a(x,x)=\ell_a(T_{xx})=\ell_a(E)$ for every $a\in \mathcal{F}.$
    Taking trace on $F(X)$ completes the proof.
\end{proof}

\begin{definition}[Gradient, symmetrization and center]\label{def: gradients of affine form}
    For each $r\in\mathcal F_+$, let $G_r\in\mathcal A$ be the unique matrix satisfying
    \(
    \ell_r(X)=\ell_r(0)+\langle G_r,X\rangle
    \)
    for every $X\in\mathcal A$, and let $S_r:=\frac{G_r+G_r^{\tt T}}{2}.$
    We call $G_r$ the \emph{gradient} of $\ell_r$ and $S_r$ its \emph{symmetrization}.
    Define the \emph{center} of $\mathcal A$ by $Z(\mathcal{A}):=\{C\in \mathcal{A}\mid CX=XC\text{ for every }X\in \mathcal{A}\}$.
\end{definition}

\begin{lemma}\label{lem:central-gradients}
For every $r\in\mathcal F_+$, $G_r\in Z(\mathcal A)$.
\end{lemma}

\begin{proof}
To prove commutativity, we form Lie  bracket from 
infinitesimal conjugations as in Lie theory.

Fix arbitrary $Y\in \mathcal{A}$. 
For each real $t$, define
\[
V(t):=E+\sum_{k=1}^{\infty}\frac{t^kY^k}{k!}.
\]
This series is absolutely convergent.
Every partial sum is in $\mathcal{A}$, so $V(t)\in \mathcal{A}$ since the finite dimensional real vector space $\mathcal{A}$ is closed.
Also, it is easy to verify that $V(t)V(-t)=E$ for every $t\in \R.$
\cref{lem: homomorphism property,lem:polynomial-identity} give
\[
F(V(t))F(V(-t))=F(E)=I,
\qquad F(V(t)XV(-t))=F(V(t))F(X)F(V(-t))
\]
for every $t\in \R$ and every $X\in \mathcal{A}.$
Taking traces,
\[
\operatorname{Tr}(F(V(t)XV(-t)))=\operatorname{Tr}(F(V(-t))F(V(t))F(X))=\operatorname{Tr}(F(X)).
\]
Using \cref{lem: trace FX},
\begin{equation}\label{eq: central-gradients}
    \prod_{r\in\mathcal F_+}\left(\frac{\ell_r(V(t)XV(-t))}{\ell_r(0)}\right)^{\ell_r(E)}=\prod_{r\in\mathcal F_+}\left(\frac{\ell_r(X)}{\ell_r(0)}\right)^{\ell_r(E)}\qquad (\forall t\in \R,\,\forall X\in \mathcal{A}).
\end{equation}
Note that for every $t\in \R$, $\ell_r(V(t)XV(-t))$ is also an affine form in the coordinates of $X$.
Hence it is irreducible.
Thus, both sides of \eqref{eq: central-gradients} are unique prime factorizations.
Therefore, for every $t\in \R$ and every $r\in \mathcal{F}_+$, there exists a unique $\rho(r,t)\in \mathcal{F}_+$ and some $c(r,t)\neq 0$ such that $\ell_r(V(t)XV(-t))=c(r,t)\ell_{\rho(r,t)}(X)$ for every $X\in\mathcal{A}$, and $\ell_r(E)=\ell_{\rho(r,t)}(E).$
Taking $X=E$,
%in $\ell_r(V(t)XV(-t))=c(r,t)\ell_{\rho(r,t)}(X)$, 
we obtain $\ell_r(E)=c(r,t)\ell_{\rho(r,t)}(E)$.
Since $\ell_r(E)=\ell_{\rho(r,t)}(E)>0$, $c(r,t)=1$ for every $t\in \R$ and every $r\in \mathcal{F}_+.$
Now fix $r\in \mathcal{F}_+$.
For every $t\in \R$, $\ell_r(V(t)XV(-t))=\ell_{\rho(r,t)}(X)\in \{\ell_r(X)\}_{r\in \mathcal{F}_+}.$
As $t$ varies continuously in $\R$, $\ell_r(V(t)XV(-t))$ varies continuously in the space of real affine forms on $\mathcal{A}.$
However, $\{\ell_r(X)\}_{r\in \mathcal{F}_+}$ is a discrete set in this space.
This implies that $\ell_r(V(t)XV(-t))$ is a fixed point in $\{\ell_r(X)\}_{r\in \mathcal{F}_+}$ for every $t\in \R$.
At $t=0$, it is $\ell_r (V(0)XV(0))=\ell_r(X)$. Therefore,
%Indeed, 
$\ell_r(V(t)XV(-t))=
%\ell_r (V(0)XV(0))=
\ell_r(X),\,\forall t\in \R,\forall X\in \mathcal{A}.$
Thus, for every $t\in \R$ and every $X\in \mathcal{A}$,
\[
\ell_r(0)+\braket{G_r,V(t)XV(-t)}=\ell_r(V(t)XV(-t))=\ell_r(X)=\ell_r(0)+\braket{G_r,X}.
\]
This implies
\[
\left.\frac{d}{dt}\braket{G_r,V(t)XV(-t)}\right|_{t=0}=0\qquad  (\forall X\in\mathcal{A}).
\]
Note that $V(t)XV(-t)=(E+tY+O(t^2))X(E-tY+O(t^2))=E+t(YX-XY)+O(t^2)$ as $t\to 0.$
So $\left.\frac{d}{dt}V(t)XV(-t)\right|_{t=0}=YX-XY.$
Therefore,
\[
0=\langle G_r,YX-XY\rangle=\operatorname{Tr}(G_r^{\tt T}(YX-XY))
=\operatorname{Tr}\bigl((G_r^{\tt T}Y-YG_r^{\tt T})X\bigr)
\quad(\forall X,Y\in\mathcal A).
\]
It follows that $G_r^{\tt T}Y-YG_r^{\tt T}=0$ for every $Y\in \mathcal{A}.$
Since $\mathcal{A}$ is transpose closed, $G_r$ commutes with every $Y\in \mathcal{A}$. 
That is, $G_r\in Z(\mathcal{A}).$
\end{proof}

\begin{lemma}\label{lem: Gr and Sr}
    If $X\in \mathcal{A}$ is symmetric, then $\braket{G_r,X}=\braket{S_r,X}.$
\end{lemma}
\begin{proof}
    Direct computation shows $\braket{S_r,X}=\braket{\frac{G_r+G_r^{\tt T}}{2},X}=\frac{1}{2}\braket{G_r,X}+\frac{1}{2}\braket{G_r,X^{\tt T}}=\braket{G_r,X}$.
\end{proof}

\begin{lemma}\label{lem: nonzero Sr}
    For every $r\in \mathcal{F}_+$, $S_r\neq 0.$
\end{lemma}
\begin{proof}
    By \cref{lem: F transpose,lem: homomorphism property}, $F(X)^{\tt T}F(X)=F(X^{\tt T})F(X)=F(X^{\tt T}X)$ for every $X\in \mathcal{A}.$
    Taking traces and applying \cref{lem: trace FX},
    \begin{equation}\label{eq: nonzero Sr}
        \|F(X)\|^2=\operatorname{Tr}(F(X)^{\tt T}F(X))=\operatorname{Tr}F(X^{\tt T}X)=\prod_{r\in\mathcal F_+}\left(\frac{\ell_r(X^{\tt T}X)}{\ell_r(0)}\right)^{\ell_r(E)}.
    \end{equation}
    Since $X^{\tt T}X$ is symmetric, $\ell_r(X^{\tt T}X)=\ell_r(0)+\braket{G_r,X^{\tt T}X}=\ell_r(0)+\braket{S_r,X^{\tt T}X}$ by \cref{lem: Gr and Sr}.
    Substitution into \eqref{eq: nonzero Sr} gives
    \[
    \|F(X)\|^2=\prod_{r\in\mathcal F_+}\left(1+\frac{\braket{S_r,X^{\tt T}X}}{\ell_r(0)}\right)^{\ell_r(E)}.
    \]
    The LHS is a polynomial of degree $2d$ in the coordinates of $X$, where $d$ is the common degree in \cref{lem:common-degree}.
    By \cref{lem: complementary degree}, $\sum_{r\in\mathcal F_+}\ell_r(E)=d$.
    If some $S_r=0$, the RHS would have degree less than $2\sum_{r\in \mathcal{F}_+}\ell_r(E)=2d$, a contradiction.
    Therefore, $S_r\neq 0$ for every $r\in \mathcal{F}_+.$
\end{proof}

\begin{lemma}\label{lem:positive-gradients}
For every $r\in\mathcal F_+$, $S_r$ is positive semi-definite.
\end{lemma}
\begin{proof}
Let $Z_{\operatorname{sym}}(\mathcal{A}):=\{C\in Z(\mathcal{A})\mid C^{\tt T}=C\}$ be the symmetric center of $\mathcal{A}.$
Matrices in $Z_{\operatorname{sym}}(\mathcal{A})$ commute with each other.
Apply the simultaneous spectral theorem (\cref{thm:simultaneous-spectral}), we obtain a common orthonormal eigenbasis $\{v_1,v_2,\ldots,v_n\}$ for matrices in $Z_{\operatorname{sym}}(\mathcal{A})$.
We want to obtain a ``common eigenspace decomposition" of $\mathbb{R}^n$ for all matrices in $Z_{\operatorname{sym}}(\mathcal{A})$, defined as the finest joint eigenspace decomposition across all $C \in Z_{\operatorname{sym}}(\mathcal{A})$.
Specifically, define the relation ``$\sim$" on $[n]$ by: $i\sim j$ iff $v_i$ and $v_j$ belong to the same eigenspace for every $C\in Z_{\operatorname{sym}}(\mathcal{A}).$
By definition, ``$\sim$" is an equivalence relation.
Let $I_0,I_1,\ldots, I_b$ be the equivalent classes of ``$\sim$''.
Note that $E\in Z_{\operatorname{sym}}(\mathcal{A})$.
Also, $\mathcal{S}$ and $\mathcal{N}$ are eigenspaces corresponding to the eigenvalues $1$ and $0$ of $E$, respectively.
After reordering, we may assume that $\mathcal{S}=\operatorname{Span}_{\R}\{v_1,\ldots, v_m\}$ and $\mathcal{N}=\operatorname{Span}_{\R}\{v_{m+1},\ldots, v_n\}$.
\begin{claim}\label{claim: 0 eq class}
    The subset $\{m+1,m+2,\ldots,n\}$
    is an equivalent class of $\sim$. We will name it  $I_0$.
    %, such that $I_0=\{m+1,m+2,\ldots,n\}.$
\end{claim}
\begin{claimproof}{\ref{claim: 0 eq class}}
    For every $m<j\le n$, $v_j\in \mathcal{N}$.
    So $Cv_j=0$ for every $C\in \mathcal{A}.$
    Hence $v_j$ belongs to the 0-eigenspace of $C$, for every $C\in  Z_{\operatorname{sym}}(\mathcal{A})$.
    Thus, $\{m+1,\ldots, n\}$ are in the same equivalent class.
    For every $1\le i\le m$, $v_i$ belong the 1-eigenspace of $E$, so $i$ is not equivalent to $j$ where $m<j\le n.$
    Therefore, we may assume that  $I_0=\{m+1,m+2,\ldots,n\}.$
\end{claimproof}

By Claim~\ref{claim: 0 eq class}, we may assume that $[m]=I_1\sqcup I_2\sqcup \ldots\sqcup I_b.$
For every $i\in [b]$,
define $\mathcal{H}_i:=\operatorname{Span}_{\R}\{v_j\mid j\in I_i\}.$
Then $\mathcal{S}=\mathcal{H}_1\oplus\cdots\oplus\mathcal{H}_b$ and $\R^n=\mathcal{N}\oplus \mathcal{H}_1\oplus\cdots\oplus\mathcal{H}_b$.
This is the common eigenspace decomposition of $\R^n.$
\begin{claim}\label{claim: generic T}
    There exists $T\in Z_{\operatorname{sym}}(\mathcal{A})$ such that $Tv=\lambda_j v$ for every $v\in \mathcal{H}_j$ and $j\in [b]$, where $\lambda_1,\lambda_2,\ldots,\lambda_b\in \R$ are nonzero and 
    %mutually 
    pairwise distinct.
\end{claim}
\begin{claimproof}{\ref{claim: generic T}}
    Let $\{A_1,\ldots, A_k\}$ be a basis of $Z_{\operatorname{sym}}(\mathcal{A})$.
    Then there exists $\mu_{ij}\in \R$ such that $A_iv=\mu_{ij}v$ for every $v\in \mathcal{H}_j$ and all $i\in[k],j\in[b].$
    Let $T=\sum_{i=1}^kc_iA_i$, where $c_1,\ldots, c_k$ are to be determined.
    We have $Tv=\sum_{i=1}^kc_i\mu_{ij}v$, for every $v\in \mathcal{H}_j$ and $j\in[b].$
    For $c=(c_1,\ldots, c_k)$ and $j\in[b],$
    let $\lambda_j(c):=\sum_{i=1}^kc_i\mu_{ij}.$
    Then $Tv=\lambda_j(c)v$, for every $v\in \mathcal{H}_j$ and $j\in[b].$
    We want to find some $c\in \R^k$ such that $\{\lambda_j(c)\}_{j\in[b]}$ are nonzero and mutually distinct.
    To see this, let $\mathcal{U}_j:=\{c\in \R^k\mid \lambda_j(c)=0\}$ and $\mathcal{V}_{j\ell}:=\{c\in \R^k \mid \lambda_j(c)=\lambda_\ell(c)\}$, for all $j,\ell\in [b]$.
    For every $j\in[b]$, $(\mu_{1j},\ldots,\mu_{kj})\neq 0$. 
    Otherwise, for some $j\in[b]$, $(\mu_{1j},\ldots,\mu_{kj})= 0$.
    For this $j$ and for every $v\in \mathcal{H}_j$, we have $A_iv=\mu_{ij}v=0 \,(\forall i\in [k])$, and thus $v\in \mathcal{N}$.
    Contradiction.
    Therefore, $\mathcal{U}_j$ is a proper subspace of $\R^k$ for every $j\in [b]$.
    For all distinct $j,\ell\in[b]$, $(\mu_{1j},\ldots,\mu_{kj})\neq (\mu_{1\ell},\ldots,\mu_{k\ell})$.
    Otherwise, $A_iv=\mu_{ij}v=\mu_{i\ell}v$ for every $v\in \mathcal{H}_j$ or $v\in \mathcal{H}_{\ell}$.
    Since $\{A_i\}_{i\in [k]}$ is a basis of $Z_{\operatorname{sym}}(\mathcal{A})$, this would imply that
    every $C\in  Z_{\operatorname{sym}}(\mathcal{A})$ acts the same way
    on $\mathcal{H}_j \cup \mathcal{H}_{\ell}$, which would contradict
    that $I_j$ and $I_{\ell}$ are distinct equivalence classes of $\sim$.
    %that $j=\ell,$ a contradiction.
    Therefore, $\mathcal{V}_{j\ell}$ is a proper subspace of $\R^k$ for all distinct $j,\ell\in [b]$.
    Because a finite dimensional real vector space cannot be the union of finitely many proper subspaces, there exists $c\in\R^k$ such that $\lambda_j(c)\neq 0$ for every $j\in [b]$ and $\lambda_j(c)\neq \lambda_{\ell}(c)$ for all distinct $j,\ell\in[b].$
    Letting $T=\sum_{i=1}^kc_i A_i$ finishes the proof.
\end{claimproof}

For every $j\in [b]$, let $R_j$ be the orthogonal projection from $\R^n$ onto $\mathcal{H}_j.$
\begin{claim}\label{claim: Rj in symmetric center}
    For every $j\in[b]$, $R_j\in Z_{\operatorname{sym}}(\mathcal{A}).$
\end{claim}
\begin{claimproof}{\ref{claim: Rj in symmetric center}}
    Let $T\in Z_{\operatorname{sym}}(\mathcal{A})$ be the matrix in Claim~\ref{claim: generic T}.
    Then $Tv=\lambda_j v$ for every $v\in \mathcal{H}_j$ and $j\in [b]$, where $\lambda_1,\lambda_2,\ldots,\lambda_b\in \R$ are nonzero and mutually distinct.
    For every $j\in [b]$, define the polynomial $$p_j(t):=\frac{t}{\lambda_j}\prod_{1\le i\le b,i\neq j}\frac{t-\lambda_i}{\lambda_j-\lambda_i}.
    %\qquad (\forall X\in Z_{\operatorname{sym}}(\mathcal{A})).
    $$
    Then $p_j(0)=0$,
    and $p_j(\lambda_i)=\mathbb{I}[i=j]$ for all $i,j\in[b].$
    For every $j\in[b]$,
    $p_j(T)$ is a polynomial in $T$ with a zero constant term.
    Consequently, for every $v \in \mathcal{N}$ we have
    $p_j(T)v=0$, and for   every $v\in \mathcal{H}_i$, $p_j(T)v=p_j(\lambda_i)v=\mathbb{I}[i=j]v$ for all $i,j\in[b]$.
    Hence $p_j(T)=R_j$ for every $j\in [b].$
    Because $p_j$ is a polynomial vanishing at $0$, and $T$ belongs to the algebra $Z_{\operatorname{sym}}(\mathcal{A})$, we have $R_j=p_j(T)\in Z_{\operatorname{sym}}(\mathcal{A})$, for every $j\in [b].$
\end{claimproof}
Recall that $S_r=\frac{G_r+G_r^{\tt T}}{2}$ is the symmetrization of the gradient matrix $G_r$, for every $r\in \mathcal{F}_+.$
By \cref{lem:central-gradients}, $G_r\in Z(\mathcal{A})$.
Since $\mathcal{A}$ is transpose closed by \cref{lem:universal-algebra}, $G_r^{\tt T}\in Z(\mathcal{A}).$
Thus, $S_r\in Z_{\operatorname{sym}}(\mathcal{A})$, for every $r\in \mathcal{F}_+.$
%The matrices $\{S_r\}_{r\in \mathcal{F}_+}$ have a common orthonormal eigenbasis $\{v_1,\ldots, v_n\}.$
Therefore, for every $r\in \mathcal{F}_+$ and every $j\in[b]$, there exists $\sigma_{rj}\in \R$ such that $S_rv=\sigma_{rj}v$ for every $v\in \mathcal{H}_j.$
For every $r\in\mathcal{F}_+$, $S_r=\sum_{j=1}^b\sigma_{rj}R_j$.
Next we show the eigenvalues $\sigma_{rj}$ are nonnegative.

For every $t=(t_1,\ldots, t_b)\in \R^b$, let $X(t):=\sum_{j=1}^b t_j R_j.$
By Claim~\ref{claim: Rj in symmetric center}, $R_j\in Z_{\operatorname{sym}}(\mathcal{A})\subseteq \mathcal{A}.$
Hence $X(t)\in \mathcal{A}$ for every $t\in \R^b.$
Note that each entry of $F(X(t))$ is a multivariate polynomial of degree $d$ in the coordinates of $X(t)$, so it is also a multivariate polynomial of the same degree in $t=(t_1,\ldots,t_b).$
Expand $F(X(t))$ as
\begin{equation}\label{eq: FXt}
    F(X(t))=\sum_{\alpha\in\mathbb Z_{\geq0}^{b},\,|\alpha|\leq d}B_\alpha t_1^{\alpha_1}t_2^{\alpha_2}\cdots t_b^{\alpha_b} ,
\end{equation}
where $\alpha=(\alpha_1,\ldots,\alpha_b),\,B_\alpha \in \R^{n\times n}$ and $|\alpha|=\sum_{j=1}^{b}\alpha_j$.
By \cref{lem: homomorphism property}, $F(X(s)X(t))=F(X(s))F(X(t))$ for every $s,t\in \R^b$.
Since $X(s)X(t)=\sum_{j=1}^{b}s_jt_jR_j$, substitution into \eqref{eq: FXt} and using $F(X(s)X(t))=F(X(s))F(X(t))$ give
\[
\sum_{\alpha\in \Z_{\ge 0}^b,\,|\alpha|\le d}B_\alpha(s_1t_1)^{\alpha_1}(s_2t_2)^{\alpha_2}\cdots (s_bt_b)^{\alpha_b} =\sum_{\substack{\alpha\in \Z_{\ge 0}^b,\,|\alpha|\le d\\\beta\in \Z_{\ge 0}^b,\,|\beta|\le d}}B_\alpha B_{\beta}s_1^{\alpha_1}s_2^{\alpha_2}\cdots s_b^{\alpha _b}t_1^{\beta_1}t_2^{\beta_2}\cdots t_b^{\alpha_b},
\]
for every $s,t\in \R^b$.
Comparing coefficients gives
\[
B_\alpha B_\beta=
\begin{cases}B_\alpha,&\alpha=\beta,\\0,&\alpha\ne\beta,
\end{cases}
\]
for all $\alpha,\beta$ displayed in the above equation. 
Also, $F(X(t))^{\tt T}=F(X(t)^{\tt T})$ by \cref{lem: F transpose}.
Expanding shows $B_\alpha=B_{\alpha}^{\tt T}$ for every $\alpha\in \Z_{\ge 0}^b$, $|\alpha|\le d.$
Therefore, $\{B_{\alpha}\mid \alpha\in \Z_{\ge 0}^b,|\alpha|\le d\}$ are orthogonal projections.
Then $\operatorname{Tr}(B_\alpha)=\operatorname{rank}(B_{\alpha})$ for each $\alpha.$
%Since $F(X(t))$ is not the zero polynomial, there exists some $\alpha$ such that $B_\alpha\neq 0.$
At $t = (0,\ldots, 0)\in \R^b$, we have $X(0) = 0$, and $B_0 = F(0) = U$ of rank 1.
Taking trace on \eqref{eq: FXt},
\begin{equation}\label{eq: TrFXt 1}
    \operatorname{Tr}F(X(t))
=\sum_{\alpha\in\mathbb Z_{\geq0}^{b},\,|\alpha|\leq d}
(\operatorname{rank}B_\alpha)t_1^{\alpha_1}t_2^{\alpha_2}\cdots t_b^{\alpha_b}\geq1
\qquad(\forall t\in\mathbb R_{\geq0}^{b}).
\end{equation}
On the other hand, \cref{lem: trace FX} gives 
\begin{equation}\label{eq: TrFXt}
    \operatorname{Tr}F(X(t))
=\prod_{r\in\mathcal F_+}\left(\frac{\ell_r(X(t))}{\ell_r(0)}\right)^{\ell_r(E)}\qquad (\forall t\in \R^b).
\end{equation}
Since $X(t)$ is symmetric for every $t\in \R^b$, $\braket{G_r,X(t)}=\braket{S_r,X(t)}$ by \cref{lem: Gr and Sr}.
Hence 
\[
\begin{aligned}  \ell_r(X(t))=&\ell_r(0)+\braket{G_r,X(t)}=\ell_r(0)+\braket{S_r,X(t)}\\
=&\ell_r(0)+\left\langle{\sum_{j=1}^b \sigma_{rj}R_j,\,\sum_{j=1}^bt_jR_j}\right\rangle\\
=& \ell_r(0)+\sum_{j=1}^b\sigma_{rj}\operatorname{rank}(R_j)t_j.
\end{aligned}
\]
Plugging into \eqref{eq: TrFXt} gives
\[
\operatorname{Tr}F(X(t))=
\prod_{r\in\mathcal F_+}
\left(1+\sum_{j=1}^{b}\frac{\sigma_{rj}\operatorname{rank}(R_j)}{\ell_r(0)}t_j\right)^{\ell_r(E)}.
\]
If $\sigma_{rj}<0$ for some $r\in\mathcal{F}_+$ and $j\in [b]$, 
set $t_j=-\ell_r(0)/(\sigma_{rj}\operatorname{rank}(R_j))>0$ and $t_i=0$ for every $i\in[b]\setminus\{j\}$.
Then $\operatorname{Tr}F(X(t))=0.$
This contradicts \eqref{eq: TrFXt 1}. Thus, $\sigma_{rj}\geq0$ for every $r\in \mathcal{F}_+$ and every $j\in [b]$.
This proves $S_r$ is positive semi-definite for every $r\in \mathcal{F}_+.$
\end{proof}

Recall that $\tau:\mathcal{F}\to \mathcal{F}$ is the involution in \cref{lem:factor-transpose} such that $\ell_{\tau(a)}(X)=\ell_a(X^{\tt T})$ for every $a\in \mathcal{F}.$
Evaluating at $X=E$, $\ell_{\tau(a)}(E)=\ell_a(E),\,\forall a\in\mathcal{F}.$
Thus, $\ell_{\tau(a)}(E)>0\iff \ell_a(E)>0.$
This implies $\tau|_{\mathcal{F}_+}$ is an involution on $\mathcal{F}_+.$
It partitions
 $\mathcal{F}_+$  into orbits of sizes one or two.
 Let $\Omega$ be the set of such orbits. 
We are ready to
%to partition the set of complementary indices $\mathcal{F}_+$ into a set of orbits under the involution $\tau:\mathcal{F}\to \mathcal{F}$ is the 
%%% this espression is kind of funny, see \Omega appears outside and inside
%$\Omega=\{\omega\mid \omega\in \Omega\}$ 
%$\Omega$, such that 
%and 
define the  binary function
%%% distance 
%%% you actually said they are not quite distances on [n], (but on L_\omega classes. so here we don't say "distance".
$\Delta_{\omega}(\cdot,\cdot)$ for every $\omega\in \Omega.$

\begin{definition}\label{def:central-gradients}
Let $\omega_r$ 
% $\omega_r=\{r,\tau(r)\}$ 
%%% this notation suggests it is of size 2
be the orbit of $r\in \mathcal{F}_+$ induced by the involution $\tau|_{\mathcal{F}_+}.$ 
% It has size one or two.
Let $\Omega=\{\omega_r\mid r\in \mathcal{F}_+\}$ be the set of orbits. 
For $\omega\in\Omega$ and $x,y\in[n]$, define
\begin{equation}
W_\omega=\sum_{r\in\omega}G_r,\qquad
\Delta_\omega(x,y)=\sum_{r\in\omega}\bigl(\ell_r(E)-h_r(x,y)\bigr).
\label{eq:decrements}
\end{equation}
\end{definition}

\begin{remark}\label{rmk: Ww}
    If for an  $r\in \mathcal{F}_+$, $\tau(r)=r$, then $\omega_r=\{r\}$ has size one; otherwise, $\omega_r$ has size two.
    In the former case, $W_{\omega_r}=G_r=G_r^{\tt T}=S_r$; in the latter case, $W_{\omega_r}=G_r+G_r^{\tt T}=2S_r.$
    In either case,
     $W_\omega\in Z_{\operatorname{sym}}(\mathcal{A})$ is nonzero and positive semi-definite for every $\omega\in \Omega$ by \cref{lem: nonzero Sr,lem:positive-gradients}.
     Thus, there exists a unique positive semi-definite matrix $W_{\omega}^{\frac{1}{2}}$ such that $(W_{\omega}^{\frac{1}{2}})^2=W_{\omega}$, $\forall \omega\in \Omega.$
\end{remark}

\begin{definition}
    For every $\omega\in \Omega$, define the relation $L_\omega$ on $[n]$ by: $xL_\omega y\iff \Delta_\omega(x,y)=0.$
\end{definition}

\begin{lemma}\label{lem: expression for Delta omega}
For every $\omega\in\Omega$, the relation $L_\omega$ is an equivalence relation on $[n]$. 
Furthermore, there are at least two equivalent classes.
More precisely, 
\begin{equation}
\Delta_\omega(x,y)=\operatorname{Tr}W_\omega-n(W_\omega)_{xy}
=\frac n2\left\|W_\omega^{1/2}\mathbf e_x-W_\omega^{1/2}\mathbf e_y\right\|^2
\quad(\forall x,y\in[n]).
\label{eq:squared-distance}
\end{equation}
%where $W_\omega$ is a nonzero positive semidefinite matrix in the symmetric center of $\mathcal A$.
\end{lemma}
%%%% JYC here here edited.
\begin{proof}
Fix $\omega\in \Omega$.
For all $x,y\in[n]$, 
\[
\begin{aligned}
    \Delta_\omega(x,y)=&\sum_{r\in \omega}(\ell_r(E)-h_r(x,y))\\
    =&\sum_{r\in \omega}(\ell_r(E)-\ell_r(T_{xy}))&&\text{by \cref{lem:affine-multiplicities}}\\
    =&\sum_{r\in \omega}\braket{G_r,E-T_{xy}}&&\text{by \cref{def: gradients of affine form}}\\
    =&\langle W_\omega,E-T_{xy}\rangle=\operatorname{Tr}W_\omega-\braket{W_{\omega},T_{xy}}\\
    =&\operatorname{Tr}(W_\omega)-n(W_\omega)_{xy} &&\text{by \cref{lem:statistic-pairing}}.
\end{aligned}
\]
On the other hand,
\[
\begin{aligned}
    \left\|W_\omega^{1/2}\mathbf e_x-W_\omega^{1/2}\mathbf e_y\right\|^2
    =&(\mathbf{e}_x-\mathbf{e}_y)^{\tt T}(W_{\omega}^{\frac{1}{2}})^{\tt T}W_{\omega}^{\frac{1}{2}}(\mathbf{e}_x-\mathbf{e}_y)\\
    =&(\mathbf{e}_x-\mathbf{e}_y)^{\tt T}W_{\omega}(\mathbf{e}_x-\mathbf{e}_y)\\
    =&(W_\omega)_{xx}+(W_\omega)_{yy}-2(W_\omega)_{xy}.
\end{aligned}
\]
Since $W_\omega\in\mathcal A$ has constant diagonal by \cref{lem:normalize_rumsum_colsum_to_0}, 
$(W_\omega)_{xx}=(W_\omega)_{yy}=\frac{1}{n}\operatorname{Tr}W_\omega$.
Thus, 
\[
\frac n2\left\|W_\omega^{1/2}\mathbf e_x-W_\omega^{1/2}\mathbf e_y\right\|^2
%=\frac{n}{2}\left(\frac{2\operatorname{Tr}W_\omega}{n}-2(W_\omega)_{xy}\right)
=\operatorname{Tr}W_\omega-n(W_\omega)_{xy}=
\Delta_\omega(x,y),
\]
for all $x,y\in [n].$

Thus $xL_\omega y\iff W_\omega^{1/2}\mathbf e_x = W_\omega^{1/2}\mathbf{e}_y$. This proves self-reflexivity, symmetry and transitivity of the relation $L_\omega$. 
So $L_\omega$ is an equivalence relation.
If $L_\omega$ had only one equivalent class, then there exists $\mathbf{v}\in\R^n$ such that $(W_\omega)^{\frac{1}{2}}\mathbf{e}_x=\mathbf{v}$ for every $x\in [n].$
Since $W_\omega\in \mathcal{A}$ has zero row sum by \cref{lem:normalize_rumsum_colsum_to_0}, $W_\omega\mathbf{1}=0$.
It follows that $\|(W_\omega)^{\frac{1}{2}}\mathbf{1}\|^2=\mathbf{1}^{\tt T}W_\omega\mathbf{1}=0$.
So $(W_\omega)^{\frac{1}{2}}\mathbf{1}=0$.
Hence $n\mathbf{v}=\sum_{x\in [n]}(W_\omega)^{\frac{1}{2}}\mathbf{e}_x=(W_\omega)^{\frac{1}{2}}\mathbf{1}=0.$
Consequently, $(W_\omega)^{\frac{1}{2}}\mathbf{e}_x=\mathbf{v}=0$ for every $x\in [n].$
This means every column of $(W_\omega)^{\frac{1}{2}}$ is 0.
This would imply $W_{\omega}=W_{\omega}^{\frac{1}{2}}=0$, contradicting \cref{rmk: Ww}.
\end{proof}

\begin{definition}\label{def: eq classes}
    For every $\omega\in\Omega$, let $X_\omega=[n]/L_\omega$ be the set of equivalence classes of the equivalence relation $L_\omega.$
    We have $|X_\omega|\ge 2$ by \cref{lem: expression for Delta omega}.
    For every $x \in [n]$ and $\omega \in \Omega$, we denote  by $[x]_{L_{\omega}}$ the equivalence class of $x$ modulo the relation $L_\omega$. Where no confusion can arise, we write simply $x_\omega$.
\end{definition}

Recall that $\{B_i\}_{i\in \mathcal{I}}$ is a basis of $\mathcal{A}$ from \cref{def:endpoint-basis}.
\begin{definition}\label{def:owned-endpoints}
For every $\omega\in\Omega$, let
\(
\mathcal I_\omega=\{i\in\mathcal I:\langle W_\omega,B_i\rangle=1\},
\)
\end{definition}

\begin{lemma}\label{lem:owned-atoms}
The sets $(\mathcal I_\omega)_{\omega\in\Omega}$ form a partition of $\mathcal I$. 
Moreover, for every $\omega\in\Omega$ and $x,y\in[n]$,
\begin{equation*}
\Delta_\omega(x,y)=\sum_{i\in\mathcal I_\omega}h_i(x,y),
\qquad
\sum_{\omega\in\Omega}\Delta_\omega(x,y)=\delta(x,y).
\end{equation*}
\end{lemma}

\begin{proof}
By definition, $\Delta_\omega(x,y)\in\Z$ for every $\omega\in \Omega$ and all $x,y\in [n].$
By \cref{lem: expression for Delta omega}, $\Delta_\omega(x,y)\in\Z_{\ge 0}.$
Also, in \cref{lem: expression for Delta omega} we proved
\(
\Delta_{\omega}(x,y)=\braket{W_\omega,E-T_{xy}}.
\)
\begin{claim}\label{claim: owned-atoms}
    For every $\omega\in \Omega$ and $i\in \mathcal{I}$, $\braket{W_\omega,B_i}\in \{0,1\}.$
    Moreover, $\sum_{\omega\in \Omega}\braket{W_\omega,B_i}=1,\,\forall i\in\mathcal{I}.$
\end{claim}
\begin{claimproof}{\ref{claim: owned-atoms}}
    Fix $i\in \mathcal{I}$.
    By \cref{lem:endpoint-coordinates}, there exist $x,y\in [n]$ such that $E-T_{xy}=B_i$ and $\alpha(x,y)=\mathbf{u}_i$, where $\mathbf{u}_i\in \mathbb{Z}_{\ge 0}^{|\mathcal{I}|}$ is the unit vector at $i.$
    Hence \(\braket{W_\omega,B_i}=\braket{W_\omega,E-T_{xy}}=\Delta_\omega(x,y)\in \mathbb{Z}_{\ge 0}\).
    Also,
    \[
    \begin{aligned}
        \sum_{\omega\in \Omega}\braket{W_\omega,B_i}=&\sum_{\omega\in \Omega}\Delta_\omega(x,y)\\
        =&\sum_{r\in \mathcal{F}_+}(\ell_r(E)-h_r(x,y))&&\text{by \cref{def:central-gradients}}\\
        =&~ d-\sum_{r\in \mathcal{F}_+}h_r(x,y)&&\text{by \cref{lem: complementary degree}}\\
        =&\sum_{j\in \mathcal{I}}h_j(x,y)&&\text{by \cref{lem:common-degree}}\\
        =&~ 1 &&\text{by $\alpha(x,y)=\mathbf{u}_i$}.
    \end{aligned}
    \]
    Therefore, $\braket{W_\omega,B_i}\in \{0,1\}$ for every $\omega\in \Omega$ and every $i\in \mathcal
    I$.
\end{claimproof}
By Claim~\ref{claim: owned-atoms}, each $i\in \mathcal{I}$ falls into exactly one $\mathcal{I}_\omega$, for some $\omega\in \Omega.$
Thus, $(\mathcal{I}_\omega)_{\omega\in\Omega}$ forms a partition of $\mathcal{I}.$
For every $\omega\in \Omega$ and all $x,y\in [n]$,
\[
\begin{aligned}
    \Delta_\omega(x,y)=&\braket{W_\omega,E-T_{xy}}\\
    =& \left\langle{W_\omega,\sum_{i\in \mathcal{I}}h_i(x,y)B_i}\right\rangle&&\text{by \cref{lem: E-Txy}}\\
    =& \sum_{i\in\mathcal{I}}h_i(x,y)\braket{W_\omega,B_i}\\
    =&\sum_{i\in \mathcal{I}_\omega}h_i(x,y) &&\text{by definition of $\mathcal{I}_\omega$ and Claim~\ref{claim: owned-atoms}}.
\end{aligned}
\]
This proves the first statement of the lemma.
For the second, $$\sum_{\omega\in \Omega}\Delta_\omega(x,y)=\sum_{\omega\in \Omega}\sum_{i\in \mathcal{I}_\omega}h_i(x,y)=\sum_{i\in \mathcal{I}}h_i(x,y)=\delta(x,y).$$
\end{proof}

\subsection{The Color Bijection}\label{subsec: color bijection}

In this subsection, we show that the map $\phi:[n]\to \prod_{\omega\in \Omega}X_\omega,\,x\mapsto ([x]_{L_{\omega}})_{\omega\in \Omega}$ is a bijection (\cref{lem: phi is bijection}).
Also, we show that $\Delta_\omega(x,y)$ only depends on $[x]_{L_\omega}$ and $[y]_{L_\omega}$ (\cref{lem:coordinate-locality}).
Thus, $\Delta_\omega(\cdot,\cdot)$ is a well-defined binary function on $X_\omega.$
In the next section, we show that it is a distance on $X_\omega$ and, more specifically, a highly regular graph distance.

\begin{definition}\label{def:quotient-coordinates}
For every $S\subseteq\Omega$, define
\(
R_S=\bigcap_{\omega\in\Omega\setminus S}L_\omega,
\)
where an empty intersection is the universal relation $[n]\times[n]$. For relations $R,T$ on $[n]$, write $x(R\circ T)y$ if there is $z\in[n]$ with $xRz$ and $zTy$.
\end{definition}

\begin{lemma}\label{lem:coordinate-composition}
For all $S,T\subseteq\Omega$,  $R_S\circ R_T=R_{S\cup T}$. If $S\cap T=\varnothing$, for $x,y\in[n]$ with $xR_{S\cup T}y$, there exists a unique $z\in[n]$ 
%such that $xRz$ and $zTy$, 
%%%% I thinbk you had a typo there.
such that $xR_Sz$ and $zR_Ty$.
\end{lemma}
\begin{proof}
Fix arbitrary $S,T\subseteq \Omega.$
We first show $R_S\circ R_T\subseteq R_{S\cup T}$.
The binary relation
$R_S$ is to impose the relations of $L_{\omega}$, for all $\omega \in \Omega$ \emph{not} in $S$.
For any $S \subseteq S'$, we have $\Omega \setminus S' \subseteq \Omega \setminus S$.
Thus, $R_S=\bigcap_{\omega\in\Omega\setminus S}L_\omega
\subseteq \bigcap_{\omega\in\Omega\setminus S'}L_\omega  = R_{S'}$.
As both $S, T \subseteq S\cup T$, we have both $R_S \subseteq R_{S\cup T}$ and 
$R_T  \subseteq R_{S\cup T}$. Therefore their composition $R_S\circ R_T\subseteq R_{S\cup T}$,
since $R_{S\cup T}$ is an equivalence relation, which is closed under composition. 

%Suppose $x(R_S\circ R_T )y.$
%Then there exists $z\in [n]$ with $xR_Sz$ and $zR_T y.$
%If $S\cup T=\Omega$, then $R_{S\cup T}=[n]\times [n]$ and the inclusion becomes trivial.
%If $S\cup T\subsetneq\Omega,$ pick $\omega\in \Omega\setminus(S\cup T).$
%Then $\omega\in \Omega\setminus S$ and $\omega\in \Omega\setminus T.$ 
%Hence $xR_Sz$ implies $xL_\omega z.$
%Similarly, $zR_Ty$ implies $zL_\omega y.$
%Since $L_\omega$ is transitive, $xL_\omega y.$
%This holds for arbitrary $\omega\in \Omega\setminus (S\cup T).$
%Thus, $R_S\circ R_T\subseteq R_{S\cup T}.$

Next we show $R_{S\cup T}\subseteq R_S\circ R_T.$
Suppose $xR_{S\cup T} y.$
Then $\Delta_\omega(x,y)=0$ for every $\omega\in \Omega\setminus (S\cup T).$
For every $i\in \mathcal{I},$ define
\[
\beta_i=\begin{cases}
h_i(x,y),&i\in\bigcup_{\omega\in S}\mathcal I_\omega,\\
0,&\text{otherwise}.
\end{cases}
\]
Then $0\le \beta\le \alpha(x,y)$, entrywise.
\cref{lem:endpoint-box} gives  $z\in [n]$ with
$\alpha(x,z)=\beta$ and $\alpha(z,y)=\alpha(x,y)-\beta$.
We claim $xR_S z$ and $zR_Ty.$
Indeed, for every $\omega\in \Omega\setminus S,$ 
\[
\Delta_\omega(x,z)=\sum_{i\in \mathcal{I}_\omega}h_i(x,z)=\sum_{i\in \mathcal{I}_\omega}\beta_i=0.
\]
The first equality holds by \cref{lem:owned-atoms}; the second holds by $\alpha(x,z)=\beta$; while the last holds by the definition of $\beta.$
Thus, $xL_\omega z$ for every $\omega \in \Omega\setminus S$.
That is, $xR_S z.$
Now we check that $zR_T y.$
Fix arbitrary $\omega\in \Omega\setminus T$.
If $\omega\in S$, then 
\[
\Delta_{\omega}(z,y)=\sum_{i\in \mathcal{I}_\omega}h_i(z,y)=\sum_{i\in \mathcal{I}_\omega}(h_i(x,y)-\beta_i)=0.
\]
The first equality holds by \cref{lem:owned-atoms}; the second holds by $\alpha(z,y)=\alpha(x,y)-\beta$; while the last holds by the definition of $\beta.$

If $\omega\notin S,$ then since we assumed that $\omega\notin T$,
we have $\omega\notin S\cup T.$ 
For every $i \in  \mathcal{I}_\omega$, $\beta_i =0$, and $h_i(z,y)= h_i(x,y)$.
We have
\[
\Delta_{\omega}(z,y)=\sum_{i\in \mathcal{I}_\omega}h_i(z,y)=
%\sum_{i\in \mathcal{I}_\omega}(h_i(x,y)-\beta_i)=
\sum_{i\in \mathcal{I}_\omega}h_i(x,y)=\Delta_\omega(x,y)=0.
\]
In either case, $\Delta_\omega(z,y)=0$ for arbitrary $\omega\in \Omega\notin T.$
Hence $zR_T y.$

To prove uniqueness, we first claim $R_S\cap R_T=R_{S\cap T}.$
To see this,
    $$R_S\cap R_T=\bigcap_{\omega\in\Omega\setminus S}L_\omega\cap \bigcap_{\omega\in\Omega\setminus T}L_\omega=\bigcap_{\omega\in(\Omega\setminus S)\cup (\Omega\setminus T)}L_\omega=\bigcap_{\omega\in \Omega\setminus(S\cap T)}L_\omega=R_{S\cap T}.$$
Now suppose $S\cap T=\emptyset$, and there exist $z,w\in [n]$ such that $xR_SzR_Ty$ and $xR_SwR_Ty.$
By symmetry and transitivity, $zR_SxR_Sw$ and $zR_TyR_T w$.
Hence $z(R_S\cap R_T)w.$
By $R_S\cap R_T=R_{S\cap T}=R_{\emptyset}$, $zR_{\emptyset }w.$
That is, $\Delta_\omega(z,w)=0$ for every $\omega\in \Omega.$
By \cref{lem:owned-atoms}, $\delta(z,w)=\sum_{\omega\in \Omega}\Delta_\omega(z,w)=0.$
Since $\delta$ is a distance by \cref{lem:endpoint-distance}, $z=w.$
This finishes the proof of uniqueness.
\end{proof}

\begin{lemma}[Color bijection]\label{lem: phi is bijection}
The map
\begin{equation*}
\phi: [n]\to \prod_{\omega\in \Omega}X_\omega,\qquad x\longmapsto\bigl([x]_{L_\omega}\bigr)_{\omega\in\Omega}
\end{equation*}
is a bijection.
\end{lemma}

\begin{proof}
    We first prove injectivity.
    Suppose $x,y\in [n]$ satisfy $[x]_{L_\omega}=[y]_{L_\omega}$ for every $\omega\in \Omega.$
    Then $\Delta_\omega(x,y)=0 \,(\forall \omega\in \Omega).$
    By \cref{lem:owned-atoms}, $\delta(x,y)=\sum_{\omega\in \Omega}\Delta_\omega(x,y)=0.$
    Since $\delta$ is a distance, $x=y.$
    Next we prove surjectivity.
    Enumerate $\Omega$ by $\Omega=\{\omega_1,\ldots, \omega_k\}.$
    For every $j\in [k]$, pick an arbitrary $[y_j]_{L_{\omega_j}}\in X_{\omega_j}$ with representative $y_j\in [n].$
    We want to show that there exists $x\in [n]$ such that $xL_{\omega_j}y_j$ for every $j\in [k].$
    For arbitrary $\omega\in \Omega$, apply \cref{lem:coordinate-composition} for $S=\{\omega\}$ and $T=\Omega\setminus \{\omega\}.$
    Then $R_{\{\omega\}}\circ L_\omega=R_{\Omega}=[n]\times [n].$
    Thus, for every $x,y\in [n],$ $x(R_{\{\omega\}}\circ L_\omega)y.$
    That is, there exists $z\in [n]$ such that $xR_{\{\omega\}}zL_\omega y.$

    Start from an arbitrary $x_0\in [n],$ the above argument gives $x_1\in [n]$ such that $x_0R_{\{\omega_1\}}x_1L_{\omega_1} y_1.$
    By the same argument, there exists $x_2\in [n]$ with $x_1R_{\{\omega_2\}}x_2L_{\omega_2} y_2.$
    Then $x_2L_{\omega_2}y_2$, and $x_2L_{\omega_1}x_1L_{\omega_1}y_1.$
    Repeat the process, the $k$-th step gives $x_k\in [n]$ with $x_{k-1}R_{\{\omega_k\}}x_kL_{\omega_k} y_k.$
    Then $x_kL_{\omega_k}y_k$, and $x_kL_{\omega_j}x_{k-1}L_{\omega_j}y_j$ for every $1\le j<k.$
    Letting $x=x_k$ proves surjectivity.
\end{proof}

\begin{lemma}\label{lem:coordinate-locality}
For every $\omega\in\Omega$ and $i\in\mathcal I_\omega$, for all $x,y\in[n]$ the value $h_i(x,y)$ depends only on $([x]_{L_\omega},[y]_{L_\omega})$. The same holds for $\Delta_\omega(x,y)$.
\end{lemma}

\begin{proof}
Fix $\omega\in \Omega$ and $i\in \mathcal{I}_\omega.$
We want to show that $h_i(x,y) = h_i(x',y')$ for all $x, y, x', y' \in [n]$ satisfying $x L_\omega x'$ and $y L_\omega y'$.
It suffices to show the following:
\begin{enumerate}
    \item\label{item: locality 1} Fix any $x\in[n]$, $h_i(x,y)=h_i(x,y')$ for all $y,y'\in[n]$ satisfying $yL_\omega y';$ and
    \item\label{item: locality 2} Fix any $y\in[n]$, $h_i(x,y)=h_i(x',y)$ for all $x,x'\in[n]$ satisfying $xL_\omega x';$ 
\end{enumerate}
We only show the first, the other is similar.
Fix any $x\in[n]$, for every $y\in [n]$ and $j\in \mathcal{I}$ define
\[\beta_j(y)=
\begin{cases}
    h_j(x,y),& j\in \mathcal{I}_\omega;\\
    0,& \text{otherwise.}
\end{cases}
\]
Then $0\le \beta(y)\le \alpha(x,y)$ coordinate wise.
Now fix arbitrary $y\in[n].$
By \cref{lem:endpoint-box}, there exists $z\in [n]$ such that $\alpha(x,z)=\beta(y)$ and $\alpha(z,y)=\alpha(x,y)-\beta.$
Hence for $\nu \in \Omega\setminus\{\omega\}$, $$\Delta_{\nu}(x,z)=\sum_{j\in \mathcal{I}_\nu}h_j(x,z)=\sum_{j\in \mathcal{I}_{\nu}}\beta_j(y)=0.$$
The first equality holds by \cref{lem:owned-atoms}; the second holds by $\alpha(x,z)=\beta(y);$ while the third holds because $\nu\neq \omega$ and $\{\mathcal{I}_{\eta}\}_{\eta\in \Omega}$ is a partition of $\mathcal{I}.$
Thus, $xL_\nu z$ for every $\nu\in \Omega\setminus \{\omega\}.$
This implies $xR_{\{\omega\}}z.$
Meanwhile, 
\[
\Delta_\omega(z,y)=\sum_{j\in \mathcal{I}_\omega}h_j(z,y)=\sum_{j\in \mathcal{I}_\omega}(h_j(x,y)-\beta_j(y))=0.
\]
The first equality holds by \cref{lem:owned-atoms}; the second holds by $\alpha(z,y)=\alpha(x,y)-\beta(y);$ while the third holds by the definition of $\beta(y).$
Thus, $zL_\omega y.$
%This implies $zR_{\Omega\setminus \{\omega\}}y.$
This is the same as $zR_{\Omega\setminus \{\omega\}}y.$
To summarize, $xR_{\{\omega\}}zR_{\Omega\setminus \{\omega\}}y$ and $h_i(x,y)=\beta_i(y)=h_i(x,z),$ (for the initially fixed $i\in \mathcal{I}_\omega$).

Now suppose $y'\in [n]$ with $yL_\omega y'.$
Repeat the above argument, there exists $z'\in [n]$ such that $xR_{\{\omega\}}z'R_{\Omega\setminus \{\omega\}}y'$ and $h_i(x,y')=\beta_i(y')=h_i(x,z').$
Since $yL_\omega y'$, which is the same as $y' R_{\Omega\setminus \{\omega\}}y$, we have $xR_{\{\omega\}}z'R_{\Omega\setminus \{\omega\}}y.$
Since $\{\omega\}$ and $\Omega\setminus\{\omega\}$ are disjoint, $z=z'$ by \cref{lem:coordinate-composition}.
Therefore, $h_i(x,y)=h_i(x,z)=h_i(x,z')=h_i(x,y').$
This proves statement~\ref{item: locality 1}.
For statement~\ref{item: locality 2}, fix any $y\in [n]$ and define 
\[\beta_j(x)=
\begin{cases}
    0,& j\in \mathcal{I}_\omega;\\
    h_j(x,y),& j\in \mathcal{I}\setminus \mathcal{I}_{\omega}.
\end{cases}
\]
Similarly we can prove $h_i(x,y)=h_i(x',y)$ for any fixed $y\in [n]$ and all $x,x'\in [n]$ with $xL_\omega x'.$
Finally, since $\Delta_\omega(x,y)=\sum_{i\in \mathcal{I}_\omega}h_i(x,y)$, $\Delta_\omega(x,y)$ also depends only on $([x]_{L_\omega},[y]_{L_\omega})$.
\end{proof}

\begin{definition}\label{def:quotient-distances}
For $\omega\in\Omega$ and $a,b\in X_\omega$, define $\Delta_\omega(a,b)$ to be $\Delta_\omega(x,y)$ for any representatives with $[x]_{L_\omega}=a$ and $[y]_{L_\omega}=b$. \cref{lem:coordinate-locality} makes this well-defined.
\end{definition}

\begin{lemma}\label{lem: Delta omega satisfies positivity and symmetry}
    For every $\omega\in \Omega$, $\Delta_\omega(\cdot,\cdot)$ on $X_\omega$ satisfies positivity and symmetry, i.e., for every $a,b\in X_\omega$, (1) $\Delta_\omega(a,b)\ge 0$, and $\Delta_\omega(a,b)=0 \iff a=b$;  (2) $\Delta_\omega(a,b)=\Delta_\omega(b,a).$
\end{lemma}
\begin{proof}
    Fix $\omega\in \Omega$ and arbitrary $a,b\in X_\omega$.
    Since $\phi:[n]\to \prod_{\eta\in\Omega}X_\eta, \, x\mapsto ([x]_{L_\eta})_{\eta\in \Omega}$ is a bijection by \cref{lem: phi is bijection}, we may let $x,y\in [n]$ such that $[x]_{L_\omega}=a$ and $[y]_{L_\omega}=b.$
    Then $\Delta_\omega(a,b)=\Delta_\omega(x,y) \ge 0$ by \cref{def:quotient-distances}.
    Thus, $\Delta_\omega(a,b)=0\iff \Delta_\omega(x,y)=0 \iff xL_\omega y\iff a=b.$
    Also, $\Delta_\omega(x,y)=\Delta_\omega(y,x)$ by \cref{lem: expression for Delta omega}.
    Hence $\Delta_\omega(a,b)=\Delta_\omega(b,a).$
\end{proof}
\section{Recover Hamming Intersection Numbers}\label{sec:hamming}

In the last section, we proved that $\Delta_\omega(\cdot,\cdot)$ is a well-defined binary function on $X_\omega$ for every $\omega\in \Omega$.
This allows us to define the coordinate graphs $\Gamma_\omega$ for $\omega\in\Omega$:

\begin{definition}[Coordinate graphs]\label{def:coordinate-graphs}
For every $\omega\in\Omega$, define the \emph{coordinate graph} $\Gamma_\omega$ as the simple graph on $X_\omega$ in which $a,b\in X_\omega$ are adjacent exactly when $\Delta_\omega(a,b)=1$.
\end{definition}

In this section, we will prove that each coordinate graph $\Gamma_\omega$ is distance-regular with graph distance $\Delta_\omega$ (\cref{lem:coordinate-graph-distance}). 
We then show that the intersection numbers of $\Gamma_\omega$ coincide with those of the Hamming graph $H(D_\omega,Q_\omega)$, by proving that the corresponding generating functions agree (\cref{lem:hamming-numbers}). 
To prepare for this comparison, we begin in \cref{subsec: distance polynomials} by introducing the distance polynomials of Hamming graphs and studying the generating functions for their intersection numbers.
Then in \cref{subsec: Distance Polynomials of Coordinate Graphs}, we similarly define the distance polynomials of the coordinate graphs and study their basic properties.
Finally in \cref{subsec: Hamming Intersection Numbers of Coordinate  Graphs}, we show that each $\Gamma_\omega$ has the same intersection numbers as $H(D_\omega,Q_\omega).$

\subsection{Distance Polynomials of Hamming Graphs}\label{subsec: distance polynomials}

\begin{definition}\label{def:distance-poly-for-GH}
For integers $D\geq1$, $Q\geq2$, and $0\leq k\leq D$, define the \emph{distance polynomial}
\[
f_{D,Q,k}(t):=Q^{-D}(1+(Q-1)t)^{D-k}(1-t)^k
\quad(\forall t\in\mathbb R).
\]
\end{definition}

\begin{lemma}\label{lem: Hamming intersection number generating polynomial}
    Let $p_{ij}^k \,(0\le i,j,k\le D)$ be the intersection numbers of the Hamming graph $H(D,Q)$.
    Then 
    \[\sum_{i=0}^{D}\sum_{j=0}^{D}p^k_{ij}f_{D,Q,i}(s)f_{D,Q,j}(t)=f_{D,Q,k}(st)
\qquad(\forall s,t\in\mathbb R).\]
\end{lemma}
\begin{proof}
Let $I_Q$ be the identity matrix of order $Q$, $J_Q$ be the all one matrix of order $Q$, and $U_Q=J_Q/Q.$
For $t\in \R$, define the $Q\times Q$ matrix $M_Q(t)$ by
\[M_Q(t)_{pq}=\begin{cases}
    \frac{1+(Q-1)t}{Q},&p=q\\
    \frac{1-t}{Q}, &p\neq q.
\end{cases}\]
In other words, $M_Q(t)=U_Q+t(I_Q-U_Q).$
Define $N(t):=M_Q(t)^{\otimes D}$, and identify rows and columns of $N(t)$ with $[Q]^D.$
For $a=(p_1,\ldots, p_D),b=(q_1,\ldots,q_D)\in [Q]^D$, let $d_H(a,b)=|\{1\le i \le D\mid p_i\neq q_i\}|$ be the Hamming distance between $a$ and $b$.
We have
\begin{equation}\label{eq: Hamming intersection number generating polynomial}
    N(t)_{ab}=\left(\frac{1+(Q-1)t}{Q}\right)^{D-d_H(a,b)}\left(\frac{1-t}{Q}\right)^{d_H(a,b)}=f_{D,Q,d_H(a,b)}(t) \quad (\forall t
    \in \R).
\end{equation}
Note that for every $s,t\in \R,$ $$M_Q(s)M_Q(t)=(U_Q+s(I_Q-U_Q))(U_Q+t(I_Q-U_Q))=U_Q+st(I_Q-U_Q)=M_Q(st).$$ 
Then $N(s)N(t)=N(st)$.
Now fix $a=(p_1,\ldots, p_D),b=(q_1,\ldots, q_D)\in[Q]^D$ at Hamming distance $k$, we have 
\[
\sum_{c\in [Q]^D}N(s)_{ac}N(t)_{cb}=N(st)_{ab} \qquad (\forall s,t\in \R).
\]
Plug in \eqref{eq: Hamming intersection number generating polynomial} and group the sum by $(i,j)=(d_H(a,c),d_H(c,b))$, we obtain
\[
\sum_{i=0}^D\sum_{j=0}^D p_{ij}^kf_{D,Q,i}(s)f_{D,Q,j}(t)=f_{D,Q,k}(st) \qquad (\forall s,t\in \R).
\]
\end{proof}

\begin{lemma}\label{lem: distance polynomial are linearly independent}
Fix integers $D\ge 1,Q\ge 2$.
The polynomials $\{f_{D,Q,i}(s)f_{D,Q,j}(t)\mid 0\le i,j\le D\}$ are linearly independent in $\R[s,t]$.
\end{lemma}

\begin{proof}
We first prove that the $D+1$ distance polynomials $\{f_{D,Q,k}(t) \mid0\leq k\leq D\}$ are linearly independent in $\R[t]$. 
%JYC, no need for contradiction... just show c_i =0.
%Suppose for a contradiction that, 
Suppose there are $c_0,\ldots,c_D\in \R$ such that
$\sum_{k=0}^{D}c_kf_{D,Q,k}(t)=0\,(t\in \R).$  We show that all $c_k=0$.
By \cref{def:distance-poly-for-GH}, we have 
%By \cref{def: distance polynomial}, we have 
%%% was a typo for this ref. it refers to a later def. 7.4
\[
\sum_{k=0}^D c_kQ^{-D}(1+(Q-1)t)^{D-k}(1-t)^k=0 \qquad (\forall t\in \R).
\]
Dividing by $Q^{-D}(1+(Q-1)t)^D$ when it is non-vanishing, we obtain the following:
\[\sum_{k=0}^D c_k \left(\frac{1-t}{1+(Q-1)t}\right)^k=0 \quad (\forall t\in \R,\,1+(Q-1)t\neq 0).\]
Let $z=\frac{1-t}{1+(Q-1)t}$, then $z$ takes infinitely many distinct values when $t\in \R$ and $1+(Q-1)t\neq 0.$
Thus, the polynomial $P(z):=\sum_{k=0}^D c_kz^k$ vanishes on infinitely many distinct points.
Hence it is the zero polynomial.
It follows that $c_0=\cdots=c_D=0.$

Now suppose for some $(c_{ij})_{0\le i,j\le D}$, $\sum_{i,j=0}^D c_{ij}f_{D,Q,i}(s)f_{D,Q,j}(t)=0 \,~(\forall s,t\in \R)$.
Grouping the sum, we have $\sum_{i=0}^D\left(\sum_{j=0}^D c_{ij}f_{D,Q,j}(t)\right)f_{D,Q,i}(s)=0.$
Linear independence of $\{f_{D,Q,i}(s)\mid 0\le i\le D\}$ gives $\sum_{j=0}^D c_{ij}f_{D,Q,j}(t)=0, \,\forall 0\le i\le D$, and for all $t\in \R$.
Then linear independence of $\{f_{D,Q,j}(t)\mid 0\le j\le D\}$ gives $c_{ij}=0,\,\forall 0\le i,j\le D.$ 
\end{proof}

\subsection{Distance Polynomials of Coordinate Graphs}\label{subsec: Distance Polynomials of Coordinate Graphs}

\begin{definition}\label{def:scalar-parameters}
For every $r\in\mathcal F_+$, define $Q_r=\frac{\ell_r(E)}{\ell_r(0)}$.
\end{definition}

\begin{lemma}\label{lem:integer-factors}
For every $r\in\mathcal F_+$, $Q_r\in \mathbb{Z}_{\ge 2}$. 
Further,
$n=\prod_{r\in\mathcal F_+}Q_r^{\ell_r(E)}.$
\end{lemma}
\begin{proof}
Recall the definition of $d$ in \cref{def: polynomial homomorphism}. It is
%Let $d$ be 
also the common degree of $F(X)_{xy}$ from  \cref{lem:common-degree}.
Put $p(t)=\operatorname{Tr}F(tE),\,t\in \R$.
By \cref{def: polynomial homomorphism}, write 
\[
F(tE)=U+tE+\sum_{k=2}^d t^kh_k(E).
\]
If $d =1$ then the sum over $k$ is empty.
Since $E$ and $\{h_k(E)\}_{k=2}^d$ are all orthogonal projections by \cref{lem: hkE mutually orth proj}, $\operatorname{Tr}(E)=\operatorname{rank}(E)$ and $\operatorname{Tr}(h_k(E))=\operatorname{rank}(h_k(E))$ for every $2\le k\le d.$
Taking trace,
\begin{equation}\label{eq: integer factors 1}
    p(t)=1+(\operatorname{rank}E)t+
\sum_{k=2}^{d}(\operatorname{rank}h_k(E))t^k \qquad (\forall t\in \R).
\end{equation}
So $\deg(p)=d.$
Every coefficient is a nonnegative integer, so $p(t)\geq1$ for $t\geq0$.

On the other hand, \cref{lem: trace FX} gives 
\[
p(t)=\prod_{r\in \mathcal{F}_+}\left(\frac{\ell_r(tE)}{\ell_r(0)}\right)^{\ell_r(E)}\qquad (\forall t\in \R).
\]
By being an affine function, $\ell_r(tE)=\ell_r(tE+(1-t)0)=t\ell_r(E)+(1-t)\ell_r(0).$
So
\begin{equation}\label{eq: integer factors 2}
    p(t)=\prod_{r\in \mathcal{F}_+}\left(\frac{t\ell_r(E)+(1-t)\ell_r(0)}{\ell_r(0)}\right)^{\ell_r(E)}=\prod_{r\in \mathcal{F}_+}\left(1+(Q_r-1)t\right)^{\ell_r(E)}\qquad (\forall t\in \R).
\end{equation}
By \cref{lem: complementary degree}, $\sum_{r\in\mathcal F_+}\ell_r(E)=d$.
If $Q_r=1$ for some $r\in \mathcal{F}_+$, then $\deg(p)<\sum_{r\in \mathcal{F}_+}\ell_r(E)=d$, a contradiction.
Hence $Q_r\neq 1$ for every $r\in \mathcal{F}_+.$
Consequently, $p(t)$ has zeros $-\frac{1}{Q_r-1},$ for all $r\in \mathcal{F}_+.$
We claim these zeros are all rationals.
To see this, $\ell_r(E)=h_r(T_{xx})=h_r(x,x)\in \Z$ and $\ell_r(0)\in \Q_{>0}$ by \cref{lem:factor-positivity}.
So $Q_r=\frac{\ell_r(E)}{\ell_r(0)}\in \Q.$
It follows that $-\frac{1}{Q_r-1}\in \Q$ for every $r\in \mathcal{F}$.
Moreover, $-\frac{1}{Q_r-1}<0$ for every $r\in \mathcal{F}_+$, since $p(t)\ge 1$ for $t\ge 0.$
Write $\frac{1}{Q_r-1}=\frac{a_r}{b_r}$, where $a_r, b_r \in \Zp$,
%%% I changes to name + \frac{1}{Q_r-1}= a_r/b_r
%and $b_r\in \Z_{<0}$, 
for every $r\in \mathcal{F}_+$.
Since $-a_r/b_r$ is a root of $p(t) \in \Z[t]$,
$a_r$ divides the constant term of $p(t)$, which is $1$ by \eqref{eq: integer factors 1}.
Thus, $a_r=1$ for every $r\in \mathcal{F}_+.$
This implies $Q_r=b_r+1\in \Z_{\ge 2}$ for every $r\in \mathcal{F}_+.$
Finally, $F(E)=I$ (\cref{lem:polynomial-identity}) gives $p(1)=n$.
By \eqref{eq: integer factors 2}, this is $n=\prod_{r\in \mathcal{F}_+}Q_r^{\ell_r(E)}.$
\end{proof}
    Note that $Q_r=Q_s$, if $r,s$ are in the same orbit $\omega.$
    We denote their common value by $Q_\omega.$

\begin{definition}\label{def: distance polynomial}
    For each $\omega\in\Omega$, let $Q_\omega$ be the common integer $Q_r$ for $r\in\omega$, and let $D_\omega=\sum_{r\in\omega}\ell_r(E)$. For $a,b\in X_\omega$, define the polynomial
\begin{equation}
F_\omega(t)_{ab}:=f_{D_\omega,Q_\omega,\Delta_\omega(a,b)}(t)=Q_\omega^{-D_\omega}
\bigl(1+(Q_\omega-1)t\bigr)^{D_\omega-\Delta_\omega(a,b)}
(1-t)^{\Delta_\omega(a,b)} \quad (\forall t\in \R).
\label{eq:scalar-kernel}
\end{equation}
\end{definition}
%Here $D_\omega$ is a positive integer.
\begin{remark}\label{rmk:distance polynomial}
    The exponents in \eqref{eq:scalar-kernel} are nonnegative integers because
    \[
    D_\omega-\Delta_\omega(x,y)=\sum_{r\in \omega}\ell_r(E)-\sum_{r\in \omega}(\ell_r(E)-h_r(x,y))=\sum_{r\in\omega}h_r(x,y)\geq0
    \quad(\forall \omega\in\Omega,\forall x,y\in[n]).
    \]
    Also, $F_\omega(t)_{ab}>0$ for $0<t<1$.
    Moreover, $F_\omega(1)=I$ because $\Delta_\omega(a,b)=0 \iff a=b$ by \cref{lem: Delta omega satisfies positivity and symmetry}.
\end{remark}

\begin{lemma}\label{lem:scalar-tensor}
Under the bijection $\phi$ in  \cref{lem: phi is bijection},
\(
F(tE)=\bigotimes_{\omega\in\Omega}F_\omega(t),(t\in\mathbb R).
\)
Entrywise,
\[
F(tE)_{xy}=\prod_{\omega\in \Omega}F_\omega(t)_{x_\omega y_\omega}\qquad (\forall x,y\in[n],\,\forall t\in \R).
\]
Here $x_w$ and $y_w$ are short for $[x]_{L_{\omega}}$ and $[y]_{L_{\omega}}$.
\end{lemma}
\begin{proof}
Plug $X=tE$ into \cref{lem:affine-factors}, 
\begin{equation}
F(tE)_{xy}=\frac1n\prod_{a\in\mathcal F}
\left(\frac{\ell_a(tE)}{\ell_a(0)}\right)^{h_a(x,y)}.
\label{eq:F(tE)xy}
\end{equation}
For every $a\in \mathcal{F}$, write $$\ell_a(tE)=\ell_a(tE+(1-t)0)=t\ell_a(E)+(1-t)\ell_a(0).$$
For $i\in\mathcal I$, $\ell_i(E)=0$ gives $\ell_i(tE)/\ell_i(0)=1-t$. 
For $r\in\mathcal F_+$,
\(
{\ell_r(tE)}/{\ell_r(0)}
=tQ_r+1-t=1+(Q_r-1)t.
\)
Substitute into \eqref{eq:F(tE)xy},
\begin{equation*}
    \begin{aligned}
        F(tE)_{xy}=&\frac{1}{n}\prod_{i\in \mathcal{I}}(1-t)^{h_i(x,y)}\prod_{r\in \mathcal{F}_+}(1+(Q_r-1)t)^{h_r(x,y)}\\
        =&\frac{1}{n}\prod_{\omega\in \Omega}\prod_{i\in \mathcal{I}_\omega}(1-t)^{h_i(x,y)}\prod_{\omega\in\Omega}\prod_{r\in \omega}(1+(Q_r-1)t)^{h_r(x,y)}\\
        =&\frac{1}{n}\prod_{\omega\in \Omega}(1-t)^{\Delta_\omega(x,y)}\prod_{\omega\in \Omega}(1+(Q_\omega-1)t)^{D_{\omega}-\Delta_{\omega}(x,y)}\\
        =&\prod_{\omega\in \Omega}Q_{\omega}^{-D_{\omega}}(1+(Q_\omega-1)t)^{D_{\omega}-\Delta_\omega(x,y)}(1-t)^{\Delta_{\omega}(x,y)}\\
        =&\prod_{\omega\in \Omega}F_{\omega}(t)_{x_\omega y_\omega}.
    \end{aligned}
\end{equation*}
The second and third equalities hold because $(\mathcal{I}_{\omega})_{\omega\in \Omega}$ forms a partition of $\mathcal{I}$, and $\Delta_{\omega}(x,y)=\sum_{i\in \mathcal{I}_\omega}h_i(x,y)$ by \cref{lem:owned-atoms}; 
The fourth equality holds because $n=\prod_{\omega\in \Omega}Q_{\omega}^{D_{\omega}}$ by \cref{lem:integer-factors};
The last equality holds because $\Delta_{\omega}(x,y)=\Delta_{\omega}(x_\omega,y_\omega)$ by \cref{def:quotient-distances}.
The lemma is proved.
\end{proof}

\begin{lemma}\label{lem:scalar-normalization}
For every $\omega\in\Omega$, $a\in X_\omega$, and $t\in\mathbb R$,
\(
\sum_{b\in X_\omega}F_\omega(t)_{ab}=1.
\)
In particular, $|X_\omega|=Q_\omega^{D_\omega}$.
\end{lemma}
\begin{proof}
For $\omega\in \Omega$, $a\in X_{\omega}$ and $t\in \R$, denote the row sum $\sum_{b\in X_\omega}F_{\omega}(t)_{ab}$ by $r_\omega(a,t)$. 
We have $r_{\omega}(a,t)>0$ for every $t\in (0,1).$
For $x\in [n]$, abbreviate $[x]_{L_{\omega}}$ as $x_\omega$.
Now fix an arbitrary $x\in [n]$.
The global row sum identity \cref{lem:row-normalization} and \cref{lem:scalar-tensor} give
\[\begin{aligned}
    1=&\sum_{y=1}^nF(tE)_{xy}
    =\sum_{y=1}^n\prod_{\omega\in \Omega}F_{\omega}(t)_{x_\omega y_\omega} \quad (\forall t\in \R).
\end{aligned}\]
Under the bijection $\phi:[n]\to \prod_{\omega\in \Omega}X_{\omega},\,[y]\mapsto (y_\omega)_{\omega\in \Omega}$, when $y$ ranges from $1$ to $n$, $(y_\omega)_{\omega\in \Omega}$ ranges over $\prod_{\omega\in \Omega}X_{\omega}.$
So
\begin{equation}
    1=\sum_{(y_\omega)_{\omega\in \Omega}\in \prod_{\omega\in \Omega}X_{\omega}}\prod_{\omega\in \Omega}F_{\omega}(t)_{x_\omega y_\omega}
=\prod_{\omega\in \Omega}\sum_{y_\omega\in X_\omega}F_\omega(t)_{x_\omega y_\omega}=\prod_{\omega\in \Omega}r_\omega(x_{\omega},t)\quad (\forall t\in \R).
\label{eq:product of rw}
\end{equation}

Since $x\in [n]$ is arbitrarily fixed, we may choose the coordinates $(x_{\omega})_{\omega\in \omega}\in \prod_{\omega\in \Omega}X_\omega$ freely under the bijection $\phi.$
Now fix an arbitrary $\omega\in \Omega$, let $a,a'\in X_{\omega}$.
We may choose $x,x'\in [n]$ such that $x_\omega=a$ and $x'_\omega=a'$, and $x_{\eta}=x'_\eta\in X_{\eta}$ are arbitrary for $\eta\in \Omega\setminus \{\omega\}$.
Compare \eqref{eq:product of rw} for $x$ and $x'$, $r_\eta(x_\eta,t)$  cancels for every $\eta\in \Omega\setminus\{\omega\}$ when $t\in (0,1)$.
This leaves us with $r_\omega(a,t)=r_\omega(a',t),\,\forall t\in (0,1)$.
Thus, for every $\omega\in \Omega$ and $t\in (0,1)$, $F_\omega(t)$ has constant row sum $r_\omega(t)$.
When $t\in (0,1)$, \eqref{eq:product of rw} gives the polynomial identity $\prod_{\omega\in \Omega}r_\omega(t)=1$.
Consequently, $r_\omega(t)$ is a positive constant polynomial for every $\omega\in \Omega.$
%%% I added the word "positive" so that the product =1 ==> each factor=1
Since $F_\omega(1)=I$ (see \cref{rmk:distance polynomial}), $r_\omega(t)=r_\omega(1)=1.$
That is, $F_\omega(t)$ has constant row sum 1.

At $t=0$, 
%every entry of 
%%% I removed the word "of".   $F_\omega(t)_{ab}$  is already just an entry.
every entry $F_\omega(t)_{ab}$ is $Q_\omega^{-D_\omega}$, so
$1=r_\omega(0)=|X_\omega|Q_\omega^{-D_\omega}$, as required.
\end{proof}

\begin{lemma}\label{lem: F omega is homomorphism}
For every $\omega\in\Omega$ and every $s,t\in\mathbb R$,
\(
F_\omega(s)F_\omega(t)=F_\omega(st).
\)
\end{lemma}
\begin{proof}
The global homomorphism property \cref{lem: homomorphism property} and \cref{lem:scalar-tensor} give
\[
\bigotimes_{\eta\in\Omega}\bigl(F_\eta(s)F_\eta(t)\bigr)
=F(sE)F(tE)=F(stE)=\bigotimes_{\eta\in\Omega}F_\eta(st).
\]
Entrywise,
\begin{equation}\label{eq:F omega hom}
    \prod_{\eta\in \Omega}(F_\eta(s)F_\eta(t))_{x_\eta y_\eta}=\prod_{\eta\in \Omega}F_\eta(st)_{x_\eta y_\eta}\qquad (\forall x,y\in [n]).
\end{equation}
Now fix an arbitrary $\omega\in \Omega$ and $a,b\in X_\omega$, we want to show $(F_\omega(s)F_\omega(t))_{ab}=F_{\omega}(st)_{ab}.$
Since $\phi:[n]\to \prod_{\eta\in \Omega}X_\eta$ is a bijection, we may fix some $x\in [n]$ such that $x_\omega=a.$
Summing \eqref{eq:F omega hom} over all $y\in [n]$ with $y_\omega=b$, we obtain 
\[
\begin{aligned}
    &\sum_{\substack{y\in [n] \\ y_\omega=b}} \prod_{\eta\in \Omega}(F_\eta(s)F_\eta(t))_{x_\eta y_\eta} = \sum_{\substack{y\in [n] \\ y_\omega=b}} \prod_{\eta\in \Omega}F_\eta(st)_{x_\eta y_\eta}\\
    \iff &
    \sum_{\substack{(y_\eta)_{\eta\in \Omega} \\ y_\omega=b}} \prod_{\eta\in \Omega}(F_\eta(s)F_\eta(t))_{x_\eta y_\eta} = \sum_{\substack{(y_\eta)_{\eta\in\Omega} \\ y_\omega=b}} \prod_{\eta\in \Omega}F_\eta(st)_{x_\eta y_\eta}\\
    \iff &
    (F_\omega(s)F_\omega(t))_{ab}\sum_{(y_\eta)_{\eta\in \Omega\setminus\{\omega\}}} \prod_{\eta\in \Omega\setminus\{\omega\}}(F_\eta(s)F_\eta(t))_{x_\eta y_\eta} 
    = F_\omega(st)_{ab}\sum_{(y_\eta)_{\eta\in\Omega\setminus\{\omega\}} } \prod_{\eta\in \Omega\setminus\{\omega\}}F_\eta(st)_{x_\eta y_\eta}\\
    \iff &
    (F_\omega(s)F_\omega(t))_{ab}\prod_{\eta\in \Omega\setminus\{\omega\}}\sum_{y_\eta\in X_\eta}(F_\eta(s)F_\eta(t))_{x_\eta y_\eta} 
    = F_\omega(st)_{ab}\prod_{\eta\in \Omega\setminus\{\omega\}} \sum_{y_\eta\in X_\eta}F_\eta(st)_{x_\eta y_\eta}.\\
\end{aligned}
\]
By \cref{lem:scalar-normalization}, for every $\eta\in \Omega\setminus\{\omega\}$ and $s,t\in \R$, $F_\eta(s),F_\eta(t)$ and $F_\eta(st)$ have row sum 1.
Thus $F_\eta(s)F_\eta(t)$ also has row sum $1$, as $F_\eta(s)F_\eta(t)\mathbf{1}=F_\eta(s)\mathbf{1}=\mathbf{1}$.
So the above equation simplifies to $(F_\omega(s)F_\omega(t))_{ab}=F_{\omega}(st)_{ab}.$
The lemma is then proved.
\end{proof}

\subsection{Hamming Intersection Numbers of Coordinate  Graphs}\label{subsec: Hamming Intersection Numbers of Coordinate  Graphs}

\begin{lemma}\label{lem:delta omega is a distance}
    For every $\omega\in \Omega$, $\Delta_\omega(\cdot,\cdot)$ is a distance on $X_\omega$.
\end{lemma}

\begin{proof}
Fix $\omega\in\Omega$, and in this proof write $D=D_\omega$, $Q=Q_\omega$, and $\Delta=\Delta_\omega$. 
By \cref{lem: Delta omega satisfies positivity and symmetry}, $\Delta(\cdot,\cdot)$ satisfies positivity (nonnegativity and 0 only at equal inputs) and symmetry.
We only need to prove the triangle inequality.
Recall \cref{def: distance polynomial},

\begin{equation*}
    F_\omega(t)_{ab}=Q^{-D}\bigl(1+(Q-1)t\bigr)^{D-\Delta(a,b)}(1-t)^{\Delta(a,b)} \quad (t\in \R).
\end{equation*}
Let $u=1-t$, we have
\begin{equation}\label{eq: asymptotic for F omega}
    \begin{aligned}
        F_\omega(t)_{ab}
        =&Q^{-D}\bigl(1+(Q-1)(1-u)\bigr)^{D-\Delta(a,b)}u^{\Delta(a,b)} \\
        =&Q^{-D}(Q-(Q-1)u)^{D-\Delta(a,b)}u^{\Delta(a,b)}\\
        =&Q^{-\Delta(a,b)}(1-\frac{Q-1}{Q}u)^{D-\Delta(a,b)}u^{\Delta(a,b)}\\
        =&Q^{-\Delta(a,b)}u^{\Delta(a,b)}+O(u^{\Delta(a,b)+1}),\quad \text{as } u\to 0.
    \end{aligned}
\end{equation}
Now fix arbitrary $a,b\in X_\omega.$
By \cref{lem: F omega is homomorphism}, $F_\omega(s)F_\omega(t)=F_\omega(st)$ for every $s,t\in \R$.
Set $s=t=1-u$ and let $u$ approaches to $0$,
we have
\begin{equation}\label{eq: F omega 1-u}
    \sum_{c\in X_\omega}F_\omega(1-u)_{ac}F_\omega(1-u)_{cb}=F_\omega(1-(2u-u^2)),\quad \text{as } u\to 0.
\end{equation}
Plug in \eqref{eq: asymptotic for F omega} and compare the lowest degree term on both sides of \eqref{eq: F omega 1-u}.
On the RHS, the lowest degree term is $Q^{-\Delta(a,b)}(2u)^{\Delta(a,b)},$
with a nonzero coefficient.
Let $m=\min_{c\in X_\omega}\{\Delta(a,c)+\Delta(c,b)\}$.
Then the lowest degree term on the LHS is 
\[
\sum_{\substack{c\in X_\omega :\Delta(a,c)+\Delta(c,b)=m}}Q^{-\Delta(a,c)}u^{\Delta(a,c)}Q^{-\Delta(c,b)}u^{\Delta(c,b)}=\sum_{\substack{c\in X_\omega :\Delta(a,c)+\Delta(c,b)=m}}Q^{-m}u^{m},
\]
where all coefficients are positive and therefore do not cancel.
The equation \eqref{eq: F omega 1-u} is a polynomial identity, thus the lowest degree on both sides must agree.
Therefore
\(
\Delta(a,b)=m\le \Delta(a,c)+\Delta(c,b),\,(\forall a,b,c\in X_\omega).
\)
\end{proof}

\begin{lemma}\label{lem:coordinate-graph-distance}
For every $\omega\in\Omega$, the graph $\Gamma_\omega$ is connected and its graph distance is $\Delta_\omega$.
\end{lemma}
\begin{proof}
Fix $\omega\in\Omega$, and in this proof write $D=D_\omega$, $Q=Q_\omega$, and $\Delta=\Delta_\omega$.
Let $d_\Gamma$ be the graph distance of $\Gamma_\omega$.
Fix $a,b\in X_\omega$,
we want to prove $d_\Gamma(a,b)=\Delta(a,b).$

Suppose $a$ and $b$ are connected in $\Gamma_\omega$ by a path,
%%% I added the word "by a path" as we also used the word to say there is an edge (a,b)
%and there is a 
and let $a_0a_1\ldots a_k$,  from   $a_0=a$ to $a_k=b$
be a path of the shortest length $k$.
\cref{lem:delta omega is a distance} gives the triangle inequality for $\Delta$.
So 
\[\Delta(a,b)=\Delta(a_0,a_k)\le\sum_{i=0}^{k-1}\Delta(a_i,a_{i+1})=k=d_\Gamma(a,b).\]

On the other hand, for $k=\Delta(a,b)$, we want to construct a path from $a$ to $b$ of length $k.$
By \cref{lem: F omega is homomorphism}, $F_\omega(1-s)F_\omega(1-t)=F_\omega((1-s)(1-t))$ for every $s,t\in \R$.
Entrywise,
\begin{equation}\label{eq: F omega 1-s}
%    \sum_{c\in X_\omega}F_\omega(1-s)_{ac}F_\omega(1-t)_{cb}=F_\omega(1-(s+t-st)).
%typo, corrected, add subscrot ab
    \sum_{c\in X_\omega}F_\omega(1-s)_{ac}F_\omega(1-t)_{cb}=F_\omega(1-(s+t-st))_{ab}.
\end{equation}
Let $s,t\to 0$, substitute the asymptotic expansion \eqref{eq: asymptotic for F omega} into \eqref{eq: F omega 1-s}, and compare the lowest total degree terms on both sides.
On the RHS, the lowest total degree term is $Q^{-\Delta(a,b)}(s+t)^k.$
On the LHS, every term indexed by $c\in X_\omega$ has total degree $\Delta(a,c)+\Delta(c,b)\ge \Delta(a,b)=k$ by the triangle inequality.
Since \eqref{eq: F omega 1-s} is a polynomial identity in $\R[s,t]$, the lowest total degree on the LHS must equal to $k$, and 
\begin{equation}\label{eq: s+t to k}
    \sum_{c\in X_\omega:\,\Delta(a,c)+\Delta(c,b)=k}s^{\Delta(a,c)}t^{\Delta(c,b)}=(s+t)^{k}.
\end{equation}
For $k\geq1$, the coefficient of $st^{k-1}$ on the RHS of \eqref{eq: s+t to k} is $k>0$, so on the LHS the term $st^{k-1}$ also appears.
Thus, there is $c\in X_\omega$ with $\Delta(a,c)=1$ and $\Delta(c,b)=k-1$.
Induction produces a path from $a$ to $b$ of length $k$. 
Hence $d_\Gamma(a,b)\le \Delta(a,b)$.
\end{proof}

\begin{lemma}\label{lem:hamming-numbers}
For every $\omega\in\Omega$, the coordinate graph $\Gamma_\omega$ is distance regular, and has the same intersection numbers as the Hamming graph $H(D_\omega,Q_\omega).$
\end{lemma}
\begin{proof}
Fix $\omega\in \Omega$ and abbreviate $D=D_\omega$, $Q=Q_\omega$. 
Fix arbitrary vertices $a,b\in X_\omega$ at distance $k$, let 
$$p_{ij}^k(a,b):=|\{c\in X_\omega\mid \Delta(a,c)=i\text{ and }\Delta(c,b)=j\}|.$$
By \cref{lem: F omega is homomorphism},
\begin{equation}\label{eq: hamming intersection numbers}
    \sum_{c\in X_\omega}F_\omega(s)_{ac}F_\omega(t)_{cb}=F_\omega(st)_{ab}.
\end{equation}
By \cref{def: distance polynomial}, $F_\omega(s)_{ac}=f_{D,Q,\Delta(a,c)}(s)$, $F_\omega(t)_{cb}=f_{D,Q,\Delta(c,b)}(t)$ and $F_\omega(st)_{ab}=f_{D,Q,k}(st)$.
Grouping the sum \eqref{eq: hamming intersection numbers} by $(i,j)=(\Delta(a,c),\Delta(c,b))$ gives
\[
\sum_{i=0}^{D}\sum_{j=0}^{D}p^k_{ij}(a,b)f_{D,Q,i}(s)f_{D,Q,j}(t)=f_{D,Q,k}(st)
\qquad(s,t\in\mathbb R).
\]
Combined with \cref{lem: Hamming intersection number generating polynomial},
\[
\sum_{i=0}^{D}\sum_{j=0}^{D}(p^k_{ij}(a,b)-p_{ij}^k)f_{D,Q,i}(s)f_{D,Q,j}(t)=0
\qquad(s,t\in\mathbb R)
\]
where $(p_{ij}^k)_{0\le i,j,k\le D}$ are intersection numbers of the Hamming graph $H(D,Q).$
\cref{lem: distance polynomial are linearly independent} gives $p_{ij}^k(a,b)=p_{ij}^k.$
Since $a,b\in X_\omega$ are arbitrary vertices at distance $k$, we conclude that $\Gamma_\omega$ is a distance regular graph and it has the same intersection numbers as $H(D_\omega,Q_\omega).$
\end{proof}

\section{Completing the Dichotomy}\label{sec:dichotomy}

In this section, we will complete the proof for the dichotomy of $\PlGH(M)$ for entrywise positive and full rank matrix $M$.
The integers $Q_\omega\in \Z_{\ge 0}$ obtained in the previous section capture the sizes of the tensor factors of $M$. 
%Technically, they serve as the alphabet sizes for the Hamming graphs $H(D_\omega,Q_\omega)$, each associated with a group of tensor factors of $M$ with the same size. 
We will prove that $\PlGH(M)$ is \#P-hard whenever at least one $Q_\omega \ge 3$ (\cref{lem:potts-hardness}), whereas $M$ is an Ising tensor under a suitable simultaneous permutation of
the rows and columns when $Q_\omega = 2$ for all $\omega\in \Omega$ (\cref{lem:binary}).

\subsection{Hardness Proof}
In this subsection, we prove that $\PlGH(M)$ is \#P-hard if some $Q_\omega\ge 3$ (\cref{lem:potts-hardness}).
\begin{lemma}\label{lem:coordinate-extraction}
For every $\omega\in\Omega$ and real $t\in(0,1)$,
\(
\PlGH(F_\omega(t))\leq_T\PlGH(M)
\).
\end{lemma}

\begin{proof}
Fix $\omega\in \Omega$ and $t\in (0,1).$
For $x,y\in[n]$, put
\[
(\Theta_\omega)_{xy}=\sum_{\eta\in\Omega\setminus\{\omega\}}\Delta_\eta(x,y).
\]
\cref{lem: expression for Delta omega} gives the matrix identity
\[
\Theta_\omega=
\left(\sum_{\eta\in\Omega\setminus\{\omega\}}\operatorname{Tr}W_\eta\right)J
-n\sum_{\eta\in\Omega\setminus\{\omega\}}W_\eta.
\]
From \cref{rmk: Ww}, for every $\eta\in \Omega$,
either  $W_\eta = S_r$ or $2 S_r$ for $r \in \eta$,
and so from \cref{def: gradients of affine form}
%{lem:positive-gradients} 
we have that $W_\eta\in\mathcal{A}$.
Thus $\exp_\circ{(s\cdot\Theta_\omega)}$ is a positive scalar multiple of the entrywise exponential of an element of $\mathcal A$, for every $s\in \R$. 
It is then realizable by \cref{lem:algebra-exponentials}. 
\cref{lem:polytope-positivity} also makes $F(tE)$ realizable.

Let $G$ be a planar input graph with $m$ edges. 
We next show that $Z_{F_\omega(t)}(G)$ is computable in polynomial time 
given oracle access to $\PlGH(M).$
%%% ZZT 10-19 resolved.
%if $\PlGH(M)$ is tractable.
%%% JYC 10-9: we have to phrase it as a reducibility. logically, without
% condition "P \ne #P", we can't say the opposite of "not #P-hard" is tractable.
%%% so, we should say $Z_{F_\omega(t)}(G)$ is reducible to   $\PlGH(M)$ 
%%%
Without loss of generality, we may assume that $G$ is connected.
Define the polynomial
\[
R_G(z)=\sum_{\sigma:V(G)\to[n]}
\left(\prod_{uv\in E(G)}F(tE)_{\sigma(u)\sigma(v)}\right)
 z^{\sum_{uv\in E(G)}(\Theta_\omega)_{\sigma(u)\sigma(v)}}.
\]
Its degree in $z$ is at most $B=m\sum_{\eta\in\Omega\setminus\{\omega\}}D_\eta$
(by \cref{rmk:distance polynomial}), which is polynomial size in terms of $G$.
Note that for $z>0$, $R_G(z)=Z_{M(z)}(G)$, where $M(z)=F(tE)\circ \exp_\circ(\log z \cdot\Theta_\omega)$.
By \cref{cor:thickening}, $M(z)$ is realizable for every $z>0$.
So $R_G(z)$ is computable in polynomial time for every $z>0$ with oracle access to $\PlGH(M).$
%%% ZZT 10-19: resolved
%%% JYC 10-9: 
%%% have to say, computable in polynomial time with oracle access to $\PlGH(M)$ 
Polynomial interpolation then recovers $R_G(0)$ in polynomial time.
For $\sigma:V(G)\to [n]$,
let $e(\sigma)=\sum_{uv\in E(G)}(\Theta_\omega)_{\sigma(u)\sigma(v)}.$
We have 
\[
R_G(0)=\sum_{\substack{\sigma:V(G)\to[n] \\ e(\sigma)=0}}
\prod_{uv\in E(G)}F(tE)_{\sigma(u)\sigma(v)}.
\]
Note that $e(\sigma)=0\iff \Delta_{\eta}(\sigma(u),\sigma(v))=0$ for every $uv\in E(G)$ and $\eta\in \Omega\setminus\{\omega\}$ $\iff [\sigma(u)]_{L_\eta}=[\sigma(v)]_{L_\eta}$ for every $uv\in E(G)$ and $\eta\in \Omega\setminus\{\omega\}$.
Suppose $\Omega=\{\omega_1,\ldots,\omega_k\}.$
Without loss of generality, we may assume $\omega=\omega_1$.
Under the bijection $\phi:[n]\to \prod_{\eta\in \Omega} X_\eta,\, x\mapsto (x_\eta)_{\eta\in \Omega}$, write $\sigma=(\sigma_1,\ldots, \sigma_k)$, where $\sigma_i:V(G)\to X_{\omega_i}, \, x\mapsto x_{\omega_i}$.
So $e(\sigma)=0$ if and only if $\sigma_i(u)=\sigma_i(v)$ for every $uv\in E(G)$ and $2\le i\le k.$

By \cref{lem:scalar-tensor}, $F(tE)_{\sigma(u)\sigma(v)}=\prod_{i=1}^kF_{\omega_i}(t)_{\sigma_i(u)\sigma_i(v)}$.
Plug into the expression of $R_G(0)$,
we have
\begin{equation}\label{eq: RG0}
    \begin{aligned}
    R_G(0)=&\sum_{\substack{\sigma_i:V(G)\to X_{\omega_i} (1\le i\le k)\\ e(\sigma_1,\ldots,\sigma_k)=0}}\prod_{uv\in E(G)}\prod_{i=1}^k F_{\omega_i}(t)_{\sigma_i(u)\sigma_i(v)}\\
\end{aligned}
\end{equation}
For $2\le i\le k$, since $\sigma_i(u)=\sigma_i(v)$ in every term of the sum \eqref{eq: RG0}, by \cref{def: distance polynomial} we have $F_{\omega_i}(t)_{\sigma_i(u)\sigma_i(v)}=Q_{\omega_i}^{-D_{\omega_i}}(1+(Q_{\omega_i}-1)t)^{D_{\omega_i}}.$
Let $c_\omega(t)=\prod_{i=2}^k Q_{\omega_i}^{-D_{\omega_i}}(1+(Q_{\omega_i}-1)t)^{D_{\omega_i}}.$
We obtain
\[
R_G(0)=\sum_{\sigma_1:V(G)\to X_{\omega} }\prod_{uv\in E(G)}F_\omega(t)_{\sigma_1(u)\sigma_1(v)}\cdot c_\omega(t)^{m}\cdot \frac{n}{|X_\omega|}.\\
\]
The factor $\frac{n}{\vert{}X_\omega\vert{}}$ arises from the fact that $G$ is connected. 
Note that for $2\le i\le k$ the choices for $\sigma_i$ must satisfy
$\sigma_i(u)=\sigma_i(v)$ for every edge $uv$, which is equivalent to
satisfying $\Delta_{\eta}(\sigma(u),\sigma(v))=0$ for every $uv\in E(G)$ and $\eta\in \Omega\setminus\{\omega\}$.
Thus this amounts to picking one vertex
 in each $\Gamma_{\omega_i}$ for connected $G$.
% Specifically, 
So, there are $\vert{}X_{\omega_i}\vert{}$ choices for each $\sigma_i$ (for $2 \le i \le k$), yielding a total of $\prod_{i=2}^k \vert{}X_{\omega_i}\vert{} = \frac{n}{\vert{}X_\omega\vert{}}$ choices for the tuple $(\sigma_2, \ldots, \sigma_k)$.
Therefore, $R_G(0)=c_\omega(t)^{m}\cdot \frac{n}{|X_\omega|}\cdot Z_{F_\omega(t)}(G).$
Oracle access to $R_G(0)$ gives $Z_{F_\omega(t)}(G)$.
The lemma is proved.
\end{proof}

\begin{lemma}\label{lem:potts-hardness}
If $Q_\omega\geq3$ for some $\omega\in\Omega$, then $\PlGH(M)$ is hard.
\end{lemma}
\begin{proof}
Fix $\omega\in \Omega$ such that $Q_\omega\ge 3$, and write $Q=Q_\omega$, $D=D_\omega$.
By \cref{lem:hamming-numbers}, the coordinate graph $\Gamma_\omega$ is distance-regular and has the same intersection numbers as the Hamming graph $H(D_\omega,Q_\omega)$. 
By \cref{thm:hamming-classification}, either $\Gamma_\omega\cong H(D_\omega,Q_\omega)$, or $\Gamma_\omega$ is isomorphic to a Doob graph.
In the latter case, we can prove $\PlGH(M)$ is \#P-hard.
See \cref{lem:doob_is_hard} in \autoref{app:doob}.
In the following we assume $\Gamma_\omega\cong H(D_\omega,Q_\omega)$.

We identify the vertex set $X_\omega$ with $[Q_\omega]^{D_\omega} = [Q]^D$.
For $a,b\in X_\omega$ at Hamming distance $k$, we have
\begin{equation*}
F_\omega(t)_{ab}=Q^{-D}
\bigl(1+(Q-1)t\bigr)^{D-k}
(1-t)^{k} \quad (t\in \R).
\end{equation*}
For $t\in \R$, define the $Q\times Q$ matrix $M_Q(t)$ by
\[M_Q(t)_{pq}=\begin{cases}
    \frac{1+(Q-1)t}{Q},&p=q\\
    \frac{1-t}{Q}, &p\neq q.
\end{cases}\]
Then $F_\omega(t)=M_Q(t)^{\otimes D}.$
\cref{lem:coordinate-extraction} proves that $F_\omega(t)$ is realizable for every $t\in (0,1).$
Note that for any input planar graph $G$, $Z_{F_{\omega}(t)}(G)=(Z_{M_Q(t)}(G))^D.$
Since $F_\omega(t)$ and $M_Q(t)$ are entrywise positive for $t\in (0,1)$, $Z_{F_{\omega}(t)}(G)$ and $Z_{M_Q(t)}$ are both positive.
Then oracle access to $Z_{F_{\omega}(t)}(G)$ recovers $Z_{M_Q(t)}$ by taking the $D$-th root.
Thus, $\PlGH(M_Q(t))\le_T \PlGH(F_\omega(t))\le_T \PlGH(M)$, for every $t\in (0,1).$

At $t=1/(Q+1)$, $M_Q(t)$ has diagonal value $\frac2{Q+1}$
and off-diagonal value $\frac1{Q+1}$. By \cref{thm:potts-source}, $\PlGH(M_Q(t))$ is \#P-hard if $Q\ge 3.$
Therefore, $\PlGH(M)$ is \#P-hard.
\end{proof}

\subsection{Recover Ising Tensor}

In this subsection, we will prove that if $Q_\omega=2$ for every $\omega\in {\Omega}$ (this is equivalent to $Q_r=2$ for every $r\in\mathcal{F}_+$), then $M$ is a tensor product of Ising factors $\left(\begin{smallmatrix}
    a_i & b_i\\
    b_i & a_i
\end{smallmatrix}\right)$ up to a simultaneous permutation of rows and columns
(\cref{lem:binary}), in which case $\PlGH(M)$ is tractable.

\begin{lemma}\label{lem:binary-scalar}
If $Q_\omega=2$ for every $\omega\in\Omega$, 
%then $n=2^d$, where $d$ is the common degree in \cref{def:probability-degree}, and
then
\begin{equation*}
F(tE)_{xy}=\frac1n(1+t)^{d-\delta(x,y)}(1-t)^{\delta(x,y)}
\qquad(x,y\in[n],\ t\in\mathbb R).
\end{equation*}
\end{lemma}
\begin{proof}
In the proof of \cref{lem:scalar-tensor} we have shown
\[F(tE)_{xy}=\frac{1}{n}\prod_{i\in \mathcal{I}}(1-t)^{h_i(x,y)}\prod_{r\in \mathcal{F}_+}(1+(Q_r-1)t)^{h_r(x,y)}\qquad(x,y\in[n],\ t\in\mathbb R).\]
By $Q_r=2 \,(\forall r\in \mathcal{F}_+)$,
\[F(tE)_{xy}=\frac{1}{n}(1-t)^{\sum_{i\in \mathcal{I}}h_i(x,y)}(1+t)^{\sum_{r\in \mathcal{F}_+}h_r(x,y)}\qquad(x,y\in[n],\ t\in\mathbb R).\]

We have $\sum_{i\in\mathcal{I}}h_i(x,y)=\delta(x,y)$ by \cref{def: delta and alpha}, and $\sum_{r\in \mathcal{F}_+}h_r(x,y)=d-\delta(x,y)$ by \cref{lem:common-degree}. 
Therefore, the lemma is proved.
\end{proof}

\begin{lemma}\label{lem:binary-cube}
If $Q_\omega=2$ for every $\omega\in\Omega$, there is a bijection $x\mapsto s_x$ from $[n]$ onto $\{-1,1\}^d$ such that
\[
E_{xy}=\frac1n\sum_{i=1}^{d}s_{xi}s_{yi}
\qquad(x,y\in[n]).
\]
\end{lemma}
\begin{proof}
    For every $\omega\in \Omega$, we have $|X_\omega|=Q_\omega^{D_\omega}=2^{D_\omega}$ by \cref{lem:scalar-normalization} and $Q_\omega=2.$
    So we may identify $X_\omega$ with $\{-1,1\}^{D_\omega}$.
    Since $\sum_{\omega\in \Omega}D_\omega=\sum_{r\in \mathcal{F}_+}\ell_r(E)=d$ by \cref{lem:common-degree}, we may identify $[n]=\prod_{\omega\in \Omega}X_\omega$ by $\{-1,1\}^d.$
    The bijection $\phi:[n]\to \prod_{\omega}X_\omega,\,x\mapsto(x_\omega)_{\omega\in \Omega}$
    may be written as $\phi=(\phi_\omega)_{\omega\in \Omega}$, where $\phi_\omega:[n]\mapsto \{-1,1\}^{D_\omega},\,x\mapsto x_\omega$.

    In \cref{lem:coordinate-graph-distance,lem:hamming-numbers}, we proved that $\Gamma_\omega$ has the same intersection numbers as $H(D_\omega,Q_\omega)$, and  $\Delta_\omega(\cdot,\cdot)$ is its graph distance.
    By \cref{thm:hamming-classification}, and since a Doob exception would require $Q_\omega=4$,
    it follows that $\Gamma_{\omega}\cong H(D_\omega,Q_\omega).$
    Thus, $\Delta_\omega(\cdot,\cdot)$ is the Hamming distance on $X_\omega=\{-1,1\}^{D_\omega}$.
    In particular, for every $x,y\in [n]$, we have
    \(\Delta_{
    \omega}(\phi_\omega(x), \phi_\omega(y))=|\{i\in [D_\omega]\mid \phi_\omega(x)_i\neq \phi_{\omega}(y)_i\}|.\)
 %%% I added a few lines
 By \cref{lem:coordinate-locality}
and \cref{def:quotient-distances}, 
$\Delta_\omega(x,y)=\Delta_\omega(x_\omega,y_\omega) =\Delta_\omega(\phi_\omega(x),\phi_\omega(y))$.
    Then by \cref{lem:owned-atoms}, 
    \begin{equation}
        \label{eq: delta x y in binary cube}\delta(x,y)=\sum_{\omega\in \Omega}\Delta_\omega(x,y)=\sum_{\omega\in \Omega}\Delta_\omega(\phi_\omega(x),\phi_\omega(y))=|\{i\in[d]\mid \phi(x)_i\neq \phi(y)_i\}|
    \end{equation}
    For $x\in [n]$, let $s_x=\phi(x)\in\{-1,1\}^d.$
    Then $\sum_{i=1}^d s_{xi}s_{y_i}=d-2\delta(x,y)$ by \eqref{eq: delta x y in binary cube}, 
    as the sum $\sum_{i=1}^d s_{xi}s_{y_i}$ 
    is the number $|\{i\in[d]\mid \phi(x)_i = \phi(y)_i\}| -
    |\{i\in[d]\mid \phi(x)_i\neq \phi(y)_i\}|$, 
    and the sum of the two counts is $d$, for every $x,y\in [n].$
    On the other hand, $F(tE)=U+tE+\sum_{k=2}^dt^ke_k \,(t\in \R)$.
    So $E_{xy}=\left.\frac{d}{dt}F(tE)_{xy}\right|_{t=0}.$
    By \cref{lem:binary-scalar}, $F(tE)_{xy}=\frac1n(1+t)^{d-\delta(x,y)}(1-t)^{\delta(x,y)}$.
    So $\left.\frac{d}{dt}F(tE)_{xy}\right|_{t=0}=\frac{1}{n}(d-2\delta(x,y))$.
    Thus, $E_{xy}=\frac{1}{n}(d-2\delta(x,y))=\frac{1}{n}\sum_{i=1}^d s_{xi}s_{yi}$ as required.
\end{proof}

\begin{lemma}\label{lem:binary}
If $Q_\omega=2$ for every $\omega\in\Omega$, then $M$ has the form~\eqref{eq:ising-tensor} up to a simultaneous permutation of rows and columns.
\end{lemma}
\begin{proof}
Let $\phi:[n]\to \{-1,1\}^d,\,x\mapsto s_x$ be the bijection in \cref{lem:binary-cube}, and let $S$ be the $n$ by $d$ matrix with entries $S_{xi}=s_{xi}$.
\cref{lem:binary-cube} gives $E=\frac{1}{n} SS^{\tt T}$. 
Let $C_0=P\log_\circ(M)P$.
It is clear that $C_0\in \mathcal{A}$.
So $$C_0=EC_0E=\frac{1}{n^2}(SS^{\tt T})C_0(SS^{\tt T})=STS^{\tt T},\quad \text{ where } T=\frac{1}{n^2}S^{\tt T}C_0S\in \Sym{d}{}{\R}.$$
Entrywise, 
\begin{equation}\label{eq: C0 xy}
     (C_0)_{xy}=\sum_{i=1}^d\sum_{j=1}^dT_{ij}s_{xi}s_{yj}\qquad (x,y\in[n]).
\end{equation}
%%% i added a line to justify 
% why s^{\tt T}Ts is constant for all s \in {-1, 1}^d
We note that $\phi$ is a bijection, so that every vector $s \in \{-1,1\}^d$
appears as some $\phi(x) = (s_{x1}, \ldots, s_{xd})$.
$C_0$ has constant diagonal by \cref{lem:normalize_rumsum_colsum_to_0}.
This means that $(C_0)_{xx}=\sum_{i=1}^d\sum_{j=1}^dT_{ij}s_{xi}s_{xj}$
is the same  for all $x \in [n]$. Thus,
$s^{\tt T}Ts$ is a constant $c$ for all $s\in\{-1,1\}^d$.
We claim that $T$ is a diagonal matrix.
Fix $1\le i<j\le d$ and group the quadratic form $s^{\tt T}Ts$ by $s_i$ and $s_j$, we have
\(
s^{\tt T}Ts=A+B s_i+Cs_j+2T_{ij}s_is_j.
\)
Let $s^{++}, s^{+-}, s^{-+}, s^{--} \in \{-1,1\}^d$ be the four vectors whose $i$-th and $j$-th entries are $(1,1), (1,-1), (-1,1),$ and $(-1,-1)$, respectively, with all other entries equal to $1$.
Then $$8T_{ij}=(s^{++})^{\tt T} Ts^{++}
 -(s^{+-})^{\tt T} Ts^{+-}
 -(s^{-+})^{\tt T} Ts^{-+}
 +(s^{--})^{\tt T} Ts^{--}=c-c-c+c=0.$$
Thus, $T_{ij}=0$ whenever $1\le i<j\le d$, i.e., $T$ is diagonal. 

Then \eqref{eq: C0 xy} becomes $(C_0)_{xy}=\sum_{i=1}^d T_{ii}s_{xi}s_{yi}$, for every $x,y\in[n].$
Let $\gamma_i=T_{ii}$ for every $i\in[d]$.
%Then $(C_0)_{xy}=\sum_{i=1}^d \gamma_i s_{xi}s_{yj}$, for every $x,y\in[n].$
Since $C_0\in \mathcal{A}$, \cref{lem:normalize_rumsum_colsum_to_0} gives $\log_\circ M=C_0+\mu J$ for some $\mu\in\mathbb R$. 
We have
\[
M_{xy}=e^\mu\exp\!\left(\sum_{i=1}^{d}\gamma_i s_{xi}s_{yi}\right)
=e^\mu\prod_{i=1}^{d}e^{\gamma_i s_{xi}s_{yi}} \quad (x,y\in[n]).
\]
Under the bijection $\phi:[n]\to \{-1,1\}^d,\,x\mapsto s_x$, we may rearrange the rows and columns of $M$ by a simultaneous permutation and write
\[
M=e^\mu\bigotimes_{i=1}^{d}
\begin{pmatrix}e^{\gamma_i}&e^{-\gamma_i}\\
e^{-\gamma_i}&e^{\gamma_i}\end{pmatrix}.
\]
Since $M$ is full rank, $\gamma_i\neq 0$ for every $i\in[d]$. 
Taking $\lambda=e^\mu$, $a_i=e^{\gamma_i}$, and $b_i=e^{-\gamma_i}$ gives \eqref{eq:ising-tensor} with $a_i,b_i>0,\,a_i\neq b_i$.
\end{proof}

\subsection{Completing the Dichotomy}
We are ready to prove the main theorem in the body text:
\begin{proof}[Proof of \cref{thm:main}]
    Fix the matrix $M\in \Sym{n}{F}{\Rp}.$
    Suppose we are in the homogeneous setting, otherwise $\PlGH(M)$ is already \#P-hard by \cref{lem:homogeneous}.
    Let $Q_\omega$ be the numbers in \cref{def: distance polynomial}.
    \cref{lem:integer-factors} proves $Q_{\omega}\in \mathbb{Z}_{\ge 2}$ for every $\omega\in \Omega.$
    If $Q_\omega\ge 3$ for some $\omega\in \Omega$, then $\PlGH(M)$ is \#P-hard by \cref{lem:potts-hardness}.
    Otherwise, $Q_\omega=2$ for every $\omega\in \Omega$. then by \cref{lem:binary}, $M$ has the form \eqref{eq:ising-tensor}, in which case $\PlGH(M)$ is tractable.
\end{proof}

\newpage
\appendix
\section{Exclude the Doob Case}\label{app:doob}
In \cref{lem:hamming-numbers}, we proved that every coordinate graph $\Gamma_\omega$ $(\omega\in \Omega)$ has the same intersection numbers as the Hamming graph $H(D_\omega,Q_\omega)$.
By \cref{thm:hamming-classification}, either $\Gamma_\omega\cong H(D_\omega, Q_\omega)$, or $\Gamma_\omega$ is isomorphic to the Doob graph $I(p,q)$ for $p\ge1, q\ge 0$ and $2p+q=D_\omega$.
In this section, we assume that for some $\omega\in \Omega$, it is the latter case.
We will prove that $\PlGH(M)$ is \#P-hard (\cref{lem:doob_is_hard}).

Recall \cref{def:hamming-doob}, the Doob graph $I(p,q)$ is defined as the Cartisian product of $p$ Shrikhande graphs and $q$ complete graphs $K_4$.
We use $SH$ to denote the Shrikhande graph.
It is a distance-regular graph of diameter $2$ (see \cref{fig:shrikhande}).
\begin{figure}[htbp]
    \centering
    \includegraphics[width=0.5\linewidth]{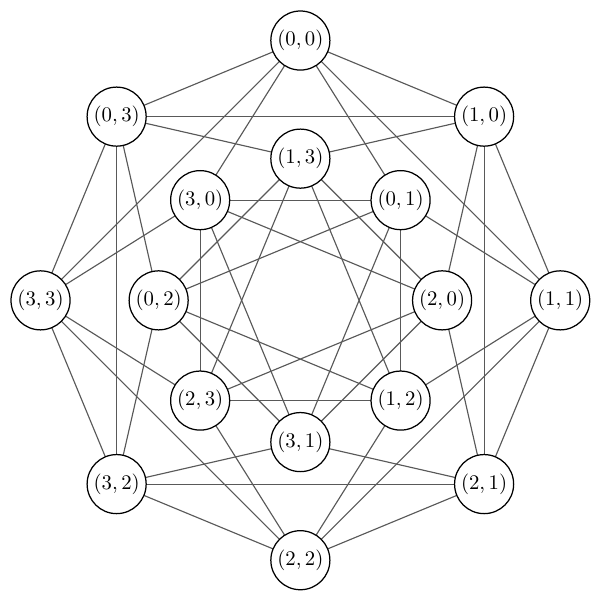}
    \caption{The Shrikhande graph $SH$. Vertices are elements in $\mathbb Z_4\times \mathbb Z_4$ and edges joining $x,y$ exactly when $y-x\in\mathcal D$, where
\(
\mathcal D=\{(1,0),(3,0),(0,1),(0,3),(1,1),(3,3)\}.
\)}
    \label{fig:shrikhande}
\end{figure}

\providecommand{\DoobMatrixUseColor}{1}
\definecolor{DoobMatrixBlue}{HTML}{DCEAF7}
\definecolor{DoobMatrixOrange}{HTML}{FCE5C5}
\definecolor{DoobMatrixGreen}{HTML}{D5EBDD}
\definecolor{DoobMatrixPurple}{HTML}{E7DDF1}
\definecolor{DoobMatrixNeutral}{HTML}{F2F3F5}
\newcommand{\DoobMatrixEntry}[1]{%
  \ifnum\DoobMatrixUseColor=1\relax
    \cellcolor{DoobMatrixNeutral}%
    \ifnum#1=4\relax\cellcolor{DoobMatrixBlue}\fi
    \ifnum#1=4624\relax\cellcolor{DoobMatrixBlue}\fi
    \ifnum#1=68\relax\cellcolor{DoobMatrixBlue}\fi
    \ifnum#1=2\relax\cellcolor{DoobMatrixOrange}\fi
    \ifnum#1=3126\relax\cellcolor{DoobMatrixOrange}\fi
    \ifnum#1=2116\relax\cellcolor{DoobMatrixGreen}\fi
    \ifnum#1=2118\relax\cellcolor{DoobMatrixPurple}\fi
  \fi
  \makebox[2.25em][c]{#1}%
}

Fix $\omega\in \Omega$ such that $\Gamma_\omega\cong I(p,q).$
Abbreviate $X=X_{\omega},D=D_\omega,Q=Q_\omega$ and $\Delta=\Delta_\omega.$
Then $Q=4$ and $2p+q=D$ by \cref{thm:hamming-classification}.
In \cref{lem:coordinate-extraction}, we proved that $F_\omega(t)$ is realizable for every $t\in(0,1)$.
Recall in \cref{def: distance polynomial}, for every $ a,b\in X_\omega$, we have

\begin{equation*}
\begin{aligned}
    F_\omega(t)_{ab}
    =&Q_\omega^{-D_\omega}\bigl(1+(Q_\omega-1)t\bigr)^{D_\omega-\Delta_\omega(a,b)}(1-t)^{\Delta_\omega(a,b)}
    =4^{-D}\bigl(1+3t\bigr)^{D-\Delta(a,b)}(1-t)^{\Delta(a,b)}.
\end{aligned}
\end{equation*}
Let $t=\frac{1}{5}$. 
Then
\begin{equation*}
\begin{aligned}
    F_\omega\left(\frac{1}{5}\right)_{ab}
    =&4^{-D}\left(\frac{8}{5}\right)^{D-\Delta(a,b)}\left(\frac{4}{5}\right)^{\Delta(a,b)}=5^{-D}2^{D-\Delta(a,b)},
\end{aligned}
\end{equation*}
and $\PlGH(F_\omega(1/5))\leq_T\PlGH(M)$ by \cref{lem:coordinate-extraction}.

Define the matrix $N\in \Sym{|X|}{}{\Zp}$ by $N_{ab}:=2^{D-\Delta(a,b)}, \,\forall a,b\in X.$
Then $N$ is realizable, since scaling by the constant $5^{-D}$ doesn't change the complexity of the $\PlGH$ problem.
Because $\Gamma_\omega\cong I(p,q)$, we may identify $X=V(\Gamma_\omega)$ by $$V(I(p,q))=\{(u_1,\ldots,u_p,x_1,\ldots, x_q)\mid u_i\in V(SH),x_j\in V(K_4),i\in[p],j\in [q]\}.$$
Let $a=(u_1,\ldots,u_p,x_1,\ldots, x_q)\in V(\Gamma_\omega)$ and $b=(v_1,\ldots,v_p,y_1,\ldots, y_q)\in V(\Gamma_\omega)$.
By the property of Cartesian product,
$$
\Delta(a,b)=d_{\Gamma_{\omega}}(a,b)=\sum_{i=1}^p d_{SH}(u_i,v_i)+\sum_{j=1}^q d_{K_4}(x_j,y_j).
$$
Thus, 
\begin{equation}\label{eq:Nab}
    N_{ab}=2^{D-\Delta(a,b)}=\prod_{i=1}^p 2^{2-d_{SH}(u_i,v_i)}\prod_{j=1}^q 2^{1-\mathbb{I}[x_j\neq y_j]}.
\end{equation}
%where $\mathbb{I}[x_j\neq y_j]$ is the indicator of $x_j\neq y_j.$

Define the matrix $S\in \Sym{16}{}{\Zp}$ by $S_{uv}:=2^{2-d_{SH}(u,v)},\,\forall u,v\in V(SH).$
Also define the matrix $K\in \Sym{4}{}{\R_{\ge 0}}$ by $K_{xy}=2^{1-\mathbb{I}[x_j\neq y_j]},\,\forall x,y\in V(K_4).$
By \eqref{eq:Nab}, we have $N=S^{\otimes p}\otimes K^{\otimes q}.$
Therefore, we proved:
\begin{lemma}
    $\PlGH(N)\equiv_T\PlGH(S^{\otimes p}\otimes K^{\otimes q})\le _T\PlGH(M).$
\end{lemma}

The matrices $S$ and $K$ are presented in \cref{fig: S and K}. 
In the following, we show that $\PlGH(S^{\otimes p}\otimes K^{\otimes q})$ is \#P-hard.
We want to make use of the following lemma:

\begin{lemma}[Lemma 13, \cite{CaiM23}]\label{lemma: full_thickeningInterpolation}
	Let $A \in \Sym{n}{}{\Zp}$ 
	with entries $(a_{ij})_{i,j\in[n]}.$
    Suppose that there exist distinct primes $p_1,\ldots, p_d$ and nonnegative integers $e_{ijk} \,(i,j\in[n],\,k\in[d])$ such that $a_{ij}=p_1^{e_{ij1}}p_2^{e_{ij2}}\cdots p_d^{e_{ijd}}$, for every $i,j\in[n].$
	Define
	$\mathcal{M}: \mathbb{R}^{d} \rightarrow \mathbb{R}^{n \times n}$ where
	$\mathcal{M}(\mathbf{g})_{ij} = \mathcal{M}(g_{1}, \dots, g_{d})_{ij} = g_{1}^{e_{ij1}}g_2^{e_{ij2}} \cdots g_{d}^{e_{ijd}}$ 
	for every $i, j \in [n]$. 
	Then, $\PlGH(\mathcal{M}(\mathbf{g})) \leq_T \PlGH(A)$ 
	for every $\mathbf{g} \in \mathbb{R}^{d}$.
\end{lemma}

The prime $2$ is the only generator for entries in $S$ and $K.$
Setting the generator to a real number $g$ 
%in $N=S^{\otimes p}\otimes K^{\otimes q}$ 
doesn't give us \#P-hardness directly.
Instead, we will first apply a ``Wheatstone gadget" (see \cref{fig:wheatstone}) to disrupt the values in $S$ and $K$, and then appeal to \cref{lemma: full_thickeningInterpolation}.

\begin{figure}[htbp]
\centering
\begin{tikzpicture}[scale=.85,every node/.style={font=\small},
  vertex/.style={circle,draw,fill=white,inner sep=2.1pt}]
  \node[vertex,label=left:$a$] (x) at (0,0) {};
  \node[vertex,label=right:$b$] (y) at (4,0) {};
  \node[vertex,label=above:$u$] (u) at (2,1.25) {};
  \node[vertex,label=below:$v$] (v) at (2,-1.25) {};
  \draw (x)--(u)--(y)--(v)--(x);
  \draw (u)--(v);
\end{tikzpicture}
\caption{The Wheatstone gadget $\mathbf{W}$.
$a$ and $b$ are distinguished vertices.
The name comes from the Wheatstone bridge in electrical circuits.}
\label{fig:wheatstone}
\end{figure}

\begin{figure}[htbp]
\centering
\begingroup
\footnotesize
\setlength{\tabcolsep}{1.5pt}
\renewcommand{\arraystretch}{1.25}
{\large $S=J+3I+A_{SH}=$}\par\medskip
\resizebox{\linewidth}{!}{%
\begin{tabular}{r|rrrr|rrrr|rrrr|rrrr}
 & \multicolumn{4}{c|}{$H$} & \multicolumn{4}{c|}{$(1,0)+H$} & \multicolumn{4}{c|}{$(0,1)+H$} & \multicolumn{4}{c}{$(1,1)+H$} \\
 & $(0,0)$ & $(0,2)$ & $(2,0)$ & $(2,2)$ & $(1,0)$ & $(1,2)$ & $(3,0)$ & $(3,2)$ & $(0,1)$ & $(0,3)$ & $(2,1)$ & $(2,3)$ & $(1,1)$ & $(1,3)$ & $(3,1)$ & $(3,3)$ \\ \hhline{-|----|----|----|----}
$(0,0)$ & \DoobMatrixEntry{4} & \DoobMatrixEntry{1} & \DoobMatrixEntry{1} & \DoobMatrixEntry{1} & \DoobMatrixEntry{2} & \DoobMatrixEntry{1} & \DoobMatrixEntry{2} & \DoobMatrixEntry{1} & \DoobMatrixEntry{2} & \DoobMatrixEntry{2} & \DoobMatrixEntry{1} & \DoobMatrixEntry{1} & \DoobMatrixEntry{2} & \DoobMatrixEntry{1} & \DoobMatrixEntry{1} & \DoobMatrixEntry{2} \\
$(0,2)$ & \DoobMatrixEntry{1} & \DoobMatrixEntry{4} & \DoobMatrixEntry{1} & \DoobMatrixEntry{1} & \DoobMatrixEntry{1} & \DoobMatrixEntry{2} & \DoobMatrixEntry{1} & \DoobMatrixEntry{2} & \DoobMatrixEntry{2} & \DoobMatrixEntry{2} & \DoobMatrixEntry{1} & \DoobMatrixEntry{1} & \DoobMatrixEntry{1} & \DoobMatrixEntry{2} & \DoobMatrixEntry{2} & \DoobMatrixEntry{1} \\
$(2,0)$ & \DoobMatrixEntry{1} & \DoobMatrixEntry{1} & \DoobMatrixEntry{4} & \DoobMatrixEntry{1} & \DoobMatrixEntry{2} & \DoobMatrixEntry{1} & \DoobMatrixEntry{2} & \DoobMatrixEntry{1} & \DoobMatrixEntry{1} & \DoobMatrixEntry{1} & \DoobMatrixEntry{2} & \DoobMatrixEntry{2} & \DoobMatrixEntry{1} & \DoobMatrixEntry{2} & \DoobMatrixEntry{2} & \DoobMatrixEntry{1} \\
$(2,2)$ & \DoobMatrixEntry{1} & \DoobMatrixEntry{1} & \DoobMatrixEntry{1} & \DoobMatrixEntry{4} & \DoobMatrixEntry{1} & \DoobMatrixEntry{2} & \DoobMatrixEntry{1} & \DoobMatrixEntry{2} & \DoobMatrixEntry{1} & \DoobMatrixEntry{1} & \DoobMatrixEntry{2} & \DoobMatrixEntry{2} & \DoobMatrixEntry{2} & \DoobMatrixEntry{1} & \DoobMatrixEntry{1} & \DoobMatrixEntry{2} \\ \hhline{-|----|----|----|----}
$(1,0)$ & \DoobMatrixEntry{2} & \DoobMatrixEntry{1} & \DoobMatrixEntry{2} & \DoobMatrixEntry{1} & \DoobMatrixEntry{4} & \DoobMatrixEntry{1} & \DoobMatrixEntry{1} & \DoobMatrixEntry{1} & \DoobMatrixEntry{1} & \DoobMatrixEntry{2} & \DoobMatrixEntry{2} & \DoobMatrixEntry{1} & \DoobMatrixEntry{2} & \DoobMatrixEntry{2} & \DoobMatrixEntry{1} & \DoobMatrixEntry{1} \\
$(1,2)$ & \DoobMatrixEntry{1} & \DoobMatrixEntry{2} & \DoobMatrixEntry{1} & \DoobMatrixEntry{2} & \DoobMatrixEntry{1} & \DoobMatrixEntry{4} & \DoobMatrixEntry{1} & \DoobMatrixEntry{1} & \DoobMatrixEntry{2} & \DoobMatrixEntry{1} & \DoobMatrixEntry{1} & \DoobMatrixEntry{2} & \DoobMatrixEntry{2} & \DoobMatrixEntry{2} & \DoobMatrixEntry{1} & \DoobMatrixEntry{1} \\
$(3,0)$ & \DoobMatrixEntry{2} & \DoobMatrixEntry{1} & \DoobMatrixEntry{2} & \DoobMatrixEntry{1} & \DoobMatrixEntry{1} & \DoobMatrixEntry{1} & \DoobMatrixEntry{4} & \DoobMatrixEntry{1} & \DoobMatrixEntry{2} & \DoobMatrixEntry{1} & \DoobMatrixEntry{1} & \DoobMatrixEntry{2} & \DoobMatrixEntry{1} & \DoobMatrixEntry{1} & \DoobMatrixEntry{2} & \DoobMatrixEntry{2} \\
$(3,2)$ & \DoobMatrixEntry{1} & \DoobMatrixEntry{2} & \DoobMatrixEntry{1} & \DoobMatrixEntry{2} & \DoobMatrixEntry{1} & \DoobMatrixEntry{1} & \DoobMatrixEntry{1} & \DoobMatrixEntry{4} & \DoobMatrixEntry{1} & \DoobMatrixEntry{2} & \DoobMatrixEntry{2} & \DoobMatrixEntry{1} & \DoobMatrixEntry{1} & \DoobMatrixEntry{1} & \DoobMatrixEntry{2} & \DoobMatrixEntry{2} \\ \hhline{-|----|----|----|----}
$(0,1)$ & \DoobMatrixEntry{2} & \DoobMatrixEntry{2} & \DoobMatrixEntry{1} & \DoobMatrixEntry{1} & \DoobMatrixEntry{1} & \DoobMatrixEntry{2} & \DoobMatrixEntry{2} & \DoobMatrixEntry{1} & \DoobMatrixEntry{4} & \DoobMatrixEntry{1} & \DoobMatrixEntry{1} & \DoobMatrixEntry{1} & \DoobMatrixEntry{2} & \DoobMatrixEntry{1} & \DoobMatrixEntry{2} & \DoobMatrixEntry{1} \\
$(0,3)$ & \DoobMatrixEntry{2} & \DoobMatrixEntry{2} & \DoobMatrixEntry{1} & \DoobMatrixEntry{1} & \DoobMatrixEntry{2} & \DoobMatrixEntry{1} & \DoobMatrixEntry{1} & \DoobMatrixEntry{2} & \DoobMatrixEntry{1} & \DoobMatrixEntry{4} & \DoobMatrixEntry{1} & \DoobMatrixEntry{1} & \DoobMatrixEntry{1} & \DoobMatrixEntry{2} & \DoobMatrixEntry{1} & \DoobMatrixEntry{2} \\
$(2,1)$ & \DoobMatrixEntry{1} & \DoobMatrixEntry{1} & \DoobMatrixEntry{2} & \DoobMatrixEntry{2} & \DoobMatrixEntry{2} & \DoobMatrixEntry{1} & \DoobMatrixEntry{1} & \DoobMatrixEntry{2} & \DoobMatrixEntry{1} & \DoobMatrixEntry{1} & \DoobMatrixEntry{4} & \DoobMatrixEntry{1} & \DoobMatrixEntry{2} & \DoobMatrixEntry{1} & \DoobMatrixEntry{2} & \DoobMatrixEntry{1} \\
$(2,3)$ & \DoobMatrixEntry{1} & \DoobMatrixEntry{1} & \DoobMatrixEntry{2} & \DoobMatrixEntry{2} & \DoobMatrixEntry{1} & \DoobMatrixEntry{2} & \DoobMatrixEntry{2} & \DoobMatrixEntry{1} & \DoobMatrixEntry{1} & \DoobMatrixEntry{1} & \DoobMatrixEntry{1} & \DoobMatrixEntry{4} & \DoobMatrixEntry{1} & \DoobMatrixEntry{2} & \DoobMatrixEntry{1} & \DoobMatrixEntry{2} \\ \hhline{-|----|----|----|----}
$(1,1)$ & \DoobMatrixEntry{2} & \DoobMatrixEntry{1} & \DoobMatrixEntry{1} & \DoobMatrixEntry{2} & \DoobMatrixEntry{2} & \DoobMatrixEntry{2} & \DoobMatrixEntry{1} & \DoobMatrixEntry{1} & \DoobMatrixEntry{2} & \DoobMatrixEntry{1} & \DoobMatrixEntry{2} & \DoobMatrixEntry{1} & \DoobMatrixEntry{4} & \DoobMatrixEntry{1} & \DoobMatrixEntry{1} & \DoobMatrixEntry{1} \\
$(1,3)$ & \DoobMatrixEntry{1} & \DoobMatrixEntry{2} & \DoobMatrixEntry{2} & \DoobMatrixEntry{1} & \DoobMatrixEntry{2} & \DoobMatrixEntry{2} & \DoobMatrixEntry{1} & \DoobMatrixEntry{1} & \DoobMatrixEntry{1} & \DoobMatrixEntry{2} & \DoobMatrixEntry{1} & \DoobMatrixEntry{2} & \DoobMatrixEntry{1} & \DoobMatrixEntry{4} & \DoobMatrixEntry{1} & \DoobMatrixEntry{1} \\
$(3,1)$ & \DoobMatrixEntry{1} & \DoobMatrixEntry{2} & \DoobMatrixEntry{2} & \DoobMatrixEntry{1} & \DoobMatrixEntry{1} & \DoobMatrixEntry{1} & \DoobMatrixEntry{2} & \DoobMatrixEntry{2} & \DoobMatrixEntry{2} & \DoobMatrixEntry{1} & \DoobMatrixEntry{2} & \DoobMatrixEntry{1} & \DoobMatrixEntry{1} & \DoobMatrixEntry{1} & \DoobMatrixEntry{4} & \DoobMatrixEntry{1} \\
$(3,3)$ & \DoobMatrixEntry{2} & \DoobMatrixEntry{1} & \DoobMatrixEntry{1} & \DoobMatrixEntry{2} & \DoobMatrixEntry{1} & \DoobMatrixEntry{1} & \DoobMatrixEntry{2} & \DoobMatrixEntry{2} & \DoobMatrixEntry{1} & \DoobMatrixEntry{2} & \DoobMatrixEntry{1} & \DoobMatrixEntry{2} & \DoobMatrixEntry{1} & \DoobMatrixEntry{1} & \DoobMatrixEntry{1} & \DoobMatrixEntry{4} \\
\end{tabular}%
}
\endgroup
\medskip
\[
K=
\left(\begin{array}{rrrr}
\DoobMatrixEntry{2} & \DoobMatrixEntry{1} & \DoobMatrixEntry{1} & \DoobMatrixEntry{1} \\
\DoobMatrixEntry{1} & \DoobMatrixEntry{2} & \DoobMatrixEntry{1} & \DoobMatrixEntry{1} \\
\DoobMatrixEntry{1} & \DoobMatrixEntry{1} & \DoobMatrixEntry{2} & \DoobMatrixEntry{1} \\
\DoobMatrixEntry{1} & \DoobMatrixEntry{1} & \DoobMatrixEntry{1} & \DoobMatrixEntry{2} \\
\end{array}\right).
\]
\caption{The matrices $S$ and $K$. 
Here $A_{SH}$ is the adjacency matrix
of the Shrikhande graph $SH$.
In $S$, the diagonal entries are $4$, adjacent pairs have weight $2$, and all
remaining pairs have weight $1$.
The row and column labels are specific to the subgroup $H=2\mathbb Z_4\times 2\mathbb{Z}_4$ of $\mathbb{Z}_4\times \mathbb{Z}_4$.
}
\label{fig: S and K}
\end{figure}

\begin{figure}[htbp]
\centering
\begingroup
\footnotesize
\setlength{\tabcolsep}{1.5pt}
\renewcommand{\arraystretch}{1.25}
{\large $\mathbf{W}(S)=$}\par\medskip
\resizebox{\linewidth}{!}{%
\begin{tabular}{r|rrrr|rrrr|rrrr|rrrr}
 & \multicolumn{4}{c|}{$H$} & \multicolumn{4}{c|}{$(1,0)+H$} & \multicolumn{4}{c|}{$(0,1)+H$} & \multicolumn{4}{c}{$(1,1)+H$} \\
 & $(0,0)$ & $(0,2)$ & $(2,0)$ & $(2,2)$ & $(1,0)$ & $(1,2)$ & $(3,0)$ & $(3,2)$ & $(0,1)$ & $(0,3)$ & $(2,1)$ & $(2,3)$ & $(1,1)$ & $(1,3)$ & $(3,1)$ & $(3,3)$ \\ \hhline{-|----|----|----|----}
$(0,0)$ & \DoobMatrixEntry{4624} & \DoobMatrixEntry{2116} & \DoobMatrixEntry{2116} & \DoobMatrixEntry{2116} & \DoobMatrixEntry{3126} & \DoobMatrixEntry{2118} & \DoobMatrixEntry{3126} & \DoobMatrixEntry{2118} & \DoobMatrixEntry{3126} & \DoobMatrixEntry{3126} & \DoobMatrixEntry{2118} & \DoobMatrixEntry{2118} & \DoobMatrixEntry{3126} & \DoobMatrixEntry{2118} & \DoobMatrixEntry{2118} & \DoobMatrixEntry{3126} \\
$(0,2)$ & \DoobMatrixEntry{2116} & \DoobMatrixEntry{4624} & \DoobMatrixEntry{2116} & \DoobMatrixEntry{2116} & \DoobMatrixEntry{2118} & \DoobMatrixEntry{3126} & \DoobMatrixEntry{2118} & \DoobMatrixEntry{3126} & \DoobMatrixEntry{3126} & \DoobMatrixEntry{3126} & \DoobMatrixEntry{2118} & \DoobMatrixEntry{2118} & \DoobMatrixEntry{2118} & \DoobMatrixEntry{3126} & \DoobMatrixEntry{3126} & \DoobMatrixEntry{2118} \\
$(2,0)$ & \DoobMatrixEntry{2116} & \DoobMatrixEntry{2116} & \DoobMatrixEntry{4624} & \DoobMatrixEntry{2116} & \DoobMatrixEntry{3126} & \DoobMatrixEntry{2118} & \DoobMatrixEntry{3126} & \DoobMatrixEntry{2118} & \DoobMatrixEntry{2118} & \DoobMatrixEntry{2118} & \DoobMatrixEntry{3126} & \DoobMatrixEntry{3126} & \DoobMatrixEntry{2118} & \DoobMatrixEntry{3126} & \DoobMatrixEntry{3126} & \DoobMatrixEntry{2118} \\
$(2,2)$ & \DoobMatrixEntry{2116} & \DoobMatrixEntry{2116} & \DoobMatrixEntry{2116} & \DoobMatrixEntry{4624} & \DoobMatrixEntry{2118} & \DoobMatrixEntry{3126} & \DoobMatrixEntry{2118} & \DoobMatrixEntry{3126} & \DoobMatrixEntry{2118} & \DoobMatrixEntry{2118} & \DoobMatrixEntry{3126} & \DoobMatrixEntry{3126} & \DoobMatrixEntry{3126} & \DoobMatrixEntry{2118} & \DoobMatrixEntry{2118} & \DoobMatrixEntry{3126} \\ \hhline{-|----|----|----|----}
$(1,0)$ & \DoobMatrixEntry{3126} & \DoobMatrixEntry{2118} & \DoobMatrixEntry{3126} & \DoobMatrixEntry{2118} & \DoobMatrixEntry{4624} & \DoobMatrixEntry{2116} & \DoobMatrixEntry{2116} & \DoobMatrixEntry{2116} & \DoobMatrixEntry{2118} & \DoobMatrixEntry{3126} & \DoobMatrixEntry{3126} & \DoobMatrixEntry{2118} & \DoobMatrixEntry{3126} & \DoobMatrixEntry{3126} & \DoobMatrixEntry{2118} & \DoobMatrixEntry{2118} \\
$(1,2)$ & \DoobMatrixEntry{2118} & \DoobMatrixEntry{3126} & \DoobMatrixEntry{2118} & \DoobMatrixEntry{3126} & \DoobMatrixEntry{2116} & \DoobMatrixEntry{4624} & \DoobMatrixEntry{2116} & \DoobMatrixEntry{2116} & \DoobMatrixEntry{3126} & \DoobMatrixEntry{2118} & \DoobMatrixEntry{2118} & \DoobMatrixEntry{3126} & \DoobMatrixEntry{3126} & \DoobMatrixEntry{3126} & \DoobMatrixEntry{2118} & \DoobMatrixEntry{2118} \\
$(3,0)$ & \DoobMatrixEntry{3126} & \DoobMatrixEntry{2118} & \DoobMatrixEntry{3126} & \DoobMatrixEntry{2118} & \DoobMatrixEntry{2116} & \DoobMatrixEntry{2116} & \DoobMatrixEntry{4624} & \DoobMatrixEntry{2116} & \DoobMatrixEntry{3126} & \DoobMatrixEntry{2118} & \DoobMatrixEntry{2118} & \DoobMatrixEntry{3126} & \DoobMatrixEntry{2118} & \DoobMatrixEntry{2118} & \DoobMatrixEntry{3126} & \DoobMatrixEntry{3126} \\
$(3,2)$ & \DoobMatrixEntry{2118} & \DoobMatrixEntry{3126} & \DoobMatrixEntry{2118} & \DoobMatrixEntry{3126} & \DoobMatrixEntry{2116} & \DoobMatrixEntry{2116} & \DoobMatrixEntry{2116} & \DoobMatrixEntry{4624} & \DoobMatrixEntry{2118} & \DoobMatrixEntry{3126} & \DoobMatrixEntry{3126} & \DoobMatrixEntry{2118} & \DoobMatrixEntry{2118} & \DoobMatrixEntry{2118} & \DoobMatrixEntry{3126} & \DoobMatrixEntry{3126} \\ \hhline{-|----|----|----|----}
$(0,1)$ & \DoobMatrixEntry{3126} & \DoobMatrixEntry{3126} & \DoobMatrixEntry{2118} & \DoobMatrixEntry{2118} & \DoobMatrixEntry{2118} & \DoobMatrixEntry{3126} & \DoobMatrixEntry{3126} & \DoobMatrixEntry{2118} & \DoobMatrixEntry{4624} & \DoobMatrixEntry{2116} & \DoobMatrixEntry{2116} & \DoobMatrixEntry{2116} & \DoobMatrixEntry{3126} & \DoobMatrixEntry{2118} & \DoobMatrixEntry{3126} & \DoobMatrixEntry{2118} \\
$(0,3)$ & \DoobMatrixEntry{3126} & \DoobMatrixEntry{3126} & \DoobMatrixEntry{2118} & \DoobMatrixEntry{2118} & \DoobMatrixEntry{3126} & \DoobMatrixEntry{2118} & \DoobMatrixEntry{2118} & \DoobMatrixEntry{3126} & \DoobMatrixEntry{2116} & \DoobMatrixEntry{4624} & \DoobMatrixEntry{2116} & \DoobMatrixEntry{2116} & \DoobMatrixEntry{2118} & \DoobMatrixEntry{3126} & \DoobMatrixEntry{2118} & \DoobMatrixEntry{3126} \\
$(2,1)$ & \DoobMatrixEntry{2118} & \DoobMatrixEntry{2118} & \DoobMatrixEntry{3126} & \DoobMatrixEntry{3126} & \DoobMatrixEntry{3126} & \DoobMatrixEntry{2118} & \DoobMatrixEntry{2118} & \DoobMatrixEntry{3126} & \DoobMatrixEntry{2116} & \DoobMatrixEntry{2116} & \DoobMatrixEntry{4624} & \DoobMatrixEntry{2116} & \DoobMatrixEntry{3126} & \DoobMatrixEntry{2118} & \DoobMatrixEntry{3126} & \DoobMatrixEntry{2118} \\
$(2,3)$ & \DoobMatrixEntry{2118} & \DoobMatrixEntry{2118} & \DoobMatrixEntry{3126} & \DoobMatrixEntry{3126} & \DoobMatrixEntry{2118} & \DoobMatrixEntry{3126} & \DoobMatrixEntry{3126} & \DoobMatrixEntry{2118} & \DoobMatrixEntry{2116} & \DoobMatrixEntry{2116} & \DoobMatrixEntry{2116} & \DoobMatrixEntry{4624} & \DoobMatrixEntry{2118} & \DoobMatrixEntry{3126} & \DoobMatrixEntry{2118} & \DoobMatrixEntry{3126} \\ \hhline{-|----|----|----|----}
$(1,1)$ & \DoobMatrixEntry{3126} & \DoobMatrixEntry{2118} & \DoobMatrixEntry{2118} & \DoobMatrixEntry{3126} & \DoobMatrixEntry{3126} & \DoobMatrixEntry{3126} & \DoobMatrixEntry{2118} & \DoobMatrixEntry{2118} & \DoobMatrixEntry{3126} & \DoobMatrixEntry{2118} & \DoobMatrixEntry{3126} & \DoobMatrixEntry{2118} & \DoobMatrixEntry{4624} & \DoobMatrixEntry{2116} & \DoobMatrixEntry{2116} & \DoobMatrixEntry{2116} \\
$(1,3)$ & \DoobMatrixEntry{2118} & \DoobMatrixEntry{3126} & \DoobMatrixEntry{3126} & \DoobMatrixEntry{2118} & \DoobMatrixEntry{3126} & \DoobMatrixEntry{3126} & \DoobMatrixEntry{2118} & \DoobMatrixEntry{2118} & \DoobMatrixEntry{2118} & \DoobMatrixEntry{3126} & \DoobMatrixEntry{2118} & \DoobMatrixEntry{3126} & \DoobMatrixEntry{2116} & \DoobMatrixEntry{4624} & \DoobMatrixEntry{2116} & \DoobMatrixEntry{2116} \\
$(3,1)$ & \DoobMatrixEntry{2118} & \DoobMatrixEntry{3126} & \DoobMatrixEntry{3126} & \DoobMatrixEntry{2118} & \DoobMatrixEntry{2118} & \DoobMatrixEntry{2118} & \DoobMatrixEntry{3126} & \DoobMatrixEntry{3126} & \DoobMatrixEntry{3126} & \DoobMatrixEntry{2118} & \DoobMatrixEntry{3126} & \DoobMatrixEntry{2118} & \DoobMatrixEntry{2116} & \DoobMatrixEntry{2116} & \DoobMatrixEntry{4624} & \DoobMatrixEntry{2116} \\
$(3,3)$ & \DoobMatrixEntry{3126} & \DoobMatrixEntry{2118} & \DoobMatrixEntry{2118} & \DoobMatrixEntry{3126} & \DoobMatrixEntry{2118} & \DoobMatrixEntry{2118} & \DoobMatrixEntry{3126} & \DoobMatrixEntry{3126} & \DoobMatrixEntry{2118} & \DoobMatrixEntry{3126} & \DoobMatrixEntry{2118} & \DoobMatrixEntry{3126} & \DoobMatrixEntry{2116} & \DoobMatrixEntry{2116} & \DoobMatrixEntry{2116} & \DoobMatrixEntry{4624} \\
\end{tabular}%
}
\endgroup
\medskip
\[
\mathbf{W}(K)=
\left(\begin{array}{rrrr}
\DoobMatrixEntry{68} & \DoobMatrixEntry{46} & \DoobMatrixEntry{46} & \DoobMatrixEntry{46} \\
\DoobMatrixEntry{46} & \DoobMatrixEntry{68} & \DoobMatrixEntry{46} & \DoobMatrixEntry{46} \\
\DoobMatrixEntry{46} & \DoobMatrixEntry{46} & \DoobMatrixEntry{68} & \DoobMatrixEntry{46} \\
\DoobMatrixEntry{46} & \DoobMatrixEntry{46} & \DoobMatrixEntry{46} & \DoobMatrixEntry{68} \\
\end{array}\right)
\]
\caption{$\mathbf{W}(S)$ and $\mathbf{W}(K)$ computed from
$\mathbf{W}(A)_{ab}=\sum_{u,v\in V(\mathbf{W})}A_{au}A_{ub}A_{av}A_{vb}A_{uv}$.
The row and column order is the same as \cref{fig: S and K}. 
In $\mathbf{W}(S)$, the diagonal $4\times4$ blocks contain $4624$ and $2116$, neither is a multiple of $3$; the off-diagonal blocks contain $3126$ and $2118$, both are divisible by $3$, and the order of $3$ is $1.$, i.e.,
$3^1 \parallel  3126, 2118$.
In $\mathbf{W}(K)$, $68$ and $46$ are not multiples of $3.$
\label{fig: S and K after Wheatstone}}
\end{figure}

For a planar input graph $G$, we obtain the planar graph $\mathbf{W}G$ by replacing every edge in $G$ by a Wheatstone gadget, with $a$ and $b$ attached to the two endpoints of an original edge.
Define the matrix $\mathbf{W}(N)\in \Sym{n}{}{\Zp}$ by $$\mathbf{W}(N)_{ab}=\sum_{u,v\in[n]}N_{au}N_{ub}N_{av}N_{vb}N_{uv}.$$
Then $Z_{\mathbf{W}(N)}(G)=Z_N(\mathbf{W}G)$ (see Definition 1, \cite{cai-maran-young}).
Oracle access to $Z_{N}(\mathbf{W}G)$ gives $Z_{\mathbf{W}(N)}(G).$
In other words, $\PlGH(\mathbf{W}(N))\le_T\PlGH(N).$

It is immediate from the definition that tensor products are preserved by edge gadgets. (This is true even for  nonplanar gadget.)
So $\mathbf{W}(N)=\mathbf{W}(S)^{\otimes p}\otimes \mathbf{W}(K)^{\otimes q}$,
for some $p\ge1, q\ge 0$.
%%% This is just writing out, for A \otimes B, every assignment from V(G) to
% V(A \otimes B) is made up of a pair of maps from  V(G) to
% V(A) and to V(B)
A direct computation yields $\mathbf{W}(S)$ and $\mathbf{W}(K)$, see \cref{fig: S and K after Wheatstone}.
Notice that In $\mathbf{W}(S)$, the diagonal $4\times4$ blocks contain $4624$ and $2116$, neither is a multiple of $3$; the off-diagonal blocks contain $3126$ and $2118$, both are divisible by $3$, and the exact order of $3$ is $1$,
i.e.,
$3^1 \parallel  3126, 2118$.
In $\mathbf{W}(K)$, $68$ and $46$ are not multiples of $3.$
Therefore, there exists a set of distinct prime generators $\{3, p_2,\ldots,p_d\}$ for entries of $N$.
Suppose $\mathbf{W}(N)_{ij}=3^{e_{ij1}}p_2^{e_{ij2}}\cdots p_d^{e_{ijd}}$, for every $i,j\in X.$
We define $\mathcal{M}(\mathbf{g})_{ij}=\mathcal{M}(g_1,\ldots,g_d)_{ij}=g_1^{e_{ij1}}g_2^{e_{ij2}}\cdots g_d^{e_{ijd}}$, for every $i,j\in X.$
By \cref{lemma: full_thickeningInterpolation}, $\PlGH(\mathcal{M}(3,1,\ldots,1))\le_T\PlGH(\mathbf{W}(N))$.
Note that
$\mathcal{M}(3,1,\ldots,1)=(P_4\otimes J_4)^{\otimes p}\otimes J_4^{\otimes q}=P_4^{\otimes p}\otimes J_4^{\otimes (p+q)}$, where $P_4=\left(\begin{smallmatrix}
    1 & 3 & 3 & 3\\
    3 & 1 & 3 & 3\\
    3 & 3 & 1 & 3\\
    3 & 3 & 3 & 1\\
\end{smallmatrix}\right)$ and $J_4$ is the 4 by 4 all one matrix.
So $Z_{\mathcal{M}(3,1,\ldots,1)}(G)=Z_{P_4}(G)^pZ_{J_4}(G)^{p+q}=4^{(p+q)|V(G)|}Z_{P_4}(G)^p$, for any input planar graph $G.$
Oracle access to $Z_{\mathcal{M}(3,1,\ldots,1)}(G)$ gives $Z_{P_4}(G).$
Thus, $\PlGH(P_4)\le_T\PlGH(\mathcal{M}(3,1,\ldots,1))\le_T\PlGH(\mathbf{W}(N))\le_T \PlGH(N)\le_T \PlGH(M)$.
Finally, by \cref{thm:potts-source} and a constant scaling, $\PlGH(P_4)$ is \#P-hard.
Therefore, we have proved the following lemma:
\begin{lemma}\label{lem:doob_is_hard}
    If there exists some $\omega\in\Omega$ such that $\Gamma_\omega$ is isomorphic to a Doob graph, then $\PlGH(M)$ is \#P-hard.
\end{lemma}

\section{Proof of Nonconstant Diagonal Hardness}\label{app:nonconstant}
In this section, we prove:
\nonconstantdiag*

\paragraph{Proof Overview}
We prove \#P-hardness for positive definite $M$.
Generalizing to full rank $M$ is relatively easy.
The proof first assumes that a minimal counterexample $M\in \Sym{n}{pd}{\Rp}$ (\cref{def:minimal-counterexample}) exists.
We then prove every realizable unary signature must satisfy stringent conditions, otherwise we may obtain a smaller counterexample (\cref{lem:unary-profile}).
Following this, we prove that the diagonals of $M$ lie in a geometric progression $M_{ii}=\kappa\rho^{k_i}$ (\cref{lem:geometric-levels}).
We call the exponents $(k_1, \ldots, k_n)$ the \emph{diagonal levels} of $M$, and the number of indices at each level,
$n_0, n_1, \ldots, n_D$, the \emph{diagonal level data} (\cref{def:diagonal-levels}).
%%% I added these names b/c later in this Proof Overview we state n_0, n_D
%%% alsdo stating these out make the plural usage "have" to be correct.
The main body of \autoref{app:nonconstant} is proving the diagonal level data have the binomial form (\cref{cor:binary-levels}).
Then in \cref{lem:boolean-power}, we reduce a $2$ by $2$ \#P-hard matrix $W$ to $M$ by a successively \emph{unpinning} process, obtaining a contradiction.

To prove that the diagonal level data have the binomial form, we show that $M$ has the $i$-th row generating polynomial in the form \(U(z)^{D-k_i}V(z)^{k_i}\), where \(U,V\) are two nonproportional linear polynomials with positive coefficients (\cref{lem:level-polynomials}).
Comparing the coefficients shows that the diagonal level data are nearly the binomial form (\cref{lem:binomial-levels}).
These are results in \cref{subsec: diagonal levels and their data}.

To squeeze out the exact binomial form, we appeal to the parity unary $u=((-1)^{k_1},\ldots,(-1)^{k_n})$.
We show that $u$ is realizable in \cref{lem:parity-unary}.
In parallel, we show that every realizable $H\in \Sym{n}{pd}{\Rp}$ respecting parity must be diagonal (\cref{lem:parity-quotient}).
We apply this to the matrix $G$ obtained by taking the geometric mean of $M$ and its parity modification, thereby showing that $G$ is diagonal (\cref{lem:geometric-mean}).
This allows us to define \(A(z)\) as the matrix-valued generating polynomial for the diagonal levels (\cref{def:normalized-family}), and show that $A(z)$ satisfies the reciprocal rule (\cref{lem:reciprocal}).
Suppose for a contradiction that there are distinct indices $i,j$ such that $k_i=k_j=0$ (or $k_i=k_j=D$).
Inspecting on $A(z)_{ii},A(z)_{jj}$ and $A(z)_{ij}$ shows a contradiction.
Therefore, the endpoint diagonal level data satisfy $n_0=n_D=1$ (\cref{lem:singleton-endpoints}).
Combined with the nearly binomial form \cref{lem:binomial-levels}, we conclude that the diagonal level data have the exact binomial form (\cref{cor:binary-levels}).
These are results in \cref{subsec: binomial diag level data}.

\begin{definition}[Minimal counterexample]\label{def:minimal-counterexample}
A \emph{minimal counterexample} is a matrix $M\in \Sym{n}{pd}{\Rp}$ with nonconstant diagonal, such that $\PlGH(M)$ is not \#P-hard and $n$ is minimal.
\end{definition}

Notice that in \cref{def:minimal-counterexample}, we require that $\PlGH(M)$ is not \#P-hard, instead of $\PlGH(M)$ is tractable.
Assume that a counterexample exists; we  fix a minimal one. 
\cref{thm:boolean-source} implies $n\geq3$. 
All realizability in \cref{app:nonconstant}
%this section 
is relative to this minimal counterexample $M.$

\subsection{Preliminaries}
\subsubsection{Polynomial Interpolation}
We collect some polynomial interpolation results that will be used in this section.
In particular, we build \cref{lem: realize unary by conformal lattice interpolation,lem: support interpolation} based on the conformal lattice interpolation (\cref{lm:conformal lattice interpolation}, \cite{CaiFST26}).

\begin{lemma}\label{lem: interpolation extract principle submatrix}
    Let $M\in \Sym{n}{pd}{\Rp}$ and $w\in \Rp^{n}.$
    Suppose $0 < w_i \le 1$, and let $C = \{i \in [n] \mid w_i = 1\}$ be non-empty. 
    Let $M_C$ denote the principal submatrix of $M$ indexed by $C$.
    Then $\PlGH(M_C)\le_T \PlHolant(\mathcal{EQ}\mid M,w).$
\end{lemma}
\begin{proof}
    Given an input planar graph $G=(V,E)$,
    define 
    \[
    F_G(k)=\sum_{\sigma:V\to [n]}\prod_{(u,v)\in E}M_{\sigma(u)\sigma(v)}\prod_{v\in V}w_{\sigma(v)}^k.
    \]
    Attach $k$ copies of $w$ to every vertex of $G$,
    oracle access to $\PlHolant(\mathcal{EQ}\mid M,w)$ gives $F_G(k)$ for any fixed $k\in \mathbb{Z}_{>0}$ in polynomial time (in terms of the size of $G$). 
    Define $X$ to be the set of all possible (distinct) product values $\prod_{v\in V}w_{\sigma(v)}$ appearing in $F_G(k)$:
    $$
    X=\left\{\prod_{i=1}^n w_i^{t_i}\mid t_i\in \mathbb{Z}_{\ge 0}\text{ and } \sum_{i=1}^n t_i=|V|\right\}. 
    $$
    There are $\binom{|V|+n-1}{n-1}=|V|^{O(1)}$ many integer sequences $(t_i)_{i\in[n]}$ such that $t_i\in \mathbb Z_{\ge 0}$ and $\sum_{i=1}^n t_i=|V|.$
    Note that $n$ is a constant.
    So $X$ has polynomial size in $G$, and elements in $X$ are enumerable in polynomial time.
    In particular, $1\in X$ since $C=\{i\in[n]\mid w_i=1\}$ is nonempty.
    Define $A_x$ to be the coefficient of $x\in X$ appearing in $F_G(k):$
    \[
    A_x=\sum_{\substack{\sigma: V\to [n]\\\prod_{v\in V}w_{\sigma(v)}=x}}\prod_{(u,v)\in E}M_{\sigma(u)\sigma(v)} \quad (x\in X).
    \]
    Then $F_G(k)=\sum_{x\in X}A_x x^{k}, k\in \Zp.$
    Query $F_G(1),F_G(2),\ldots, F_G(|X|)$, we obtain a full rank Vandermonde system in variables $(A_x)_{x\in X}.$
    Solving the Vandermonde system, we recover $A_1$ in polynomial time.
    Every term in the sum expression of $A_1$ indexed by $\sigma$ satisfies $\prod_{v\in V}w_{\sigma(v)}=1.$
    Since $0<w_i\le 1$, $\prod_{v\in V}w_{\sigma(v)}=1\iff w_{\sigma(v)}=1$ for every $v\in V\iff \sigma(v)\in C$ for every $v\in V$.
    Thus,
    \[
    A_1=\sum_{\sigma: V\to C}\prod_{(u,v)\in E}M_{\sigma(u)\sigma(v)}=Z_{M_C}(G).
    \]
    Therefore, $\PlGH(M_C)\le_T \PlHolant(\mathcal{EQ}\mid M,w).$
\end{proof}

\begin{definition}\label{definition: latticeSet}
    Let $(a_1,a_2,\ldots,a_n)\in \R_{\neq 0}^n$.
    Define
    $${\mathfrak{L}}(a_1,a_2,\ldots,a_n)
    =\left\{(x_1,x_2,\ldots, x_n)\in\Z^n \mid \prod_{i=1}^n a_i^{x_i}=1 \text{ and } \sum_{i=1}^n x_i=0 \right\}.$$
\end{definition}

\begin{lemma}[Conformal lattice interpolation \cite{CaiFST26}]\label{lm:conformal lattice interpolation}
    Let $n\in \Zp$, ${x}=(x_1,x_2,\ldots,x_n)\in \R_{\neq 0}^n$ and  ${y}=(y_1,y_2,\ldots,y_n)\in \R_{\neq 0}^n$.
    Suppose 
    ${\mathfrak{L}}(x_1,x_2\ldots,x_n) \subseteq {\mathfrak{L}}(y_1,y_2\ldots,y_n)$.
    Given the numbers $$N_\ell(x_1, x_2, \ldots, x_n) := \sum\limits_{\substack{j_1,j_2,\dots, j_n \geq 0 \\ j_1+j_2+\dots+j_n= m}} (x_1^{j_1}x_2^{j_2}\cdots x_n^{j_n})^\ell \cdot  z_{j_1,j_2,\dots,j_n}$$ for $\ell = 1,2,\dots,  \binom{m+n-1}{n-1}$, we can compute $N_1(y_1, y_2, \ldots, y_n)$ in time polynomial in $m$. (In this time complexity statement, $n$, $x$ and $y$ are considered fixed.)
\end{lemma}

\begin{lemma}\label{lem: realize unary by conformal lattice interpolation}
    Suppose ${x}=(x_1,x_1,\ldots, x_n)\in \R_{\neq 0}^n$ and $y=(y_1,y_1,\ldots, y_n)\in \R_{\neq 0}^n$ are two unary signatures.
    If ${\mathfrak{L}}(x_1,x_2\ldots,x_n) \subseteq {\mathfrak{L}}(y_1,y_2\ldots,y_n)$ and ${x}$ is realizable, then ${y}$ is also realizable.
\end{lemma}
\begin{proof}
    We prove $\PlHolant(\mathcal{EQ}\mid M, {y})\le_T \PlHolant(\mathcal{EQ}\mid M, {x}).$
    Consider a planar graph $G=(V,E)$ as an instance of the problem $\PlHolant(\mathcal{EQ} \mid M, y)$, where each vertex in $V$ is labeled by an {\sc Equality} signature, each edge is assigned the binary signature $M$, and a multiset $S \subseteq V$ is attached to the unary signature $y$.
The partition function of the instance is 
\begin{equation}\label{eq:N1y}
    \begin{aligned}
    N_1(y_1,y_2,\ldots, y_n)
    =&\sum_{\sigma:V\to [n]}\prod_{(u,v)\in E}M_{\sigma(u)\sigma(v)}\prod_{v\in S}y_{\sigma(v)}\\
    =& \sum_{\substack{j_1,j_2,\ldots, j_n\ge 0\\j_1+j_2+\cdots+j_n=|S|}}y_1^{j_1}y_2^{j_2}\cdots y_n^{j_n}\cdot z_{j_1,j_2,\ldots, j_n}.
\end{aligned}
\end{equation}
Here the coefficients $z_{j_1,j_2,\ldots, j_n}$ only depend on $G,S$ and $M$, but not depend on $y$. Indeed, $z_{j_1,j_2,\ldots, j_n}$ is the sum of the product values
$\prod_{(u,v)\in E} M_{\sigma(u)\sigma(v)}$ where the 
assignment $\sigma$ assigns values to the variables $V$ in such a way that
exactly $j_{i}$ many unaries (counting multiplicity) are attached
to variables assigned value $i \in [n]$.
Now we replace every occurrence of $y$ by $\ell$ copies of $x$, and obtain an instance of $\PlHolant(\mathcal{EQ}\mid M,x)$ with partition function
\begin{equation}\label{eq:Nld}
    \begin{aligned}
    N_\ell(x_1,x_2,\ldots, x_n)
    =&\sum_{\sigma:V\to [n]}\prod_{(u,v)\in E}M_{\sigma(u)\sigma(v)}\prod_{v\in S}x_{\sigma(v)}^{\ell}\\
    =& \sum_{\substack{j_1,j_2,\ldots, j_n\ge 0\\j_1+j_2+\cdots+j_n=|S|}}(x_1^{j_1}x_2^{j_2}\cdots x_n^{j_n})^{\ell}\cdot z_{j_1,j_2,\ldots, j_n}.
\end{aligned}
\end{equation}
The coefficients $z_{j_1,j_2,\ldots, j_n}$ in \eqref{eq:Nld} are the same as in \eqref{eq:N1y}.

By \cref{lm:conformal lattice interpolation}, we may query $N_{\ell}(d_1,\ldots, d_n)$ for $\ell=1,2,\ldots, \binom{|S|+n-1}{n-1}$ and recover $N_1(y_1,\ldots, y_n)$ in polynomial time in $|S|$.
Note that $n$ is a constant. 
So this takes polynomial time in the input size.
Therefore, $y$ is realizable.
\end{proof}

\begin{lemma}\label{lem: support interpolation}
    Let $M\in \Sym{n}{}{\R_{\ge 0}}.$
    Let $K\in \Sym{n}{}{\R_{\ge 0}}$ be the matrix obtained from $M$ by replacing every nonzero entry in $M$ by $1$.
    Then $\PlGH(K)\le_T\PlGH(M).$
\end{lemma}
\begin{proof}
    The case $M=0$ is trivial.
    Now suppose $M$ contains nonzero entries.
    Let \{$x_1,\ldots, x_d$\} be all distinct nonzero entries of $M$.
    Consider a planar graph $G=(V,E)$.
    By thickening \cref{cor:thickening}, we have oracle access to
    \begin{equation}\label{eq: support interpolation}
    \begin{aligned}
        Z_{M^{\circ \ell}}(G)
        =&\sum_{\sigma:V\to [n]}\prod_{(u,v)\in E}M_{\sigma(u)\sigma(v)}^{\ell }\\
        =&\sum_{\substack{j_1,j_2,\ldots, j_d\ge 0\\ j_1+j_2+\cdots +j_d=|E|}}(x_{1}^{j_1}x_2^{j_2}\cdots x_d^{j_d})^{\ell}\cdot z_{j_1,j_2,\ldots, j_d}
    \end{aligned} 
    \end{equation}
    for every $\ell \in \Zp.$
    Here, the coefficient $z_{j_1,j_2, \ldots, j_d}$ is the number of $\sigma:V\to [n]$ such that for every $i \in [d]$, exactly $j_i$ elements in $(M_{\sigma(u)\sigma(v)})_{(u,v)\in E}$ are equal to $x_i$.
    On the other hand,
    \begin{equation*}
    \begin{aligned}
        Z_{K}(G)
        =&\sum_{\sigma:V\to [n]}\prod_{(u,v)\in E}K_{\sigma(u)\sigma(v)}\\
        =&\sum_{\substack{j_1,j_2,\ldots, j_d\ge 0\\ j_1+j_2+\cdots +j_d=|E|}}(1^{j_1}1^{j_2}\cdots 1^{j_d})^{\ell}\cdot z_{j_1,j_2,\ldots, j_d}.
    \end{aligned} 
    \end{equation*}
    Obviously, $\mathfrak{L}(x_1,x_2,\ldots, x_d)\subseteq \mathfrak{L}(1,1,\ldots, 1)$.
    By \cref{lm:conformal lattice interpolation}, we may query $Z_{M^{\circ \ell}}(G)$ for $\ell = 1,2,\dots, \binom{|E|+d-1}{d-1}$ and recover $Z_K(G).$
    This takes polynomial time in $|E|.$
    Therefore, $\PlGH(K)\le_T \PlGH(M).$
\end{proof}

\subsubsection{Matrix Analysis}

\begin{theorem}[Schur product theorem; {\cite[Theorem 7.5.3]{horn-johnson}}]\label{thm:schur-product}
For every integer $n\geq1$ and $B,C\in\Sym{n}{pd}{\R}$, the entrywise product $B\circ C$ is still positive definite.
\end{theorem}

\begin{definition}[Matrix geometric mean; {\cite[Section 4.1]{Bhatia2007}}]\label{def:geometric-mean}
For integer $n\geq1$ and $X,Y\in \Sym{n}{pd}{\R}$, define
\(
X\#Y=X^{1/2}\bigl(X^{-1/2}YX^{-1/2}\bigr)^{1/2}X^{1/2}.
\)
\end{definition}

\begin{lemma}\label{lem:geometric-mean-equation}
For every integer $n\geq1$ and $X,Y\in \Sym{n}{pd}{\R}$, $X\#Y$ is the unique solution $G\in\Sym{n}{pd}{\R}$ of $GX^{-1}G=Y$.
\end{lemma}
\begin{proof}
Multiplying the equation $GX^{-1}G=Y$ on both sides by $X^{-1/2}$ gives
\[
\bigl(X^{-1/2}GX^{-1/2}\bigr)^2=X^{-1/2}YX^{-1/2}.
\]
The RHS is positive definite and has a unique positive-definite square root. Solving for $G$ gives exactly \cref{def:geometric-mean}.
\end{proof}

\subsection{Diagonal Levels and Their Data}\label{subsec: diagonal levels and their data}

\begin{lemma}\label{lem:unary-profile}
For every minimal counterexample $M$ and every positive realizable unary $u\in\mathbb R_{>0}^n$, there exist $c>0$ and $\gamma\in\mathbb R$ such that $u_i=c\,M_{ii}^\gamma$ for every $i\in[n]$.
\end{lemma}

\begin{proof}
    Let $m=\min_{i\in[n]}M_{ii}$ and $a=\max_{i\in[n],M_{ii}=m}u_i.$
    Since $M$ has nonconstant diagonal, there exists $i\in[n]$ such that $M_{ii}>m.$
    Define 
    \[\gamma=\max_{i\in[n],\,M_{ii}>m}\frac{\log(u_i/a)}{\log(M_{ii}/m)}\quad \text{and}\quad w_{i}=\frac{u_i}{a}\left(\frac{m}{M_{ii}}\right)^\gamma,\forall i\in [n].\]
    By \cref{lem:realize_diag_rowsum_colsum,lem:thickening_interpolation,lem:thickening}, $w=a^{-1}m^\gamma \cdot u\circ(M_{ii}^{-\gamma})_{i\in [n]}$ is a realizable signature.
    
    We claim that $0< w_i\le 1$ for every $i\in [n].$
    It is clear that $w_i>0$ for every $i\in [n]$.
    If $M_{ii}=m$, then $u_i\le a$ by the definition of $a$, and thus $w_i\le 1.$
    If $M_{ii}>m$, then $\log(u_i/a)\le \gamma\log(M_{ii}/m)$ by the definition of $\gamma$, and thus $w_i\le 1$.

    Next we show there are two distinct indices $i,j\in [n]$ such that $w_i=w_j=1$ and $M_{ii}\neq M_{jj}.$
    Let $i\in [n]$ satisfies $M_{ii}=m$ and $u_i=a$, then $w_i=1.$
    Let $j\in [n]$ be the index attaining the maximum of $\gamma$, then $w_j=1$.
    Also, $M_{jj}>m=M_{ii}$.
    
    Let $C=\{k\in [n]\mid w_k=1\}$.
    Then $i,j\in C.$
    By \cref{lem: interpolation extract principle submatrix}, $\PlGH(M_C)\le_T \PlHolant(\mathcal{EQ}\mid M,w)\le_T \PlGH(M),$
    where $M_C$ is the principle submatrix of $M$ indexed by $C$. $M_C$ is also entrywise positive and positive definite, and has nonconstant diagonal.
    Since $M$ is minimal counterexample, $M_C=M$ and $C=[n].$
    So $w_i=1$ for every $i\in [n].$
    Therefore, $u_i=am^{-\gamma}M_{ii}^\gamma =cM_{ii}^\gamma$ for every $i\in [n]$, where $c=am^{-\gamma}>0.$
\end{proof}

The following lemma says that
the distinct diagonal entries of a minimal counterexample $M$ lie in a geometric progression, possibly with some terms missing.

\begin{lemma}\label{lem:geometric-levels}
For every minimal counterexample $M$, there exist $\kappa>0$, $\rho\neq 1$, an integer $D\geq1$, and numbers $(k_1,k_2,\ldots, k_n)\in \Z_{\ge0}^n$ such that
\begin{enumerate}
    \item $M_{ii}=\kappa \rho^{k_i}$ for every $i\in [n]$;
    \item $\min_{i\in [n]}k_i=0$, $\max_{i\in [n]}k_i=D$ and $\gcd\{k_i\mid i\in [n]\}=1$.
\end{enumerate}
\end{lemma}
\begin{proof}
For $i\in [n]$, let $d_i=M_{ii}$ and  $b_i=\log(d_i/d_1)$.
Choose a basis $\{t_1,\ldots, t_r\}$ for $\mathrm{span}_{\Q}\{b_1,\ldots, b_n\}$.
After clearing denominators, we can write $b_i= \sum_{j=1}^r c_{ij}t_j,\,c_{ij}\in \Z$ for every $i\in [n].$
Since $b_1=0$, $c_{1j}=0$ for all $j\in [r].$
Since $M$ has nonconstant diagonal, there exists $j\in [r]$ with $c_{ij}\neq 0$ for at least one \(i\in [n]\).
Fix such an index $j\in [r]$, and define $y_i=e^{c_{ij}} \,(i\in [n]).$
Then $y=(y_1,\ldots, y_n)$ is nonconstant.

Next we prove that $y$ is a realizable unary signature.
First note that $d=(d_1,\ldots, d_n)$ is realizable by \cref{lem:realize_diag_rowsum_colsum}.
We want to apply \cref{lem: realize unary by conformal lattice interpolation}, so we need to check $\mathfrak{L}(d_1,d_2,\ldots, d_n)\subseteq \mathfrak{L}(y_1,y_2,\ldots, y_n)$.
Suppose $(x_1,x_2,\ldots, x_n)\in \mathfrak{L}(d_1,d_2,\ldots, d_n)$.
We have 
\[
\begin{aligned}
    &d_1^{x_1}d_2^{x_2}\cdots d_n^{x_n}=1 \text{ and } x_1+x_2+\cdots x_n=0\\
    \implies &\left(\frac{d_2}{d_1}\right)^{x_2}\cdots \left(\frac{d_n}{d_1}\right)^{x_n}=1\\
    \implies & x_2b_2+\cdots +x_nb_n=0  && \text{by } b_i=\log\left(\frac{d_i}{d_1}\right)\\
    \implies & \sum_{k=1}^r\left(\sum_{i=2}^nx_ic_{ik}\right) t_k=0 && \text{by } b_i=\sum_{k=1}^r c_{ik} t_k\\
    \implies & \sum_{i=2}^n x_i c_{ik}=0 ,\, \forall k\in [r] && \text{by rational independence of } \{t_k\}_{k\in [r] }
\end{aligned}
\]
In particular, $\sum_{i=2}^n x_i c_{ij}=0$ for the fixed index $j\in [r]$.
Since $y_1=e^{c_{1j}}=1$,
we have $\prod_{i=1}^n y_i^{x_i}=\prod_{i=2}^n (e^{c_{ij}})^{x_i}=1$.
So $(x_1,\ldots,x_n)\in \mathfrak{L}(y_1,\ldots, y_n)$.
Thus by \cref{lem: realize unary by conformal lattice interpolation}, $y$ is realizable.

By \cref{lem:unary-profile}, there exist $c>0$ and $\gamma\in \R$ such that $e^{c_{ij}}=y_i=c d_i^\gamma$ for every $i\in [n].$
By $1=y_1=cd_1^\gamma$, we have $e^{c_{ij}}=y_i=(\frac{d_i}{d_1})^\gamma$ for every $i\in [n].$
This implies $d_i=d_1(e^{1/\gamma})^{c_{ij}}$, $\forall i\in [n].$
Let $c_*=\min_{i\in [n]}\{c_{ij}\}$ and absorb $(e^{1/\gamma})^{c_*}$ into the constant $d_1$, we may assume $c_{ij}\in \mathbb{Z}_{\ge 0}$ for every $i\in [n]$ without loss of generality.
Let $\kappa =d_1>0$, $g=\gcd\{c_{_{ij}}\mid i\in [n]\}\in \mathbb{Z}_{>0}$, $\rho=(e^{1/\gamma})^g\neq 1$ and $k_i=c_{ij}/g\in \mathbb{Z}_{\ge 0}$ for every $i\in [n]$, we have $M_{ii}=d_i=\kappa \rho^{k_i}$ and $\gcd\{k_i\mid i\in [n]\}=1.$
Letting $D=\max_{i\in [n]}k_i$ finishes the proof.
\end{proof}

\begin{definition}[Diagonal levels and their data]\label{def:diagonal-levels}
For a minimal counterexample $M$, fix parameters $\kappa,\rho,D$ and $(k_1,k_2,\ldots, k_n)$ satisfying \cref{lem:geometric-levels}. 
For every integer $a\in\{0,\ldots,D\}$, let
\(
n_a=|\{i\in[n]:k_i=a\}|.
\)
We call $(k_1,\ldots,k_n)$ diagonal levels and $(n_0,\ldots, n_D)$ diagonal level data.
%Thus $n_a>0$ exactly when $a\in\mathcal E$. 
%Define the parity matrix to be\(S=\operatorname{diag}\bigl((-1)^{k_1},(-1)^{k_2},\ldots, (-1)^{k_n}\bigr).\)
\end{definition}

In the remainder of this subsection, we will show that the diagonal level data $(n_a)_{0\le a\le D}$ have a near binomial form (\cref{lem:binomial-levels}).
This is proved by showing that the $i$-th row generating polynomial of $M$ has the form $U(z)^{D-k_i}V(z)^{k_i}$ where $U(z)$ and $V(z)$ are two nonproportional linear polynomials with positive coefficients (\cref{lem:level-polynomials}), then comparing the coefficients.
\cref{lem:level-row-sums} guarantees that the row generating polynomials are well-defined.
In parallel, \cref{lem:polynomial-chain} serves as a technical preparation for \cref{lem:level-polynomials}.

\begin{lemma}\label{lem:level-row-sums}
Let $M$ be a minimal counterexample of order $n$, and let $K\in\mathbb R_{>0}^{n\times n}$ be any realizable matrix. 
For every diagonal value $\delta\in\{M_{jj}\mid j\in[n]\}$ and every $i,i'\in[n]$ with $M_{ii}=M_{i'i'}$,
\[
\sum_{\substack{j\in[n]\\M_{jj}=\delta}}K_{ij}
=
\sum_{\substack{j\in[n]\\M_{jj}=\delta}}K_{i'j}.
\]
\end{lemma}
\begin{proof}
Let $\delta_1,\ldots,\delta_r$ be distinct diagonal values in $M$. 
Let $d=(d_1,\ldots,d_n)^{\tt T}=(M_{11},\ldots, M_{nn})^{\tt T}.$
By \cref{lem:realize_diag_rowsum_colsum,lem:thickening_interpolation}, $d^{\circ k}$ is a realizable unary signature for every $k\in \R$.
Connecting $d^{\circ k}$ and the second variable of $K$ using $=_2$, we realize the unary signature $Kd^{\circ k}$.
By \cref{lem:unary-profile}, for every $k\in[r]$, there exists $c_k>0$ and $\gamma_k\in \R$ such that $Kd^{\circ k}=c_k \cdot d^{\circ \gamma_k}.$
That is, $\sum_{j=1}^n K_{ij}d_j^k=c_kd_i^{\gamma_k}$, for every $i\in[n]$ and $k\in [r]$.
Consequently, whenever \(d_i=d_{i'}\),
\begin{equation}\label{eq:level-row-sums}
    \sum_{j=1}^n K_{ij}d_j^k=\sum_{j=1}^n K_{i'j}d_j^k\qquad (k\in[r]).
\end{equation}
For $\ell\in[r]$,
define
\(
A_\ell=
\sum_{j:d_j=\delta_\ell}K_{ij}
-
\sum_{j:d_j=\delta_\ell}K_{i'j}.
\)
Grouping \eqref{eq:level-row-sums} by diagonal value gives
\[
\sum_{\ell=1}^r\delta_\ell^k A_\ell=0
\qquad(k\in[r]).
\]
This is a full rank Vandermonde system in variables $(A_\ell)_{\ell\in [n]}$.
Therefore $A_\ell=0$ for every $\ell\in [r]$, proving the lemma.
\end{proof}

\begin{lemma}\label{lem:polynomial-chain}
Let $\mathcal{E} \subseteq \Z_{\ge 0}$ be a finite set (not a multiset)
of nonnegative
integers, with $D=\max\{e \mid e \in \mathcal{E}\} \ge 1$
and $\gcd\{e\mid e\in \mathcal{E}\}=1$.
%Let $(k_1,\ldots, k_n)\in \Z_{\ge 0}^n$ be diagonal levels of $M$ and $D=\max_{i\in [n]}k_i\ge 1$.
%Let $\mathcal{E}=\{k_1,\ldots,k_n\}.$ ($\mathcal{E}$ is a set, not a multiset.)
For every $e\in\mathcal{E}$, let $P_e(z)\in\mathbb R[z]$ have degree $D$ and positive coefficients.
Suppose $\{P_e(z)\mid z\in\mathcal{E}\}$ are not all proportional and
$P_e(z)^D=P_0(z)^{D-e}P_D(z)^e$ for every $e\in\mathcal E$.
Then there exist nonproportional linear polynomials $U(z)=u_0+u_1z$ and $V(z)=v_0+v_1z$ with $u_0,u_1,v_0,v_1>0$ such that
\[
P_e(z)=U(z)^{D-e}V(z)^e \quad(\forall z\in \R)
\qquad\text{for every }e\in\mathcal E.
\]
\end{lemma}

\begin{proof}
First, we claim that $P_0$ and $P_D$ are nonproportional.
Otherwise, $P_D(z)=\lambda P_0(z)$ for some $\lambda\in \R$ and every $z\in \R.$
By assumption, $P_0$ and $P_D$ have positive coefficients, so $\lambda>0.$
The equation $P_e(z)^D=P_0(z)^{D-e}P_D(z)^e$ gives $P_e(z)^D=\lambda^{e}P_0(z)^D$, for every $e\in \mathcal{E}.$
For every \(z>0\), both \(P_e(z)\) and \(P_0(z)\) are positive. 
Taking the positive \(D\)-th root gives
\(
P_e(z)=\lambda^{e/D}P_0(z)
\)
for every $e\in \mathcal{E}.$
This holds for every $z>0$, so it is a polynomial identity. 
This contradicts the assumption that $\{P_e(z)\mid e\in \mathcal{E}\}$ are not all proportional.

Let $\mathcal F$ be the finite set of monic nonconstant irreducible factors of $P_0P_D$ in $\mathbb R[z]$.
For every $f\in\mathcal F$ and $e\in\mathcal E$, let $\nu_e(f) \in \Z_{\ge 0}$ be the multiplicity of $f$ in $P_e$,
i.e., in the UFD $R[z]$, we have $f^{\nu_e(f)} \mid P_e$ but
$f^{\nu_e(f)+1} \nmid P_e$.
%(For example, if \(f^3\mid P_e\) but \(f^4\nmid P_e\), then \(\nu_e(f)=3\).)
Fix an arbitrary $f\in \mathcal{F}$,
take multiplicity on both sides of $P_e(z)^D=P_0(z)^{D-e}P_D(z)^e$, we obtain
\[
D\nu_e(f)=(D-e)\nu_0(f)+e\nu_D(f)\implies  D\mid e\bigl(\nu_D(f)-\nu_0(f)\bigr)
\quad(\forall e\in\mathcal E).
\]
By hypothesis,
%\cref{lem:geometric-levels}, 
$\gcd\{e\mid e\in \mathcal{E}\}=1$.
%B\'ezout's identity for the integers in $\mathcal E$
Thus an integer linear combination $\sum_{e \in \mathcal E} a_e e =1$. This
implies  $D\mid\nu_D(f)-\nu_0(f)$. 
By assumption, both $P_0$ and $P_D$ have degree $D$, so $0\le \nu_0(f),\nu_D(f)\le D.$
Combined with $D\mid\nu_D(f)-\nu_0(f)$, we have \({\nu_D(f)-\nu_0(f)\in\{-D,0,D\}.}\)
That is, an irreducible factor $f\in \mathcal{F}$ either has the same multiplicity in both $P_0$ and $P_D$, or occurs exactly \(D\) times in one and is absent from the other.

Since $P_0$ and $P_D$ are nonproportional, there exists $f\in \mathcal{F}$ with different multiplicities in $P_0$ and $P_D$.
Without loss of generality, we may assume $\nu_0(f)=D$ and $\nu_D(f)=0$; the reverse case is identical.
Now $f^D\mid P_0$, so $D\deg(f)\le \deg(P_0)=D$.
Since $f$ is nonconstant, $\deg(f)=1$, and $P_0=c_0f^D$ for some $c_0\in \R.$
Choose any monic nonconstant irreducible factor $g$ of $P_D$.
Since $\nu_D(f)=0$, $f\neq g$.
So $\nu_0(g)=0$ and $\nu_D(g)=D.$
The same degree argument gives $\deg(g)=1$ and $P_D=c_Dg^D$, for some $c_D\in \R.$
Write $f(z)=z+r$ and $g(z)=z+s$, we obtain $P_0(z)=c_0(z+r)^D$ and $P_D(z)=c_D(z+s)^D.$
Since $P_0$ and $P_D$ have positive coefficients, $c_0,c_D,r,s>0$.

Now the equation $P_e(z)^D=P_0(z)^{D-e}P_D(z)^e$ gives $P_e(z)^D=(c_0(z+r)^D)^{D-e}(c_D(z+s)^D)^{e}$ for every $e\in \mathcal{E}.$
Taking the $D$-th root gives
$P_e(z)=
%(c_0^{1/D}(z+r))^{D-e}(c_D^{1/D}(z+s))^e=
%%% JYC I think omit the middle expression makes it easier to read
U(z)^{D-e}V(z)^e$ for every $e\in \mathcal{E}$, where $U(z)=c_0^{1/D}(z+r)$ and $V(z)=c_D^{1/D}(z+s).$ 
\end{proof}

%%%%%%%here

\begin{lemma}\label{lem:level-polynomials}
Let $M$ be a minimal counterexample with diagonal levels $(k_1,\ldots,k_n)$ and $D=\max_{i\in[n]}k_i.$ 
For any realizable $K\in \Sym{n}{pd}{\Rp}$, there exist nonproportional linear polynomials $U(z)=u_0+u_1z$ and $V(z)=v_0+v_1z$  with $u_0,u_1,v_0,v_1>0$ such that
\begin{equation*}
\sum_{j=1}^n K_{ij}z^{k_j}
=U(z)^{D-k_i}V(z)^{k_i} \quad(\forall z\in \R),
\qquad \text{for every }i\in[n].
%\label{eq:level-polynomials}
\end{equation*}
\end{lemma}

\begin{proof}
    Let $d=(d_1,\ldots,d_n)^{\tt T}=(M_{11},\ldots, M_{nn})^{\tt T}.$
    For every $e\in\{k_1,\ldots,k_n\}$, define the polynomial 
    \[P_e(z):=\sum_{j=1}^n K_{ij}z^{k_j},\quad 
    %\text{for every } i\in [n] \text{ such that }k_i=e.\]
    \text{by choosing any } i\in [n] \text{ such that }k_i=e.\]
    By \cref{lem:level-row-sums}, $\sum_{j\in[n],k_j=t}K_{ij}=\sum_{j\in[n],\,k_j=t}K_{i'j}$ for every $i,i'\in[n]$ such that $k_i=k_i'$ and every $t\in \{k_1,\ldots, k_n\}$.
    So $P_e(z)$ is well-defined, i.e., we have defined $P_e(z)$
    for each $e$ in the \emph{set} $\mathcal{E}$ of distinct members
    of (the multiset) $\{k_1,\ldots,k_n\}$.
    In \cref{lem:level-row-sums}, we also proved that $u^{(\ell)}=Kd^{\circ \ell}$ is a realizable unary for every $\ell\in \Zp.$ 
    By \cref{lem:unary-profile}, for every $\ell\in \Zp$, there exist $c_\ell>0$ and $\gamma_\ell\in \R$ such that $u^{(\ell)}_i=c_\ell d_i^{\gamma_\ell}\,(i\in [n]).$
    By \cref{lem:geometric-levels}, there exist $\kappa>0$ and $\rho\neq 1$ such that $d_i=\kappa \rho^{k_i}  \,(i\in [n]).$
    Thus, for $i\in[n]$ with $k_i=e$, we have
    $$\kappa^{\ell}P_e(\rho^{\ell})=\sum_{j=1}^n K_{ij}\kappa^{\ell}\rho^{k_j\ell}=\sum_{j=1}^n K_{ij}d_j^\ell=u^{(\ell)}_i=c_\ell d_i^{\gamma_\ell}=c_{\ell}\kappa^{\gamma_\ell}\rho^{k_i\gamma_\ell}=c_{\ell}\kappa^{\gamma_\ell}\rho^{e\gamma_\ell}.$$
    This implies $P_e(\rho^{\ell})=c_{\ell}\kappa^{\gamma_\ell-\ell}\rho^{e\gamma_\ell}.$
    For any fixed $\ell\in \Zp$, the sequence
    $(P_e(\rho^{\ell}))_e$ lies in a geometric progression.
    %%% I don't say it is a geometric progression b/c e may not be consecutive
    Hence 
    \[
    P_e(\rho^{\ell})^D=P_0(\rho^{\ell})^{D-e}P_D(\rho^{\ell})^e\qquad (\forall \ell\in\Zp).
    \]
    Since $\rho\neq 1$, the above equation implies the polynomial identity \(P_e(z)^D=P_0(z)^{D-e}P_D(z)^e\,(\forall z\in \R).\)
    
    To apply \cref{lem:polynomial-chain}, we need to check that $\{P_e(z)\mid e\in \mathcal{E}\}$ are not all proportional.
    Suppose for a contradiction that $P_e(z)=\lambda_eQ(z)$ for some $Q(z)\in \R[z]$ and $\lambda_e\in \R  \,(\forall e\in\mathcal{E}).$
 %   It is clear that for every $e$, $P_e(z)$ is not the zero polynomial, otherwise $M$ has constant diagonal.
     Note that for every $i\in [n]$ with $k_i=e$, we have $P_e(1)=\sum_{j=1}^n K_{ij}=(K\mathbf{1})_i>0$, 
    in particular  $\lambda_e\neq 0$ for all $e$.
     It is also clear that for every $e$, $P_e(z)$ is a polynomial of degree at least 1, otherwise $M$ has constant diagonal.
 %   So $\lambda_e\neq 0$ for all $e.$
 %%% this should follow after saying $P_e(1) >0$
%
    Also, $P_e(\rho)=\sum_{j=1}^n  K_{ij}\rho^{k_j}=\frac{1}{\kappa}\sum_{j=1}^n K_{ij}d_j=\frac{1}{\kappa}(Kd)_i$.
    Then $P_e(\rho)/P_e(1)=Q(\rho)/Q(1)$ is independent of $e$.
    Therefore, $Kd=cK\mathbf{1}$ for some $c\in\R$.
    Since $K$ is invertible, $d=c\mathbf{1}$ is constant.
    Contradiction.

    Therefore, \cref{lem:polynomial-chain} provides two nonproportional linear polynomials $U(z)$ and $V(z)$ with positive coefficients such that $P_e(z)=U(z)^{D-e}V(z)^e$ for every $e\in \{k_1,\ldots, k_n\}$, proving the lemma.
\end{proof}

%%%%%%%% here

\begin{lemma}\label{lem:binomial-levels}
For every minimal counterexample $M$ with diagonal level data $(n_a)_{0\le a\le D}$, there exist $u_1,v_0>0$ such that
\begin{equation*}
n_a=n_0\binom Da\left(\frac{u_1}{v_0}\right)^a
\qquad\text{for every }a\in\{0,\ldots,D\}.
%\label{eq:binomial-levels}
\end{equation*}
In particular, $n_a>0$ for every $0\le a\le D.$
\end{lemma}
\begin{proof}
Let $(k_1,\ldots,k_n)$ be the diagonal levels of $M.$
Use \cref{lem:level-polynomials} with $K=M$, and write $U(z)=u_0+u_1z$ and $V(z)=v_0+v_1z$. 
Fix arbitrary $a\in \{0,\ldots, D\}.$
For every $i$ with $k_i=0$, \cref{lem:level-polynomials} gives 
\[
\sum_{j=1}^n M_{ij}z^{k_j}=(u_0+u_1z)^D.
\]
Extract the coefficient of $z^a$ on both sides, we have
\begin{equation}\label{eq:binomial-levels sum 1}
    \sum_{j\in [n],\,k_j=a}M_{ij}=\binom{D}{a}u_0^{D-a}u_1^a.
\end{equation}
Also, for every $j\in[n]$ with $k_j=a$, \cref{lem:level-polynomials} gives
\[
\sum_{i=1}^n M_{ji}z^{k_i}=(u_0+u_1z)^{D-a}(v_0+v_1z)^a.
\]
Extract the constant term, we have
\begin{equation}\label{eq:binomial-levels sum 2}
    \sum_{i\in[n],\,k_i=0}M_{ji}=u_0^{D-a}v_0^a.
\end{equation}
Sum \eqref{eq:binomial-levels sum 1} over the $n_0$ choices of $i$ such that $k_i=0$, we have 
\begin{equation*}
    \sum_{\substack{i,j\in[n]\\k_i=0, k_j=a}}M_{ij}=n_0\binom{D}{a}u_0^{D-a}u_1^a.
\end{equation*}
Sum \eqref{eq:binomial-levels sum 2} over the $n_a$ choices of $j$ such that $k_j=a$, we obtain
\begin{equation*}
    \sum_{\substack{i,j\in[n]\\k_i=0, k_j=a}}M_{ji}=n_a u_0^{D-a}v_0^a.
\end{equation*}
Symmetry of $M$ gives
\[
n_0\binom Da u_0^{D-a}u_1^a
=n_a u_0^{D-a}v_0^a,
\qquad
n_a=n_0\binom Da\left(\frac{u_1}{v_0}\right)^a.
\]
\end{proof}

\subsection{Binomial Diagonal Level Data}\label{subsec: binomial diag level data}

In this subsection, we will show that the diagonal level data have the exact binomial form (\cref{cor:binary-levels}).
By \cref{lem:binomial-levels}, it suffices to show $n_0=n_D=1$ (\cref{lem:singleton-endpoints}). 

\begin{lemma}\label{lem:parity-unary}
For every minimal counterexample $M$ with diagonal levels $(k_1,\ldots, k_n)$, the following unary signatures are all realizable: $u=((-1)^{k_1},\ldots, (-1)^{k_n})$, $v=(\mathbb{I}[k_1=0],\ldots, \mathbb{I}[k_n=0])$, $w=(\mathbb{I}[k_1=D],\ldots, \mathbb{I}[k_n=D]).$
\end{lemma}

\begin{proof}
    Let $d=(d_1,\ldots, d_n)=(M_{11},\ldots, M_{nn})$.
    Then $d$ is realizable by \cref{lem:realize_diag_rowsum_colsum}.
    By \cref{lem:geometric-levels}, there exist $\kappa>0$, $\rho\neq1$ such that  $d_i=\kappa\rho^{k_i} \,(i\in [n])$.
    For every $t\in \R$, define $y(t)=(t^{k_1},\ldots, t^{k_n})$.
    Here we take the convention that $0^0=1.$
    
    We claim that for every $t\neq 0$, $\mathfrak{L}(d_1,\ldots, d_n)\subseteq \mathfrak{L}(y(t)_1,\ldots, y(t)_n)$.
    Suppose $(x_1,\ldots, x_n)\in \mathfrak{L}(d_1,\ldots, d_n).$
    Then $\prod_{i=1}^n (\kappa \rho^{k_i})^{x_i}=1$ and $\sum_{i=1}^n x_i=0$.
    This implies $\rho^{\sum_{i=1}^n k_ix_i}=1.$
    Since $\rho\neq 1$, we have $\sum_{i=1}^n k_i x_i=0$.
    Hence $\prod_{i=1}^n(t^{k_i})^{x_i}=1$, $(x_1,\ldots, x_n)\in \mathfrak{L}(y(t)_1,\ldots, y(t)_n).$
    By \cref{lem: realize unary by conformal lattice interpolation}, $y(t)$ is realizable for every $t\neq 0.$
    In particular, $u=y(-1)$ is realizable.

    Next we show that $v=y(0)$ is realizable.
    Consider a planar graph $G$ as an instance of $\PlHolant(\mathcal{EQ}\mid M,v)$, where every vertex in $V$ is labeled by an {\sc Equality} signature, every edge is labeled by the binary signature $M$, and a multi-subset of vertices $S \subseteq V$ is attached to the unary signature $v=y(0)$.
    The partition function of the instance is
    \begin{equation}\label{eq:N1y0}
    \begin{aligned}
    N_1(v_1,v_2,\ldots, v_n)
    =&\sum_{\sigma:V\to [n]}\prod_{(u,v)\in E}M_{\sigma(u)\sigma(v)}\prod_{v\in S}y(0)_{\sigma(v)}\\
    =& \sum_{\substack{j_1,j_2,\ldots, j_n\ge 0\\j_1+j_2+\cdots+j_n=|S|}}y(0)_1^{j_1}y(0)_2^{j_2}\cdots y(0)_n^{j_n}\cdot z_{j_1,j_2,\ldots, j_n}.
\end{aligned}
\end{equation}
Here the coefficients $z_{j_1,j_2,\ldots, j_n}$ only depend on $G,S$ and $M$, but not depend on $v=y(0)$.
Fix any $t\neq 0$.
We replace every occurrence of $v$ by one copy of $y(t)$, and obtain an instance of $\PlHolant(\mathcal{EQ}\mid M,y(t))$ with partition function
\begin{equation}\label{eq:Nlyt}
    \begin{aligned}
    N_1(y(t)_1,y(t)_2,\ldots, y(t)_n)
    =&\sum_{\sigma:V\to [n]}\prod_{(u,v)\in E}M_{\sigma(u)\sigma(v)}\prod_{v\in S}y(t)_{\sigma(v)}\\
    =& \sum_{\substack{j_1,j_2,\ldots, j_n\ge 0\\j_1+j_2+\cdots+j_n=|S|}}t^{k_1 j_1}t^{k_2 j_2}\cdots t^{k_n j_n}\cdot z_{j_1,j_2,\ldots, j_n}.
\end{aligned}
\end{equation}
The coefficients $z_{j_1,j_2,\ldots, j_n}$ in \eqref{eq:Nlyt} are the same as in \eqref{eq:N1y0}.

Notice that $N_{1}(y(t)_1,\ldots ,y(t)_n)=P(t)$ is a polynomial in $t$ with degree at most $D|S|$, and $P(0)=N_1(v_1,v_2,\ldots,v_n)$ under the convention $0^0=1$.
Since $y(t)$ is realizable for every $t\in \R$
, we may query $P(1),P(2),\ldots,P(D|S|),P(D|S|+1)$ and recover $P(0)$ in polynomial time.
Therefore, $v=y(0)$ is realizable.
Similarly, we may define $\tilde{y}(t)=(t^{D-k_1},t^{D-k_2},\ldots, t^{D-k_n})$ and show that $w=\tilde{y}(0)$ is realizable.
\end{proof}

\begin{definition}[Connected components]\label{def:support-components}
For $H\in\Sym{n}{}{\R_{\ge 0}}$, its \emph{connected components} are the connected components of the undirected graph on $[n]$ with an edge $(i,j)$ whenever $i\ne j$ and $H_{ij}>0$.
We do not distinguish between the connected components in the underlying graph of $H$ and their corresponding matrix components in $H.$
\end{definition}

\begin{lemma}\label{lem:parity-quotient}
Let $M$ be a minimal counterexample with diagonal levels $(k_1,\ldots, k_n)$. 
Let $H\in \Sym{n}{pd}{\R_{\ge 0}}$ be realizable. 
Assume $H_{ij}=0$ for all $i,j\in[n]$ if $k_i,k_j$ have different parity. 
Then $H$ cannot have connected component of size at least $2$, i.e., $H$ is diagonal.
%If a connected component of $H$ has size at least $2$, then a positive symmetric positive-definite target of order less than $n$ with nonconstant diagonal reduces to $M$.
\end{lemma}

\begin{proof}
Since \(H\) is positive definite, every diagonal entry is positive.
We may perform a simultaneous permutation of rows and columns of $H$ and write $H=H_1\oplus H_2\oplus \cdots \oplus  H_{q}$, where $H_1,\ldots, H_q$ are nontrivial connected components of $H$.
Suppose the $i$-th connected component has size $c_i$, i.e., $H_i$ has order $c_i$, for every $i\in [q].$
We claim that $c_i=1$ for every $i\in[q]$, and thus $H$ is diagonal.
Otherwise, $q<n.$
We will show that there exists a matrix $Q\in \Sym{q}{pd}{\Rp}$ with nonconstant diagonal, such that $\PlGH(Q)\le_T \PlGH(M)$.
If so, by the definition of a minimal counterexample, $\PlGH(Q)$ would be \#P-hard.
This would imply $\PlGH(M)$ is \#P-hard.
But $M$ is defined such that $\PlGH(M)$ is not \#P-hard (see \cref{def:minimal-counterexample}),
a contradiction.
In the following we show the construction of $Q.$

Define $K=J_{c_1}\oplus J_{c_2}\oplus \cdots \oplus J_{c_q}.$
We show that $K$ is realizable.
First, $H^n$ is realizable by \cref{lem:stretching}.
Note that $H^n=H_1^{n}\oplus H_2^n \oplus \cdots \oplus H_q^n$, where every $H_i^n$ is entrywise positive.
This is because any two vertices in the same connected component are joined by a path of length at most $n-1$.
Then any path of length at most \(n-1\)
%and 
can be extended to a walk of length \(n\) using self-loops.
(Note that $H$ has positive diagonals.)
By \cref{lem: support interpolation}, we may replace every positive entry in $H^n$ by $1$ and obtain that $K=J_{c_1}\oplus J_{c_2}\oplus \cdots \oplus J_{c_q}$ is realizable.
%Let $r=(r_1,r_2,\ldots, r_n)$ with $r_i=\sum_{j=1}^n E_{ij}=c_{p_i}$ for every $i\in[n]$, where $p_i\in[q]$ is the index of the connected component containing $i.$
%By \cref{lem:realize_diag_rowsum_colsum}, $r$ is a realizable unary signature.
%By \cref{lem:thickening_interpolation}, its reciprocal $r^{-1}=(c^{-1}_{p_1},\ldots, c^{-1}_{p_n})$ is realizable.

Define the $n\times q$ matrix $R$ by
 \[
 R_{ij}=\begin{cases}
     1,& \text{if $i$ is in the $j$-th connected component;} \\
     0, & \text{otherwise.}
 \end{cases}
 \]
In other words, $R_{ij}=\mathbb{I}[{p_i}=j]$, where $p_i\in[q]$ is the index of the connected component containing $i$, for every $i\in[n], j\in [q].$
Let $d$ be the diagonal vector of $M.$
For every integer $\theta\in\{0,\ldots,2D\}$, set $M_\theta=\operatorname{diag}(d^{\circ\theta})M\operatorname{diag}(d^{\circ\theta})$ and $Q_\theta=R^{\tt T}M_\theta R\in \Sym{q}{}{\Rp}.$
Entrywise,
\begin{equation}\label{eq: Q theta jj}
    (Q_\theta)_{jj'}=\sum_{i\in [n]}\sum_{\ell\in[n]}R_{ij}R_{\ell j'}(M_\theta)_{i\ell}
    =\sum_{\substack{i\in [n]\\p_i=j}}\sum_{\substack{\ell\in[n]\\p_{\ell}=j'}}(M_\theta)_{i\ell}
    =\sum_{\substack{i\in [n]\\p_i=j}}\sum_{\substack{\ell\in[n]\\p_{\ell}=j'}}M_{i\ell}d_i^\theta d_\ell ^{\theta}
    \quad  \forall j,j'\in[q].
\end{equation}
We claim that $Q_\theta$ is positive definite for every $\theta\in \{0,\ldots, 2D\}.$
It is easy to check that $R$ has full column rank, so $Ry=0\implies y=0$ for every $y\in \R^q.$
For every nonzero $y\in\mathbb R^{q}$,
\(
y^{\tt T}Q_\theta y=(R y)^{\tt T}M_\theta(R y)>0.
\)
So $Q_\theta$ is positive definite, hence $Q_\theta\in \Sym{q}{pd}{\Rp}$, for every $\theta\in \{0,\ldots, 2D\}.$

Fix $\theta\in \{0,\ldots, 2D\}.$
Next we show that $Q_\theta$ is realizable.
Consider a planar graph $G=(V,E)$ as an instance of $\PlGH(Q_\theta).$
Without loss of generality, we may assume that $G$ is connected.
We construct the planar graph $G'$ as follows (see \cref{fig:incidence}).
Replace every vertex $v\in V$ of by a cycle $C(v)$ of length $\deg(v)$.
The original edges in $E$ connected to $v$ are now connected to exactly one vertex in the cycle.
Label $K$ on every edge in the cycle, and label \(M_\theta\) on every original edge in $E$.
%Attach the unary \(r^{-1}\) at every vertex.
\begin{figure}[htbp]
\centering
\begin{tikzpicture}[scale=.8,every node/.style={font=\small},
  vertex/.style={circle,draw,fill=white,inner sep=2pt}]
  \begin{scope}[xshift=-3.7cm]
    \node[vertex,label=below:$v$] (v) at (0,0) {};
    \foreach \a in {30,120,210,300}{\draw (v)--(\a:1.65);}
    \node at (0,-2.05) {};
  \end{scope}
  \draw[-{Stealth[length=2mm]}] (-1.25,0)--(-.35,0);
  \begin{scope}[xshift=1.8cm]
    \draw[dashed,gray] (0,0) circle (1.28);
    \foreach \a/\i in {30/1,120/2,210/3,300/4}{
      \node[vertex] (p\i) at (\a:.82) {};
      \draw (p\i)--(\a:1.65);
    }
    \draw[thick] (p1)--(p2)--(p3)--(p4)--(p1);
    \node at (0,0) {};
    \node at (0,-2.05) {};
  \end{scope}
\end{tikzpicture}
\caption{$G=(V,E)$ is an instance of $\PlGH(Q_\theta).$
Replace every $v\in V$ by a cycle $C(v)$ of length $\deg(v)$.
The cycle edges are labeled by the binary signature $K$; 
each original edge are labeled by $M_\theta$.
%Every vertex is attached with the unary $r^{-1}$.
We obtain $G'$ as an instance of $\PlHolant(\mathcal{EQ}\mid M_\theta,E)$.}
\label{fig:incidence}
\end{figure}

Note that $d^{\circ \theta}$ is realizable by \cref{lem:realize_diag_rowsum_colsum,lem:thickening}, and thus $M_\theta$ is realizable by \cref{lem:realize_diag_K_diag}.
We also proved $K$ is realizable.
Therefore, we have oracle access to the partition function of $G'$, which is
\begin{equation}\label{eq: ZMEG}
    Z_{M,K}(G'):=
\sum_{\sigma:V(G')\to [n]}\prod_{(u,v)\in E}(M_\theta)_{\sigma(u)\sigma(v)}\prod_{v\in V(G)}\prod_{(i,j)\in E(C(v))}K_{\sigma(i)\sigma(j)}.
\end{equation}
For every $\sigma:V(G')\to [n]$ contributing a nonzero term in $Z_{M,K}(G')$,
it maps all vertices in every fixed cycle $C(v)$ to the same connected component, since $K_{\sigma(i)\sigma(j)}=1$ iff $\sigma(i)$ and $\sigma(j)$ belong to the same connected component.
Therefore, $\sigma$ first chooses a component indexed by $\tau(v)\in[q]$ for every original vertex $v$, then maps every vertex in the cycle $C(v)$ to the component $H_{\tau(v)}.$
We may stratify the sum \eqref{eq: ZMEG} according to $(\tau(v))_{v\in V}$.
The contribution of all $\sigma$ which chooses the component $H_{\tau(v)}$ for $v\in V$ in \eqref{eq: ZMEG} is
\[
\sum_{\substack{\sigma:V(G')\to [n]\\ \sigma(C(v))\subseteq H_{\tau(v)},\forall v\in V}}\prod_{(u,v)\in E}(M_\theta)_{\sigma(u)\sigma(v)}=\prod_{(u,v)\in E}\sum_{\substack{\sigma(u)\in H_{\tau(u)} \\\sigma(v)\in H_{\tau(v)}}}(M_\theta)_{\sigma(u)\sigma(v)}
\]
By \eqref{eq: Q theta jj}, $$\sum_{\substack{\sigma(u)\in H_{\tau(u)} \\\sigma(v)\in H_{\tau(v)}}}(M_\theta)_{\sigma(u)\sigma(v)}=\sum_{\substack{\sigma(u)\in [n]\\ p_{\sigma(u)}=\tau(u)}}\sum_{\substack{\sigma(v)\in [n]\\ p_{\sigma(v)}=\tau(v)}}(M_\theta)_{\sigma(u)\sigma(v)}=(Q_\theta)_{\tau(u)\tau(v)}.$$
Thus,
\[
Z_{M,K}(G')=\sum_{\tau:V\to [q]}\prod_{(u,v)\in E}(Q_\theta)_{\tau(u)\tau(v)}=Z_{Q_\theta}(G).
\]
Therefore, $\PlGH(Q_\theta)\le_T \PlGH(M)$.

It remains to choose $\theta\in \{0,\ldots, 2D\}$ so that $Q_\theta$ has nonconstant diagonal. 
By \cref{lem:geometric-levels}, there exist $\kappa>0$ and $\rho\neq 1$ such that $d_i=\kappa \rho^{k_i}$.
By \eqref{eq: Q theta jj}, $$(Q_\theta)_{jj}=\kappa^{2\theta}\sum_{\substack{i\in [n]\\p_i=j}}\sum_{\substack{\ell\in[n]\\p_{\ell}=j}}M_{i\ell}\rho^{\theta(k_i+k_\ell)}.$$
For every $j\in[q]$,
define the polynomial $$p_j(z):=\sum_{\substack{i\in [n]\\p_i=j}}\sum_{\substack{\ell\in[n]\\p_{\ell}=j}}M_{i\ell}z^{k_i+k_\ell}.$$
Then $\deg(p_j)\le 2D$ and $(Q_\theta)_{jj}=\kappa^{2\theta}p_j(\rho^\theta)$ for every $j\in [q].$

By \cref{lem:binomial-levels}, $n_1>0$.
So there exists $i_1\in [n]$ such that $k_{i_1}=1.$
Suppose $i_1$ belongs to the $j_1$-th connected component in $H.$
We claim $p_{j_1}(0)=0$.
For every $i\in [n]$ such that $p_i=j_1$, $i$ belongs to the connected component $H_{j_1}.$
So there is a path $i_1=v_0v_1\ldots v_t=i$ in the underlying graph of $H$ connecting $i_1$ and $i$.
Along the path, $k_{i_1}=k_{v_0},k_{v_1},\ldots, k_{v_s}=k_{i}$ have the same parity, since $H_{uv}=0$ whenever $k_u$ and $k_v$ have different parity.
Thus, $k_i$ is odd for every $i\in [n]$ with $p_i=j_1$, as $k_{i_1}=1$ is odd.
Consequently, every term $M_{i\ell}z^{k_i+k_\ell}$ appearing in $p_{j_1}(z)$ has degree at least $2$.
This proves $p_{j_1}(0)=0.$

On the other hand, there exists $i_0\in [n]$ with $k_{i_0}=0.$
Suppose $i_0$ belongs to the $j_0$-th connected component of $H.$
Then $(i,\ell)=(i_0,i_0)$ contributes $M_{i_0i_0}>0$ to the constant term of $p_{j_0}(z)$, and thus $p_{j_0}(z)$ has a positive constant term.
Therefore, $p_{j_0}(z)-p_{j_1}(z)$ is a nonzero polynomial of degree at most $2D$.
It cannot vanish on all of $\{\rho^\theta\}_{0\le \theta\le 2D}.$
Choose some $\theta\in \{0,\ldots, 2D\}$ such that $p_{j_0}(\rho^{\theta})\neq p_{j_1}(\rho^{\theta})$, then $(Q_{\theta})_{j_0j_0}\neq (Q_{\theta})_{j_1j_1}$.

Let $Q=Q_\theta$.
We have proved that $Q\in\Sym{q}{pd}{\Rp}$ has nonconstant diagonal, and $\PlGH(Q)\le \PlGH(M)$.
This is a contradiction.
\end{proof}

\iffalse
\begin{definition}\label{def: parity matrix S}
    Let $M$ be a minimal counterexample with diagonal levels $(k_1,\ldots, k_n)$.
    Define $S=\mathrm{diag}((-1)^{k_1},\ldots,(-1)^{k_n}).$
\end{definition}
\fi
\begin{lemma}\label{lem:geometric-mean}
Let $M$ be a minimal counterexample with diagonal levels $(k_1,\ldots, k_n)$.
Let $S=\mathrm{diag}((-1)^{k_1},\ldots, (-1)^{k_n})$.
For every realizable $X\in\Sym{n}{pd}{\R}$, the matrix $G=X\# (SXS)$ is realizable, and is diagonal with positive diagonal entries.
\end{lemma}
\begin{proof}
Put $Y=SXS$ and $G=X\#Y$. 
Since $X$ is positive definite, $Y=SXS$ 
%and 
is also positive definite.
By \cref{def:geometric-mean}, $G=X^{1/2}\bigl(X^{-1/2}YX^{-1/2}\bigr)^{1/2}X^{1/2}$ is also positive definite.
The parity unary $((-1)^{k_1},\ldots, (-1)^{k_n})$ is realizable by \cref{lem:parity-unary}, so $Y$ is realizable by \cref{lem:realize_diag_K_diag}. 
Consequently, $G$ is realizable by \cref{lem:stretching,lem:spectral-powers}.

By \cref{lem:geometric-mean-equation}, $GX^{-1}G=Y$. 
So $GY^{-1}G=X$. 
Also $S^{-1}=S$ and so $Y^{-1} = SX^{-1}S $.
Therefore
\[
(SGS)X^{-1}(SGS)
=S G(SX^{-1}S)G S
=S GY^{-1}G S
=SXS=Y.
\]
The matrix $SGS$ is positive definite, so uniqueness in \cref{lem:geometric-mean-equation} gives $SGS=G$. Entrywise, 
\(
(-1)^{k_i+k_j}G_{ij}=G_{ij}\,(i,j\in[n]),
\)
and hence $G_{ij}=0$ if $k_i$ and $k_j$ have different parity.

The entrywise square $G\circ G$ is nonnegative and realizable by \cref{cor:thickening}. 
It is positive definite by \cref{thm:schur-product}. 
Also, $(G\circ G)_{ij}=0$ if $k_i$ and $k_j$ have different parity.
\Cref{lem:parity-quotient} then implies that $G \circ G$ is diagonal, and hence so is $G$.
The positive definiteness of $G$ gives $G_{ii}>0$ for every $i\in[n]$.
\end{proof}

\begin{definition}\label{def:normalized-family}
For a minimal counterexample $M$ with diagonal levels $(k_1,\ldots, k_n)$, define $S=\mathrm{diag}((-1)^{k_1},\ldots, (-1)^{k_n})$, $G=M\#(SMS)$, $B=G^{-1/2}MG^{-1/2}$ and $A(z)=B\operatorname{diag}(z^{k_1},z^{k_2},\ldots, z^{k_n}) B$ for every $z\in \R.$
\end{definition}

    By \cref{lem:geometric-mean}, $G$ is a realizable positive diagonal matrix. 
    So $G^{-\frac{1}{2}}$ is also realizable by \cref{lem:spectral-powers}.
    Consequently, $B\in \Sym{n}{pd}{\Rp}$ is realizable by \cref{lem:stretching}. 
    For every fixed $z>0$, we proved that $(z^{k_1},\ldots, z^{k_n})$ is realizable in \cref{lem:parity-unary}.
    Thus by \cref{lem:stretching}, $A(z)\in \Sym{n}{pd}{\Rp}$ is also realizable for every fixed $z>0$.

\begin{lemma}\label{lem:parity-normalization}
Let $M$ be a minimal counterexample with diagonal levels $(k_1,\ldots, k_n)$.
Let $S=\mathrm{diag}((-1)^{k_1},\ldots, (-1)^{k_n})$, and let $B$ be the matrix in \cref{def:normalized-family}.
Then $B^{-1}=SBS$ and $BSB=S.$
\end{lemma}

\begin{proof}
The geometric-mean equation \cref{lem:geometric-mean-equation} gives $GM^{-1}G=SMS$. Since the diagonal matrices $G^{-\frac{1}{2}}$ and $S$ commute,
\[
B^{-1}
=G^{1/2}M^{-1}G^{1/2}
=G^{-1/2}(GM^{-1}G)G^{-1/2}
=G^{-1/2}SMS G^{-1/2}
=SBS.
\]
Multiplication by $B$ on the left and $S= S^{-1}$ on the right gives $BSB=S$.
\end{proof}

\begin{lemma}\label{lem:reciprocal}
For every minimal counterexample $M$, let $A(z)$ be the family in \cref{def:normalized-family}. 
For every real $z>0$, there exists a positive diagonal matrix $K(z)$ such that
\(
A(z)=K(z)A(z^{-1})K(z).
\)
\end{lemma}
\begin{proof}
Fix $z>0$ and put $K(z)=A(z)\#(SA(z)S)$. 
Since $A(z)\in \Sym{n}{pd}{\Rp}$ is realizable, 
\cref{lem:geometric-mean} applies to $A(z)$.
So $K(z)$ is positive diagonal. \cref{lem:parity-normalization} gives
\begin{equation}\label{eq: reciprocal}
\begin{aligned}
A(z)^{-1}
&=B^{-1}\operatorname{diag}(z^{-k_1},z^{-k_2},\ldots,z^{-k_n})B^{-1}\\
&=SBS\operatorname{diag}(z^{-k_1},z^{-k_2},\ldots,z^{-k_n})SBS\\
&=SB\operatorname{diag}(z^{-k_1},z^{-k_2},\ldots,z^{-k_n})BS\\
&=SA(z^{-1})S.
\end{aligned}
\end{equation}
The geometric-mean equation \cref{lem:geometric-mean-equation} gives $K(z)A(z)^{-1}K(z)=SA(z)S$. Substitute \eqref{eq: reciprocal}, we have
\(K(z)SA(z^{-1})SK(z)=SA(z)S\).
%Multiply by $S$ on the left and right. 
Since $S$ commutes with the diagonal matrix $K(z)$ and $S^2=I$, the result is $K(z)A(z^{-1})K(z)=A(z)$.
\end{proof}

\begin{lemma}\label{lem:normalized-diagonal-polynomials}
For every minimal counterexample $M$, there exist nonproportional linear polynomials $U(z)=u_0+u_1z$, $V(z)=v_0+v_1z$ with $u_0,u_1,v_0,v_1>0$, 
such that the family $A(z)$ in \cref{def:normalized-family} has the expression
\[
A(z)_{ii}=U(z)^{D-k_i}V(z)^{k_i}
\qquad\text{for every }i\in[n].
\]
\end{lemma}
\begin{proof}
The matrix $B\circ B$ is realizable, entrywise positive, and symmetric. 
It is positive definite by \cref{thm:schur-product}. \cref{lem:level-polynomials} applied to this matrix gives
\[
A(z)_{ii}=\sum_{j=1}^nB_{ij}^2z^{k_j}
=U(z)^{D-k_i}V(z)^{k_i}\qquad\text{for every }i\in[n],
\]
as claimed.
\end{proof}

\begin{lemma}\label{lem:singleton-endpoints}
For every minimal counterexample $M$ with diagonal level data $(n_a)_{0\le a\le D}$, one has $n_0=n_D=1$.
\end{lemma}
\begin{proof}
Fix $b\in\{0,D\}$.
By \cref{lem:binomial-levels}, $n_b>0.$
Suppose for a contradiction that $n_b\ge 2$, then there are distinct $i,j\in [n]$ with $k_i=k_j=b.$ 
\cref{lem:normalized-diagonal-polynomials} gives a linear polynomial $L(z)=\ell_0+\ell_1z$ with $\ell_0,\ell_1>0$ such that $A(z)_{ii}=A(z)_{jj}=L(z)^D$. (Take $L(z)=U(z)$ if $b=0$ and $L(z)=V(z)$ if $b=D$.) 
Note that $A(-1)=BSB=S$
by \cref{def:normalized-family} and  \cref{lem:parity-normalization}.
So $A(-1)_{ij}=S_{ij}=0$ as $i\neq j.$
Evaluate $A(z)_{ii}=L(z)^D$ at $z=-1$, we have $L(-1)^D=A(-1)_{ii}=S_{ii}=(-1)^{b}\neq 0.$ 
So $L(-1)\neq 0.$
In the following we show the polynomial identity $A_{ij}(z)=cL(z)^D$ for some $c>0$, thus contradicting $A(-1)_{ij}=0$ and $L(-1)\neq 0.$

Put $p(z)=A_{ij}(z)=\sum_{t=1}^nB_{it}B_{tj}z^{k_t}$. 
Grouping the sum by its degree, we write $p(z)=\sum_{r=0}^D c_r z^r,$ where $c_r=\sum_{t\in[n], k_t=r}B_{it}B_{tj}.$
By \cref{lem:binomial-levels}, $c_r>0$ for every $0\le r\le D.$
So $\deg(p)=D.$
Let $K(z)$ be the positive diagonal matrix in \cref{lem:reciprocal}.
Examine the $i$-th and $j$-th diagonal entries of $A(z)=K(z)A(z^{-1})K(z)$, we obtain
\[
L(z)^D=A(z)_{ii}=K(z)_{ii}^2 A(z^{-1})_{ii}=K(z)_{ii}^2 L(z^{-1})^D
\]
and 
\[
L(z)^D=A(z)_{jj}=K(z)_{jj}^2 A(z^{-1})_{jj}=K(z)_{jj}^2 L(z^{-1})^D.
\]
Since $\ell_0,\ell_1>0$, $L(z)>0$ when $z>0.$
So $K_{ii}(z)^2=K_{jj}(z)^2$ when $z>0$.
Since $K(z)$ has positive diagonal for every $z>0$, we have $K(z)_{ii}=K(z)_{jj}\,(z>0).$
Examining the $(i,j)$ entry of $A(z)=K(z)A(z^{-1})K(z)$ yields
\[
p(z)=A(z)_{ij}=K(z)_{ii}K(z)_{jj}A(z^{-1})_{ij}=K(z)_{ii}^2A(z^{-1})_{ij}=\frac{L(z)^D}{L(z^{-1})^D}p(z^{-1})\quad (z>0).
\]
So
\begin{equation}\label{eq: reciprocal polynomial}
    p(z)L(z^{-1})^D=p(z^{-1})L(z)^D \qquad(z>0).
\end{equation}
Define the reciprocal polynomials $L^*(z)=zL(z^{-1})=\ell_1+\ell_0z$ and $p^*(z)=z^Dp(z^{-1})=\sum_{r=0}^D c_{D-r}z^r$. 
Multiplying \eqref{eq: reciprocal polynomial} by $z^D$ gives the polynomial identity
\begin{equation}
p(z)L^*(z)^D=p^*(z)L(z)^D.
\label{eq:reciprocal-polynomial}
\end{equation}
It is a polynomial identity because it holds for every $z>0$.

We claim that $L(z)$ and $L^*(z)$ are coprime.
The roots of $L(z)$ and $L^*(z)$ are $-\ell_0/\ell_1$ and $-\ell_1/\ell_0$, respectively.
Note that the two roots are distinct, otherwise $\ell_0=\ell_1$ as $\ell_0,\ell_1>0$, contradicting $L(-1)=\ell_0-\ell_1\neq 0.$
Thus \(L(z)\) and \(L^*(z)\) have distinct roots and are coprime.
Combined with \eqref{eq:reciprocal-polynomial}, $L(z)^D\mid p(z).$
Since $\deg (p)=D$, $p(z)=cL(z)^D\,(z\in \R)$, where $c=p(0)/L(0)^D=c_0/\ell_0^D>0.$
This contradicts that $p(-1)=0$ and $L(-1)\neq 0$.
Therefore, $n_0=1$ and $n_D=1.$
\end{proof}

\begin{corollary}\label{cor:binary-levels}
For every minimal counterexample $M$ with diagonal level data $(n_a)_{0\le a\le D}$,
\[
n_a=\binom Da\quad \text{for every }a\in\{0,\ldots,D\}, \quad \text{and }
n=2^D.
\]
\end{corollary}
\begin{proof}
    By \cref{lem:binomial-levels}, there are $u_1,v_0>0$ such that 
    \[
    n_a=n_0\binom Da\left(\frac{u_1}{v_0}\right)^a \qquad\text{for every }a\in\{0,\ldots,D\}
    \]
    Since $n_0=n_D=1$ by \cref{lem:singleton-endpoints}, $\left(\frac{u_1}{v_0}\right)^D=1$, so $\frac{u_1}{v_0}=1.$
    Hence $n_a=\binom{D}{a}$ for every $a\in \{0,\ldots,D\}.$
    Finally, $n=\sum_{a=0}^D n_a=\sum_{a=0}^D\binom{D}{a}=2^D.$
\end{proof}

\subsection{Hardness Proof}
\begin{definition}\label{def:boolean-power-data}
For a minimal counterexample $M$ with diagonal levels $(k_1,\ldots, k_n)$, let $x_0$ and $x_D$ be the unique index in $[n]$ such that $k_{x_0}=0$ and $k_{x_D}=D$, respectively. 
Define the matrix $W$ indexed by $\{0,1\}$:
\[
W=\begin{pmatrix}
M_{x_0x_0}^{1/D}&M_{x_0x_D}^{1/D}\\
M_{x_0x_D}^{1/D}&M_{x_Dx_D}^{1/D}
\end{pmatrix}.
\]
\end{definition}

\begin{lemma}\label{lem:boolean-power}
For any planar graph $G$, $Z_M(G)=Z_W(G)^D.$
\end{lemma}
\begin{proof}
Fix the planar graph $G=(V,E).$
%%% I added "subset" b/c we have been using multiset earlier. here it's not multiset
For any subset $S\subseteq V$, a \emph{compatible pair of pinning on $S$} are two maps $\alpha_S:S\mapsto \{x_0,x_D\}$ and $\beta_S:S\mapsto \{0,1\}$, such that 
\[
    \alpha_S(v)=\begin{cases}
        x_D &\text{if } \beta_S(v)=1;\\
        x_0 &\text{if } \beta_S(v)=0,
    \end{cases}
\]
for every $v\in S.$
For every $S\subseteq V$ and a compatible pair of pinning $\alpha_S,\beta_S$ on $S$, we define
\[
Z_{M}(G,S,\alpha_S):=\sum_{\substack{\sigma:V\to [n]\\\sigma|_S=\alpha_S}}\prod_{(u,v)\in E}M_{\sigma(u)\sigma(v)}\quad \text{and}\quad
Z_{W}(G,S,\beta_S):=\sum_{\substack{\tau:V\to \{0,1\}\\ \tau|_S=\beta_S}}\prod_{(u,v)\in E}W_{\tau(u)\tau(v)}.
\]
It is clear that $Z_{M}(G,\emptyset,\alpha_{\emptyset})=Z_M(G)$, and $Z_W(G,\emptyset,\beta_{\emptyset})=Z_{W}(G).$
We prove the stronger claim that for every $S\subseteq V$ and every compatible pair of pinning $\alpha_S,\beta_S$ on $S$, $Z_{M}(G,S,\alpha_S)=Z_{W}(G,S,\beta_S)^D.$ 

We prove the claim by induction on the number of unpinned vertices $|V|-|S|.$
If $|S|=|V|$, then every vertex in $V$ is pinned.
For every compatible pair of pinning $\alpha_V$ and $\beta_V$, we have
$$Z_{M}(G,V,\alpha_V)=\prod_{(u,v)\in E}M_{\alpha(u)\alpha(v)}=\prod_{(u,v)\in E}W_{\beta(u)\beta(v)}^D=Z_{W}(G,V,\beta_V)^D$$
by the definition of $W$ and compatible pair of pinning.

For the inductive step, we choose an unpinned vertex $v_0.$
For every $i\in [n]$, we define 
\[
f_i:=\sum_{\substack{\sigma:V\to [n]\\ \sigma|_S=\alpha_S,\,\sigma(v_0)=i}}\prod_{(u,v)\in E}M_{\sigma(u)\sigma(v)}.
\]
We claim that $f=(f_1,\ldots, f_n)$ is a realizable unary signature.
By \cref{lem:parity-unary}, the unary pinning $\delta_{x_0}=(\mathbb{I}[k_1=0],\ldots, \mathbb{I}[k_n=0])$ and $\delta_{x_D}=(\mathbb{I}[k_1=D],\ldots, \mathbb{I}[k_n=D])$ are realizable.
Since $x_0$ is the unique index in $[n]$ such that $k_{x_0}=0$, $\delta_{x_0}(u)$ equals $1$ when $u=x_0$ and $0$ otherwise.
Similarly, $\delta_{x_D}(u)$ equals $1$ when $u=x_D$ and $0$ otherwise.
Then $f$ is realized by the following planar gadget: We take
the given planar graph $G$, transform it so that the vertex $v_0$ is on
the external face and it is the single distinguished vertex for this unary gadget;
%with the planar embedding one the dangling vertex $v$, where 
every edge is labeled by $M$, and every vertex $v$ in $S$ is attached to the unary $\delta_{\alpha_S(v)}$. Note that for every $v \in S$, 
$\delta_{\alpha_S(v)}$ is either $\delta_{x_0}$ if $\beta_S(v)=0$, or 
 $\delta_{x_D}$ if $\beta_S(v)=1$.
 %%% for the following sentence, since we technically are not talking about
 % complexity, just equality of two partition functions, so this "constant"
 % is technically meaningless. so i rewrote it.
%Notice that $G$ is a fixed planar graph, so this is a constant size planar gadget.
Notice that $G$ is a fixed planar graph, so this is a planar gadget.
(This lemma is about an equality of partition functions,
not about complexity, and so $G$ is just a planar gadget.)

By \cref{lem:unary-profile}, there exist $c>0$ and $\gamma\in \R$ such that $f_i=cd_i^\gamma$ for every $i\in [n]$, where $d_i=M_{ii}.$
Since $d_i=\kappa\rho^{k_i}$ for some $\kappa>0$ and $\rho\neq 1$ by \cref{lem:geometric-levels}, we have $f_i=c\kappa^{\gamma}\rho^{\gamma k_i}$ for every $i\in [n].$
By $k_{x_0}=0$ and $k_{x_D}=D$, we have $f_{x_0}=c\kappa^\gamma$ and $f_{x_D}=c\kappa^\gamma \rho^{\gamma D}$.
Hence \(
f_i=f_{x_0}^{\,1-k_i/D}f_{x_D}^{\,k_i/D},
\)
for every $i \in [n]$.
%%% I added a few words here. I used j=0,1 below to differentiate with i \in [n]
Now for every $j\in \{0,1\}$,
we similarly define the partial partition function for $Z_W(G)$ by
\[
g_j:=\sum_{\substack{\tau:V\to \{0,1\}\\\tau|_S=\beta_S,\, \tau(v_0)=j}}\prod_{(u,v)\in E}W_{\tau(u)\tau(v)}.
\]
When $v_0$ is pinned, the induction hypothesis gives $f_{x_0}=g_0^D$ and $f_{x_D}=g_1^D.$
Thus, $f_i=g_0^{D-k_i}g_1^{k_i}$ for every $i\in [n].$
%%% I added the following line to justify the step from \sum_i=1}^n to \sum_{a=0}^D
Recall  
\cref{def:diagonal-levels} for the diagonal levels, that 
each $i \in[n]$ belongs to a unique level $k_i$
and $n_a=|\{i\in[n]\mid k_i=a\}|.$ Therefore,
\[
Z_{M}(G,S,\alpha_S)=\sum_{i=1}^nf_i=\sum_{a=0}^D n_ag_0^{D-a}g_1^a=\sum_{a=0}^D\binom{D}{a}g_0^{D-a}g_1^a=(g_0+g_1)^D=Z_{W}(G,S,\beta_S)^D,
\]
since $n_a=\binom{D}{a}$ by \cref{cor:binary-levels}.
This completes the induction.
In particular, when $S=\emptyset$, we have $Z_M(G)=Z_W(G)^D.$
\end{proof}

\begin{lemma}\label{lem:nonconstant}
If $M\in\Sym{n}{pd}{\Rp}$ and $M$ has non-constant diagonal, then $\PlGH(M)$ is \#P-hard.
\end{lemma}
\begin{proof}
    Suppose that a minimal counterexample $M\in \Sym{n}{pd}{\Rp}$ exists.
Let $(k_1,\ldots, k_n)$ be its diagonal levels, and let $x_0$ and $x_D$ be the unique index such that $k_{x_0}=0$ and $k_{x_D}=D$.
The principal minor satisfies
\(
M_{x_0x_0}M_{x_Dx_D}>M_{x_0x_D}^2.
\)
Taking positive $D$-th roots gives $W_{00}W_{11}>W_{01}^2$. Thus $W$ is positive definite. 
By \cref{lem:geometric-levels}, there exist $\kappa>0$ and $\rho\neq 1$ such that $M_{ii}=\kappa\rho^{k_i}$.
The diagonal entries of $W$ are $W_{00}=M_{x_0x_0}^{1/D}=\kappa^{1/D}$ and $W_{11}=M_{x_Dx_D}^{1/D}=\kappa^{1/D}\rho$, which differ. 
\cref{thm:boolean-source} makes $\PlGH(W)$ hard.
\cref{lem:boolean-power} gives $\PlGH(W)\le_T \PlGH(M)$, so $\PlGH(M)$ is \#P-hard. 
This contradicts the definition of a minimal counterexample \cref{def:minimal-counterexample}.
\end{proof}

We are ready to prove the nonconstant diagonal hardness theorem for entrywise positive and full rank matrices:
\nonconstantdiag*
\begin{proof}
    Consider $M^2\in \Sym{n}{pd}{\Rp}.$
    By \cref{lem:stretching}, $\PlGH(M^2)\le_T \PlGH(M).$
    If $M^2$ has nonconstant diagonal, then $\PlGH(M^2)$ is \#P-hard by \cref{lem:nonconstant}.
    It follows that $\PlGH(M)$ is \#P-hard.
    Otherwise, $M^2$ has constant diagonal.
    Let $d=(M_{11},M_{22},\ldots,M_{nn})$.
    Then $d$ is realizable unary by \cref{lem:realize_diag_rowsum_colsum}.
    Let $M'=\operatorname{diag}(d_1,\ldots, d_n)M^2\operatorname{diag}(d_1,\ldots, d_n).$
    Then $\PlGH(M')\le \PlGH(M)$ by \cref{lem:realize_diag_K_diag}.
    The matrix $M'$ has diagonal entries $M'_{ii}=d_i^2(M^2)_{ii}$, for $i\in[n].$
    Since $M$ has nonconstant diagonal and $M^2$ has constant diagonal, $M'$ has nonconstant diagonal.
    It is also clear that $M'\in \Sym{n}{pd}{\Rp}.$
    Thus, $\PlGH(M')$ is \#P-hard by \cref{lem:nonconstant}.
    It follows that $\PlGH(M)$ is \#P-hard.
\end{proof}

\section{Extension to Nonnegative Matrices}
\label{app: nnextension}
In this section, we extend the dichotomy for entrywise positive and full rank matrices (\cref{thm:main}) to nonnegative matrices (\cref{thm: nonnegative}).
More precisely, we will prove \cref{thm:nn-dichotomy}.

\begin{theorem}[Dichotomy for nonnegative matrices]
\label{thm:nn-dichotomy}
Let $M\in \Sym{n}{}{\R_{\geq0}}$. 
The problem $\PlGH(M)$ is \#P-hard unless after removing zero rows and columns, every connected component of $M$, after a
simultaneous permutation of rows and columns, has one of the forms:
\begin{enumerate}
    \item\label{item: tractable 1} $K\otimes ww^{\tt T},\,K \text{ is an Ising tensor and } w \text{ is an entrywise positive vector;}$
    \item\label{item: tractable 2} $\begin{pmatrix}O &K\otimes uv^{\tt T}\\K\otimes vu^{\tt T}&O \end{pmatrix},$
    $K$ is an Ising tensor and $u,v$ are entrywise positive vectors;
\end{enumerate}
in which case $\PlGH(M)$ is tractable.
By an Ising tensor we mean a matrix in the form \eqref{eq:ising-tensor}.
The Ising tensor $K$ in different connected components may differ.
An empty Ising tensor is allowed in both \cref{item: tractable 1,item: tractable 2}, which means a nonzero order one matrix.
\end{theorem}

%Recall from \cref{def:support-components} that the connected components of $M$ are defined as the connected components of its underlying graph.Throughout this section, we do not distinguish between a connected component of the underlying graph and the corresponding matrix component in $M$.

\subsection{Preliminaries}\label{subsec: appc preliminary}
The underlying graph of $M\in \Sym{n}{}{\R_{\ge 0}}$ is the undirected graph with vertex set $[n]$, where $i$ and $j$ are adjacent iff $M_{ij}>0.$
(If $i=j$ this is a self-loop.)
We say $M$ is bipartite (resp, nonbipartite) if its underlying graph is bipartite (resp, nonbipartite).
We say $M$ is connected if its underlying graph is connected.
We do not distinguish the connected components of the underlying graph of $M$ and their corresponding matrix components in $M.$

\begin{lemma}\label{lem: realize TA}
    Let $A\in \Sym{n}{}{\R_{\ge0}}$, with $0\le A_{ij}\le 1$ for all $i,j\in [n]$.
    Let $T_A\in \Sym{n}{}{\R_{\ge 0}}$ satisfy $(T_A)_{ij}=\mathbb{I}[A_{ij}=1],$ for all $i,j\in [n].$
    Then $\PlGH(T_A)\le_T \PlGH(A).$
\end{lemma}
\begin{proof}
    The proof is similar to that of \cref{lem: interpolation extract principle submatrix}.
    The case $A_{ij}<1$ for all $i,j\in [n]$ is trivial, as $T_A$ is the zero matrix in this case.
    In the following we assume that $T_A$ is not the zero matrix.
    Input a planar graph $G=(V,E)$,
    define 
    \[
    F_G(k)=\sum_{\sigma:V\to [n]}\prod_{(u,v)\in E}A_{\sigma(u)\sigma(v)}^k.
    \]
    By thickening (\cref{cor:thickening}), $A^{\circ k}$ is realizable for every $k\in \Zp.$
    Oracle access to $\PlGH(A)$ gives $F_G(k)$ for any $k\in \mathbb{Z}_{>0}$ in polynomial time (in terms of the size of $G$). 
    Define $X$ to be the set of all possible product values $\prod_{(u,v)\in E}A_{\sigma(u)\sigma(v)}$ appearing in $F_G(k)$:
    $$
    X=\left\{\prod_{i,j=1}^n A_{ij}^{t_{ij}}\mid t_{ij}\in \mathbb{Z}_{\ge 0}\text{ and } \sum_{i,j=1}^n t_{ij}=|E|\right\}. 
    $$
    There are $\binom{|E|+n^2-1}{n^2-1}=|E|^{O(1)}$ many integer sequences $(t_{ij})_{i,j\in[n]}$ such that $t_{ij}\in \mathbb Z_{\ge 0}$ and $\sum_{i,j=1}^n t_{ij}=|E|.$
    Note that $n$ is a constant.
    So $X$ has polynomial size in $G$, and elements in $X$ are enumerable in polynomial time.
    In particular, $1\in X$ since $T_A$ is not the zero matrix.
    Define $A_x$ to be the coefficient of $x\in X$ appearing in $F_G(k):$
    \[
    A_x=\left|\left\{\sigma:V\to [n]\mid \prod_{(u,v)\in E}A_{\sigma(u)\sigma(v)}=x\right\}\right|\quad (\forall x\in X).
    \]
    Then $F_G(k)=\sum_{x\in X}A_x x^{k}, k\in \Zp.$
    Query $F_G(1),F_G(2),\ldots, F_G(|X|)$, we obtain a full rank Vandermonde system in variables $(A_x)_{x\in X}.$
    Solving the Vandermonde system, we recover $A_1$ in polynomial time.
    Note that $\prod_{(u,v)\in E}A_{\sigma(u)\sigma(v)}=1$ $\iff$ $A_{\sigma(u)\sigma(v)}=1$ for every $(u,v)\in E$ $\iff$ $(T_A)_{\sigma(u)\sigma(v)}=1$ for every $(u,v)\in E$ $\iff$ $\prod_{(u,v)\in E}(T_A)_{\sigma(u)\sigma(v)}=1$.  
    Thus,
    \[
    A_1=\sum_{\sigma: V\to [n]}\prod_{(u,v)\in E}(T_A)_{\sigma(u)\sigma(v)}=Z_{T_A}(G).
    \]
    This proves $\PlGH(T_A)\le_T \PlGH(A).$
\end{proof}

Also, \cref{lem:thickening_interpolation} can be generalized to the nonnegative setting:

\begin{lemma}\label{lem:nn-thickening_interpolation}
    Let $k\in\{1,2\}$.  
    Let $f$ be a real  signature of arity $k$ with nonnegative values on domain $[n].$
    Then for arbitrary signature set $\mathcal{G},$
    $\PlHolant(\mathcal{EQ}\mid f^t,\mathcal{G})\le \PlHolant(\mathcal{EQ}\mid f,\mathcal{G})$ for every $t\in \R$. Here $f^t$ is the arity $k$ signature satisfying $f^t({\bf x})=(f({\bf x}))^t$ if $f(\mathbf{x})>0$, and $f^t({\bf x})=0$ if $f(\mathbf{x})=0$, $\forall \mathbf{x}\in[n]^k.$
\end{lemma}

Its proof is similar to that of \cref{lem: support interpolation}, making use of the conformal lattice interpolation (\cref{lm:conformal lattice interpolation}).
For brevity, we omit the details here.

The following lemma is Lemma 4.6 in~\cite{cai2013graph} for general
$\GH$, but is also valid for
$\PlGH$; 
the proof uses Lemma 4.1 in~\cite{cai2013graph}, called the first pinning lemma, 
which uses a planar construction.

\begin{lemma}[{\cite[Lemma 4.6]{cai2013graph}}]
\label{lem:nn-component}
Let $m\in \Zp$ and $N \in \mathrm{Sym}_m(\mathbb{R}_{\ge 0})$.
Suppose $N_C$ is an arbitrary connected component of $N$.
Then $\PlGH(N_C)\le_T \PlGH(N).$
\end{lemma}

In the following, we fix a nonzero matrix $M\in \Sym{n}{}{\R_{\ge 0}}$.
Removing 0 rows and columns clearly does not change the complexity
of $\PlGH(M)$.
%%%% by M being symmetric, they are  the same indices. just isolated vertices w/o loop
So we assume $M$ has no 0 rows and columns. We also assume that $M$ is connected.
Assume the problem $\PlGH(M)$ is not \#P-hard, 
we show $M$ has the form in either \cref{item: tractable 1} or \cref{item: tractable 2}.
All realizability is relative to this $M$ unless otherwise stated.

\subsection{Remove Proportional Rows}
\label{subsec:nn-quotient}

\begin{definition}[Proportional relation]
    %Let $0\neq M\in\Sym{n}{}{\R_{\ge 0}}$ be connected.
    Define the relation ``$\sim$" on $[n]$ by: $i\sim j\iff $ there exists $\lambda\neq 0$ such that $M_{ik}=\lambda M_{jk}$ for every $k\in [n].$
    Clearly, this is an equivalence relation.
    Let $P_1,\ldots, P_q$ denote its equivalence classes.
    We call them \emph{proportional classes}.
    Let $n_a:=|P_a|$ for every $a\in [q].$
    {\rm (}$M$ has no 0 rows, $\{P_1,\ldots, P_q\}$ is a partition of $[n]$,
    $\sum_{a\in [q]} n_a = n$.{\rm )}
\end{definition}

\begin{definition}[Compressed matrix]\label{def: compressed matrix}
    Define $s_i:=\sqrt{(M^2)_{ii}}$ for every $i\in [n]$, let $s=(s_1,\ldots, s_n).$
    Define the matrix $C\in \Sym{q}{}{\R_{\ge 0}}$, with entries 
    \[
    C_{ab}:=\frac{M_{ij}}{s_i s_j} \qquad (\forall a,b\in [q]),
    \]
    %%% with M having >=0, but no 0 rows/col we still have s_i >0. but C can have
    % 0 entries still. 
    where $i\in P_a$ and $j\in P_b.$
    We call $C$ the \emph{compressed matrix} of $M.$
\end{definition}

Since $M\in \Sym{n}{}{\R_{\ge 0}}$ has no zero rows, $s_i=\sqrt{\sum_{k\in [n]}M_{ik}^2}>0$, for every $i\in [n].$
It is easy to verify that the definition of $C_{ab}$ is independent of the choices of representatives $i\in P_a$ and $j\in P_b$.
Thus, $C$ is well-defined.
Also, $C$ has pairwise nonproportional rows.

\begin{lemma}\label{lem: realize s}
    The unary $s\in \R^n_{>0}$ is realizable.
\end{lemma}
\begin{proof}
    By \cref{lem:stretching}, $M^2$ is realizable.
    By \cref{lem:realize_diag_rowsum_colsum}, $((M^2)_{11},\ldots, (M^2)_{nn})\in \Rp^{n}$ is realizable.
    By \cref{lem:thickening_interpolation}, $s$ is realizable.
\end{proof}

\begin{lemma}\label{lem: realize compressed matrix C}
    $\PlGH(C)\le_T \PlGH(M)$, where $C$ is the compressed matrix of $M.$
\end{lemma}
\begin{proof}
    Consider the matrix $A=\diag(s_1^{-1},\ldots, s_n^{-1})M^2 \diag (s_1^{-1},\ldots,s_n^{-1}).$
    By \cref{lem: realize s}, $s$ is realizable.
    So $s^{-1}=(s_1^{-1},\ldots, s_{n}^{-1})$ is realizable by \cref{lem:thickening_interpolation}.
    It follows that $A$ is realizable from \cref{lem:realize_diag_K_diag}.
    Since $M$ is symmetric, we can write $A$ entrywise as
    \[
    A_{ij}=\frac{\sum_{k=1}^n M_{ik}M_{kj}}{\sqrt{\sum_{k=1}^nM_{ik}^2}\sqrt{\sum_{k=1}^nM_{jk}^2}}=\frac{\braket{M_{i,*}\,,M_{j,*}}}{\|M_{i,*} \|\cdot \|M_{j,*}\|}
    \qquad (\forall i,j\in[n])
    \]
    where $M_{i,*}$ is the $i$-th row of $M.$
    By Cauchy-Schwarz inequality, $0\le A_{ij}\le 1$, and $A_{ij}=1$ iff $M_{i,*}$ and $M_{j,*}$ are proportional, i.e., $i$ and $j$ belong to the same proportional class.
    Let $T_A\in \Sym{n}{}{\R_{\ge0}}$ be the matrix with entries
    \[
    (T_A)_{ij}=\begin{cases}
        1, &\text{if }i\sim j,\\
        0, &\text{otherwise}
    \end{cases}
    \qquad(\forall i,j\in [n]).
    \]
    Then $(T_A)_{ij}=\mathbb{I}[A_{ij}=1].$
    By \cref{lem: realize TA}, $T_A$ is realizable.
    By \cref{lem:realize_diag_rowsum_colsum}, the row sum of $T_A$, denoted by $r=(r_1,\ldots, r_n)$, is realizable.
    Note that $r_i=n_a$, whenever $i\in P_a$.
    Consequently, the matrix $R_A=\diag(r_1^{-1},\ldots, r_n^{-1})T_A$ is realizable by \cref{lem:realize_diag_rowsum_colsum}.
    We have 
    \[
    (R_A)_{ij}=\begin{cases}
        \frac{1}{n_a},& \text{if }i,j\in P_a,\\
        0,&\text{otherwise} 
    \end{cases}
    \qquad(\forall i,j\in [n]).
    \]

    Now we are ready to prove the reduction $\PlGH(C)\le_T\PlGH(M).$
    Consider a planar graph $G=(V,E)$ as an instance of $\PlGH(C).$
    Without loss of generality, we may assume that $G$ is connected.
We construct the planar graph $G'$ as follows (see \cref{fig:incidence2}).
Replace every vertex $v\in V$ of  $G$ by a cycle $C(v)$ of length $\deg(v)$.
The original edges in $E$ connected to $v$ are now connected to exactly one vertex in the cycle.
Label $R_A$ on every edge in the cycle, and label \(M\) on every original edge in $E$.
Finally, attach the unary \(s^{-1}=(s_1^{-1},\ldots,s_{n}^{-1})\) at every vertex in $G'.$
We obtain an instance of $\PlHolant(\mathcal{EQ}\mid M,R_A,s).$

\begin{figure}[htbp]
\centering
\begin{tikzpicture}[scale=.8,every node/.style={font=\small},
  vertex/.style={circle,draw,fill=white,inner sep=2pt}]
  \begin{scope}[xshift=-3.7cm]
    \node[vertex,label=below:$v$] (v) at (0,0) {};
    \foreach \a in {30,120,210,300}{\draw (v)--(\a:1.65);}
    \node at (0,-2.05) {};
  \end{scope}
  \draw[-{Stealth[length=2mm]}] (-1.25,0)--(-.35,0);
  \begin{scope}[xshift=1.8cm]
    \draw[dashed,gray] (0,0) circle (1.28);
    \foreach \a/\i in {30/1,120/2,210/3,300/4}{
      \node[vertex] (p\i) at (\a:.82) {};
      \draw (p\i)--(\a:1.65);
    }
    \draw[thick] (p1)--(p2)--(p3)--(p4)--(p1);
    \node at (0,0) {};
    \node at (0,-2.05) {};
  \end{scope}
\end{tikzpicture}
\caption{$G=(V,E)$ is an instance of $\PlGH(C).$
Replace every $v\in V$ by a cycle $C(v)$ of length $\deg(v)$.
The cycle edges are labeled by the binary signature $R_A$; 
each original edge are labeled by $M$.
Every vertex is attached with the unary $s^{-1}$.
We obtain $G'$ as an instance of $\PlHolant(\mathcal{EQ}\mid M,R_A,s)$.}
\label{fig:incidence2}
\end{figure}

Since $R_A$ and $s$ are both realizable, we can query
\begin{equation}\label{eq: ZMRAsG}
    Z_{M,R_A,s}(G'):=
\sum_{\sigma:V(G')\to [n]}\prod_{(u,v)\in E}M_{\sigma(u)\sigma(v)}\prod_{v\in V(G)}\prod_{(i,j)\in E(C(v))}(R_A)_{\sigma(i)\sigma(j)}\prod_{u\in V(G')}s^{-1}_{\sigma(u)}.
\end{equation}

For every $\sigma:V(G')\to [n]$ contributing a nonzero term in \eqref{eq: ZMRAsG},
it maps all vertices in every fixed cycle $C(v)$ to the same proportional class, since $(R_A)_{\sigma(i)\sigma(j)}$ is nonzero iff $\sigma(i)\sim \sigma(j)$.
Therefore, $\sigma$ first chooses a proportional class indexed by $\tau(v)\in[q]$ for every original vertex $v\in V$, then maps every vertex in the cycle $C(v)$ to the proportional class $P_{\tau(v)}.$
We may stratify the sum \eqref{eq: ZMRAsG} according to $(\tau(v))_{v\in V}$.
The contribution of all $\sigma$ which chooses the proportional class $P_{\tau(v)}$ for $v\in V$ in \eqref{eq: ZMRAsG} is

\begin{equation}\label{eq: stratify by tau}
    \sum_{\substack{\sigma:V(G')\to [n]\\ \sigma(C(v))\subseteq P_{\tau(v)},\forall v\in V}}\prod_{(u,v)\in E}M_{\sigma(u)\sigma(v)}\prod_{v\in V}\left(\frac{1}{n_{\tau(v)}}\right)^{\deg(v)}\prod_{u\in V(G')}s^{-1}_{\sigma(u)}.
\end{equation}
Plugging in $M_{\sigma(u)\sigma(v)}=C_{\tau(u)\tau(v)}s_{\sigma(u)}s_{\sigma(v)}$, \eqref{eq: stratify by tau} becomes 
\begin{equation}\label{eq: stratify by tau 2}
    \sum_{\substack{\sigma:V(G')\to [n]\\ \sigma(C(v))\subseteq P_{\tau(v)},\forall v\in V}}\prod_{(u,v)\in E}C_{\tau(u)\tau(v)}\prod_{v\in V}\left(\frac{1}{n_{\tau(v)}}\right)^{\deg(v)}.
\end{equation}
For each cycle $C(v)$ of length $\deg(v)$ being mapped to $P_{\tau(v)}$, there are $|P_{\tau(v)}|^{\deg(v)}=n_{\tau(v)}^{\deg(v)}$ many choices.
Then \eqref{eq: stratify by tau 2} becomes
\[
\prod_{(u,v)\in E}C_{\tau(u)\tau(v)}\prod_{v\in V}\left(\frac{1}{n_{\tau(v)}}\right)^{\deg(v)}\prod_{v\in V}n_{\tau(v)}^{\deg(v)}=\prod_{(u,v)\in E}C_{\tau(u)\tau(v)}.
\]
Summing over all possible $(\tau(v))_{v\in V}\in [q]^{V}$ gives $Z_{M,R_A,s}(G')=Z_C(G).$
Thus, $\PlGH(C)\le_T\PlHolant(\mathcal{EQ}\mid M,R_A,s)\le_T\PlGH(M).$
\end{proof}

The following lemma generalizes \cref{lem: realize compressed matrix C}.
Note that when $f_i=s_i^{-1}$, it gives \cref{lem: realize compressed matrix C}.
\begin{lemma}
\label{lem:nn-quotient}
Let $C\in \Sym{q}{}{\R_{\ge 0}}$ be the compressed matrix of $M.$
For every realizable unary $f\in \Rp^n$, put
\begin{equation}
 u(f)_a=\frac1{n_a}\sum_{i\in P_a}s_if_i,
 \qquad (\forall a\in [q])
 \label{eq:nn-weighted-quotient}
\end{equation}
and let $C_f=\diag(u(f))C\diag(u(f)).$
Then $\PlGH(C_f)\leq_T\PlGH(M)$. 
\end{lemma}

\begin{proof}
    The proof is similar to that of \cref{lem: realize compressed matrix C}.
    Instead of attaching $s^{-1}$ to every vertex of $G'$, we attach the unary $f.$
    One can check $Z_{M,R_A,f}(G')=Z_{C_f}(G)$, and then the lemma follows.
\end{proof}

\begin{lemma}\label{lem: nonproportional rows matrix powers}
    Let $m\in \Zp$ and $D\in \Sym{m}{}{\R_{\ge 0}}.$
    If $D$ has pairwise nonproportional rows, then so does $D^k$ for every $k\in \Zp.$
\end{lemma}
\begin{proof}
    Fix arbitary $k\in\Zp.$
    Suppose for a contradiction that $D^k_{i,*}=\lambda D^k_{j,*}$ for two distinct rows $i,j\in [m]$ and some $\lambda\in \R_{\neq 0}.$
    Let $v=\mathbf{e}_i-\lambda \mathbf{e}_j\in \R^{m}$, where $\mathbf{e}_i$ and $\mathbf{e}_j$ are unit vectors at $i\in [m]$ and $j\in [m]$, respectively.
    Then $D^kv=0$, as $D$ is symmetric. Hence $v\in \ker (D^k).$
    We note that $\ker(D^k)=\ker(D)$.
    This follows from the real spectral theorem (\cref{thm:spectral-theorem}): raising $D$ to the $k$-th power preserves its $0$-eigenspace.
    Thus, $v\in \ker(D).$
    This would imply that the $i$-th row and $j$-th row of $D$ are proportional, a contradiction.
\end{proof}

\subsection{Nonbipartite Case}\label{subsec: nonbipartite}
In this subsection, we fix the nonzero and connected matrix $M\in \Sym{n}{}{\R_{\ge 0}}$, and assume that $M$ is nonbipartite.
We will show that $M$ has the form \cref{item: tractable 1}, otherwise $\PlGH(M)$ is \#P-hard.

\begin{definition}[Cube space]\label{def: cube space}
    Let $d\in \Zp$.
    For any bijection $\phi:[2^d]\to  \{1,-1\}^d$ and $i\in [d]$, define the matrix $S(\phi ,i)\in \R^{2^d\times 2^d}$ by $S(\phi,i)_{xy}:=\phi(x)_i\phi(y)_i,$ and the vector space
\begin{equation*}
    \mathcal V_{\phi}:=\left\{cJ_{2^d}+\sum_{i=1}^d\gamma_i S({\phi,i})\mid c,\gamma_1,\ldots,\gamma_d\in\R\right\}.
\end{equation*}
We call $\mathcal{V}_{\phi}$ the \emph{cube space} associated with $\phi.$
\end{definition}

\begin{lemma}\label{lem: ising tensor and vector space}
    Let $m=2^d$ and $K\in \Sym{m}{F}{\R_{\ge 0}}$.
    Then $K$ is an Ising tensor after a simultaneous permutation of rows and columns iff there exists a bijection $\phi:[m]\to \{1,-1\}^d$ such that $\log_\circ K\in \mathcal{V}_{\phi}.$
\end{lemma}
\begin{proof}
If $K$ is an Ising tensor, then it can be written as
\[
K=e^c\bigotimes_{i=1}^{d}
\begin{pmatrix}e^{\gamma_i}& e^{-\gamma_i}\\e^{-\gamma_i}&e^{\gamma_i}\end{pmatrix},
\quad c\in \R,\quad \gamma_i\in \R_{\neq 0},
\]
after a simultaneous permutation of rows and columns.
Let $\phi:[m]\to \{1,-1\}^d,\,x\mapsto \phi(x)$ denote this permutation.
Then 
\[
K_{xy}=e^c\prod_{i=1}^d\exp(\phi(x)_i\phi(y)_i\gamma_i)\implies \log K_{xy}=c+\sum_{i=1}^d\phi(x)_i\phi(y)_i \gamma_i\qquad(\forall x,y\in [m]).
\]
Consequently, $\log_\circ K=cJ+\sum_{i=1}^d \gamma_iS(\phi,i)\in\mathcal{V}_\phi.$
The converse also holds by reversing the argument.
\end{proof}

\begin{lemma}\label{lem: cube basis}
    For $d\in \Zp$ and a bijection  $\phi:[2^d]\to \{1,-1\}^d$, $\{J_{2^d},S(\phi,1),\ldots, S(\phi,d)\}$ is a basis for $\mathcal{V}_{\phi}.$
\end{lemma}
\begin{proof}
    Suppose there are real numbers $c,\gamma_1,\ldots,\gamma_d$ such that $cJ_{2^d}+\sum_{i=1}^d\gamma_i S(\phi,i)=O_{2^d}.$
    By the proof of \cref{lem: ising tensor and vector space}, this implies
    $$J_{2^d}=\exp_\circ(O_{2^d})=e^c\prod_{i=1}^d\begin{pmatrix}e^{\gamma_i}& e^{-\gamma_i}\\e^{-\gamma_i}&e^{\gamma_i}\end{pmatrix}$$
    after a simultaneous permutation of rows and columns.
    Therefore, $c=\gamma_1=\cdots=\gamma_d=0.$
\end{proof}

\begin{lemma}\label{lem: uniform tensor}
    Let $m\in \Zp$ and $\{R(t)\}_{t\in \R}\subseteq \Sym{m}{}{\Rp}$ be an entrywise analytic family of matrices such that $\PlGH(R(t))$ is not \#P-hard for every $t\in \R.$
    If there exists some $t_0 \in \mathbb{R}$ such that $R(t_0)\in \Sym{m}{F}{\Rp}$, then there is a bijection $\phi \colon [m] \to \{-1,1\}^d$ such that $\log_\circ R(t) \in \mathcal{V}_\phi$ for every $t \in \mathbb{R}$.
\end{lemma}
\begin{proof}
    Let $f(t)=\det(R(t)),\,\forall t\in \R.$
    This is an analytic function in $t$, and $f(t_0)\neq 0.$
    Hence there is a neighborhood $U_1$ of $t_0$ so that $R(t)$ is invertible for every $t\in U_1.$
    Since $R(t_0)$ is entrywise positive, there is a neighborhood $U_2$ of $t_0$ so that $R(t)$ is entrywise positive for every $t\in U_2.$
    Let $U=U_1\cap U_2.$
    For every $t\in U_1\cap U_2$, $R(t)\in \Sym{n}{F}{\Rp}$, so \cref{thm:main} applies.
    We obtain that $m=2^d$ for some $d\in \Zp$, and $R(t)$ is an Ising tensor up to a simultaneous permutation of rows and columns.
    By \cref{lem: ising tensor and vector space}, there exists a bijection $\phi_t:[m]\to \{1,-1\}^d$ such that $\log_\circ R(t)\in \mathcal{V}_{\phi_t}$, for every $t\in U.$
    Since there are only finitely many bijections from $[m]$ to $\{-1, 1\}^d$, there exists a fixed bijection $\phi \colon [m] \to \{-1, 1\}^d$ and a sequence $\{t_k\}_{k \in \Zp}$ converging to $t_0$ such that $\log_\circ R(t_k) \in \mathcal{V}_{\phi}$ for all $k$.

    Since $\mathcal{V}_\phi$ is a linear subspace of $\R^{n\times n}$, there are linear functionals $\ell_1,\ldots,\ell_r:\R^{n\times n}\to \R$ such that $X\in \mathcal{V}_{\phi}\iff \ell_1(X)=\cdots=\ell_r(X)=0.$
    Now for every $j\in [r]$, the analytic function $\ell_j(\log_\circ R(t))$ vanishes on the converging sequence $\{t_k\}.$
    This implies $\ell_j(\log_\circ R(t))\equiv 0$ for every $j\in [r]$.
    That is, $\log_\circ R(t)\in \mathcal{V}_{\phi}$ for every $t\in \R.$
\end{proof}

\begin{lemma}
\label{lem:nn-positive-kernel}
Let $R\in \Sym{m}{}{\R_{>0}}$ have pairwise nonproportional rows. 
If
$\PlGH(R)$ is not $\#\mathrm P$-hard, then $R$ is an Ising tensor up to a simultaneous permutation of rows and columns.
\end{lemma}

\begin{proof}
Put $h_i=\sqrt{(R^2)_{ii}}>0$ for every $i\in[m]$ and consider
\(
A=\diag(h)^{-1}R^2\diag(h)^{-1}.
\)
Similar to the proof of \cref{lem: realize compressed matrix C}, $A$ is realizable (relative to $R$), and $0<A_{ij}\le 1$ for all $i,j\in [m]$, by Cauchy-Schwarz inequality.
The equality holds if and only if the $i$-th row and the $j$-th row of $R$ are proportional, i.e., $i=j.$
Thus, $A_{ii}=1$ and $A_{ij}<1$ for all $i,j\in [m]$ and $ i \ne j$.
A sufficiently
large fixed integer $k$ makes $B=A^{\circ k}$ strictly diagonally
dominant, hence invertible by the Gershgorin circle theorem.
The matrix $B$ is realizable relative to $R$ by \cref{cor:thickening}.

For every $t\in \R$, define $R(t)=B^{\circ t}\circ R^{\circ(1-t)}.$
Since both $B$ and $R$ are entrywise positive, $R(t)$ is well-defined
for all $t\in \R$.
Then $R(t)$ is an entrywise analytic family of matrices, with $R(0)=R$ and $R(1)=B\in \Sym{m}{F}{\Rp}$.
By \cref{cor:thickening,cor:thickening_interpolation}, $\PlGH(R(t))\le_T\PlGH(R)$ for every $t\in \R.$
Since $\PlGH(R)$ is not \#P-hard, so is $\PlGH(R(t))$ for every $t\in \R.$
By \cref{lem: uniform tensor}, there is a bijection $\phi \colon [m] \to \{-1,1\}^d$ such that $\log_\circ R(t) \in \mathcal{V}_\phi$ for every $t \in \mathbb{R}$.
In particular, $R=R(0)\in \mathcal{V}_\phi$.
By \cref{lem: ising tensor and vector space}, $R$ is an Ising tensor up to a simultaneous permutation of rows and columns.
\end{proof}

\begin{lemma}
\label{lem:nn-nonbip-core}
Let $R\in \Sym{m}{}{\R_{\ge 0}}$ be nonzero, connected and nonbipartite, with pairwise nonproportional rows.
If $\PlGH(R)$ is not
$\#\mathrm P$-hard, then $R$ is an Ising tensor up to a simultaneous permutation of rows and columns.
\end{lemma}

\begin{proof}
Since $R$ is nonzero, connected and nonbipartite, $H=R^\ell$ is entrywise
positive for a sufficiently large odd integer $\ell$.
By stretching (\cref{lem:stretching}),
$\PlGH(H)\le_T\PlGH(R)$, so $\PlGH(H)$ is not $\#\mathrm P$-hard.
Moreover, \cref{lem: nonproportional rows matrix powers} implies that
$H$ has pairwise nonproportional rows.
Applying \cref{lem:nn-positive-kernel}, we conclude that $H$ is an
Ising tensor up to a simultaneous permutation of rows and columns.
Thus there exists a permutation matrix $P\in\R^{m\times m}$ such that
\[
P^{\tt T}R^\ell P
=\lambda\bigotimes_{i=1}^{d}
\begin{pmatrix}a_i&b_i\\b_i&a_i\end{pmatrix},
\qquad
\lambda>0,\quad a_i,b_i>0,\quad a_i\neq b_i.
\]

Since $\ell$ is odd, there is a unique real symmetric matrix $X$
satisfying
\(
X^\ell=P^{\tt T}R^\ell P.
\)
It is obtained by taking the real $\ell$-th root of each eigenvalue
of $P^{\tt T}R^\ell P$, while preserving the corresponding eigenvectors.
The odd spectral root factors over the displayed tensor decomposition:
Let $H_2 = \tfrac{1}{\sqrt{2}}
\begin{psmallmatrix}1&1\\1&-1\end{psmallmatrix}$, then
$
\begin{psmallmatrix}a&b\\b&a\end{psmallmatrix}
= H_2^{\tt T} 
\begin{psmallmatrix}\lambda &0\\0& \mu\end{psmallmatrix}
H_2$, with
$\lambda = a+b,\, \mu = a-b$,
and its real $\ell$-th root is 
$H_2^{\tt T}  \begin{psmallmatrix}\lambda^{1/\ell} &0\\0& \mu^{1/\ell}\end{psmallmatrix} H_2$, a matrix of the 2 by 2 Ising form.
So $X$ is again an Ising tensor.

On the other hand, $P^{\tt T}RP$ is real symmetric and satisfies
\(
(P^{\tt T}RP)^\ell=P^{\tt T}R^\ell P.
\)
By uniqueness, $P^{\tt T}RP=X$.
Hence $R$ is an Ising tensor up to a simultaneous permutation of
rows and columns.
\end{proof}

Recall $s_i=\sqrt{(M^2)_{ii}}\,(\forall i\in [n])$, and $P_a$ is the proportional class of size $n_a$.

\begin{lemma}
\label{lem:nn-moments-nonbipartite}
Let $M\in \Sym{n}{}{\R_{\ge 0}}$ be nonzero, connected and nonbipartite.
If $\PlGH(M)$ is not
$\#\mathrm P$-hard, then the multisets $\{s_i\mid i\in P_a\}$ agree for all $a\in [q].$
\end{lemma}

\begin{proof}
    Let $C\in \Sym{q}{}{\R_{\ge 0}}$ be the compressed matrix of $M$.
    Note that $C$ is still nonzero, connected and nonbipartite.
    To see that, suppose $i, j\in P_a$ are in the same proportional class.
    Then for any $k$,  $M_{ik} \ne 0$ iff $M_{jk} \ne 0$, and so any path
    or cycle of the underlying graph of $M$ has a corresponding  path
    or cycle of the underlying graph of $C$.
    Moreover, $C$ has pairwise nonproportional rows.
    By \cref{lem: realize compressed matrix C}, $\PlGH(C)\le_T\PlGH(M)$.
    Since we assume that $\PlGH(M)$ is not \#P-hard, neither is $\PlGH(C)$.
    \cref{lem:nn-nonbip-core} applied to $C$ shows that $C$ is an Ising tensor up to a simultaneous permutation of rows and columns.
    In particular, $C$ is invertible, entrywise positive and has constant diagonal.

    Now consider an arbitrary realizable unary $f\in \Rp^n.$
    Let \[u(f)_a=\frac1{n_a}\sum_{i\in P_a}s_if_i\,\qquad (\forall a\in [q])\]
    and $C_f=\diag(u(f))C\diag(u(f)).$
    By \cref{lem:nn-quotient}, $\PlGH(C_f)\leq_T\PlGH(M)$. 
    Then $\PlGH(C_f)$ is not \#P-hard.
    Since $u(f)\in \Rp^n$, $C_f$ is invertible and entrywise positive.  
    By \cref{lem:nn-positive-kernel}, $C_f$ is an Ising tensor up to a simultaneous permutation of rows and columns.
    In particular, $C_f$ has constant diagonal.
    But $(C_f)_{aa}=C_{aa}u(f)_{a}^2,\,\forall a\in[q]$.
    This implies that $u(f)$ is a constant unary, for any realizable $f\in \Rp^n.$
    
    Let $r_i=n_a$ if $i\in P_a$, for every $i\in [n].$
    In the proof of \cref{lem: realize compressed matrix C}, we showed that $r=(r_1,\ldots, r_n)$ is realizable.
    Consider $f_i=\frac{r_i}{s_i}=\frac{n_a}{s_i}$ if $i\in P_a$, for every $i\in [n].$
    Then $f\in \Rp^n$ and $f$ is realizable by \cref{lem:thickening_interpolation,lem:thickening}.
    For this $f$, we have 
    \[
    u(f)_a=\frac{1}{n_a}\sum_{i\in P_a}s_i\frac{n_a}{s_i}=n_a \qquad (\forall a\in [q]).
    \]
    Thus, $n_a=c$ is constant for every $a\in [q].$
    Now for any fixed $j\in \Zp$, let $f_i=s_i^{j-1},\,\forall i\in [n].$
    Then $f$ is realizable by \cref{lem:thickening}.
    For this $f$, 
    \[
    u(f)_a=\frac{1}{n_a}\sum_{i\in P_a}s_i\cdot s_i^{j-1}=\frac{1}{c}\sum_{i\in P_a}s_i^j \qquad (\forall a\in [q]).
    \]
    Then $\sum_{i\in P_a}s_i^j$ agree for all $a\in [q]$, for any fixed $j.$
    By Newton's identities, the elementary symmetric polynomials in $\{s_i\mid i\in P_a\}$ (as a multiset) agree for all $a\in [q].$
    Define the polynomial $F_a(z)=\prod_{i\in P_a}(z-s_i)$ for every $a\in[q].$ 
    Then $F_a$ are the same polynomial for all $a\in [q].$
    This shows that the multisets $\{s_i\mid i\in P_a\}$ agree for all $a\in [q].$
\end{proof}

\begin{lemma}\label{lem: nonbipartite case}
    Let $M\in \Sym{n}{}{\R_{\ge 0}}$ be nonzero, connected and nonbipartite.
    If $\PlGH(M)$ is not $\#\mathrm P$-hard, then $M=K\otimes ww^{\tt T}$ up to a simultaneous permutation of rows and columns, where $K$ is an Ising tensor and $w$ is an entrywise positive vector.
\end{lemma}
\begin{proof}
    Let $C\in \Sym{q}{}{\R_{\ge 0}}$ be the compressed matrix of $M$.
    As in the proof of \cref{lem:nn-moments-nonbipartite}, $C$ is an Ising tensor up to a simultaneous permutation of rows and columns.
    More precisely, there exists a permutation $\pi:[q]\to [q]$ and an Ising tensor $K$ such that $K_{ab}=C_{\pi(a)\pi(b)}$ for all $a,b\in [q].$
    By \cref{lem:nn-moments-nonbipartite}, we may write $\{s_i\mid i\in P_a\}=\{w_1,\ldots,w_k\},\,\forall a\in [q]$. Here $\{w_1,\ldots,w_k\}$ is the same for all
    $a \in [q]$.
    For every $a\in [q],$
    We may reorder the $k$ elements in $P_{\pi(a)}$ by $(i_{a,1},i_{a,2},\ldots, i_{a,k})$ such that $s_{i_{a,t}}=w_t$, for every $t\in [k].$
    Finally, we order the $n$ colors by $(i_{1,1},\ldots, i_{1,k},i_{2,1},\ldots, i_{2,k},\ldots, i_{q,1},\ldots,i_{q,k}).$
    We obtain the matrix ${M'}$ after this reordering, which satisfies
    \[
    M'_{(a,s),(b,t)}=M_{i_{a,s},i_{b,t}}=C_{\pi(a)\pi(b)}s_{i_{a,s}}s_{i_{b,t}}=K_{ab}w_sw_t\qquad (\forall a,b\in[q],\forall s,t\in[k]).
    \]
    That is, $M'=K\otimes ww^{\tt T}.$
\end{proof}

\subsection{Bipartite Case}
\label{subsec: bipartite}
In this subsection, we prove that a nonzero connected bipartite matrix is in the form of \cref{item: tractable 2}, otherwise it defines a \#P-hard $\PlGH$ problem.

\begin{lemma}
\label{lem:nn-matching}
Let $\ell,r,m\in \Zp$ and $\ell+r=m$.
Let $R=\left(\begin{smallmatrix}O_\ell&D\\D^{\tt T}&O_r\end{smallmatrix}\right)\in\Sym{m}{}{\R_{\ge0}}$ have pairwise nonproportional rows, where   $D\in\Rp^{\ell\times r}$.
If $\PlGH(R)$ is not $\#\mathrm P$-hard, then $\ell=r$, and $D$ is an Ising tensor up to independent row and column permutations. 
\end{lemma}

\begin{proof}
All realizability in this proof is relative to $R.$
By \cref{lem:nn-thickening_interpolation},
we can realize 
$$\left(\begin{matrix}O_\ell&D^{\circ t}\\(D^{\tt T})^{\circ t}&O_r\end{matrix}\right)\qquad (\forall t\in \R).$$
Put
\[
c=\frac1{\ell r}\sum_{x=1}^{\ell}\sum_{y=1}^r\log D_{xy},
\qquad
L=\log_\circ D-cJ_{\ell\times r},
\qquad
U(t)=\exp_\circ(tL).
\]
Then we can realize the matrix
$$H(t):=\left(\begin{matrix}
    O_\ell & e^{-ct}D^{\circ t}\\
    e^{-ct}(D^{\tt T})^{\circ t} & O_r
\end{matrix}\right)=\left(\begin{matrix}
    O_\ell & \exp_\circ(tL)\\
    \exp_\circ(tL)^{\tt T} & O_r
\end{matrix}\right)=\left(\begin{matrix}
    O_\ell & U(t)\\
    U(t)^{\tt T} & O_r
\end{matrix}\right)\quad(\forall t\in \R).$$
Since $R$ has pairwise nonproportional rows, so does $H(t)$ for $t \in \mathbb{R}_{\neq 0}$.
Let $F(t)=U(t)U(t)^{\tt T}\in \Sym{\ell}{psd}{\Rp}$ and $G(t)=U(t)^{\tt T}U(t)\in \Sym{r}{psd}{\Rp}$ for every $t\in \R.$
By stretching (\cref{lem:stretching}), we can realize
\[
H(t)^2=
\left(\begin{matrix}
    F(t) & O_{\ell\times r}\\
    O_{r\times \ell} & G(t)
\end{matrix}\right)\qquad(\forall t\in \R).
\]
By \cref{lem: nonproportional rows matrix powers}, $H(t)^2$ also has pairwise nonproportional rows, as do $F(t)$ and $G(t)$, for every $t\in \R_{\neq 0}.$
By \cref{lem:nn-component}, $\PlGH(F(t))\le_T\PlGH(R)$ and $\PlGH(G(t))\le_T \PlGH(R)$, for every $t\in \R.$
Since $\PlGH(R)$ is not \#P-hard, neither $\PlGH(F(t))$ nor $\PlGH(G(t))$ is \#P-hard.
%Since $R$ has pairwise nonproportional rows, $D$ has pairwise nonproportional rows and columns, as does $U(t)$ for every $t \in \mathbb{R}_{\neq 0}$.
Then \cref{lem:nn-positive-kernel} applies to $F(t)$ and $G(t)\,(\forall t\neq 0)$.
We obtain that $\ell=2^d$ and $r=2^e$; $F(t)$ and $G(t)$ are Ising tensors up to possibly distinct simultaneous permutations of rows and columns.
In particular, $F(t)$ is invertible, and so is $G(t)$, for every $t\neq 0$.
Fix some $t\in \R_{\neq 0}.$
By invertibility, 
\(
\ell=\operatorname{rank}(F(t))=\operatorname{rank}(U(t))=\operatorname{rank}(G(t))=r.
\)
In the following we prove that $D$ is an Ising tensor up to independent row and column permutations. 

\begin{claim}\label{claim: L zero row col sum}
    $L$ has zero row and column sums.
\end{claim}
\begin{claimproof}{\ref{claim: L zero row col sum}}
Consider
\begin{equation*}
    F(t)_{xx}=\sum_{y=1}^r U(t)_{xy}^2=\sum_{y=1}^re^{2tL_{xy}}\qquad (\forall x\in [\ell]).
\end{equation*}
Differentiating at $t=0$ gives 
$$F'(0)_{xx}=2\sum_{y=1}^r{L_{xy}}\qquad (\forall x\in [\ell]).$$
If $t\neq 0$, then $F(t)$ is an Ising tensor up to a simultaneous permutations of rows and columns, and thus $F(t)$ has constant diagonal.
%%% your sentence here is a trubcated. i edited
%If $t=0$, $F(0)$ also has constant 
By continuity of the expression of $F(t)_{xx}$, and taking limit $t \rightarrow 0$, we see that $F(0)$ also has constant diagonal.
Thus, $F'(0)_{xx}=\lim_{t\to 0}\frac{1}{t}(F(t)_{xx}-F(0)_{xx})$ is independent of $x.$
Hence the above equation implies that $L$ has constant row sum.
The similar calculation for $G'(0)_{yy}$ implies that $L$ has constant column sum.
Since the sum of all entries of $L$ is zero by the definition of $L$, every row and column
sum is zero.
\end{claimproof}

Apply \cref{lem: uniform tensor} to $\{F(t)\}_{t\in \R}$ (note that $F(t)\in \Sym{\ell}{F}{\Rp}$ when $t\neq 0$), we obtain a bijection $\phi_{\rm L}:[\ell]\to \{1,-1\}^d$ such that $\log_\circ F(t)\in \mathcal{V}_{\phi_{\rm L}}$ for every $t\in \R.$
Similarly, there is a bijection $\phi_{\rm R}:[\ell]\to \{1,-1\}^d$ such that $\log_\circ G(t)\in \mathcal{V}_{\phi_{\rm R}}$ for every $t\in \R.$
Entrywise, $$F(t)_{xx'}=\sum_{y=1}^\ell  U(t)_{xy}U(t)_{x'y}=\sum_{y=1}^\ell e^{t\left(L_{xy}+L_{x'y}\right)} \qquad (\forall x,x'\in[\ell]).$$ 
Thus, $F(0)_{xx'}=\ell $, $F'(0)_{xx'}=\sum_{y=1}^\ell (L_{xy}+L_{x'y})=0$, and $F''(0)_{xx'}=\sum_{y=1}^\ell (L_{xy}+L_{x'y})^2$, for all $x,x'\in [\ell].$
Hence
\[
\left.\frac{d^2}{dt^2}\log F(t)_{xx'}\right|_{t=0}=\frac{F''(0)_{xx'}}{F(0)_{xx'}}-\frac{F'(0)_{xx'}^2}{F(0)_{xx'}^2}=\frac{1}{\ell }\sum_{y=1}^r(L_{xy}+L_{x'y})^2\quad(\forall x,x'\in [\ell]).
\]
Therefore, with $P=I_\ell-J_\ell/\ell$, 
\begin{equation}
 \left.\frac{d^2}{dt^2}
       P\log_{\circ}(F(t))P\right|_{t=0}
 =\frac2\ell LL^{\tt T}.
 \label{eq:nn-second-order}
\end{equation}
This can be verified by setting $B=  \left.\frac{d^2}{dt^2}\log F(t)\right|_{t=0}
=\tfrac{1}{\ell} (M_1 +M_2 +2LL^{\tt T})$, and $A = PBP$,
where $M_1= {\bf 1} {\bf d}^{\tt T}$, $M_2=   {\bf d}{\bf 1}^{\tt T}$,
${\bf d}= {\rm diag}(LL^{\tt T})$. Using Claim~\ref{claim: L zero row col sum}
we have $LL^{\tt T} J = J LL^{\tt T} =O$.  After some calculation, one can check that $A = \tfrac{2}{\ell} LL^{\tt T}$.
% B = [M1 + M2 + 2M3]/l
% write M_3 = LL^{\tt T}, A = (I-J/l)B(I-J/l) = B -JB/l -BJ/l +JBJ/l^2, 
% J = 1 1^T. 
% JM_1 = Tr . J,  M_2 J = Tr . J
% JM_2 = l M_2,   M_1 J = l M_1
% J M_3 = M_3 J =0
% JB = [ JM1 + J M2 + 0]/l = [Tr J]/l  +M_2
% BJ = [ l M1 + Tr . J] /l = [Tr J]/l  +M1
% JBJ = 2 Tr J
% A = B -JB/l -BJ/l +JBJ/l^2 = B - [[Tr J]/l^2  +M_2/l] - [[Tr J]/l^2  +M1/l]
% + 2/l^2 [Tr J] = 2/l  M3

Since $\log_\circ F(t)\in \mathcal{V}_{\phi_{\rm L}}$, 
by \cref{lem: ising tensor and vector space}
we can write $\log_\circ F(t)=c(t)J+\sum_{i=1}^d\gamma_i(t) S(\phi_{\rm L},i)$.
By \cref{lem: cube basis}, this expansion is unique, and smoothness of $\log_\circ F(t)$ implies the smoothness $c(t)$ and $\gamma_i(t)$. 
%where $S(\phi_{\rm L},i)_{xx'}=\phi_{\rm L}(x)_i\phi_{\rm L}(x')_i$.
Define 
$S_{\ell}\in \R^{\ell\times d}$ by $(S_\ell)_{xi}=\phi_{\rm L}(x)_i$, then $\log_\circ F(t)=c(t)J+S_{\ell}\diag (\gamma_1(t),\ldots,\gamma_d(t))S_\ell^{\tt T}.$
Note that the matrix $S_\ell$ has zero column sum, since $\phi_{\rm L}$ is a bijection from $[\ell]$ to $\{1,-1\}^d$.
Hence $PS_{\ell}=S_{\ell}.$
Also $PJ =O$.
Thus, $P\log_\circ(F(t))P=S_{\ell}\diag (\gamma_1(t),\ldots,\gamma_d(t))S_\ell^{\tt T}$, for every $t\in \R.$
Then \eqref{eq:nn-second-order} implies $LL^{\tt T}=\frac{\ell}{2}S_\ell \diag(\gamma_1''(0),\ldots, \gamma_d''(0))S_{\ell}^{\tt T}.$
In particular, $\operatorname{im}(LL^{\tt T})\subseteq \operatorname{im}(S_{\ell}).$
For the real matrix $L$, the rank of $L L^{\tt T}$ is the same as $L$.
%As $L L^{\tt T} v =0 \Longrightarrow \|L^{\tt T}v\| =0  \Longrightarrow  L^{\tt T}v=0$,
%$L L^{\tt T}$ has the same rank as $L$.
So, $\operatorname{im}(LL^{\tt T})=\operatorname{im}(L)$, and we obtain $\operatorname{im}(L)\subseteq\operatorname{im}(S_\ell).$
Define $S_r\in \R^{\ell \times d}$ by $(S_r)_{yj}=\phi_{\rm R}(y)_j.$
The same calculation for $G(t)$ yields $\operatorname{im}(L^{\tt T})\subseteq \operatorname{im}(S_{r}).$
\begin{claim}\label{claim: L and A}
    Let $A=\frac{1}{\ell ^2}S_\ell^{\tt T}LS_r\in \R^{d\times d}$, then $L=S_{\ell} A S_r^{\tt T}.$
\end{claim}
\begin{claimproof}{\ref{claim: L and A}}
    Note that
    \(S_\ell^{\tt T}S_\ell=\ell I_d\) and \(S_r^{\tt T}S_r=\ell I_d\), since $\phi_{\rm L}:[\ell]\to \{1,1\}^d$ and $\phi_{\rm R}:[\ell]\to \{1,1\}^d$ are bijecions.
    Since $\operatorname{im}(L)\subseteq\operatorname{im}(S_\ell)$, there exists $B\in \R^{d\times \ell}$ such that $L=S_\ell B.$
    Then \(
    \frac1\ell S_\ell S_\ell^{\tt T}L
    =\frac1\ell S_\ell S_\ell^{\tt T}S_\ell B
    =\frac1\ell S_\ell(\ell I_d)B
    =S_\ell B
    =L.
    \)
    Similarly, by $\operatorname{im}(L^{\tt T})\subseteq \operatorname{im}(S_{r})$ we have \(\frac1\ell LS_rS_r^{\tt T}=L.\)
    Combining these two identities, we obtain $L=\frac{1}{\ell^2} S_\ell S_\ell^{\tt T} LS_rS_r^{\tt T}=S_{\ell}AS_{r}^{\tt T}.$
\end{claimproof}
Next we show that the matrix $A$ in Claim~\ref{claim: L and A} must be a weighted permutation matrix, i.e., there is exactly one nonzero element in every row and column of $A.$

\begin{claim}\label{claim: A distinct rows}
    For all distinct $i,k\in [d]$ and every $j\in[d]$, $A_{ij}A_{kj}=0.$
\end{claim}
\begin{claimproof}{\ref{claim: A distinct rows}}
By $\log_{\circ}F(1)\in \mathcal{V}_{\phi_{\rm L}}$, for arbitrary $x,y\in[\ell]$ we can write
\begin{equation}\label{eq: logF1 1}
    \log F(1)_{xy}=\mu+\sum_{j=1}^d \phi_{\rm L}(x)_j\phi_{\rm L}(y)_j\gamma_j
\end{equation}
for some $\mu,\gamma_1,\ldots,\gamma_d\in \R.$
On the other hand, by $F(1)=U(1)U(1)^{\tt T}=\exp_\circ (L)\exp_\circ (L)^{\tt T}$,
\begin{equation}\label{eq: F1xy 1}
   F(1)_{xy}=\sum_{z=1}^\ell e^{L_{xz}+L_{yz}}. 
\end{equation}
From $L=S_{\ell}AS_r^{\tt T}$, $(S_\ell)_{xj}=\phi_{\rm L}(x)_j$ and $(S_r)_{zj}=\phi_{\rm R}(z)_j$, we have
\begin{equation}\label{eq: Lxz plus Lyz}
    \begin{aligned}
    L_{xz}+L_{yz}=&\sum_{j=1}^d\sum_{h=1}^d A_{hj}(\phi_{\rm L}(x)_h\phi_{\rm R}(z)_j+\phi_{\rm L}(y)_h\phi_{\rm R}(z)_j)\\
    =&\sum_{j=1}^d\left(\sum_{h=1}^dA_{hj}(\phi_{\rm L}(x)_h+\phi_{\rm L}(y)_h)\right)\phi_{\rm R}(z)_j\\
    =&\sum_{j=1}^d w_j \eta_j,
\end{aligned}
\end{equation}
where $w_j=\sum_{h=1}^dA_{hj}(\phi_{\rm L}(x)_h+\phi_{\rm L}(y)_h)$ and $\eta_j=\phi_{\rm R}(z)_j$, for every $j\in[d].$
Since $\phi_{\rm R}:[\ell]\to\{1,-1\}^d$ is a bijection, when $z$ runs over $1$ to $\ell$, $\eta=\phi_{\rm R}(z)$ runs over $\{1,-1\}^d.$
Plug \eqref{eq: Lxz plus Lyz} into \eqref{eq: F1xy 1}, we have
\[
F(1)_{xy}=\sum_{\eta\in\{1,-1\}^d} \exp\left(\sum_{j=1}^d w_j\eta_j \right)
=\sum_{\eta\in\{1,-1\}^d} \prod_{j=1}^de^{w_j \eta_j}
=\prod_{j=1}^d(e^{w_j}+e^{-w_j})
=2^d\prod_{j=1}^d\cosh{w_j}.
\]
Taking logarithm and substituting the expression for $w_j$, 
\begin{equation}\label{eq: logF1 2}
    \log F(1)_{xy}
=\log \ell+
\sum_{j=1}^d
\log\cosh\left(\sum_{h=1}^dA_{hj}(\phi_{\rm L}(x)_h+\phi_{\rm L}(y)_h)\right).
\end{equation}

The equations \eqref{eq: logF1 1} and \eqref{eq: logF1 2} hold for arbitrary $x,y\in[\ell]$.
We may pick $x,y\in [\ell]$ such that $\phi_{\rm L}(x)=-\phi_{\rm L}(y).$
Such $x,y$ exist because $\phi_{\rm L}:[\ell]\to \{1,-1\}^d$ is a bijection.
Then 
\begin{equation*}
    \log F(1)_{xy}=\mu-\sum_{j=1}^d \gamma_j=\log\ell.
\end{equation*}
Hence $\mu=\log \ell+\sum_{j=1}^d \gamma_j.$
Plugging this into \eqref{eq: logF1 1},
\begin{equation}\label{eq: logF1 3}
\begin{aligned}
    \log F(1)_{xy}=\log \ell
    %+\sum_{j=1}^d \gamma_i
    +\sum_{j=1}^d \left( 1 + \phi_{\rm L}(x)_j\phi_{\rm L}(y)_j \right) \gamma_j
    =\log \ell +\sum_{\substack{j\in [d]\\\phi_{\rm L}(x)_j=\phi_{\rm L}(y)_j}} 2 \gamma_j.
\end{aligned} 
%%% I put 2 inside, it looks better
\end{equation}

Now fix distinct $i,k\in [d].$ 
We choose $x,y\in [\ell]$ such that $\phi_{\rm L}(x)_i=\phi_{\rm L}(y)_i=1$ and $\phi_{\rm L}(x)_j=-\phi_{\rm L}(y)_j$ for every $j\in [d]\setminus \{i\}.$
Equations \eqref{eq: logF1 2} and \eqref{eq: logF1 3} give
\begin{equation}\label{eq: log cosh Aij}
    \sum_{j=1}^d \log\cosh \left(2A_{ij}\right)=2\gamma_i.
\end{equation}
\iffalse
Similarly, we may choose $x,y\in [n]$ such that $\phi_{\rm L}(x)_k=\phi_{\rm L}(y)_k=1$ and $\phi_{\rm L}(x)_j=-\phi_{\rm L}(y)_j$ for every $j\in [d]\setminus \{k\}.$
Then
\begin{equation}\label{eq: log cosh Akj}
    \sum_{j=1}^d \log\cosh \left(2A_{kj}\right)=2\gamma_k.
\end{equation}
\fi
Also, we may choose $x,y\in [n]$ such that $\phi_{\rm L}(x)_i=\phi_{\rm L}(y)_i=1$, $\phi_{\rm L}(x)_k=\phi_{\rm L}(y)_k=\varepsilon\in\{1,-1\}$, and $\phi_{\rm L}(x)_j=-\phi_{\rm L}(y)_j$ for every $j\in [d]\setminus \{i,k\}.$
Equations \eqref{eq: logF1 2} and \eqref{eq: logF1 3} imply
\begin{equation}\label{eq: log cosh Aij Akj}
    \sum_{j=1}^d\log\cosh\bigl(2A_{ij}
    %\pm
    + 2 \varepsilon A_{kj}\bigr)=2\gamma_i+2\gamma_k.
\end{equation}
Combining \eqref{eq: log cosh Aij}, and substituting $k$ for $i$ in \eqref{eq: log cosh Aij}, and \eqref{eq: log cosh Aij Akj},
for both signs $\pm$,
\begin{equation}\label{eq: log cosh}
    \sum_{j=1}^d\log\cosh\bigl(2A_{ij}\pm 2A_{kj}\bigr)
=
\sum_{j=1}^d\log\cosh(2A_{ij})
+\sum_{j=1}^d\log\cosh(2A_{kj}).
\end{equation}
Adding the two identities in \eqref{eq: log cosh}, combining logarithm and using
\[
\frac{\cosh(u+v)\cosh(u-v)}{\cosh^2u\cosh^2v}
=1-\tanh^2u\,\tanh^2v
\]
yields
\begin{equation}\label{eq: tanhAij}
    \sum_{j=1}^d\log\!\left(
1-\tanh^2(2A_{ij})\tanh^2(2A_{kj})
\right)=0.
\end{equation}
Every summand in \eqref{eq: tanhAij} is nonpositive, with equality precisely when $\tanh^2(2A_{ij})\tanh^2(2A_{kj})=0,\,\forall j\in[d]$.
This is equivalent to
$A_{ij}A_{kj}=0$ for every $j\in [d]$.
\end{claimproof}

\begin{claim}\label{claim: A nonzero row}
    No row of $A$ is zero.
\end{claim}
\begin{claimproof}{\ref{claim: A nonzero row}}
    Suppose for a contradiction that the $i$-th row of $A$ is zero.
    Choose $x,x'\in[\ell]$ such that $\phi_{\rm L}(x)_i=-\phi_{\rm L}(x')_i$ and $\phi_{\rm L}(x)_h=\phi_{\rm L}(x')_h$ for every $h\in [d]\setminus\{i\}.$
    Then $x\neq x'.$
    From $L=S_{\ell}AS_r^{\tt T}$ and the definition of $S_\ell,S_r$, we have 
    \begin{equation}\label{eq: Lxy and Lx'y}
        L_{xy}=\sum_{h=1}^d\sum_{j=1}^dA_{hj}\phi_{\rm L}(x)_h\phi_{\rm R}(y)_j\quad \text{and} \quad 
        L_{x'y}=\sum_{h=1}^d\sum_{j=1}^dA_{hj}\phi_{\rm L}(x')_h\phi_{\rm R}(y)_j
    \end{equation}
    for every $y\in [\ell].$
    The terms in the summations in \eqref{eq: Lxy and Lx'y} indexed by $h = i$ vanish because we assumed that the $i$-th row of $A$ is zero. 
    The remaining terms indexed by $h \neq i$ are identical.
    It follows that $L_{xy} = L_{x'y}$ for every $y \in [\ell]$.
    Thus, $D_{xy}=D_{x'y}$ for every $y\in [\ell]$, contradicting that $D$ has pairwise nonproportional rows.
\end{claimproof}

By Claims~\ref{claim: A distinct rows} and~\ref{claim: A nonzero row}, there is exactly one nonzero element in every row of $A$, i.e., $A$ is a weighed permutation matrix.
There exists a permutation $\pi:[d]\to [d]$, such that $A_{ij}=\gamma_i\neq 0$ if $j=\pi(i)$, and $A_{ij}=0$ if $j\neq \pi(i)$, for all $i,j\in [d].$
From $L=S_{\ell}AS_r^{\tt T}$ and the definition of $S_\ell,S_r$, we have 
\[
L_{xy}
=\sum_{i=1}^d\gamma_i
\phi_{\rm L}(x)_i\phi_{\rm R}(y)_{\pi(i)}\qquad (\forall x,y\in [\ell]).
\]
Define
\(
\widetilde\phi_{\rm R}(y)_i:=\phi_{\rm R}(y)_{\pi(i)}
\)
for every $y\in [\ell].$
This remains a bijection from $[\ell]$ to $\{1,-1\}^d$, and 
\[
L_{xy}
=\sum_{i=1}^d\gamma_i
\phi_{\rm L}(x)_i\widetilde\phi_{\rm R}(y)_{i}\qquad (\forall x,y\in [\ell]).
\]
By $\log_\circ D=L+cJ_{\ell}$,
\[
D_{xy}
=e^c\prod_{i=1}^d
\exp\bigl(\gamma_i\phi_{\rm L}(x)_i\widetilde\phi_{\rm R}(y)_i\bigr) \qquad(\forall x,y\in [\ell]).
\]
Therefore, $D$ is the Ising tensor $e^c\bigotimes_{i=1}^d \left(\begin{matrix}
    e^{\gamma_i}& e^{-\gamma_i}\\
    e^{-\gamma_i} & e^{\gamma_i}
\end{matrix}\right)$  after independent permutations of rows and columns.
\end{proof}

The following lemma generalizes \cref{lem:nn-matching} from entrywise positive $D$ to nonnegative $D$.

\begin{lemma}
\label{lem:nn-bip-core}
Let $\ell,r,m\in \Zp$ and $\ell+r=m$.
Let $R=\left(\begin{smallmatrix}O_\ell&D\\D^{\tt T}&O_r\end{smallmatrix}\right)\in\Sym{m}{}{\R_{\ge0}}$
be connected with pairwise nonproportional rows, where $D\in\R_{\ge0}^{\ell\times r}$.
If $\PlGH(R)$ is not $\#\mathrm P$-hard, then $\ell=r$, and $D$ is an Ising tensor up to independent row and column permutations. 
\end{lemma}

\begin{proof}
All realizability in this proof is relative to $R.$
For every $k\in\Zp$, by stretching (\cref{lem:stretching}) we can realize
\[
R^{2k+1}=
\left(\begin{matrix}
    O_{\ell} & D\\
    D^{\tt T} & O_r
\end{matrix}\right)^{2k+1}=
\left(\begin{matrix}
    O_\ell & D(D^{\tt T}D)^k\\
    (D^{\tt T}D)^kD^{\tt T} & O_r
\end{matrix}\right).
\]
Since $R$ is connected, for sufficiently large $k$, $D(D^{\tt T}D)^k$
is entrywise positive. 
Fix this $k$ and let $H=D(D^{\tt T}D)^k.$ 
Note that $R^{2k+1}$ has pairwise nonproportional rows by \cref{lem: nonproportional rows matrix powers}, so \cref{lem:nn-matching} applies to $R^{2k+1}$.
We obtain that $\ell=r$, and $H$ is an Ising tensor up to independent row and column permutations. 
That is, there exist $\ell\times\ell$ permutation matrices $P$ and $Q$ such that $H'=PHQ$ is an Ising tensor.
In particular, $H'$ is symmetric and invertible.
Let $D'=PDQ$, then $H=D(D^{\tt T}D)^k$ implies $H'=D'(D'^{\tt T}D')^k$.
We may assume that $H$ itself is symmetric and invertible.
%In particular, $H$ is symmetric and invertible.
%Hence $D$ is square and invertible.

Write $A=DD^{\tt T}$. 
Then $H=A^k D$ and
$H^2=HH^{\tt T}=A^{2k+1}$.
Since $2k+1$ is odd, the $(2k+1)$-st root of any symmetric and invertible matrix exists
and is unique, and it is obtained by taking the $(2k+1)$-st root of each eigenvalue of the matrix, while preserving the corresponding eigenvectors.
%%% I added the line below.
%$H^2$ 
%%is positive definite.
%%% I changed .... as we need to talk about 2k+1 root of H too.
Thus,  $A=H^{\frac{2}{2k+1}}$, the unique $(2k+1)$-th real root
of $H^2$, and 
\[
 D=A^{-k}H=H^{-\frac{2k}{2k+1}}H=H^{\frac{1}{2k+1}}.
\]
%Since $2k+1$ is odd, the $(2k+1)$-st root of the symmetric and invertible matrix $H$ is unique, and it is obtained by taking the $(2k+1)$-st root of each eigenvalue of $H$, while preserving the corresponding eigenvectors.
Thus, $D$ is still an Ising tensor, up to the same row and column permutations $P$ and $Q.$
\end{proof}

%%%%%%%%%%%%%%%here

The following lemma is the bipartite version of \cref{lem:nn-moments-nonbipartite}.
\begin{lemma}
\label{lem:nn-moments-bipartite}
Let $M\in\Sym{n}{}{\R_{\ge0}}$ be nonzero, connected and bipartite.
If $\PlGH(M)$ is not $\#\mathrm P$-hard, then the multisets
$\{s_i\mid i\in P_a\}$ agree among all proportional classes
on the same side of the bipartition.
\end{lemma}

\begin{proof}
Let $C=\left(\begin{smallmatrix}
    O_\ell & D\\
    D^{\tt T} & O_r
\end{smallmatrix}\right)\in \Sym{q}{}{\R_{\ge 0}}$ be the compressed matrix of $M$, where $\ell,r\ge 1$ and $\ell+r=q.$
The matrix $C$ is connected and bipartite, with pairwise
nonproportional rows.
By \cref{lem: realize compressed matrix C},
$\PlGH(C)\le_T\PlGH(M)$, so $\PlGH(C)$ is not $\#\mathrm P$-hard.
Applying \cref{lem:nn-bip-core}, we conclude that $\ell=r$, and $D$ is an Ising tensor up to independent row and column permutations.

Now consider an arbitrary realizable unary $f\in\Rp^n$.
Put
\[
u(f)_a=\frac1{n_a}\sum_{i\in P_a}s_if_i
\qquad(\forall a\in[q]),
\]
and let $C_f=\diag(u(f))C\diag(u(f))$.
By \cref{lem:nn-quotient},
$\PlGH(C_f)\le_T\PlGH(M)$, so $\PlGH(C_f)$ is not
$\#\mathrm P$-hard.
Since $u(f)\in\Rp^q$, the matrix $C_f$ still has pairwise nonproportional rows.
Thus \cref{lem:nn-bip-core} also applies to $C_f$.
Its cross block
\[
(D_f)_{ab}=u(f)_aD_{ab}u(f)_{b+\ell}
\qquad(\forall a,b\in [\ell])
\]
is therefore an Ising tensor up to independent row and column permutations.
Every row of an Ising tensor contains the same multiset of entries,
and the same holds for its columns.
Consequently, both $D$ and $D_f$ have constant row product
and constant column product.
For every $a\in [\ell]$,
\[
\prod_{b\in [\ell]}(D_f)_{ab}
=
u(f)_a^{\ell}
\prod_{b\in [\ell]}D_{ab}
\prod_{b\in [\ell]}u(f)_{b+\ell}.
\]
The LHS and the last two factors are positive and
independent of $a$.
Hence $u(f)_a$ is constant for every $a\in [\ell]$.
The analogous argument using column product shows that
$u(f)_a$ is constant for every $\ell<a\le q$ as well.

Similar to the proof of \cref{lem:nn-moments-nonbipartite},
$f=(f_1,\ldots, f_n)$ is a realizable unary, where
\(
f_i=\frac{r_i}{s_i}=\frac{n_a}{s_i}\,(i\in P_a, \forall i\in [n]).
\)
For this choice,
\[
u(f)_a
=\frac1{n_a}\sum_{i\in P_a}s_i\frac{n_a}{s_i}
=n_a.
\]
Thus $n_a$ are constant on each side of the bipartition, i.e., there exist $c_{\ell},c_{r}\in \Zp$ such that $n_a=c_\ell$ for every $1\le a\le \ell$ and $n_a=c_r$ for every $\ell< a\le q.$

For any fixed $j\in\Zp$, take $f_i=s_i^{j-1}$.
The unary $f=(f_1,\ldots,f_n)$ is realizable, and
\[
u(f)_a=
\begin{cases}
\displaystyle\frac1{c_\ell}\sum_{i\in P_a}s_i^j,
    &1\le a\le \ell\\[6pt]
\displaystyle\frac1{c_r}\sum_{i\in P_a}s_i^j,
    &\ell <a \le q.
\end{cases}
\]
Since $u(f)$ is constant on each side, the power sums
$\sum_{i\in P_a}s_i^j$ agree among all classes on the same side.
Applying Newton's identities separately on the two sides, we
conclude that the elementary symmetric polynomials in
$\{s_i\mid i\in P_a\}$ agree among proportional classes on the same side.
Hence the polynomials
\(
F_a(z)=\prod_{i\in P_a}(z-s_i)
\)
are identical for all $1\le a\le\ell$, and are identical for all
$\ell<a\le q$.
Their roots, counted with multiplicity, therefore agree within
each side, as claimed.
\end{proof}

\begin{lemma}\label{lem: bipartite case}
Let $M=\left(\begin{smallmatrix}
    O & B\\
    B^{\tt T} & O
\end{smallmatrix}\right)\in\Sym{n}{}{\R_{\ge0}}$ be nonzero, connected and bipartite.
If $\PlGH(M)$ is not $\#\mathrm P$-hard, then
\[
M=
\begin{pmatrix}
O&K\otimes uv^{\tt T}\\
K\otimes vu^{\tt T}&O
\end{pmatrix}
\]
up to a simultaneous permutation of rows and columns, where $K$
is an Ising tensor and $u,v$ are entrywise positive vectors.
\end{lemma}

\begin{proof}
Let $C=\left(\begin{smallmatrix}
    O_{\ell} & D\\
    D^{\tt T} & O_{r}
\end{smallmatrix}\right)\in \Sym{q}{}{\R_{\ge 0}}$ be the compressed matrix of $M$, where $\ell,r\ge 1$ and $\ell+r=q.$
The matrix $C$ is nonzero, connected and bipartite, with pairwise
nonproportional rows.
By \cref{lem: realize compressed matrix C},
$\PlGH(C)\le_T\PlGH(M)$, so $\PlGH(C)$ is not $\#\mathrm P$-hard.
Applying \cref{lem:nn-bip-core}, we obtain
$\ell=r$, and $D$ is an Ising tensor up to independent permutations of rows and columns.
Thus, there exist bijections $\pi_{\ell}:[\ell]\to [\ell]$ and $\pi_r:[\ell]\to \{\ell+1,\ldots,q\}$
such that
\(
K_{ab}=C_{\pi_\ell(a)\pi_r(b)}\,
(\forall a,b\in[\ell])
\)
is an Ising tensor.

By \cref{lem:nn-moments-bipartite}, there are positive vectors
$u=(u_1,\ldots,u_{k_\ell})$ and $v=(v_1,\ldots,v_{k_r})$ such that
\[
\{s_i\mid i\in P_{a}\}
=\{u_1,\ldots,u_{k_\ell}\}
\quad \text{and} \quad
\{s_i\mid i\in P_{b}\}
=\{v_1,\ldots,v_{k_r}\}
\]
as multisets, for every $1\le a\le \ell$ and $\ell<b\le q.$
We may therefore enumerate
\begin{equation}\label{eq: prop order}
    P_{\pi_\ell(a)}=\{i_{a,1},\ldots,i_{a,k_\ell}\}
\quad \text{and} \quad
P_{\pi_r(a)}=\{j_{a,1},\ldots,j_{a,k_r}\}
\end{equation}
so that
\(
s_{i_{a,s}}=u_s,\,s_{j_{a,t}}=v_t
\)
for every $a\in[\ell]$, $s\in[k_\ell]$ and $t\in[k_r]$.

Finally, order the color set $[n]$ by listing the left proportional classes in the order $\pi_\ell(1),\ldots,\pi_\ell(\ell)$, followed by the right proportional classes in the order
$\pi_r(1),\ldots,\pi_r(\ell)$, using the chosen order \eqref{eq: prop order} within each proportional class.
Apply this same ordering to the rows and columns of $M$, and denote
the resulting matrix by $M'$.
The upper-right block $B'$ of $M'$ satisfies
\[
B'_{(a,s),(b,t)}
=M_{i_{a,s},j_{b,t}}
=C_{\pi_\ell(a)\pi_r(b)}
  s_{i_{a,s}}s_{j_{b,t}}
=K_{ab}u_sv_t
\]
for all $a,b\in[\ell]$, $s\in[k_\ell]$ and $t\in[k_r]$.
Thus $B'=K\otimes uv^{\tt T}$.
Therefore,
\[
M'=
\begin{pmatrix}
O&B'\\
B'^{\tt T}&O
\end{pmatrix}
=
\begin{pmatrix}
O&K\otimes uv^{\tt T}\\
K\otimes vu^{\tt T}&O
\end{pmatrix},
\]
as claimed.
\end{proof}

\subsection{Dichotomy for Nonnegative Matrices}
\label{subsec:nn-completion}
We are now ready to prove \cref{thm:nn-dichotomy}, the dichotomy theorem for $\PlGH(M)$ when $M$ is a real symmetric nonnegative matrix.

\paragraph{Proof of tractability:}
From the definition of the partition function $Z_M(G)$, it is clear that $Z_{M_1\otimes M_2}(G)=Z_{M_1}(G)Z_{M_2}(G)$ and $Z_{M_1\oplus M_2}(G)=Z_{M_1}(G)+Z_{M_2}(G)$ for arbitrary real symmetric matrices $M_1,M_2$ and connected planar graph $G$.
Note that $\PlGH(M)$ is tractable, when $M=\left(\begin{smallmatrix}
    a_i & b_i\\
    b_i & a_i
\end{smallmatrix}\right)$ with $a_i,b_i>0$ and $a_i\neq b_i$ \cite{fisher,kasteleyn1963dimer,temperley1961dimer,cai2017complexitybook}, or $M=ww^{\tt T}$ \cite{bulatov2005complexity}, or $M=\left(\begin{smallmatrix}
    O & uv^{\tt T}\\
    vu^{\tt T} & O
\end{smallmatrix}\right)$ \cite{bulatov2005complexity}, where $u,v,w$ are entrywise positive vectors.
A zero block is trivially tractable.
These are basic tractable blocks.
Thus, if $M$ satisfies the conditions of \cref{thm:nn-dichotomy}, then $\PlGH(M)$ is tractable because $M$ decomposes into a direct sum of tensor products of basic tractable blocks.

\paragraph{Proof of hardness:}
Suppose $\PlGH(M)$ is not $\#\mathrm P$-hard.
If $M$ is the zero matrix, then the theorem is trivial.
Now suppose $M$ is nonzero.
Let $M_0$ be an arbitrary connected component of $M$ after removing zero rows and columns.
By \cref{lem:nn-component}, $\PlGH(M_0)\le_T\PlGH(M).$
So $\PlGH(M_0)$ is not \#P-hard.
Let $C$ be the compressed matrix of $M_0.$
If $M_0$ is nonbipartite, then by \cref{lem: nonbipartite case}, $M_0$ has form \cref{item: tractable 1} up to a simultaneous permutation of rows and columns.
If $M_0$ is bipartite, then by \cref{lem: bipartite case}, $M_0$ has form \cref{item: tractable 2} up to a simultaneous permutation of rows and columns.
\qed

\section*{Acknowledgement}
The authors had the idea of defining the entrywise logarithmic matrix algebra, and uploaded into Codex many of our unpublished  non-AI generated manuscripts over the past several months.
We used ChatGPT 5.6 sol and 6 Astra to develop mathematical proofs.
The entire manuscript was written by human authors.
The authors have verified, reorganized, and simplified all proofs, and take full responsibility for their correctness and originality.
The authors would like to thank Ashwin Maran and Ben Young. 
Our work is built on their substantial previous work.
Chenghua Liu and Boning Meng informed us that
they have independently obtained the same result.

\clearpage
\phantomsection
\addcontentsline{toc}{section}{References}
\bibliographystyle{plainnat}
\bibliography{references}
\end{document}